%% file: arxiv_version.tex
\documentclass[final,3p,preprint,times]{elsarticle}
\input{preamble.tex}
\usepackage{lineno}

\journal{Journal of Computer and System Sciences (JCSS)}
\begin{document}

  \acrodef{AIG}{And-Inverter Graph}
  \acrodef{API}{Application Programming Interface}
  \acrodef{AST}{Abstract Syntax Tree}
  \acrodef{BDD}{Binary Decision Diagram}
  \acrodef{BMC}{Bounded Model Checking}
  \acrodef{BV}{Bitvector Arithmetic}
  \acrodef{CDCL}{Conflict-Driven Clause Learning}
  \acrodef{CEGIS}{Counterexample-Guided Inductive Synthesis}
  \acrodef{CNF}{Conjunctive Normal Form}
  \acrodef{DNF}{Disjunctive Normal Form}
  \acrodef{DQBF}{Dependency Quantified Boolean Formulas}
  \acrodef{EPR}{Effectively Propositional Logic}
  \acrodef{FOL}{First-Order Logic}
  \acrodef{GUI}{Graphical User Interface}
  \acrodef{GR(1)}{Generalized Reactivity of Rank~1}
  \acrodef{IDE}{Integrated Development Environment}  
  \acrodef{LIA}{Linear Integer Arithmetic}
  \acrodef{LTL}{Linear Temporal Logic}
  \acrodef{MAX-SMT}{Maximum Satisfiability Modulo Theories}
  \acrodef{MBD}{Model-Based Diagnosis}
  \acrodef{PCNF}{Prenex Conjunctive Normal Form}
  \acrodef{PDR}{Property Directed Reachability}
  \acrodef{QBF}{Quantified Boolean Formula}
  \acrodef{RHS}{Right-Hand Side (of an assignment)}
  \acrodef{S1S}{Monadic Second Order Logic of One Successor}
  \acrodef{SAT}{Satisfiability}
  \acrodef{SFD}{Single-Fault Diagnosis}
  \acrodef{SMT}{Satisfiability Modulo Theories}
  \acrodef{SSA}{Static Single Assignment}
  \acrodef{TCAS}{Traffic Collision Avoidance System}

\begin{frontmatter}

\title{Satisfiability-Based Methods for Reactive Synthesis from 
Safety Specifications\tnoteref{funding_label}}
\tnotetext[funding_label]{This work was supported in part by the Austrian Science Fund (FWF)
through the projects RiSE (S11406-N23, S11408-N23, S11409-N23) and QUAINT 
(I774-N23), and
by the European Commission through the projects STANCE (317753) and IMMORTAL 
(644905).}

\author[label1]{Roderick Bloem}
\author[label2]{Uwe Egly}
\author[label1]{Patrick Klampfl}
\author[label1]{Robert K\"onighofer\corref{cor1}}
\author[label2]{Florian Lonsing}
\author[label3]{Martina Seidl}
\cortext[cor1]{Corresponding author. 
               Email: \url{robert.koenighofer@gmail.com},
               Tel.: \texttt{+43 664 1112277},
               Fax: \texttt{+43 316 873 5520}}

\address[label1]{Institute of Applied Information Processing and Communications,
                 Graz University of Technology, 
                 Inffeldgasse 16a, 8010 Graz, Austria}
\address[label2]{Institute of Information Systems 184/3, 
                 Vienna University of Technology,
                 Favoritenstraße 9-11, 1040 Vienna, Austria}
\address[label3]{Institute for Formal Models and Verification,
                 Johannes Kepler University, 
                 Altenbergerstr. 69, 4040 Linz, Austria}

\begin{abstract}
Existing approaches to synthesize reactive systems from declarative 
specifications mostly rely on Binary Decision Diagrams (BDDs), inheriting their 
scalability issues.  We present novel algorithms for safety specifications that 
use decision procedures for propositional formulas (SAT solvers), Quantified 
Boolean Formulas (QBF solvers), or Effectively Propositional Logic (EPR).  Our 
algorithms are based on query learning, templates, reduction to EPR, QBF 
certification, and interpolation.  A parallelization combines multiple 
algorithms.  Our optimizations expand quantifiers and utilize unreachable 
states and variable independencies.  Our approach outperforms a simple 
BDD-based tool and is competitive with a highly optimized one.  It won two 
medals in the \syntcomp competition.
\end{abstract}


\end{frontmatter}

\input{01intro_arxiv.tex}

\newpage
\input{02prelim.tex}

\input{03strategy.tex}

\input{04circuit.tex}
\input{05experiments.tex}
\input{06relwork.tex}

\input{07conclusion.tex}

\section*{Acknowledgements}

We thank Aaron R. Bradley for fruitful discussions about using IC3-concepts in 
synthesis, Andreas Morgenstern for his support in 
re-implementing~\cite{MorgensternGS13} and translating benchmarks, Bettina 
K\"onighofer for providing benchmarks, and Fabian 
Tschiatschek and Mario Werner for their BDD-based synthesis tool.

%

%
%
%

%
%

%
%
%
%
\bibliographystyle{elsarticle-num} 
\bibliography{refs-num}

\end{document}

%% file: preamble.tex
\usepackage{multicol}
\usepackage{multirow}
\usepackage{breakcites}
\usepackage{subfig}
\usepackage{wrapfig}
\usepackage{paralist} %

\usepackage{xcolor}
\definecolor{darkgreen}{rgb}{0,0.2,0}
\definecolor{darkblue}{rgb}{0,0,0.2}
\definecolor{darkred}{rgb}{0.2,0,0}
\definecolor{lightgrey}{gray}{0.8}
\definecolor{darkgrey}{gray}{0.15}
\definecolor{mygreen}{rgb}{0,0.4,0}

\usepackage{algorithm}
\usepackage[noend]{algpseudocode} %
\algrenewcommand\algorithmicindent{1em}

\algdef{SE}[PROCEDURE]{ProcedureRet}{EndProcedureRet}%
[3]{\algorithmicprocedure\ 
\textproc{#1}\ifthenelse{\equal{#2}{}}{}{$\left(#2\right)$}%
, \textbf{returns}: {#3}}%
{\algorithmicend\ \algorithmicprocedure}%

\algdef{SE}[PROCEDURE]{ProcedureRetL}{EndProcedureRetL}%
[4]{\algorithmicprocedure\ 
\textproc{#1}\ifthenelse{\equal{#2}{}}{}{$\left(#2\right)$}%
, \hspace{#4} \; \textbf{returns}: {#3}}%
{\algorithmicend\ \algorithmicprocedure}%

\newcommand{\LineIf}[2]{%
    \State\algorithmicif\ {#1}\ \algorithmicthen\ {#2}}

\usepackage{amsmath}
\usepackage{amsfonts}
\usepackage{amssymb}
\usepackage{stmaryrd} %
\usepackage{booktabs}
\usepackage{temporal}
\usepackage{xspace}

\usepackage{url}

\usepackage[printonlyused]{acronym}
\usepackage{rotating}
\usepackage{hyperref}  

\usepackage{tikz}
\usetikzlibrary{shapes}
\usetikzlibrary{automata}
\usetikzlibrary{positioning}
\usetikzlibrary{arrows}
\tikzset{initial text={}}
\tikzstyle{every edge}=[->,draw]
\tikzstyle{every node}=[minimum size=5mm, inner sep=0pt, font=\scriptsize]

\newtheorem{theorem}{Theorem}
\newproof{proof}{Proof}
\newtheorem{lemma}[theorem]{Lemma}

\newdefinition{definition}[theorem]{Definition}
\newdefinition{example}[theorem]{Example}
\newenvironment{exa}{
\refstepcounter{theorem}
\textbf{Example \thetheorem.}
}{}
\newenvironment{exaq}{
\refstepcounter{theorem}
\textbf{Example \thetheorem.}
}{\qed}

\newcommand{\murl}[1]{\url{#1}}
\newcommand{\mypara}[1]{\textbf{#1}}

\newcommand{\icthree}{\textsf{IC3}\xspace}

\newcommand{\aiger}{\href{http://fmv.jku.at/aiger/}{\textsf{AIGER}}\xspace}
\newcommand{\buchi}{B\"uchi\xspace}
\newcommand{\syntcomp}{\href{http://www.syntcomp.org/}{\textsf{SyntComp}}\xspace}
\newcommand{\syntcompy}{\href{http://www.syntcomp.org/}{\textsf{SyntComp} 2014}\xspace}

\newcommand{\bloqqer}{\textsf{Bloqqer}\xspace}
\newcommand{\ratsy}{\textsf{Ratsy}\xspace}
\newcommand{\grasp}{\textsf{GRASP}\xspace}
\newcommand{\depqbf}{\textsf{DepQBF}\xspace}
\newcommand{\quantor}{\textsf{Quantor}\xspace}
\newcommand{\rareqs}{\textsf{RAReQS}\xspace}
\newcommand{\qube}{\textsf{QuBE}\xspace}
\newcommand{\nenofex}{\textsf{Nenofex}\xspace}
\newcommand{\cudd}{\textsf{CUDD}\xspace}
\newcommand{\iprover}{\textsf{iProver}\xspace}

\newcommand{\dimacs}{\textsf{DIMACS}\xspace}
\newcommand{\qdimacs}{\textsf{QDIMACS}\xspace}
\newcommand{\abc}{\textsf{ABC}\xspace}
\newcommand{\minisat}{\textsf{MiniSat}\xspace}
\newcommand{\lingeling}{\textsf{Lingeling}\xspace}
\newcommand{\picosat}{\textsf{PicoSAT}\xspace}
\newcommand{\qbfcert}{\textsf{QBFCert}\xspace}
\newcommand{\demiurge}{\href{https://www.iaik.tugraz.at/content/research/opensource/demiurge}{\textsf{Demiurge}}\xspace}
\newcommand{\lily}{\textsf{Lily}\xspace}
\newcommand{\acaciap}{\textsf{Acacia+}\xspace}
\newcommand{\unbeast}{\textsf{Unbeast}\xspace}
\newcommand{\mathsat}{\textsf{MathSAT}\xspace}
\newcommand{\abssynthe}{\textsf{AbsSynthe}\xspace}

\newcommand{\propsat}{\textsc{PropSat}}
\newcommand{\propsatmodel}{\textsc{PropSatModel}}
\newcommand{\propsatcore}{\textsc{PropUnsatCore}}
\newcommand{\propsatmincore}{\textsc{PropMinUnsatCore}}
\newcommand{\qbfsat}{\textsc{QbfSat}}
\newcommand{\qbfsatmodel}{\textsc{QbfSatModel}}

\newcommand{\interpol}{\textsc{Interpol}}
\newcommand{\EQ}{\mathsf{EQ}}
\newcommand{\SU}{\mathsf{SUB}}
\newcommand{\mrk}[1]{{\color{blue}#1}}
\newcommand{\tm}[1]{{\bf\color{blue}#1}}

\newcommand{\B}{\mathbb{B}}

\newcommand{\D}{\mathbb{D}}

\newcommand{\FE}{\mathsf{Force}^e_1}
\newcommand{\FS}{\mathsf{Force}^s_1}
\newcommand{\reach}{\mathsf{Reach}_1}
\newcommand{\uneq}{=\!\!\!\!\!\!\textcolor{red}{\lightning}\;}
\DeclareMathOperator{\ld}{ld}
\usepackage{bold-extra} %

\newcommand{\add}{{\textsf{add}}}
\newcommand{\mult}{{\textsf{mult}}}
\newcommand{\cnt}{{\textsf{cnt}}}
\newcommand{\mv}{{\textsf{mv}}}
\newcommand{\bs}{{\textsf{bs}}}
\newcommand{\stay}{{\textsf{stay}}}
\newcommand{\genbuf}{{\textsf{genbuf}}}
\newcommand{\amba}{{\textsf{amba}}}
\newcommand{\demo}{{\textsf{demo}}}
\newcommand{\fa}{{\textsf{fact}}}
\newcommand{\gb}{{\textsf{gb}}}
\newcommand{\load}{{\textsf{load}}}
\newcommand{\ltltodba}{{\textsf{ltl2dba}}}
\newcommand{\ltltodpa}{{\textsf{ltl2dpa}}}
\newcommand{\mo}{{\textsf{mov}}}
\newcommand{\driver}{{\textsf{driver}}}

\newcommand{\wbdd}{\textsf{BDD}\xspace}
\newcommand{\wifm}{\textsf{IFM}\xspace}  
\newcommand{\wabs}{\textsf{ABS}\xspace}
\newcommand{\wq}{\textsf{Q}\xspace}   
\newcommand{\wqb}{\textsf{QB}\xspace}   
\newcommand{\wqgb}{\textsf{QGB}\xspace}
\newcommand{\wqgab}{\textsf{QGAB}\xspace}   
\newcommand{\wqgcb}{\textsf{QGCB}\xspace}   
\newcommand{\wqi}{\textsf{QI}\xspace}   
\newcommand{\ws}{\textsf{S}\xspace}   
\newcommand{\wsg}{\textsf{SG}\xspace}   
\newcommand{\wsgc}{\textsf{SGC}\xspace}   
\newcommand{\wse}{\textsf{SE}\xspace}   
\newcommand{\wsge}{\textsf{SGE}\xspace}   
\newcommand{\wtqc}{\textsf{TQC}\xspace}   
\newcommand{\wtbc}{\textsf{TBC}\xspace}   
\newcommand{\wtsc}{\textsf{TSC}\xspace}   
\newcommand{\wepr}{\textsf{EPR}\xspace}  
\newcommand{\wptwo}{\textsf{P2}\xspace}  
\newcommand{\wpthree}{\textsf{P3}\xspace}  
\newcommand{\ebdd}{\textsf{BDD}\xspace}
\newcommand{\eabs}{\textsf{ABS}\xspace}
\newcommand{\ecert}{\textsf{QC}\xspace}
\newcommand{\ecertn}{\textsf{QCN}\xspace}
\newcommand{\eq}{\textsf{QL}\xspace}
\newcommand{\eqb}{\textsf{QLB}\xspace}
\newcommand{\eqi}{\textsf{QLI}\xspace}
\newcommand{\ei}{\textsf{ID}\xspace}
\newcommand{\es}{\textsf{SL}\xspace}
\newcommand{\esd}{\textsf{SLD}\xspace}
\newcommand{\esdm}{\textsf{SLDM}\xspace}
\newcommand{\eptwo}{\textsf{P2}\xspace}  
\newcommand{\epthree}{\textsf{P3}\xspace} 

\newcommand{\ts}{\textsuperscript}

%% file: 01intro_arxiv.tex
\section{Introduction}
\label{sec:intro}

A common criticism of formal verification techniques such as model 
checking~\cite{ClarkeE81,QueilleS82} is that they are only applied after the 
implementation is completed.  Synthesis~\cite{PnueliR89} is more ambitious: it 
constructs an implementation from a declarative specification automatically.  
The specification may only express \emph{what} the system shall do, but not 
\emph{how}.  Hence, writing a specification can be significantly easier than 
implementing it. Another advantage is that synthesized implementations are 
\emph{correct-by-construction}, i.e., guaranteed to satisfy the specification 
from which they have been constructed.  Assuming that the specification 
expresses the design intent correctly and completely, this eliminates the need 
for verification and debugging of the implementation.  
This effort reduction is illustrated in Figure~\ref{fig:sy}.

\mypara{Applications of synthesis.}  Synthesis is particularly well suited for 
\emph{rapid prototyping}, where a working implementation needs to be available 
quickly.  A synthesized prototype can later be exchanged by a (manual) 
implementation that is more optimized.  Another interesting application is 
\emph{program sketching}~\cite{Solar-Lezama13, lezamaphd}, where the programmer 
can leave ``holes'' in the code.  A synthesizing compiler then fills the holes 
such that a given specification is satisfied.  This mix of imperative and 
declarative programming is appealing because some aspects of the program may be 
easy to implement, while others may be easier to specify.  In \emph{controller 
synthesis}, a plant needs to be controlled such that some specification is 
satisfied.  Synthesizing such a controller is similar to program sketching in 
that a given part (the plant) is combined with a synthesized part (the 
controller).  Another related application is automatic program 
repair~\cite{JobstmannSGB12,KonighoferB11}, where potentially faulty program 
parts (identified by some error localization algorithm) are replaced by 
synthesized corrections.  In all these applications, automatic synthesis  
contributes to keeping the manual development effort low.

\mypara{Systems.}  
This article is concerned with synthesis algorithms for \emph{reactive 
systems}~\cite{HarelP89}, which interact with their environment in a synchronous 
way: in every time step, the environment provides input values and the system 
responds with output values.  This is repeated ad infinitum, i.e., reactive 
systems conceptually never terminate.  Thus, reactive systems can directly model 
(synchronous) hardware designs, but also other non-terminating systems such as 
an operating system, a server implementing some protocol, etc.  In contrast, 
\emph{transformational systems} terminate after processing their input.  
They are thus suited to model procedures of a 
software program, e.g., a sorting algorithm.

\mypara{Specifications.}  
We focus on synthesis of reactive systems from \emph{safety specifications}, 
which express that certain ``bad things'' never happen.  This stands in 
contrast to liveness properties, which stipulate that certain ``good things'' 
must happen eventually.  Synthesis algorithms for safety specifications can be 
useful even for specifications that contain liveness properties.  First, bounded 
synthesis approaches~\cite{Ehlers12,FiliotJR11} can reduce synthesis from 
richer specifications, such as Linear Temporal Logic (LTL)~\cite{Pnueli77}, to 
safety synthesis problems by setting a bound on the reaction time.
For instance, instead of requiring that some event happens eventually, one may 
require that it happens within at most $k$ steps.  Clearly, a realization of the 
latter is also a realization of the former.  By choosing $k$ as low as possible 
(such that a solution still exists), we may even get systems that react 
faster.
A second reason why safety specifications are important is that safety 
properties often make up the bulk of a specification and they 
can be handled in a compositional manner: the safety synthesis problem can be 
solved before the other properties are handled~\cite{SohailS13}. 

\begin{figure}
\centering
\subfloat%
  {\includegraphics[width=0.41\textwidth]{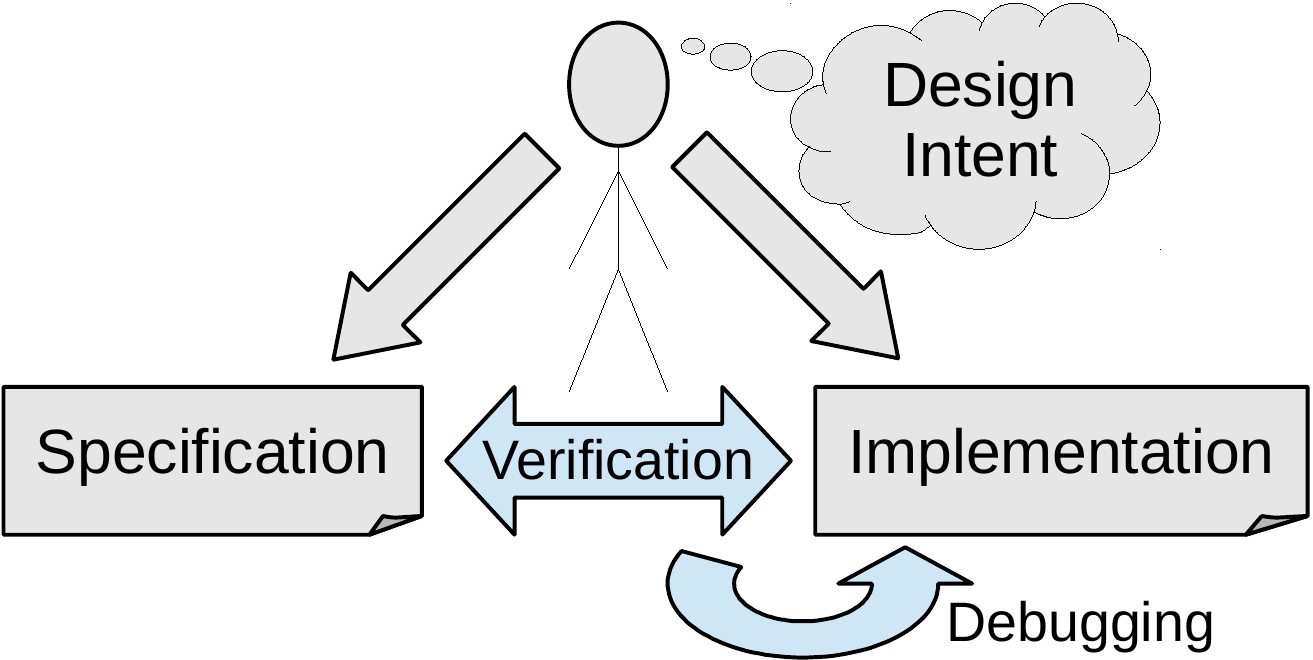}}
\hspace{11mm}
\subfloat%
  {\includegraphics[width=0.41\textwidth]{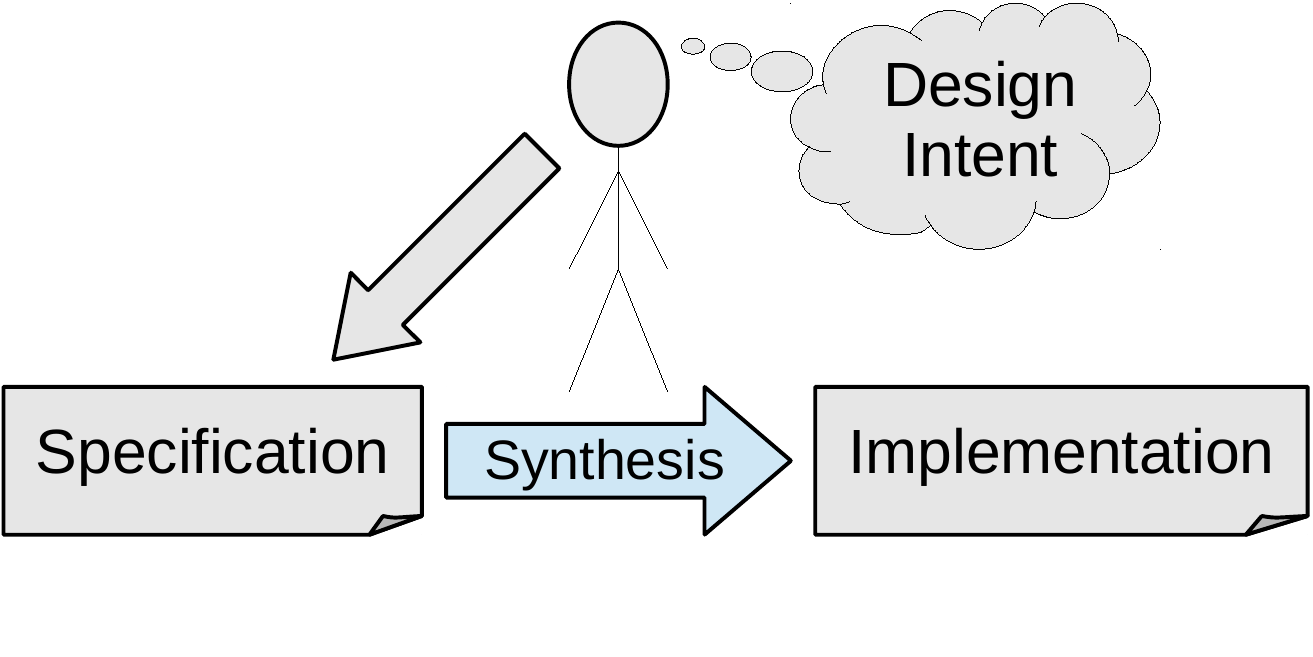}}
\caption{Reduction of the development effort due to synthesis}
\label{fig:sy}
\end{figure}

\mypara{Synthesis is a game.}  
Model checking can be understood as (exhaustive) search for inputs under which a 
(model of the) system violates its specification.  That is, the inputs are the 
only source of  non-determinism. Synthesis, on the other hand, needs to handle 
two sources of non-determinism: the unknown inputs and the (yet) unknown system 
implementation. Synthesis can thus be seen as a game between two players: The 
environment player controls the inputs of 
\begin{wrapfigure}[5]{r}{0.44\textwidth}
\centering
\vspace{-3mm}
\includegraphics[width=0.43\textwidth]{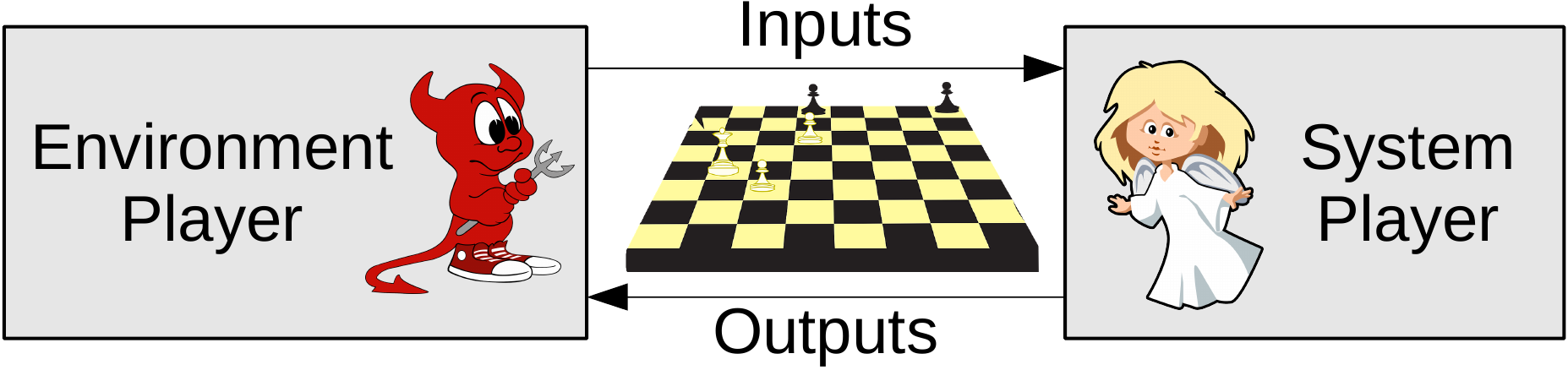}
\label{fig:game}
\end{wrapfigure}
the system to be synthesized.  The 
system player controls the outputs and attempts to satisfy the specification for 
\emph{every} environment behavior. 
The environment player has the role of 
the antagonist, trying to violate the specification.  
The game-based approach to synthesis computes a \emph{strategy} for the system 
player to win the game (i.e., to satisfy the specification) against every 
environment player.  An implementation of such a winning strategy forms the 
solution.  Computing a winning strategy involves dealing with 
alternating quantifiers because for every input (or environment behavior) there 
must exist some output (or system behavior) satisfying the specification. This 
stands in contrast to model checking, where existential quantification suffices.

\mypara{Scalability.}
Synthesis is computationally hard.  For safety specifications, the worst-case 
time complexity is exponential~\cite{BrenguierPRS14,PapadimitriouY86} in the 
size of the specification. For \acs{LTL}, it is even doubly 
exponential~\cite{Rosner92}.   Measures to improve the performance in practice 
include limiting the expressiveness of the 
specification~\cite{BloemJPPS12, AlurT04}, 
limiting the size of systems to construct~\cite{FinkbeinerS13}, and applying 
symbolic algorithms~\cite{BurchCMDH90}, which use formulas as a compact 
representation of state sets instead of enumerating states explicitly.  These 
formulas can in turn be represented using Binary Decision Diagrams 
(BDDs)~\cite{Bryant86}, a graph-based representation for propositional formulas. 
However, for certain structures, BDDs are known to explode in size and thus 
scale insufficiently~\cite{Bryant86}. This is one reason why BDDs have largely 
been displaced by SAT solvers in model checking.  Yet, in reactive synthesis, 
BDDs are still the predominant symbolic reasoning engine.  This is witnessed by 
the fact that all submissions to the reactive synthesis competition \syntcomp in 
2014~\cite{sttt_syntcomp} and 2015~\cite{JacobsBBKPRRSST16}, except for our own, 
were BDD-based.  One reason is that synthesis inherently deals with alternating 
quantifiers (see above).  BDDs provide universal and existential quantifier 
elimination to deal with that.

\subsection*{Contributions and Outline}  

To offer additional alternatives to BDDs in reactive synthesis, we present novel 
synthesis algorithms for safety specifications using decision procedures for the 
satisfiability of propositional formulas (SAT solvers), Quantified Boolean 
Formulas (QBF solvers), or Effectively Propositional Logic (EPR), which is a 
subset of first-order logic.  Our algorithms exploit solver features such as 
incremental solving and unsatisfiable cores by design.  Similar to existing 
solutions, our approach consists of two steps: computing a strategy and building 
a circuit that implements this strategy.

\mypara{Preliminaries.}
Before we present our algorithms, Chapter~\ref{sec:prelim} introduces background 
and notation.  It starts by defining logics and decision procedures.  Readers 
who are familiar with SAT and QBF can focus on Section~\ref{sec:sat_not} 
and \ref{sec:qbf_not} to understand our notation.  In Section~\ref{sec:prel:hw}, 
we define the addressed synthesis problem and give a textbook solution.  
Synthesis experts can focus on Definition~\ref{def:safety}.  Finally, we
introduce query learning~\cite{Angluin87} and 
\acf{CEGIS}~\cite{Solar-Lezama13} as algorithmic principles underlying many of 
our algorithms.

\mypara{Strategy computation.}
Chapter~\ref{sec:hw:win} presents our algorithms and optimizations for computing 
a strategy to satisfy the specification.  Section~\ref{sec:qbf_learn} starts 
with a learning algorithm that uses a QBF solver.  In 
Section~\ref{sec:sat_learn}, we modify this algorithm to use a plain SAT solver 
while exploiting incremental solving and unsatisfiable cores.  This 
turns out to be significantly faster.  Both these sections contain correctness 
proofs and discuss possible variations and an efficient implementation.  In 
Section~\ref{sec:hw_exp}, we reduce the number of iterations (and thereby the 
execution time) of the SAT solver based solution by partially expanding 
quantifiers.  Section~\ref{sec:hw_reach} continues with optimizations that 
exploit unreachable states based on concepts from the model checking algorithm 
\icthree~\cite{Bradley11}.  Both optimizations give a speedup of more than 
one order of magnitude each.  In Section~\ref{sec:hw:templ}, we describe a 
completely different approach, which fixes the structure of the solution using a 
template.  We compute solutions either with a single call to a QBF solver or by 
calling a SAT solver repeatedly using (an extension of) 
\acs{CEGIS}~\cite{Solar-Lezama13}.  
Section~\ref{sec:red_to_epr} is similar in spirit but avoids the template by 
formulating the problem in \acs{EPR}.  Since different algorithms perform well 
in different cases, we finally present a parallelization that combines various 
methods and configurations in multiple threads while exchanging fine-grained 
information.

\mypara{Circuit computation.}
Chapter~\ref{sec:hw:circ} is devoted to computing an implementation in the form 
of a circuit from a given strategy.  The goal is to obtain \emph{small} circuits 
\emph{efficiently}.  To this end, implementation freedom available in the 
strategy needs to be exploited wisely.  We present a number of 
satisfiability-based methods that not only work for safety specifications but 
also for strategies to satisfy other objectives.  For each method, we thus 
present the general solution as well as an efficient realization for the special 
case of safety synthesis problems.  We start with an approach based on QBF 
certification~\cite{NiemetzPLSB12} in Section~\ref{sec:hw:circ_qbfcert}.  In 
Section~\ref{sec:hw:extr:qbflearn}, we use a QBF solver in a learning algorithm. 
This performs better, especially when using incremental QBF solving.  
Section~\ref{sec:hw:int} adopts the interpolation-based approach by Jiang et 
al.~\cite{JiangLH09} and extends it with an optimization to exploit variable 
(in)dependencies.  In Section~\ref{sec:hw:extr:satlearn}, we combine the 
approach by Jiang et al.~\cite{JiangLH09} with query learning as a special 
interpolation procedure.  This improves the speed and the resulting circuit size 
by around two orders of magnitude.  Finally, we present a parallelization that 
combines multiple methods in different threads with the aim to inherit their 
strengths and to compensate their weaknesses. 

\mypara{Tool.}  We implemented our methods in an open-source tool 
named \demiurge.  It supports the input format of 
the reactive synthesis competition \syntcomp~\cite{sttt_syntcomp} and won two 
medals in this competition.   \demiurge is extendable and highly configurable 
regarding solvers, methods and optimizations to use.  We describe \demiurge
in~Section~\ref{sec:dem_impl}.

\mypara{Experiments.}  In Chapter~\ref{sec:hw:exp}, we evaluate our approach on 
the \syntcomp benchmarks.  We compare our different methods and evaluate the 
effect of optimizations.  We also investigate the performance of different 
methods on different classes of benchmarks.  Our parallelization turns out to be 
faster than a BDD-based tool by one order of magnitude and produces circuits 
that are smaller by two orders of magnitude.  Our tool is even competitive with 
\abssynthe~\cite{BrenguierPRS14}, a BDD-based tool that implements advanced 
concepts such as abstraction/refinement.  

\mypara{Conclusion.}
Since our approach is particularly superior for certain benchmark classes, we 
conclude that it forms a valuable complement to existing approaches.  Moreover, 
decision procedures for satisfiability are an active field of 
research, and enormous scalability improvements are witnessed by various 
competitions over the years.  Since our algorithms use such decision procedures 
as a black box, they directly benefit from future developments in this field.

\mypara{Relation to previous work.}  This is the manuscript of an article that
has been submitted to the Journal of Computer and System Sciences (JCSS).  It is 
based on earlier work by the authors~\cite{BloemKS14, BloemEKKL14}, which has 
been extended with additional optimizations and variations of algorithms, as 
well as a more elaborate experimental evaluation.  This entire work forms the 
basis of a dissertation~\cite{phd}.

%% file: 02prelim.tex
\section{Preliminaries and Notation} \label{sec:prelim}

We will use upper case letters for sets, lower case letters for set 
elements, and calligraphic fonts for tuples defining more complex structures.
We denote the Boolean domain by $\B=\{\true, \false\}$ and write \emph{iff} 
for ``if and only if''.

\subsection{Logics} \label{sec:prelim:lo}

We will use various kinds of logics to solve synthesis problems.  This section 
introduces these logics.  Decision procedures and reasoning engines for these 
logics will then be introduced in Section~\ref{sec:prelim:dp}.

\mypara{Variables and formulas.}
We will use lower case letters for variables and capital letters to denote 
formulas.  Recall that capital letters are also used to denote sets, but this 
is 
no coincidence since we will later use formulas to represent sets (see 
Section~\ref{sec:prelim:syenc}).  Vectors of variables will be written with an 
overline.  For clarity, we will often write the variables that occur freely in 
a 
formula in brackets.  For instance, $F(\overline{x})$ denotes a formula over 
the 
variables $\overline{x} = (x_1, x_2, \ldots, x_n)$.  If the variables are clear 
from the context, we will sometimes omit the brackets, i.e., write only $F$ 
instead of $F(\overline{x})$.  Furthermore, we will use the brackets to denote 
variable substitutions: if $F(\ldots,x,\ldots)$ is a formula, we denote by 
$F(\ldots,y,\ldots)$ the same formula but with all occurrences of $x$ replaced 
by $y$.  With a slight abuse of notation, we will also treat vectors of 
variables like sets if the order of the elements is irrelevant.  For instance, 
$\overline{x} \cup \overline{y}$ denotes a concatenation of two variable 
vectors, and $\overline{x}\setminus \{x_i\}$ denotes the variable 
vector~$\overline{x}$ but with element $x_i$ removed.

\mypara{Operator precedence.}  Save for cases where too many brackets hamper 
readability, we will avoid ambiguities in operator precedence.  However, for 
the 
avoidance of doubt, will will use the following precedence order (from stronger 
to weaker binding) for operators in formulas: $\neg, \wedge, \vee, \rightarrow, 
\leftrightarrow, \forall, \exists$.

\subsubsection{Propositional Logic}  \label{sec:prelim:prop}

All variables in propositional logic are Boolean, i.e., take values from 
$\B=\{\true, \false\}$.  We will use the Boolean connectives 
$\neg$, $\wedge$, $\vee$, $\rightarrow$, $\leftrightarrow$, encoding 
negation, conjunction, disjunction, implication, and equivalence, respectively.

\mypara{\acfp{CNF}.}
A \emph{literal}  is a Boolean variable or its negation.  A 
\emph{clause} is a disjunction of literals. A \emph{cube} 
is a conjunction of literals.  We will sometimes treat clauses and 
cubes as sets of literals.  For instance, given that $l$ is a literal and $c_1, 
c_2$ are clauses, we write $l \in c_1$ to denote that $l$ occurs as a disjunct 
in clause $c_1$, and we write $c_1 \subseteq c_2$ to denote that all literals 
of 
clause $c_1$ also occur in clause $c_2$. A propositional formula is in 
\emph{\acf{CNF}} if it is written as a conjunction  of clauses.  There are two 
reasons why \ac{CNF} representations are important.  First, decision procedures 
for satisfiability usually require the input formula to be in \ac{CNF}.  
Second, 
every formula can be transformed into an equisatisfiable formula in \ac{CNF} 
by introducing at most a linear amount of auxiliary variables.  
This is called \emph{Tseitin 
transformation}~\cite{Tseitin83}. An improvement 
by exploiting the polarity (even or odd number of negations) of subformulas to 
obtain smaller \acs{CNF} encodings has been proposed by Plaisted and 
Greenbaum~\cite{PlaistedG86}.

\mypara{Variable assignments.}
We use cubes to describe (potentially partial) truth assignments to 
variables: unnegated variables of the cube are set to $\true$, negated ones are 
$\false$.  We use bold letters to denote cubes.  For instance, 
$\mathbf{x}$ 
denotes a cube over the variables $\overline{x}$.  An 
\emph{$\overline{x}$-minterm} is a cube that contains all 
variables of $\overline{x}$ either negated or unnegated (but not both).  Thus, 
minterms describe \emph{complete} assignments to Boolean variables.  We write 
$\mathbf{x} \models F(\overline{x})$ to denote that the $\overline{x}$-minterm 
$\mathbf{x}$ satisfies the formula $F(\overline{x})$.  Given a formula 
$F(\ldots,\overline{x},\ldots)$ and an $\overline{x}$-minterm $\mathbf{x}$, we 
write $F(\ldots,\mathbf{x},\ldots)$ to denote the formula $F$ but with all 
occurrences of the variables  $\overline{x}$ replaced by their respective truth 
value defined by~$\mathbf{x}$.

\mypara{Unsatisfiable cores.}  Let $F$ be an 
unsatisfiable formula in \acs{CNF}.  A \emph{clause-level unsatisfiable core} 
is 
a subset of the clauses of $F$ that is still unsatisfiable.  While this 
definition is widely used, many applications require the minimization of 
``interesting'' constraints while the remaining constraints remain fixed.  For 
such problems, Nadel~\cite{Nadel10} coined the term \emph{high-level 
unsatisfiable core}.  To support such high-level unsatisfiable cores, we use 
the 
following definition. Let $\mathbf{x}$ be a cube and let 
$F(\overline{x},\overline{y})$ be a formula such that $\mathbf{x} \wedge F$ is 
unsatisfiable.  An \emph{unsatisfiable core} of $\mathbf{x}$ with respect to 
$F$ 
is a subset $\mathbf{x}'\subseteq \mathbf{x}$ of the literals in $\mathbf{x}$ 
such that $\mathbf{x}' \wedge F$ is still unsatisfiable.  An unsatisfiable core 
$\mathbf{x}'$ is \emph{minimal} if no proper subset $\mathbf{x}''$ of 
$\mathbf{x}'$ 
makes $\mathbf{x}'' \wedge F$ unsatisfiable.  With this definition, high-level 
unsatisfiable cores can be computed by adding conjuncts of the form $x_i 
\rightarrow G(\overline{y})$ for $x_i\in \overline{x}$ to 
$F(\overline{x},\overline{y})$.  This way, the constraint $G(\overline{y})$ can 
be enabled or disabled via the truth value of $x_i$.  Moreover, this notion of 
unsatisfiable cores is directly supported by many solver.

\mypara{Interpolants.}  
Let $A(\overline{x},\overline{y})$ and $B(\overline{x},\overline{z})$ be two 
propositional formulas such that $A\wedge B$ is unsatisfiable, and 
$\overline{y}$ and $\overline{z}$ are disjoint.  A \emph{Craig 
interpolant}~\cite{Craig57a} is a formula $I(\overline{x})$ 
such that $A \rightarrow I \rightarrow \neg B$.  Intuitively, the interpolant 
is 
a formula that is weaker than $A$, but still strong enough to make $I \wedge B$ 
unsatisfiable.  In addition to that, the interpolant references only the 
variables $\overline{x}$ that occur both in $A$ and in $B$.

\mypara{Cofactors.} 
Let $F(\ldots,x,\ldots)$ be a propositional formula.  The \emph{positive 
cofactor} of $F$ regarding $x$ is the formula $F(\ldots,\true,\ldots)$, where 
all occurrences of $x$ have been replaced by $\true$.  Analogously, the 
\emph{negative cofactor} of $F$ regarding $x$ is the formula 
$F(\ldots,\false,\ldots)$.

\subsubsection{Quantified Boolean Formulas} \label{sec:prelim:qbf}

\acfp{QBF}~\cite{BuningB09}  
extend propositional logic with universal (denoted $\forall$) and existential 
(denoted $\exists$) quantification of variables.  The quantifiers have their 
expected semantics: Since propositional variables can only be either $\true$ or 
$\false$, $\exists x_i \scope F(\ldots, x_i, \ldots)$ can be seen as a 
shorthand 
for $F(\ldots, \true, \ldots) \vee F(\ldots, \false, \ldots)$.  Likewise, 
$\forall x_i \scope F(\ldots, x_i, \ldots)$ is short for $F(\ldots, \true, 
\ldots) \wedge F(\ldots, \false, \ldots)$.  Using these rules, a \ac{QBF} can 
always be transformed into a purely propositional formula.  However, this 
usually causes a significant blow-up in formula size.

\mypara{\acsp{PCNF}.} 
A \ac{QBF} is in \emph{\acf{PCNF}} if it is written in the form 
\[Q_1\overline{x}_1 \scope Q_2\overline{x}_2 \scope \ldots Q_k\overline{x}_k 
\scope
  F(\overline{x}_1,\overline{x}_2, \ldots,\overline{x}_k),\]
where $Q_i \in \{\forall, \exists\}$ and $F$ is a propositional formula in 
\ac{CNF}.  In this formulation, we use $Q_i\overline{x}_i$ as a shorthand for 
$Q_i x_{i,1} \scope \ldots Q_i x_{i,n}$ 
with $\overline{x}_i = (x_{i,1}, \ldots, x_{i,n})$.  We refer to 
$Q_1\overline{x}_1 \scope Q_2\overline{x}_2 \scope \ldots Q_k\overline{x}_k$ as 
the \emph{quantifier prefix} 
 and call 
$F(\overline{x}_1,\overline{x}_2, \ldots,\overline{x}_k)$ the 
\emph{matrix}  of the \ac{PCNF}. We require 
every 
\ac{PCNF} to be \emph{closed}  in the sense that all 
variables occurring in the matrix must be quantified either existentially or 
universally. Hence, a \ac{QBF} in \ac{PCNF} can only be valid (equivalent to 
$\true$) or unsatisfiable (equivalent to $\false$).

\mypara{Skolem functions.}  Let 
$
 \exists \overline{a}_1 \scope 
 \forall \overline{b}_1 \scope
 \ldots
 \exists \overline{a}_k \scope 
 \forall \overline{b}_k \scope
 \exists \overline{c} \scope
 Q_1 \overline{d}_1 \scope
 \ldots
 Q_l \overline{d}_l \scope
 F(\overline{a}_1,\overline{b}_1,
   \ldots,
   \overline{a}_k,\overline{b}_k,
   \overline{c},
   \overline{d}_1,
   \ldots,
   \overline{d}_l
   )
$
with $Q_i \in \{\forall, \exists\}$ be a \ac{QBF} in \ac{PCNF} that is valid.  
A \emph{Skolem function} for the existentially quantified variables 
$\overline{c}$ is a function $f: 2^{|\overline{b}_1|} \times \ldots \times
2^{|\overline{b}_k|} \rightarrow 2^{|\overline{c}|}$ that defines 
the values of the variables $\overline{c}$ based on the universally quantified 
variables $\overline{b}_1,\ldots,\overline{b}_k$ occurring before 
$\overline{c}$ in the quantifier prefix such that 
\[
 \exists \overline{a}_1 \scope 
 \forall \overline{b}_1 \scope
 \ldots
 \exists \overline{a}_k \scope 
 \forall \overline{b}_k \scope
 Q_1 \overline{d}_1 \scope
 \ldots
 Q_l \overline{d}_l \scope
 F\bigl(\overline{a}_1,\overline{b}_1, 
   \ldots,
   \overline{a}_k,\overline{b}_k,
   f(\overline{b}_1,\ldots,\overline{b}_k),
   \overline{d}_1,
   \ldots,
   \overline{d}_l
   \bigr)
\]
is still valid.  The function $f$ can be seen as a \emph{certificate} 
 to show that values for the variables 
$\overline{c}$ making the \acs{QBF} $\true$ exist (for any values of the 
variables $\overline{b}_1,\ldots,\overline{b}_k$).  Note that $f$ cannot depend 
on the variables $\overline{d}_1, \ldots, \overline{d}_l$ occurring after 
$\overline{c}$ in the quantifier prefix, independent of whether some 
$\overline{d}_i$ is quantified universally or existentially.

\mypara{Herbrand functions.}  
A Herbrand function is the dual of a Skolem 
function for a \ac{QBF} that is unsatisfiable. Let 
$
 \exists \overline{a}_1 \scope 
 \forall \overline{b}_1 \scope
 \ldots
 \exists \overline{a}_k \scope 
 \forall \overline{b}_k \scope
 \forall \overline{c} \scope
 Q_1 \overline{d}_1 \scope
 \ldots
 Q_l \overline{d}_l \scope
 F(\overline{a}_1,\overline{b}_1,
   \ldots,
   \overline{a}_k,\overline{b}_k,
   \overline{c},
   \overline{d}_1,
   \ldots,
   \overline{d}_l
   )
$
be an unsatisfiable \ac{QBF}. A \emph{Herbrand function} for the 
universally quantified variables 
$\overline{c}$ is a function $f: 2^{|\overline{a}_1|} \times \ldots \times
2^{|\overline{a}_k|} \rightarrow 2^{|\overline{c}|}$ that defines 
the values of the variables $\overline{c}$ based on the existentially 
quantified variables $\overline{a}_1,\ldots,\overline{a}_k$ occurring before 
$\overline{c}$ in the quantifier prefix such that 
$
 \exists \overline{a}_1 \scope 
 \forall \overline{b}_1 \scope
 \ldots
 \exists \overline{a}_k \scope 
 \forall \overline{b}_k \scope
 Q_1 \overline{d}_1 \scope
 \ldots
 Q_l \overline{d}_l \scope
 F\bigl(\overline{a}_1,\overline{b}_1, 
   \ldots,
   \overline{a}_k,\overline{b}_k,
   f(\overline{b}_1,\ldots,\overline{b}_k),
   \overline{d}_1,
   \ldots,
   \overline{d}_l
   \bigr)
$
is still unsatisfiable.

\mypara{Universal expansion.} 
Let $G=Q_1\overline{x}_1 \scope \ldots Q_k\overline{x}_k \scope \forall y 
\scope 
\exists \overline{z} \scope F(\overline{x}_1,\ldots,\overline{x}_k, y, 
\overline{z})$ be a 
\ac{QBF} in \ac{PCNF}. The \emph{universal expansion}~\cite{BubeckB07} of 
variable $y$ in $G$ is the formula 
$G'= 
Q_1\overline{x}_1 \scope \ldots
Q_k\overline{x}_k \scope \exists \overline{z}, \overline{z}' \scope 
F(\overline{x}_1, \ldots, \overline{x}_k, \true, \overline{z}) \wedge 
F(\overline{x}_1, \ldots, \overline{x}_k, \false, \overline{z}'),$ 
where $\overline{z}'$ is a fresh copy of the variables 
$\overline{z}$. This transformation is equivalence preserving~\cite{BubeckB07}. 
In our formulation, the universally quantified variable $y$ to expand must only 
be followed by existential quantifications in the prefix.  The variables 
$\overline{z}$ may depend on $y$ in $G$, i.e., may take different values for 
different truth values of $y$.  Hence, they need to be renamed in one copy of 
the matrix when turning the universal quantification into a conjunction.  Note 
that $G'$ is in \ac{PCNF} again because the conjunction of two \acp{CNF} is 
again a \ac{CNF}.

\mypara{One-point rule.}  
Let $\mathbf{x}$ be an $\overline{x}$-minterm.  We 
have that
\begin{equation}
\Bigl(\forall \overline{x} \scope \mathbf{x} \rightarrow F(\overline{x}, 
\overline{y})\Bigr)
\leftrightarrow
\Bigl(F(\mathbf{x}, \overline{y})\Bigr)
\leftrightarrow
\Bigl(
\exists \overline{x} \scope \mathbf{x} \wedge F(\overline{x}, \overline{y})
\Bigr)
\label{eq:op0}
\end{equation}
holds true because, in all three formulations, $F$ has to hold for a given
$\overline{y}$-assignment if and only if the variables $\overline{x}$ 
have the specific truth values defined by $\mathbf{x}$.  A slightly more 
complicated 
instance of this rule can be formulated as follows.  Let $T(\overline{z}, 
\overline{x})$ be a formula that defines the variables $\overline{x}$ uniquely 
based on the values of some other variables $\overline{z}$.  Formally, we 
assume that $\forall \overline{z} \scope \exists \overline{x} \scope 
T(\overline{z}, \overline{x})$ and $\forall \overline{z},\overline{x}_1, 
\overline{x}_2 \scope \bigl(T(\overline{z}, \overline{x}_1) \wedge 
T(\overline{z}, \overline{x}_2) \bigr) \rightarrow (\overline{x}_1 = 
\overline{x}_2)$.  We have that
\begin{equation}
\Bigl(\forall \overline{x} \scope 
T(\overline{z}, \overline{x}) \rightarrow 
F(\overline{x}, \overline{y})\Bigr)
\leftrightarrow
\Bigl(
\exists \overline{x} \scope 
T(\overline{z}, \overline{x}) \wedge 
F(\overline{x}, \overline{y})
\Bigr)
\label{eq:op1}
\end{equation}
holds true because for a given $\overline{z}$-assignment $\mathbf{z}$ and a 
given 
$\overline{y}$-assignment $\mathbf{y}$, $F$ needs to hold only for the 
$\overline{x}$-assignment $\mathbf{x}$ that is uniquely defined by $T$ in both 
formulations.  We will use these dualities in various proofs and transformations.

\subsubsection{First-Order Logic} \label{sec:prel:fol}

\acf{FOL}~\cite{0017977} \index{First order logic@\acf{FOL}} is a more 
expressive logic, which enables reasoning about elements from arbitrary 
domains. 
Let $\D$ be a (potentially infinite) domain and let $\overline{x} = (x_1, x_2, 
\ldots, x_k)$ be variables ranging over this domain. Furthermore, let 
$\overline{y} = (y_1, y_2, \ldots, y_l)$ be Boolean variables ranging over 
$\B$, 
let $f_1, f_2, \ldots, f_m$ be function symbols and let $p_1, p_2, \ldots, p_n$ 
be predicate symbols.  Each function symbol and each predicate symbol has a 
certain \emph{arity}, i.e., number of arguments to which it can be applied.  A 
\emph{term} \index{term (of a \acs{FOL} formula)} in first-order logic is 
either a domain variable $x_i$ (with $1\leq i\leq k$) or a function application 
$f_i(t_1, \ldots, t_a)$, where $f_i$ is a function symbol with arity $a$, 
and all $t_i$ (with $1\leq i\leq a$) are terms.  Intuitively, a term evaluates 
to an element of $\D$.  An \emph{atom} \index{atom (of a \acs{FOL} 
formula)} is either a propositional variable $y_i$ (with $1\leq i\leq l$) or a 
predicate application $p_i(t_1, \ldots, t_a)$ where $p_i$ is a predicate symbol 
with arity $a$, and all $t_i$ (with $1\leq i\leq a$) are terms. Thus, 
intuitively, an atom evaluates to a truth value from $\B$.  Finally, a 
\emph{\acf{FOL}} formula is one of 
\[
a, \quad
\neg F_1, \quad
F_1\vee F_2, \quad
F_1\wedge F_2,  \quad
F_1\rightarrow F_2, \quad
F_1\leftrightarrow F_2, \quad
\exists x_i \scope F_1, \text{ or } \quad
\forall x_i \scope F_1,
\]
where $F_1$ and $F_2$ are \acl{FOL} formulas themselves and $a$ is an atom. The 
semantics of the Boolean connectives and the quantifiers are as expected. A 
\emph{model} of a \ac{FOL} formula \index{model (of a \acs{FOL} formula)} is a 
structure that satisfies the formula. It consists of concrete values for all 
variables that are not explicitly quantified, as well as concrete realizations 
of all functions $f_i$ and predicates $p_i$.  Similar to propositional logic, 
we 
refer to an atom or the negation of an atom as a \index{first-order 
literal} \emph{first-order literal}. A \emph{first-order clause} 
\index{first-order 
clause} is a disjunction of first-order literals.  A \emph{first-order 
\ac{CNF}} 
is a conjunction of first-order clauses.  A \ac{FOL} formula is 
\emph{quantifier-free} if it contains no occurrences of $\exists$ and 
$\forall$. 

\subsubsection{Effectively Propositional Logic} \label{sec:prelim:epr}

\emph{\acf{EPR}}~\cite{Lewis80}, \index{Effectively Propositional 
Logic@\acf{EPR}} also known as Bernays-Sch\"onfinkel class, is a subset of 
first-order logic that contains formulas of the form $\exists \overline{x}\scope 
\forall \overline{y} \scope F$, where $\overline{x}$ and $\overline{y}$ are 
disjoint vectors of variables ranging over domain $\D$, and $F$ is a 
function-free first-order \ac{CNF}. The formula $F$ can contain predicates over 
$\overline{x}$ and $\overline{y}$, though.

\subsection{Decision Procedures and Reasoning Engines} \label{sec:prelim:dp}

In the following, we will discuss decision procedures and reasoning engines for 
the logics introduced in the previous section from a user's perspective.  

\subsubsection{\aclp{BDD}} \label{sec:prelim:bdds}

{
\setlength{\fboxsep}{0.5mm}
\acfp{BDD}~\cite{Bryant86}  are a graph-based representation 
for formulas in propositional logic.  The graphs are rooted and acyclic.  There 
are two terminal nodes, which we denote by \fbox{0} and \fbox{1}.  Non-terminal 
nodes are labeled by a variable, have exactly two outgoing edges, and act as 
decisions: when traversing the graph from the root node, depending on the truth 
value of the variable labelling a node, one of the outgoing edges is taken.  If 
the terminal node \fbox{0} is reached during such a traversal, then this means 
that the formula evaluates to $\false$ for this assignment.  If \fbox{1} is 
reached, the formula evaluates to $\true$.

\begin{wrapfigure}[8]{r}{0.2\textwidth}
\centering
\vspace{-0.8cm}
\scalebox{1}{\input{fig/ex_bdd.tex}}
\end{wrapfigure}
\begin{exa} \label{ex:bdd}
A \ac{BDD} for the formula $f = (x\vee y) \wedge \neg z$ is shown on the right. 
The root node, representing $f$, is 
marked with an incoming arrow.  Non-terminal nodes are drawn as circles. The 
solid outgoing edge is taken if the variable written in the node is $\true$, the 
dashed edge is taken if the variable written in the node is $\false$.  The two 
terminal nodes are drawn as boxes.  The graph can be read as follows:  If 
$z=\true$, the entire formula $f$ is $\false$. Otherwise, $x$ is considered.  If 
$x$ is $\true$ (and $z=\false$), the formula is $\true$.  Otherwise, $y$ is 
considered. If $y=\true$ (and $z=x=\false$), then $f$ is $\true$.  If $y=\false$ 
(and $z=x=\false$), then $f$ is $\false$.\qed

\end{exa}

\mypara{Orderdness and Reducedness.} \acp{BDD} are \emph{ordered} 
 in the sense that for all paths from the root to the terminal 
nodes, decisions on the variables are always taken in the same order.  We will 
refer to this order as the \emph{variable order} of the \ac{BDD}. 
For instance, the variable order in 
Example~\ref{ex:bdd} is $z,x,y$.  Furthermore, \acp{BDD} are \emph{reduced} 
 in the sense that redundant vertices (where the 
$\true$- and the $\false$-successor are the same node) and isomorphic subgraphs 
have been eliminated.  This reduction serves two purposes.  First, it reduces 
the size of the \acp{BDD}.  Second, for a fixed variable order, it makes 
\acp{BDD} a \emph{canonical} representation of a propositional formula.  

\mypara{Canonicity.} 
A \ac{BDD} is a \emph{canonical} representation of a propositional formula in 
the sense that for a fixed variable order, the same formula will always be 
represented by isomorphic graphs.  This property makes equivalence checks 
between propositional formulas simple: once the \acp{BDD} have been built, all 
that needs to be done is to compare the graphs.  In particular, a 
satisfiability 
check can be performed by comparing the \ac{BDD} with that for $\false$ (which 
has the terminal node \fbox{0} as its root).  \ac{BDD} libraries are usually 
implemented in such a way that multiple formulas are represented by a single 
graph with several root nodes~\cite{MinatoIY90}.  If two formulas are 
equivalent, they are represented by the same graph node.  This saves 
memory (because common subgraphs are stored only once) and allows for 
equivalence checks between formulas in constant time: all that needs to be done 
is to check if the root nodes are identical.

\mypara{Variable (re)ordering.}  
In practice, the size of a \ac{BDD} crucially depends on the variable ordering 
that is imposed.  For example, a certain sum-of-products 
formula~\cite{Bryant86} 
can be represented with a linear number of nodes in the best ordering, and with 
an exponential number of nodes in the worst ordering.  Unfortunately, it can be 
shown~\cite{Sieling02} that the problem of computing a variable ordering that 
results in at most $k$ times the \ac{BDD} nodes of the optimal ordering is 
NP-complete.  That is, finding a good variable ordering is a computationally 
hard problem.  As a consequence, \ac{BDD} libraries mostly rely on heuristics.  
Particularly important are dynamic reordering heuristics~\cite{Rudell93}, which 
try to reduce the \ac{BDD} size automatically while constructing and 
manipulating \acp{BDD}.  Additionally (or alternatively), the user of a 
\ac{BDD} 
library can also trigger reorderings with specified heuristics manually.  

Variable reordering heuristics are certainly effective in improving the 
scalability of \acp{BDD}, especially in industrial applications such as formal 
verification of hardware circuits~\cite{Rudell93}.  However, there exist 
formulas for which no variable ordering yields a small \ac{BDD}.  Even worse, 
such characteristics cannot only be observed on artificial examples, but also 
on 
structures that occur frequently in industrial applications.  For instance, for 
an $n$-bit multiplier, it can be shown~\cite{Bryant86} that at least one of the 
output functions requires at least $2^{n/8}$ \ac{BDD} nodes for any variable 
ordering.  Together with the recent progress in efficient \acs{SAT} solving 
(see 
below), these scalability issues are among the reasons why \acp{BDD} are 
increasingly displaced in applications like model checking.

\mypara{Operations on \acp{BDD}.}
\ac{BDD} libraries like \cudd~\cite{cudd} provide a rich set of operations.  
Besides the basic Boolean connectives $\neg$, $\vee$, $\wedge$, etc., they 
offer 
universal and existential quantification of variables.  Hence, \acp{BDD} can 
also be used to reason about \acfp{QBF}.  Other useful operations are the 
computation of positive and negative cofactors, as well as swapping of 
variables 
in the formula.  Satisfying assignments can be computed by traversing some path 
from the root to the terminal node \fbox{1}.  \ac{BDD} libraries often also 
provide combined operations that can be computed more efficiently than 
performing the operations in isolation.  One example of such a combined 
operation is $\exists \overline{x} \scope F_1(\overline{x}, \overline{y}) 
\wedge 
F_2(\overline{x}, \overline{z})$, i.e., conjunction followed by existential 
quantification of some variables. Because of this rich set of operations, it is 
often not difficult to realize symbolic algorithms (we will introduce this 
term in Section~\ref{sec:prelim:syenc}) using \acp{BDD} as the underlying 
reasoning engine.
}

\subsubsection{\acs{SAT} solvers}  

A \acs{SAT} solver decides whether a 
given propositional formula in \ac{CNF} is satisfiable.  This problem is 
NP-complete, i.e., given solutions can be checked in polynomial 
time, 
but no polynomial algorithms to compute solutions are known\footnote{Even more, 
if P$\neq$NP, which is widely believed but not proven, no polynomial algorithm 
exists.}.  Despite this relatively high complexity\footnote{Well, in comparison 
to the complexities that have to be dealt with in synthesis it is actually not 
so high.} there have been enormous scalability improvements over the last 
decades.  Today, modern \acs{SAT} solvers can handle industrial problems 
with millions of variables and clauses~\cite{KatebiSS11}.  

\mypara{Working principle.}
Modern \acs{SAT} solvers~\cite{KatebiSS11} are based on the concept of 
\emph{\ac{CDCL}}, where 
partial assignments that falsify the 
formula are eliminated by adding a blocking clause to forbid the partial 
assignment.  The current assignment in the search is not just negated to obtain 
the clause.  Instead, a conflict graph is analyzed with the goal of eliminating 
irrelevant variables and thus learning smaller blocking clauses.  This idea is 
combined with aggressive (so-called \emph{non-chronological}) backtracking to 
continue the search.  This general principle was introduced in 1996 with 
the \acs{SAT} solvers \grasp~\cite{SilvaS96}.  Modern solvers still follow the 
same principle~\cite{KatebiSS11}, but extended with clever data structures for 
constraint propagation, heuristics to choose variable assignments, restarts of 
the search, and other improvements.  We refer to~\cite{hos} for more details on 
these techniques.

\mypara{SAT competition.}
One driving force for research in efficient \acs{SAT} solving is the annual 
\href{http://www.satcompetition.org/}{\acs{SAT} competition}~\cite{JarvisaloBRS12}, held since 2002.  It 
also defines a simple textual format for \acsp{CNF}, which is called 
\dimacs~\cite{Prestwich09}  and supported by virtually 
all 
\acs{SAT} solvers.  A comparison~\cite{JarvisaloBRS12} of the best solvers from 
2002 to 2011 shows that the number of benchmark instances (of the 2009 
benchmark 
set) solved within 1200 seconds increased from around $50$ to more than $170$ 
during this time span.  Conversely, the maximum solving time for the $50$ 
simplest benchmarks dropped from around $1100$ seconds to around $10$ seconds.  
The plot in~\cite{JarvisaloBRS12} summarizing this data does not show any signs 
of saturation over the years. Hence, further performance improvements can also 
be expected for the coming years.  Our \acs{SAT} solver based synthesis methods 
will directly benefit from such improvements.

\paragraph{Solver Features and Notation} \label{sec:sat_not} ~\newline

In the algorithms presented in this article, we will denote a call to a 
\acs{SAT} solver by 
  $\textsf{sat} := \propsat\bigl(F(\overline{x})\bigr),$ 
where $F(\overline{x})$ is a propositional formula in \ac{CNF}.  The variable 
$\textsf{sat}$ is assigned $\true$ if $F(\overline{x})$ is satisfiable, and 
$\false$ otherwise.  

\mypara{Satisfying assignments.}
Modern \acs{SAT} solvers do not only decide satisfiability, but 
can also compute a satisfying assignment for the variables in the formula.  We 
will write 
  $(\textsf{sat},\mathbf{x}, \mathbf{y}, \ldots) := 
\propsatmodel\bigl(F(\overline{x}, \overline{y}, \ldots)\bigr)$ 
to denote a call to the solver where we also extract a satisfying assignment in 
the form of cubes $\mathbf{x}, \mathbf{y}, \ldots$ over the variables
$\overline{x}, \overline{y}, \ldots$ occurring in the formula $F$.
The cubes may be incomplete if the value of the missing 
variables is irrelevant for $F$ to be $\true$.  The returned cubes are 
meaningless if $\textsf{sat}$ is $\false$.  

\mypara{Unsatisfiable cores.}
Another feature of modern \acs{SAT} solvers is the efficient computation of 
unsatisfiable cores, as defined in Section~\ref{sec:prelim:prop}.  Given that
$\mathbf{x} \wedge F(\overline{x},\overline{y})$ is unsatisfiable, we will write
  $\mathbf{x}' := \propsatcore\bigl(\mathbf{x}, 
    F(\overline{x},\overline{y})\bigr)$
to denote the extraction of an unsatisfiable core $\mathbf{x}'\subseteq 
\mathbf{x}$ such that $\mathbf{x}' \wedge F(\overline{x},\overline{y})$ is 
still 
unsatisfiable.  Natively, \acs{SAT} solvers usually compute unsatisfiable 
cores that are not necessarily minimal.  However, a computed core can easily be 
minimized by trying to drop literals of $\mathbf{x}'$ one by one and checking 
if unsatisfiability is still preserved. We will denote the computation of a 
minimal unsatisfiable core by
  $\mathbf{x}' := \propsatmincore\bigl(\mathbf{x}, 
    F(\overline{x},\overline{y})\bigr).$
In our algorithms, we use unsatisfiable core computations to 
generalize discovered facts.  In our experience, good generalizations (in the 
form of small cores) are usually more beneficial than fast ones.  Thus, we will 
usually compute minimal unsatisfiable cores.

\mypara{Interpolation.}
Given two \acp{CNF} $A(\overline{x},\overline{y})$ and 
$B(\overline{x},\overline{z})$ with $A\wedge B=\false$, we denote the 
computation of a Craig interpolant $I(\overline{x})$ (such that $A\rightarrow I 
\rightarrow \neg B$; cf.\ Section~\ref{sec:prelim:prop}) by 
  $I := \interpol(A,B).$
While \acs{SAT} solvers usually cannot compute interpolants natively, many of 
them can output unsatisfiability proofs.  An interpolant can then be computed 
from such an unsatisfiability proof for $A\wedge B$ using different 
methods~\cite{DSilvaKPW10}.

\mypara{Incremental solving.}
Modern \ac{CDCL}-based \acs{SAT} solvers can solve sequences of similar 
\ac{CNF} 
queries more efficiently than by processing the queries in isolation.  For 
instance, if clauses are only added but not removed between satisfiability 
checks, all the 
clauses learned so far can be retained and do not have to be rediscovered again 
and again.  Removing clauses is more problematic. Certain learned clauses may 
become invalid and need to be removed as well. Clause removals are supported by 
different solvers in different ways (or not at all). One wide-spread approach 
is 
to provide an interface for pushing the current state of the solver onto a 
stack 
and restoring it later.  A related feature that is supported by many \acs{SAT} 
solvers is \emph{assumption literals},  which can be 
asserted temporarily.  In the algorithms presented in this article, we will 
mostly avoid removing clauses from incremental \acs{SAT} sessions and use 
assumption literals to enable or disable parts of a formula instead.  In this 
context, will also refer to variables that are introduced for the purpose of 
enabling or disabling formula parts as \emph{activation variables}.

In general, we will present our synthesis algorithms in a non-incremental way 
and discuss the use of incremental solving separately.  This way, we do not have 
to introduce notation for adding clauses, resetting the state of a solver, etc., 
which improves the readability of the algorithms.

\subsubsection{\acs{QBF} Solvers}

A \acs{QBF} solver decides whether a given \acl{QBF} in \ac{PCNF} is 
satisfiable.  This problem is PSPACE-complete~\cite{BuningB09}, i.e., solving it 
requires a polynomial amount of memory.  No NP-time algorithms are 
known\footnote{And it is widely believed, but not proven, that no such 
algorithms exist.}, so from a complexity point of view, \acs{QBF} problems are 
(likely to be) strictly harder than \acs{SAT} problems.  

\mypara{Working principle.}
While most modern \acs{SAT} solvers follow the concept of \ac{CDCL}, the set of 
techniques applied for \acs{QBF} solving is more diverse.  For instance, the 
solver \depqbf~\cite{LonsingB10} uses a search-based algorithm (called QDPLL) 
with conflict-driven clause learning (similar to \ac{CDCL} \acs{SAT} solvers) 
and solution-driven cube learning.  The solver \quantor~\cite{Biere04} uses 
variable elimination in order to transform the problem into a purely 
propositional formula.  The solver \rareqs~\cite{JanotaKMC12} follows the idea 
of counterexample-guided refinement of solution candidates, where plain 
\acs{SAT} solvers are used to compute solution candidates as well as to refute 
and refine them.  None of these techniques is clearly superior --- different 
techniques appear to work well on different benchmarks.

\mypara{Preprocessing.} 
An important topic in \acs{QBF} solving is preprocessing.  A \acs{QBF} 
preprocessor simplifies a \acs{QBF} before the actual solver is called.  It is 
also possible that the preprocessor solves a \acs{QBF} problem directly, or 
reduces it to a propositional formula, for which a \acs{SAT} solver can be 
used. 
\bloqqer~\cite{BiereLS11} is an example of a modern \acs{QBF} preprocessor 
implementing many techniques.  It has been shown to have a very positive impact 
on the performance of various solvers~\cite{BiereLS11}: when using 
\bloqqer, the  \acs{QBF} solvers \depqbf~\cite{LonsingB10}, 
\quantor~\cite{Biere04}, \qube~\cite{GiunchigliaNT01} and 
\nenofex~\cite{LonsingB08} can solve between 20\,\% and 40\,\% more benchmarks 
(of the benchmark set from the QBFEVAL 2010 competition within 900 seconds).  
The median execution time decreases by up to a factor of $50$ (achieved for 
\qube) due to \bloqqer~\cite{BiereLS11}.

\mypara{Competitions.}
Similar to \acs{SAT} solving, there are also competitions in \acs{QBF} solving 
(\href{http://www.qbflib.org/index_eval.php}{QBFEVAL} and the 
\href{http://qbf.satisfiability.org/gallery/}{QBF Gallery}) with the aim 
of 
collecting benchmarks as well as assessing and advancing the state of the art 
in 
\acs{QBF} research and tool development.  The input format for these 
competitions is called \qdimacs,  and is essentially 
just an extension of the \dimacs format with a quantifier prefix.  While the 
\ac{QBF} competitions definitely witness solid progress in scalability over the 
years, it seems that \acs{QBF} has not yet reached the maturity of \acs{SAT}, 
especially when it comes to industrial applications such as formal 
verification, 
where the scalability is often insufficient~\cite{BenedettiM08}.  However, 
because \acs{QBF} is a much younger research field than \acs{SAT}, 
future scalability improvements may be even more significant.  The 
\acs{QBF}-based synthesis algorithms presented in this article would directly 
benefit from such developments.

\paragraph{Solver Features and Notation} \label{sec:qbf_not} ~\newline

Similar to our notation for \acs{SAT} solvers, we will write
 $\textsf{sat} := \qbfsat\bigl(Q_1\overline{x}
   \scope Q_2\overline{y} \scope \ldots F(\overline{x},\overline{y}, \ldots)
   \bigr)$ 
to denote a call to a \acs{QBF} solver, where $F$ is a propositional formula in
\ac{CNF}, and $Q_i\in\{\exists,\forall\}$.  As before, \textsf{sat} will be
assigned $\true$ if the \acs{QBF} is satisfiable and $\false$ otherwise.

\mypara{Satisfying assignments.} 
Many existing \acs{QBF} solvers cannot only decide the satisfiability of 
formulas, but also compute satisfying assignments for variables that are 
quantified existentially on the outermost level.  We will write
  $(\textsf{sat}, \mathbf{a}, \mathbf{b} \ldots) :=
     \qbfsatmodel\bigl(\exists\overline{a} \scope 
     \exists\overline{b} \scope \ldots
     Q_1 \overline{x} \scope Q_2 \overline{y} \scope \ldots
     F(\overline{a},\overline{b}, \ldots,
     \overline{x},\overline{y}, \ldots)\bigr)$
to denote the extraction of such a satisfying assignment in the form of cubes 
$\mathbf{a}, \mathbf{b}, \ldots$ over the variable vectors $\overline{a}, 
\overline{b}, \ldots$ quantified existentially on the outside.  In general, 
satisfying assignments cannot be extracted when applying \ac{QBF} preprocessing, 
because preprocessing techniques are often not model preserving.  However, 
recently, an extension of the popular \ac{QBF} preprocessors \bloqqer to 
preserve satisfying assignments has been proposed~\cite{SeidlK14}.  This 
extension enables using \ac{QBF} preprocessing in synthesis algorithms that 
require satisfying assignments. 

\mypara{Unsatisfiable cores.}
Certain \acs{QBF} solvers, such as \depqbf~\cite{LonsingE14}, can compute 
unsatisfiable cores natively.  However, this feature cannot be used with 
preprocessing straightforwardly.  Furthermore, we did not encounter significant 
performance improvements in our experiments compared to minimizing the core in 
an explicit loop.  Hence, we do not introduce notation for unsatisfiable 
\ac{QBF} cores and use explicit minimization loops in our algorithms instead.

\mypara{Incremental solving.}
Comprehensive approaches for incremental \acs{QBF} solving have only been 
proposed recently~\cite{LonsingE14}.  However, incremental solving cannot 
yet be used in combination with \acs{QBF} preprocessing, because existing 
preprocessors are inherently non-incremental.  We experimented with incremental 
solving in our synthesis algorithms.  For many cases, preprocessing 
turned out to more beneficial than incremental solving.  We will thus 
refrain from introducing notation for incremental \ac{QBF} solving, and discuss 
possibilities for incremental solving separately.  

\subsubsection{First-Order Theorem Provers}

First-order logic is undecidable~\cite{0017977}, that is, an algorithm to decide 
the satisfiability (or validity) of every possible first-order logic formula 
cannot exist. Yet, incomplete algorithms and tools \emph{do} exist, and they 
perform well on many practical problems.  Similar to \acs{SAT} and \acs{QBF}, 
there is also a competition for automatic theorem provers \index{theorem prover} 
to solve problems in first-order logic and subsets thereof.  It is called 
\href{http://www.cs.miami.edu/~tptp/CASC/}{CASC}~\cite{SutcliffeS06} 
and exists since 1996.  Benchmarks for the competition are taken from 
the TPTP library~\cite{SutcliffeS98}, which defines a common format for 
first-order logic problems.

In this work, we are particularly interested in the subset called \acf{EPR}.  In 
contrast to full first-order logic, 
\ac{EPR} is actually decidable~\cite{Lewis80} (the problem is 
NEXPTIME-complete).  The CASC competition also features a track for \ac{EPR}. 
From 2008 to 2014, this track was always won by \iprover~\cite{Korovin08}. 
\iprover is an instantiation-based solver and can thus not only decide the 
satisfiability of \ac{EPR} formulas, but also compute models in form of concrete 
realizations for the predicates.  This feature makes \iprover particularly 
suitable for synthesis.

\subsection{Symbolic Encoding and Symbolic Computations} 
\label{sec:prelim:syenc}

Formal methods for verification or synthesis must be able to deal with large 
sets of states or large sets of possible inputs efficiently.  \emph{Symbolic 
encoding}~\cite[page 383]{0017977} is a way to 
represent large sets of elements compactly using formulas. 
 Set elements are represented by assignments to variables.  Formulas 
over these variables characterize which elements are contained in a set: if the 
formula evaluates to $\true$ for a particular variable assignment, then the 
corresponding element is part of the set, otherwise not. Such a formula is 
called the 
\emph{characteristic formula} of the set.

\begin{exaq}\label{ex:syenc}
Consider the set $A$ of all integers from $0$ to $65535$.  We can use $16$ 
Boolean variables $\overline{x} = (x_0,\ldots, x_{15})$ to encode subsets of 
$A$ 
symbolically.  The variables represent the bits of the binary encoding of a 
number, with $x_0$ being the least significant bit.  An explicit representation 
of the set $A_0=\{0,2,4,\ldots,65534\}$ of all even numbers would have to 
enumerate $32768$ elements.  In a symbolic representation, the set of even 
numbers can be represented by the propositional formula $F_0(\overline{x}) = 
\neg x_0$, requiring that the least significant bit is $\false$ and all other 
bits are arbitrary.  The set $A_1=\{49152,49153,\ldots,65535\}$ of all numbers 
greater or equal to $49152$ can be represented symbolically using the formula 
$F_1(\overline{x}) = x_{15} \wedge x_{14}$, stating that the two most 
significant bits must be set.
\end{exaq}

Characteristic formulas cannot only be used to \emph{represent} sets.  We can 
also perform set operations directly on the formulas.  A set union $A_0 \cup 
A_1$ can be realized as disjunction of the corresponding characteristic 
formulas 
$F_0$ and $F_1$, intersection corresponds to conjunction, and a complement to 
the negation of the characteristic formula. The formula $\false$ represents the 
empty set, the formula $\true$ represents the set of all elements in the domain.

\begin{exaq}\label{ex:sycomp}
Continuing Example~\ref{ex:syenc}, the set $A_0 \cap A_1$ of even numbers 
greater 
or equal to $49152$ can be computed symbolically as $F_0(\overline{x}) \wedge 
F_1(\overline{x}) = \neg x_0 \wedge x_{15} \wedge x_{14}$.  The set $A_1 
\setminus A_0$ of odd numbers greater or equal to $49152$ can be computed 
symbolically as $F_1(\overline{x}) \wedge \neg F_0(\overline{x}) = x_{15} 
\wedge 
x_{14} \wedge x_0$.
\end{exaq}

In this article, we will often handle sets and their symbolic representations 
interchangeably.  For instance, we may say ``the set of states 
$F(\overline{x})$'' although $F$ is a formula over state variables 
$\overline{x}$, representing the set symbolically.

\subsection{Reactive Synthesis from Safety Specifications} \label{sec:prel:hw}

This section defines the reactive synthesis problem from safety 
specifications and the relevant concepts from game theory.  We also 
present a standard textbook solution.  It will serve as 
baseline for our satisfiability-based methods.  

\subsubsection{Safety Specifications} 

\begin{wrapfigure}[9]{r}{0.55\textwidth}
\vspace{-4mm}
\centering
\includegraphics[width=0.54\textwidth]{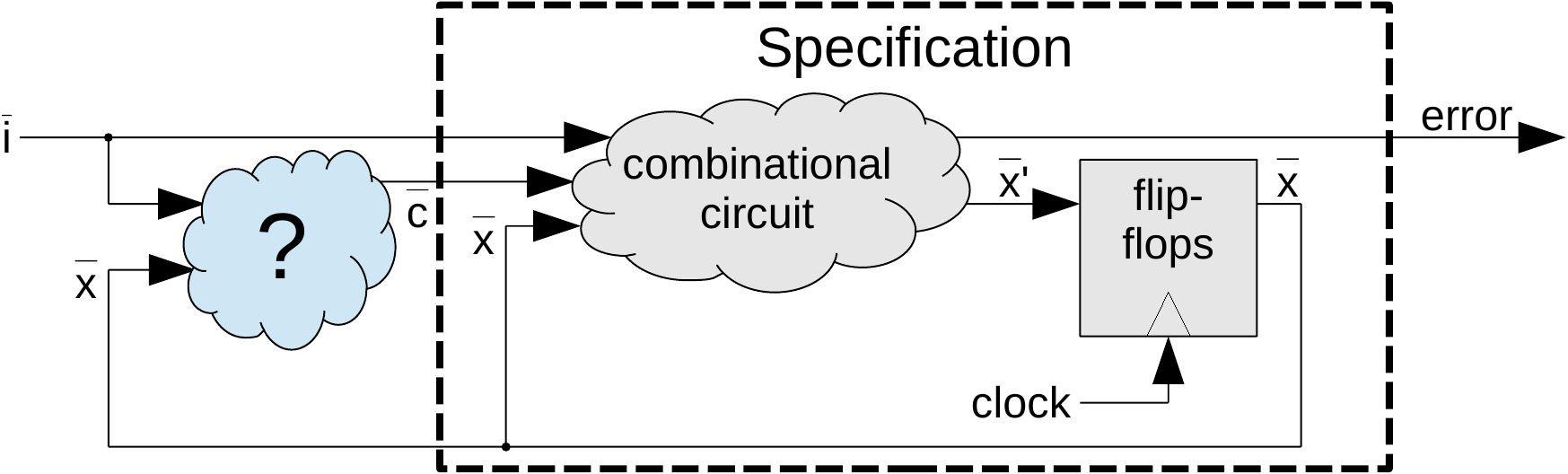}
\caption{Circuit representation of a safety specification.}
\label{fig:syntcomp_spec}
\end{wrapfigure}
A safety specification expresses that certain ``bad things'' 
never happen in a system.  We follow the framework of the 
\syntcomp~\cite{sttt_syntcomp} 
synthesis competition, which defines safety specification benchmarks as hardware 
circuits in \aiger format, as 
illustrated in Figure~\ref{fig:syntcomp_spec}.  The circuits have uncontrollable 
inputs 
 $\overline{i}$, controllable inputs 
 $\overline{c}$, flip-flops to store a number of 
state bits $\overline{x}$, and one output ``error'' signaling specification 
violations.  The corresponding synthesis problem is to construct a circuit that 
defines the controllable inputs $\overline{c}$ based on the uncontrollable 
inputs $\overline{i}$ and the state $\overline{x}$ in such a way that the error 
output can never become $\true$. This unknown circuit to be 
constructed 
is denoted with a question mark in Figure~\ref{fig:syntcomp_spec}. We will also 
refer to the controllable inputs as \emph{control signals} to emphasize that 
these signals are not intended to be inputs of the 
final system. 

The specification illustrated in Figure~\ref{fig:syntcomp_spec} can be seen as 
a runtime monitor, declaratively encoding the design intent for the system to be 
synthesized.  Another view is that the specification is a plant which needs to 
be controlled, or a sketch of a hardware circuit where the implementation for 
certain signals is still missing.  Hence, this format flexibly fits various 
applications of synthesis. Formally, we define a safety specification as 
follows.

\begin{definition}[Safety Specification] \label{def:safety}
A safety specification is a tuple
$\mathcal{S} = (\overline{x}, \overline{i}, \overline{c}, I, T, P)$, where 
\begin{compactitem}
\item $\overline{x}$ is a vector of Boolean state variables, 
\item $\overline{i}$ is a vector of uncontrollable, Boolean input variables, 
\item $\overline{c}$ is a vector of controllable, Boolean input variables,
\item $I(\overline{x})$ is an initial condition, expressed as a propositional 
formula over the state variables,
\item $T(\overline{x},\overline{i},\overline{c},\overline{x}')$ is a transition
relation, expressed as a propositional formula over 
the variables 
$\overline{x}$, $\overline{i}$, $\overline{c}$, and $\overline{x}'$, where 
$\overline{x}'$ denotes the next-state copy of $\overline{x}$,
\item the transition relation 
$T(\overline{x},\overline{i},\overline{c},\overline{x}')$ is complete in the 
sense that $\forall \overline{x}, \overline{i}, \overline{c} \scope \exists 
\overline{x}' \scope T(\overline{x},\overline{i},\overline{c},\overline{x}')$,
\item $T$ is deterministic, meaning that $\forall 
\overline{x}, \overline{i}, 
\overline{c},
\overline{x}_1', \overline{x}_2' \scope
\bigl(T(\overline{x},\overline{i},\overline{c},\overline{x}_1') \wedge
T(\overline{x},\overline{i},\overline{c},\overline{x}_2')\bigr)
\rightarrow (\overline{x}_1' = \overline{x}_2')
$, and
\item $P(\overline{x})$ is a propositional formula representing the set of safe 
states in $\mathcal{S}$.
\end{compactitem}
\end{definition}

\noindent
A \emph{state} of $\mathcal{S}$ is an 
assignment to all state variables 
$\overline{x}$.  We represent such assignments (and thus states ) as 
$\overline{x}$-minterms $\mathbf{x}$.  In the spirit of 
symbolic encoding as introduced Section~\ref{sec:prelim:syenc}, a formula 
$F(\overline{x})$ over the state variables $\overline{x}$ represents the set of 
all states $\mathbf{x}$ for which $\mathbf{x} \models F(\overline{x})$ holds. 
In 
this way, the formula $I(\overline{x})$ defines a set of initial states, 
and $P(\overline{x})$ defines the safe states. Similarly, the 
formula 
$T$ defines allowed state transitions: a transition from the current state 
$\mathbf{x}$ to the next state $\mathbf{x}'$ is allowed with input $\mathbf{i}$ 
and $\mathbf{c}$ iff $\mathbf{x} \wedge \mathbf{i} \wedge \mathbf{c} 
\wedge \mathbf{x}' \models 
T(\overline{x},\overline{i},\overline{c},\overline{x}')$.  
Definition~\ref{def:safety} requires that the transition relation $T$ is both 
deterministic and complete.  That is, for any state $\mathbf{x}$ and input 
$\mathbf{i},\mathbf{c}$, the next state $\mathbf{x}'$ is uniquely 
defined.

\subsubsection{Safety Games}

A specification $\mathcal{S} = (\overline{x}, \overline{i}, \overline{c}, I, T, 
P)$ can be seen as a game between two players: the \emph{environment} and the 
\emph{system} we wish to synthesize.  Depending on the context, we will thus 
refer to $\mathcal{S}$ either as a specification or as a game.

\mypara{Plays.} 
The game starts in one of the initial states (chosen by the environment), and 
is 
played in rounds.  In every round $j$, the environment first chooses an 
assignment $\mathbf{i}_j$ to the uncontrollable inputs~$\overline{i}$.  Next, 
the system picks an assignment $\mathbf{c}_j$ to the controllable inputs 
$\overline{c}$.  The transition relation $T$ then computes the next state 
$\mathbf{x}_{j+1}$.  This is repeated indefinitely.  The resulting sequence 
$\mathbf{x}_0, \mathbf{x}_1 \ldots $ of states is called a \emph{play}.  
Formally, we have that $\mathbf{x}_0\models I(\overline{x})$ and 
$\mathbf{x}_j \wedge \mathbf{x}_{j+1}' \wedge  
T(\overline{x},\overline{i},\overline{c},\overline{x}')$ is satisfiable (with 
some $\mathbf{i}_j$ and $\mathbf{c}_j$ chosen by the players) for all $j\geq0$. 
A play $\mathbf{x}_0, \mathbf{x}_1 \ldots $ is \emph{won} by the system and 
lost 
by the environment if $\forall j\scope \mathbf{x}_j \models P(\overline{x})$, 
i.e., if only safe states are visited.  Otherwise, the play is \emph{lost} by 
the system and won by the environment.

\mypara{Preimages.} 
Let $F(\overline{x})$ be a formula representing a certain 
set of states. The mixed preimage 
$\FS\bigl(F(\overline{x})\bigr) = \forall \overline{i} \scope 
                                   \exists \overline{c},\overline{x}' \scope 
  T(\overline{x},\overline{i},\overline{c},\overline{x}') \wedge 
  F(\overline{x}')
$
represents all states from which the system can enforce  
that some state of $F$ is reached in exactly one step. Analogously, 
$\FE\bigl(F(\overline{x})\bigr) = \exists \overline{i} \scope 
                        \forall \overline{c}\scope \exists \overline{x}'\scope 
  T(\overline{x},\overline{i},\overline{c},\overline{x}') \wedge 
  F(\overline{x}')
$ 
gives all states from which the environment can enforce that $F$ is visited in 
one step.  We also define the cooperative preimage 
$ \reach\bigl(F(\overline{x})\bigr) = 
   \exists \overline{i},\overline{c},\overline{x}' \scope 
  T(\overline{x},\overline{i},\overline{c},\overline{x}') \wedge 
  F(\overline{x}')$
denoting the set of all states from which $F$ can be reached cooperatively by 
the two players.  The following dualities can easily be shown:
\begin{compactitem}
\item $\neg\FS(F) = \FE(\neg F)$ holds because, intuitively, the states 
from which the system cannot enforce that $F$ is reached must be the states 
from which the environment can enforce that $\neg F$ is reached.
\item $\neg\FE(F) = \FS(\neg F)$ holds because, dually, the states from 
which the environment cannot enforce that $F$ is reached must be the states 
from which the system can enforce that $\neg F$ is reached.
\end{compactitem}
Furthermore, we have that $\reach(F_1) \vee \reach(F_2) = \reach(F_1 \vee 
F_2)$. 
Yet, the following equivalence does \textbf{not} hold in general: $\FS(F_1) 
\vee 
\FS(F_2) \uneq \FS(F_1 \vee F_2)$.  The reason is that there may be 
states from which the environment controls whether $F_1$ or $F_2$ is visited 
next, and the system can only ensure that one of the two regions is reached.  
Such states falsify $\FS(F_1 \vee F_2) \rightarrow \FS(F_1) \vee 
\FS(F_2)$. This difference in compositionality between $\reach$ and $\FS$ 
explains why some ideas from verification cannot be ported to synthesis 
straightforwardly.

\mypara{Strategies.} 
We focus on memoryless strategies because these strategies are 
sufficient\footnote{``Sufficient'' means: If a strategy to win a given safety 
game exists, then there also exists a memoryless strategy 
to win the safety game.} for safety games~\cite{Thomas95}.  A (memoryless) 
\emph{strategy} 
for the system player in the game $\mathcal{S}$ is a formula 
$S(\overline{x},\overline{i},\overline{c},\overline{x}')$ that specializes $T$ 
in the sense that
\begin{compactitem}
\item $S(\overline{x},\overline{i},\overline{c},\overline{x}') \rightarrow
       T(\overline{x},\overline{i},\overline{c},\overline{x}')$ and
\item $\forall \overline{x}, \overline{i} \scope
       \exists \overline{c}, \overline{x}' \scope
      S(\overline{x},\overline{i},\overline{c},\overline{x}').$
\end{compactitem}
The first bullet requires that the strategy may only allow 
state transitions that are also allowed by the transition relation.  The second 
bullet requires the strategy to be complete with respect to the 
current 
state and uncontrollable input: for every state $\mathbf{x}$ and input 
$\mathbf{i}$, the strategy must contain some way to choose 
$\mathbf{c}$ (and some next state, but the next state is uniquely defined by 
$T$ 
already).  For a particular situation, the strategy 
can allow many possibilities to choose $\mathbf{c}$, though.  A strategy for 
the 
system is \emph{winning} if all plays that can be 
constructed by following $S$ 
instead of $T$ are won by the system.  The \emph{winning region} 
$W(\overline{x})$ is the set of all states from which a winning strategy 
exists. 
That is, if the play would start in some arbitrary state of the winning region, 
the system player would have a strategy to win the~game.

\mypara{System implementations.} 
A \emph{system implementation} is a function $f: 2^{\overline{x}} \times 
2^{\overline{i}} \rightarrow 2^{\overline{c}}$ to uniquely define the control 
signals $\overline{c}$ based on the current state and the uncontrollable inputs 
$\overline{i}$.  A system implementation $f$ \emph{implements} a strategy $S$ 
if 
$\forall \overline{x}, \overline{i} \scope
       \exists \overline{x}' \scope
      S\bigl(\overline{x}, 
\overline{i}, f(\overline{x},\overline{i}), \overline{x}'\bigr)$, 
that is, if for every state $\mathbf{x}$ and input $\mathbf{i}$, the control 
value $\mathbf{c} = f(\mathbf{x}, \mathbf{i})$ computed by $f$ is allowed 
by the strategy $S$.  A system implementation $f$ \emph{realizes} a safety 
specification $\mathcal{S} = \bigl(\overline{x}, \overline{i}, \overline{c}, 
I(\overline{x}), T(\overline{x},\overline{i},\overline{c},\overline{x}'), 
P(\overline{x})\bigr)$ if all plays of $\mathcal{S}' = \bigl(\overline{x}, 
\overline{i}, 
\emptyset, I(\overline{x}), T(\overline{x},\overline{i},f(\overline{x}, 
\overline{i}),\overline{x}'), P(\overline{x})\bigr)$ are won by the system 
player, 
i.e., visit only safe states.  Here, $\mathcal{S}'$ is a simplified version of 
the game $\mathcal{S}$ where the moves of the system player are defined 
by $f$, i.e., the system has no choices left.  A safety specification is 
\emph{realizable} if a system 
implementation that realizes it exists.  Given a 
winning strategy $S$ for a safety specification $\mathcal{S}$, every 
implementation $f$ of the winning strategy $S$ realizes the 
specification~$\mathcal{S}$.  This follows from the definition of the 
winning strategy.  Hence, a system implementation for a safety 
specification $\mathcal{S}$ can be constructed by computing a winning strategy 
$S$ for $\mathcal{S}$ and then computing an implementation $f$ of~$S$.

\subsubsection{Synthesis Algorithms for Safety Specifications}
\label{sec:prelim:sasyalg}

Given an explicit representation of the safety specification $\mathcal{S}$ as a 
game graph (with vertices representing states and edges representing state 
transition) the problem of deciding the realizability of a safety specification 
is solvable in linear time~\cite{Thomas95}.  When starting from our 
symbolic representation $\mathcal{S}$, the problem is EXP-time 
complete~\cite{PapadimitriouY86}.

A \emph{synthesis algorithm}  for safety 
specifications takes as input a safety specification $\mathcal{S}$ and computes 
a system implementation realizing this specification if such an implementation 
exists.  If no such implementation exists, the algorithm reports 
unrealizability.  
Wolfgang 
Thomas~\cite{Thomas95} sketches the standard textbook algorithm for solving 
this 
problem.  It proceeds in two steps.  First, a winning strategy is computed.  
Second, the winning strategy is implemented in a circuit.  This process is 
elaborated in the following two subsections.

\subsubsection{Computing a Winning Strategy} \label{sec:prelim:winstrat}

The computation of a winning strategy $S(\overline{x}, \overline{i}, 
\overline{c}, 
\overline{x}')$ for the game $\mathcal{S}\! =\! \bigl(\overline{x}, 
\overline{i}, 
\overline{c}, I(\overline{x}), 
T(\overline{x}, \overline{i}, \overline{c}, \overline{x}'), 
P(\overline{x})\bigr)$ is achieved by computing the winning region 
$W(\overline{x})$ of the game $\mathcal{S}$ using the procedure 
\textsc{SafeWin}, shown in Algorithm~\ref{alg:SafeWin}.  The winning region $W$ 
is built up in the variable $F$.  Initially, $F$ represents the set of all safe 
states~$P$.  Line~\ref{alg:SafeWin:4} retains only those states of $F$ from 
which the system player can enforce that the play stays in a state of $F$ 
also in the next 
step.  This operation is repeated as long as the state set $F$ changes. If the 
set of initial states $I$ is not contained in $F$ any more, the procedure 
aborts, returning $\false$ to signal unrealizability of the specification. 
Otherwise, the final version of $F$ is returned as the winning region.  All 
operations that are performed in this algorithm can easily be realized using 
\acp{BDD}.

\begin{figure}
\begin{minipage}{0.36\textwidth}
\begin{algorithm}[H]
\caption{\textsc{SafeWin}: Computes a 
winning region in a safety game.}
\label{alg:SafeWin}
\begin{algorithmic}[1]
\ProcedureRetL{SafeWin}
             {(\overline{x}, \overline{i}, \overline{c}, I, T, P)}
             {The winning region or $\false$}
             {1mm}
  \State $F := P$
  \While{$F$ changes}             \label{alg:SafeWin:3}
    \State $F := F \wedge \FS(F)$ \label{alg:SafeWin:4}
    \If{$I \not \rightarrow F$}   \label{alg:SafeWin:5}
      \State \textbf{return} $\false$
    \EndIf
  \EndWhile
  \State \textbf{return} $F$
\EndProcedure  
\end{algorithmic}
\end{algorithm}
\vspace*{-1mm}
\end{minipage}
\hspace{1mm}
\begin{minipage}{0.62\textwidth}
\begin{algorithm}[H]
\caption[\textsc{CofSynt}: A cofactor-based algorithm for computing an 
implementation of a strategy]
{\textsc{CofSynt}: A cofactor-based algorithm for computing an 
implementation of a strategy.}
\label{alg:CofSynt}
\begin{algorithmic}[1]
\ProcedureRet{CofSynt}
             {S(\overline{x},\overline{i},\overline{c},\overline{x}')}
             {$f_1,\ldots,f_n:2^{\overline{x}} \times 2^{\overline{i}}
              \rightarrow \B$}
  \For{$c_j \in \overline{c}$}
    \State $C_1(\overline{x}, \overline{i}) := \exists \overline{x}', 
                \overline{c} \scope
      S\bigl(\overline{x}, 
       \overline{i},
       (c_0,\ldots,c_{j-1},\true,c_{j+1},\ldots,c_n),
       \overline{x}'
       \bigr)$\label{alg:CofSynt:c1}
     \State $C_0(\overline{x}, \overline{i}) := \exists \overline{x}', 
                 \overline{c}\scope
      S\bigl(\overline{x}, 
       \overline{i},
       (c_0,\ldots,c_{j-1},\false,c_{j+1},\ldots,c_n),
       \overline{x}'
       \bigr)$\label{alg:CofSynt:c0}
     \State $C(\overline{x}, \overline{i}) := 
            \neg C_1(\overline{x}, \overline{i})  \vee
            \neg C_0(\overline{x}, \overline{i})$\label{alg:CofSynt:c}
     \State $F_j(\overline{x}, \overline{i}) := \textsf{simplify}(C_1, C)$
            \label{alg:CofSynt:simpl}
     \State $S(\overline{x},\overline{i},\overline{c},\overline{x}') := 
             S(\overline{x},\overline{i},\overline{c},\overline{x}') \wedge 
            \bigl(c_j \leftrightarrow F_j(\overline{x}, \overline{i})\bigr)$
            \label{alg:CofSynt:resub}
  \EndFor
  \State \textbf{return} $F_1,\ldots, F_n$
\EndProcedure  
\end{algorithmic}
\end{algorithm}
\end{minipage}
\end{figure}

If the specification is realizable, i.e., \textsc{SafeWin} did not return 
$\false$, a winning strategy $S$ is computed from the winning region $W$. For 
safety specifications, $S$ can be defined as 
$S(\overline{x},\overline{i},\overline{c},\overline{x}') =   
T(\overline{x},\overline{i},\overline{c},\overline{x}') \wedge \bigl(   
W(\overline{x}) \rightarrow W(\overline{x}')\bigr).
$ 
That is, the transition relation must always be respected.  Furthermore, if the 
current state is in the winning region, then the next state must be contained 
in 
the winning region as well.  This strategy will enforce the specification 
because $I \rightarrow W$, i.e., all initial states are contained in the 
winning 
region (otherwise \textsc{SafeWin} would have signaled unrealizability).  When 
starting from a state of the winning region, the strategy ensures that the next 
state will be in the winning region again.  Finally, the winning region $W$ can 
only contain safe states, i.e., $W\rightarrow P$.  Hence, only safe states can 
be visited when following the~strategy.

\subsubsection{Computing a System Implementation from a Winning Strategy}

The second step is to compute a system implementation that implements 
the strategy, and to realize this implementation in form of a circuit.  This 
can be done by computing a Skolem function
for the variables $\overline{c}$ in the formula
$\forall \overline{x}, \overline{i} \scope
  \exists \overline{c}, \overline{x}' \scope 
  S(\overline{x},\overline{i},\overline{c},\overline{x}'),
$
i.e., a function $f: 2^{\overline{x}} \times 2^{\overline{i}} \rightarrow 
2^{\overline{c}}$ such that
$\forall \overline{x}, \overline{i} \scope
  \exists \overline{x}' \scope 
  S\bigl(\overline{x},\overline{i},f(\overline{x}, 
  \overline{i}),\overline{x}'\bigr)
$
holds.  Usually, we prefer simple functions that can be implemented in small 
circuits.  A survey of existing methods to solve this problem can be found in 
the work by Ehlers et al.~\cite{EhlersKH12}.  One widely used method is 
presented in the following.

\mypara{The cofactor-based method.} The cofactor-based method 
presented by Bloem et al.~\cite{BloemJPPS12} can be considered as the
``standard 
method'' for computing an implementation from a strategy.  It is outlined in 
Algorithm~\ref{alg:CofSynt}.  The input is a strategy $S$, the output is a set 
of functions $f_1,\ldots,f_n:2^{\overline{x}} \times 2^{\overline{i}} 
\rightarrow \B$, each one defining one control signal of $\overline{c} = 
(c_1,\ldots, c_n)$.  Together, these functions define $f: 2^{\overline{x}} 
\times 2^{\overline{i}} \rightarrow 2^{\overline{c}}$.  \textsc{CofSynt} 
computes one $f_j$ after the other.  In Line~\ref{alg:CofSynt:c1}, a 
formula $C_1(\overline{x}, \overline{i})$ is constructed.  It represents the 
set 
of all valuations of $\overline{x}$ and $\overline{i}$ in which $c_j=\true$ is 
allowed by the strategy.  It is computed as the positive cofactor of $S$ with 
respect to $c_j$, while all signals that are currently not relevant are 
quantified existentially.  Similarly, Line~\ref{alg:CofSynt:c0} computes all 
situations where $c_j=\false$ is allowed by the strategy.  Our definition of a 
strategy implies that $C_1(\overline{x}, \overline{i}) \vee C_0(\overline{x}, 
\overline{i}) = \true$, i.e., one of the two values is always allowed (but 
sometimes both are allowed). Next, Line~\ref{alg:CofSynt:c} computes the 
\emph{care set} $C$, i.e., the set of all situations in which 
the output 
matters.  Outside of this care set, the value of $c_j$ can be set arbitrarily.  
Line~\ref{alg:CofSynt:simpl} uses this information to simplify $C_1$:  The 
procedure \textsf{simplify} returns some $F_j$ which is equal to $C_1$ wherever 
$C$ is $\true$, and arbitrary where $C$ is $\false$.  When using \acp{BDD} as 
reasoning engine, this simplification can be implemented with 
the 
operation \textsf{Restrict}~\cite{CoudertM90}.  However, this is an 
optional optimization to obtain smaller circuits.  Setting $F_j=C_1$ would work 
as well.  Finally, Line~\ref{alg:CofSynt:resub} refines the strategy $S$ with 
the computed implementation for the control signal $c_j$.  This step is 
necessary because some control signals may depend on others, so fixing the 
implementation of one control signal may restrict other control signals.  

\begin{wrapfigure}[8]{r}{0.35\textwidth}
\vspace{-3mm}
\centering
\includegraphics[width=0.34\textwidth]{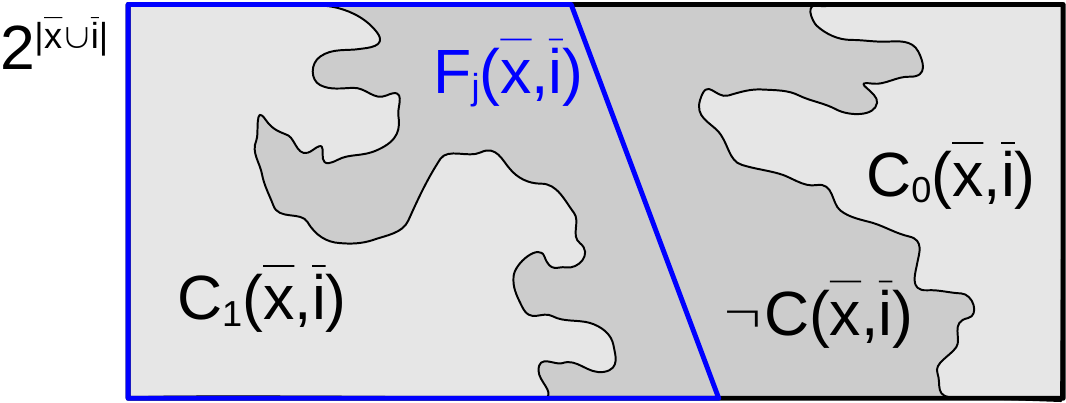}
\caption{Working principle of \textsc{CofSynt}.}
\label{fig:cofsynt}
\end{wrapfigure}
\mypara{Illustration.}
Figure~\ref{fig:cofsynt} illustrates one iteration of the \textsc{CofSynt} 
procedure graphically.  The box represents the set 
of all possible assignments 
to the variables $\overline{x}$ and $\overline{i}$.  The region $C_1$ contains 
all situations where $c_j=\true$ is allowed.   Similarly, $C_0$ contains all 
situations where $c_j=\false$ is allowed.  The overlap of the two regions is 
colored in dark gray.  Hence, the dark gray region is the set of situations 
where both $c_j=\true$ and $c_j=\false$ is allowed.  It corresponds to the 
negation $\neg C$ of the care set $C$.  Note that each point in the box is 
either contained in $C_1$ or in $C_0$ (or in both).  The function $F_j$ 
defining 
$c_j$ is shown in blue.  Outside of the dark gray don't-care area $\neg C$ it 
matches $C_1$ precisely.  In the don't-care area it can be different, though. 
These properties are enforced by the procedure \textsf{simplify}, called in 
Line~\ref{alg:CofSynt:simpl} of \textsc{CofSynt}.  Exploiting the freedom in 
the 
don't-care region can result in simpler formulas and thus in smaller circuits.  
In Figure~\ref{fig:cofsynt}, this is indicated by $F_j$ being much more 
regular than $C_1$.

\mypara{Computing circuits.}
In order to obtain an implementation $f$ in form of a hardware circuit, the 
individual functions $f_j$, defined as formulas $F_j$, need to be transformed 
into a network of gates.  In principle, this is not difficult:  each $F_j$ is a 
propositional formula (if quantifiers are left, they can be expanded) and the 
structure of the formula can directly be translated into gates.  If \acp{BDD} 
are used, each \acs{BDD} node can be translated into a multiplexer.

\subsection{Learning by Queries}\label{sec:prelim:ql}

In this section, we discuss concepts for learning propositional formulas based 
on queries, as introduced by Angluin~\cite{Angluin87}.  We refer to 
Crama and Hammer~\cite[Chapter 7]{Crama:2010:BMM:1855106} for a more elaborate 
discussion.

\subsubsection{Basic Concept}

\begin{wrapfigure}[7]{r}{0.33\textwidth}
\vspace{-6mm}
\centering
\includegraphics[width=0.32\textwidth]{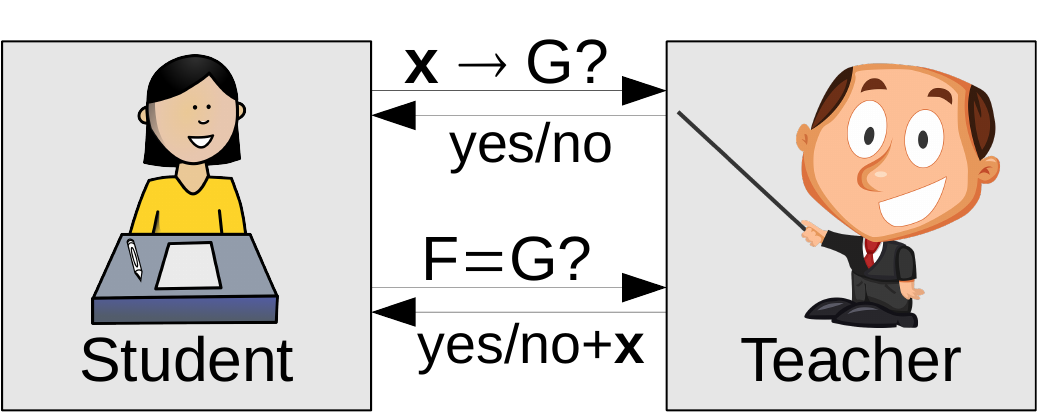}
\caption{Student and teacher in query learning.}
\label{fig:learn}
\end{wrapfigure}
The goal of query learning is to compute a small representation $F$ of a 
propositional formula $G(\overline{x})$ over a given set $\overline{x}$ of 
Boolean variables.  As illustrated in Figure~\ref{fig:learn}, this is achieved 
by two parties in interaction: the \emph{student} (or learner) and the 
\emph{teacher} (or oracle).  The student can ask two kinds of questions:
\begin{compactitem}
\item A \emph{subset query} asks if a given (potentially incomplete) cube 
$\mathbf{x}$ is fully contained in $G(\overline{x})$, i.e., if the implication 
$\mathbf{x} \rightarrow G$ holds.  The answer to this question is either 
\textsf{yes} or \textsf{no}.  In algorithms, we will denote such queries by 
$\SU(\mathbf{x}, G)$.
\end{compactitem}

\begin{compactitem}
\item An \emph{equivalence query} asks if a given 
candidate formula 
$F(\overline{x})$ is equivalent to $G(\overline{x})$.  The answer is again 
either 
\textsf{yes} or \textsf{no}.  However, in the \textsf{no}-case, the teacher 
also 
returns a \emph{counterexample} $\mathbf{x}$ in form of an 
$\overline{x}$-minterm witnessing the difference.  A counterexample is either a 
\emph{false-positive} with $\mathbf{x}\models F$ and $\mathbf{x}\not\models G$ 
or a \emph{false-negative} with $\mathbf{x}\not\models F$ and 
$\mathbf{x}\models G$.  In algorithms, we will denote equivalence queries by 
$\EQ(F, G)$.
\end{compactitem}
A \emph{membership query} is a special form of a 
subset query where 
$\mathbf{x}$ 
is an $\overline{x}$-minterm, i.e., a complete cube.

\subsubsection{Learning Algorithms}

The general pattern for query learning algorithms is that they start with some 
initial ``guess'' of the target function.  In a loop, they then perform 
equivalence queries.  If counterexamples are returned, the guess of the target 
function is refined to eliminate the counterexample.  The refinement may 
involve 
membership- and subset queries, and distinguishes the algorithms.  
Concrete algorithms are presented in the following.

\begin{figure}
\begin{minipage}{0.49\textwidth}
\begin{algorithm}[H]
\caption[\textsc{DnfLearn}: A DNF learning algorithm]
{\textsc{DnfLearn}: A DNF learning algorithm.}
\label{alg:DnfLearn}
\begin{algorithmic}[1]
\ProcedureRetL{DnfLearn}
             {G(\overline{x})}
             {A DNF representation $F(\overline{x})$ of $G(\overline{x})$}
             {3cm}
  \State $F := \false$
  \While{$\EQ(F,G)$ returns a counterexample $\mathbf{x}$} 
         \label{alg:DnfLearn:3}
    \State $\mathbf{x}_g := \mathbf{x}$
    \For{each literal $l$ in $\mathbf{x}$}
      \If{$\SU(\mathbf{x}_g\setminus \{l\}, G)$}
        \State $\mathbf{x}_g := \mathbf{x}_g\setminus \{l\}$
      \EndIf
    \EndFor
    \State $F := F \vee \mathbf{x}_g$
  \EndWhile
  \State \textbf{return} $F$
\EndProcedure  
\end{algorithmic}
\end{algorithm}
\end{minipage}
\hspace{1mm}
\begin{minipage}{0.49\textwidth}
\begin{algorithm}[H]
\caption[\textsc{CnfLearn}: A CNF learning algorithm]
{\textsc{CnfLearn}: A CNF learning algorithm.}
\label{alg:CnfLearn}
\begin{algorithmic}[1]
\ProcedureRetL{CnfLearn}
             {G(\overline{x})}
             {A CNF representation $F(\overline{x})$ of $G(\overline{x})$}
             {3cm}
  \State $F := \true$
  \While{$\EQ(F,G)$ returns a counterexample $\mathbf{x}$} 
         \label{alg:CnfLearn:3}
    \State $\mathbf{x}_g := \mathbf{x}$
    \For{each literal $l$ in $\mathbf{x}$}
      \If{$\SU(\mathbf{x}_g\setminus \{l\}, \neg G)$}
        \State $\mathbf{x}_g := \mathbf{x}_g\setminus \{l\}$
      \EndIf
    \EndFor
    \State $F := F \wedge \neg \mathbf{x}_g$
  \EndWhile
  \State \textbf{return} $F$
\EndProcedure  
\end{algorithmic}
\end{algorithm}
\end{minipage}
\end{figure}

\mypara{Learning a \acs{DNF}.}
\textsc{DnfLearn}~\cite[Chapter 7]{Crama:2010:BMM:1855106} in 
Algorithm~\ref{alg:DnfLearn} computes a \acs{DNF} representation 
of a given formula $G(\overline{x})$ using equivalence- and subset 
queries. It starts with the initial guess $F = \false$.  This guess is then 
refined based on the counterexamples returned by the equivalence 
queries in Line~\ref{alg:DnfLearn:3}.  The algorithm maintains the invariant 
$F\rightarrow G$.  Hence, a counterexample $\mathbf{x}$ can only be a 
false-negative, i.e., $\mathbf{x}\not\models F$ but $\mathbf{x}\models G$.  In 
principle, the counterexample $\mathbf{x}$ can be eliminated by updating $F$ to 
$F \vee \mathbf{x}$ without executing the inner \textbf{for}-loop.  However, in 
order to (potentially) reduce the number of iterations and also the size of 
$F$, 
the counterexamples are generalized:  The inner loop drops literals from the 
cube $\mathbf{x}$ as long as the reduced cube $\mathbf{x}_g$ still implies $G$, 
i.e., represents only variable assignments that must be mapped to $\true$ in 
the 
end. Thus, the subsequent update $F := F \vee \mathbf{x}_g$ does not only 
eliminate the original counterexample $\mathbf{x}$, but may also eliminate many 
other counterexamples that have not been encountered yet.  Note that this inner 
loop actually computes an unsatisfiable core $\mathbf{x}_g := 
\propsatmincore(\mathbf{x}, \neg G)$.  If no more counterexamples are left, the 
algorithm terminates and returns $F$, which is a disjunction of cubes, i.e., a 
\acs{DNF} that is equivalent to $G$.

\mypara{Learning a \acs{CNF}.}
A \ac{CNF} representation of a given formula $G(\overline{x})$ can be computed 
with $F = \neg\textsc{DnfLearn}(\neg G)$, i.e., by computing a \acs{DNF} for 
$\neg G$ and negating the result.  Alternatively, the procedure 
\textsc{DnfLearn} can easily be rewritten to compute \acp{CNF} directly.  This 
is shown in Algorithm~\ref{alg:CnfLearn}.  The working principle remains the 
same, but $F$ is initialized to $\true$ and refined with clauses that are 
computed from the false-positives returned by the equivalence queries.

More query learning algorithms can be found in the literature.  For instance, an 
algorithm to learn formulas in form of a conjunction of \acsp{DNF} can be 
defined using Bshouty's \emph{monotone} theory~\cite{Bshouty95}.  Ehlers et 
al.~\cite{EhlersKH12} show how various learning algorithms can be used 
effectively in circuit synthesis using \acp{BDD}.  In this article we focus on 
satisfiability-based synthesis methods.  \acs{SAT}- and \acs{QBF} solvers 
operate on 
\acs{CNF} representations of a formula.  Hence, our algorithms will mostly rely 
on the \acs{CNF} learning approach.  We therefore refrain from introducing more 
complicated learning methods here in detail, and refer the interested reader to 
the book by Crama and Hammer~\cite[Chapter 7]{Crama:2010:BMM:1855106}.

\subsection{\acf{CEGIS}}\label{sec:prelim:cegis}

The basic principle of query learning, namely refining an initial ``guess'' 
of 
the solution iteratively based on counterexamples, has also been applied to 
other synthesis-related problems.  One example is 
\acf{CEGIS}~\cite{Solar-Lezama13, lezamaphd}, which 
was introduced in the 
context of program sketching as a method to compute satisfying assignments for 
quantified formulas of the form
$\exists \overline{e} \scope \forall \overline{u} \scope 
  F(\overline{e}, \overline{u}).$
The goal is to compute concrete values $\mathbf{e}$ for the variables 
$\overline{e}$ such that $\forall \overline{u} \scope F(\mathbf{e}, 
\overline{u})$ holds. While the general principle is independent of the logic, 
we will assume that $F$ is a propositional formula.  Hence, $\overline{e}$ and 
$\overline{u}$ are vectors of Boolean variables, and we can use a \acs{SAT} 
solver to reason about $F$ (without the quantifiers).  

\begin{wrapfigure}[9]{r}{0.56\textwidth}
\vspace{-4mm}
\centering
\includegraphics[width=0.55\textwidth]{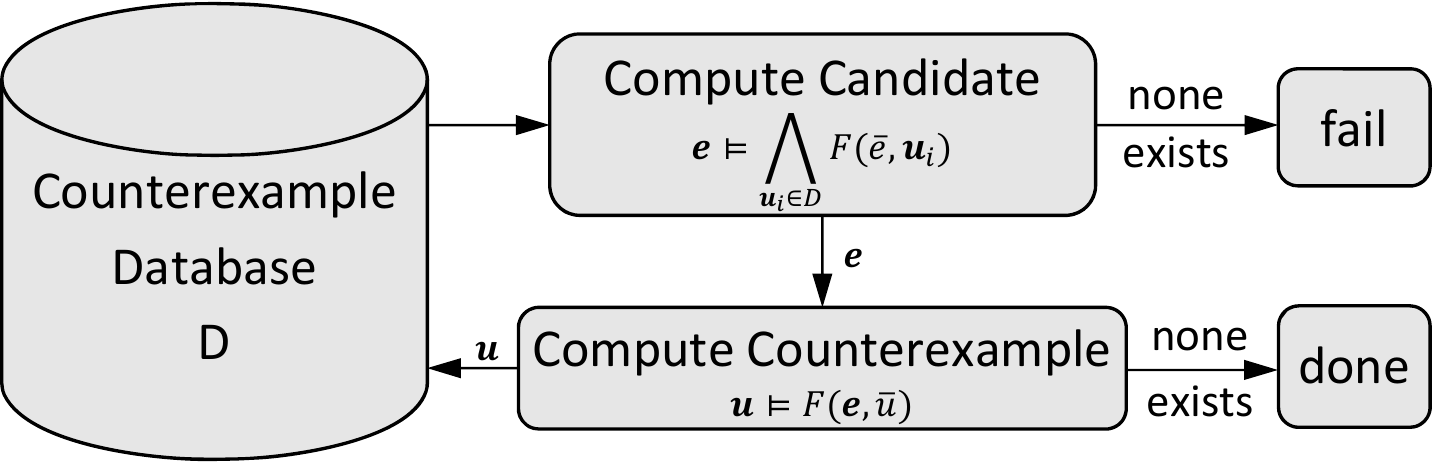}
\caption{Working principle of \acs{CEGIS}.}
\label{fig:cegis}
\end{wrapfigure}
\mypara{Working principle.}  Similar to query learning, a candidate $\mathbf{e}$ 
for a solution is iteratively refined based on counterexamples, which are 
concrete assignments to the variables $\overline{u}$ witnessing that $\forall 
\overline{u} \scope F(\mathbf{e}, \overline{u})$ does not yet hold.  This 
refinement loop is illustrated in Figure~\ref{fig:cegis}. There is a database 
$D$ of counterexamples $\mathbf{u}_i$, which is initially empty.  The first step 
of the loop is to compute a candidate assignment $\mathbf{e}\models 
\bigwedge_{\mathbf{u}_i\in D} F(\overline{e}, \mathbf{u}_i)$ that satisfies $F$ 
for all counterexamples that have been encountered previously.  This is a 
necessary but not a sufficient condition for $\forall \overline{u} \scope 
F(\mathbf{e}, \overline{u})$.  Hence, if no such candidate $\mathbf{e}$ exists, 
this means that $\exists \overline{e} \scope \forall \overline{u} \scope 
F(\overline{e}, \overline{u})$ is unsatisfiable, so the algorithm aborts.  If a 
candidate $\mathbf{e}$ was found, the next step is to check if $F(\mathbf{e}, 
\overline{u})$ holds \emph{for all} $\overline{u}$ and not just for the concrete 
$\overline{u}$-values stored in $D$.  This check is performed by searching for a 
counterexample $\mathbf{u}\models \neg F(\mathbf{e}, \overline{u})$ for which 
$F$ does not (yet) hold with the given~$\mathbf{e}$.  If no such counterexample 
exists, then $\mathbf{e}$ must be a solution, and the algorithm terminates.
Otherwise, the counterexample $\mathbf{u}$ is added to $D$ and another iteration 
is performed.  The candidate that is computed in the next iteration is already 
``better'' in the sense that it satisfies $F$ also for the counterexample from 
the previous iteration (and all iterations before).  For a propositional formula 
$F$ over finite vectors $\overline{e}$ and $\overline{u}$ of Boolean variables,
the \acs{CEGIS} algorithm must terminate eventually.  The reason is that every 
iteration excludes (at least) one candidate.  Moreover, there is only a finite 
set of counterexamples to encounter.

\begin{wrapfigure}[11]{r}{0.59\textwidth}
\vspace{-8mm}
\centering
\begin{minipage}{0.58\textwidth}
\begin{algorithm}[H]
\caption{\textsc{CegisSat}: \acs{CEGIS} implemented using a \acs{SAT} solver.}
\label{alg:CegisSmt}
\begin{algorithmic}[1]
\ProcedureRetL{CegisSat}
             {F(\overline{e}, \overline{u})}
             {An assignment $\mathbf{e}$ for $\overline{e}$ such 
              that $\forall \overline{u} \scope 
              F(\mathbf{e}, \overline{u})$ or~``\textsf{fail}''}
             {4.5cm}
  \State $G(\overline{e}) := \true$ \label{alg:CegisSmt:init}
  \While{$\true$}
    \If{$\textsf{sat} = \false$ in
        $(\textsf{sat},\mathbf{e}) := 
            \propsatmodel\bigl(G(\overline{e}) \bigr)$
       }\label{alg:CegisSmt:comp}
      \State \textbf{return} ``\textsf{fail}''
    \EndIf
    \If{$\textsf{sat} = \false$ in
        $(\textsf{sat},\mathbf{u}) := 
           \propsatmodel\bigl(\neg F(\mathbf{e}, \overline{u}) \bigr)$
       }\label{alg:CegisSmt:check}
      \State \textbf{return} $\mathbf{e}$
    \EndIf
    \State $G(\overline{e}) := G(\overline{e}) \wedge 
                               F(\overline{e}, \mathbf{u})$
                               \label{alg:CegisSmt:refine}
  \EndWhile
\EndProcedure  
\end{algorithmic}
\end{algorithm}
\end{minipage}
\end{wrapfigure}
\mypara{Algorithm.}  Algorithm~\ref{alg:CegisSmt} implements \acs{CEGIS} using a 
\acs{SAT} solver.  Line~\ref{alg:CegisSmt:comp} computes candidates 
and Line~\ref{alg:CegisSmt:check} performs the candidate check as 
well as the counterexample computation in the straightforward way.  Instead of 
storing a database of counterexamples, the algorithm directly refines the 
constraints for a candidate in Line~\ref{alg:CegisSmt:refine}.  Note that 
constraints are only added to $G$, so the algorithm is well suited for 
incremental solving.

%% file: fig/ex_bdd.tex
\begin{tikzpicture}[node distance=5cm,auto]
\tikzstyle{every edge}=[->,draw]
\node[circle,draw,initial,initial text={$f$}]    at  (0,-1)  (z) {$z$};
\node[circle,draw]            at  (0,-2)    (x) {$x$};
\node[circle,draw]            at  (0,-3)  (y) {$y$};
\node[rectangle,draw]         at  (-0.75,-4)   (t0){$0$};
\node[rectangle,draw]         at  (0.75,-4)    (t1){$1$};

\path
(z) edge[bend right]
    node[rotate=50,xshift=8mm,yshift=6.3mm] {$\true$} 
    (t0)
(z) edge[dashed]                 
    node[rotate=90,xshift=-2.5mm,yshift=1.5mm] {$\false$} 
    (x)
(x) edge[dashed]
    node[rotate=90,xshift=-2.5mm,yshift=1.5mm] {$\false$} 
    (y)
(x) edge[bend left]                 
    node[rotate=-50,xshift=-9mm,yshift=-0.5mm] {$\true$} 
    (t1)
(y) edge[dashed]
    node[rotate=55,xshift=-2mm,yshift=4mm] {$\false$} 
    (t0)
(y) edge                 
    node[rotate=-55,xshift=-3mm,yshift=-1mm] {$\true$} 
    (t1)

;
\end{tikzpicture}

%% file: 03strategy.tex
\section{From Safety Specifications to Strategies}
\label{sec:hw:win}

As discussed in Section~\ref{sec:prelim:sasyalg}, a strategy $S$ for 
realizing a safety specification $\mathcal{S} = \bigl(\overline{x}, 
\overline{i}, \overline{c}, I(\overline{x}), $ $ 
T(\overline{x},\overline{i},\overline{c},\overline{x}'), P(\overline{x})\bigr)$ 
can be constructed by computing the winning region $W(\overline{x})$ in the game 
defined by~$\mathcal{S}$.  Recall that the winning region is the set of all 
states from which the system player can enforce that only safe states are 
visited. Once 
the winning region is available, the corresponding strategy can be defined as 
$S(\overline{x},\overline{i},\overline{c},\overline{x}') =   
T(\overline{x},\overline{i},\overline{c},\overline{x}') \wedge \bigl(   
W(\overline{x}) \rightarrow W(\overline{x}')\bigr).
$
However, a winning strategy can also be computed by different means.  One 
option is to use a winning area, defined as follows.
\begin{definition}[Winning Area]\label{def:winset}
A \emph{winning area} \index{winning area} for a safety specification 
$\mathcal{S} = 
\bigl(\overline{x}, \overline{i}, \overline{c}, I(\overline{x}),  
T(\overline{x},\overline{i},\overline{c},\overline{x}'), $ $ 
P(\overline{x})\bigr)$ is a state set $F$, represented symbolically as a formula 
$F(\overline{x})$, with the following three properties:
\begin{compactitem}
\item Every initial state is contained in $F$, i.e., $I(\overline{x}) 
\rightarrow F(\overline{x})$.
\item $F$ contains only safe states, i.e., $F(\overline{x}) \rightarrow 
P(\overline{x})$.
\item The system player can enforce that the play stays in $F$, i.e.,
$F(\overline{x}) \rightarrow \FS\bigl(F(\overline{x})\bigr)$.
\end{compactitem}
\end{definition}
These properties are sufficient to ensure that 
$T(\overline{x},\overline{i},\overline{c},\overline{x}') \wedge \bigl(   
F(\overline{x}) \rightarrow F(\overline{x}')\bigr)$ is a winning strategy.  The 
reason is the same as for the winning region 
(Section~\ref{sec:prelim:winstrat}): the control signals can always be set 
such that the next state is in~$F$ again, and $F$ contains only safe 
states.  In fact, the winning region is just a special winning area, namely the 
largest~one.  

The following sections will present different methods for computing the winning 
region or a winning area using decision procedures for the satisfiability of 
formulas.  We will use the terms ``\emph{satisfiability-based}'' or 
``\emph{SAT-based}'' to indicate the use of any such decision procedures, 
including SAT-, QBF- and EPR solvers. We will write ``\emph{SAT solver based}'' 
to specifically indicate the use of propositional SAT solvers.

\subsection{\acs{QBF}-Based Learning} \label{sec:qbf_learn}

The \textsc{SafeWin} procedure presented in Algorithm~\ref{alg:SafeWin} can be 
implemented with \acp{BDD} using their capability of quantifier elimination in a 
rather straightforward manner.  However, a realization with plain \acs{SAT} 
solvers is not easily possible because the preimage operation $\FS$ in 
Line~\ref{alg:SafeWin:4} contains a universal quantification.  Therefore, a 
natural option is to use a \acs{QBF} solver, which can handle universal 
quantifications without expanding the formula.

\subsubsection{A Straightforward \acs{QBF} Realization of \textsc{SafeWin}}
\label{sec:hw:qbf:sf}

A direct realization of \textsc{SafeWin} with \acs{QBF} solving was presented by 
Staber and Bloem~\cite{StaberB07}.  We briefly review this existing method 
and its drawbacks before presenting our learning-based algorithms.  For this 
discussion, we will refer to the different values of the variable $F$ in 
Algorithm~\ref{alg:SafeWin} with indices.  That is, $F_0=P$ denotes the initial 
value of $F$ and $F_j = F_{j-1} \wedge \FS(F_{j-1})$ is the value after the 
$j$\ts{th} iteration.  The termination check in Line~\ref{alg:SafeWin:3} is 
performed by checking two subsequent values $F_j$ and $F_{j-1}$ for 
equivalence.  Since $F_j \rightarrow F_{j-1}$, i.e., the set $F$ of states 
can only get smaller from iteration to iteration, it is sufficient to check if
$F_{j-1} \rightarrow F_{j}$.  Thus, the first check of ``F changes'' can be
realized with the \ac{QBF} query
$ \neg \qbfsat\bigl(
\forall \overline{x},\overline{i} \scope
\exists \overline{c},\overline{x}' \scope
 P(\overline{x}) \rightarrow 
 \bigl(
 T(\overline{x},\overline{i}, \overline{c},\overline{x}') \wedge 
 P(\overline{x}')
 \bigr)
\bigr).$
The second check if $F$ changes translates to
$ \neg \qbfsat\bigl(
\forall \overline{x},\overline{i} \scope
\exists \overline{c},\overline{x}' \scope
\forall \overline{i}' \scope
\exists \overline{c}',\overline{x}'' \scope
 \bigl(
 P(\overline{x}) \wedge 
 T(\overline{x},\overline{i}, \overline{c},\overline{x}') \wedge 
 P(\overline{x}')
 \bigr)
 \rightarrow 
 \bigl(
 T(\overline{x}',\overline{i}', \overline{c}',\overline{x}'') \wedge 
 P(\overline{x}'')
 \bigr)
\bigr),$ and so on.
In general, the check if $F$ changed in iteration $j$ requires solving a 
\ac{QBF} with $2\cdot j-1$ quantifier alternations and $j$ copies of the 
transition relation $T$.  The checks if $I\rightarrow F_j$ in 
Line~\ref{alg:SafeWin:5} of Algorithm~\ref{alg:SafeWin} work in a similar way, 
also requiring $2\cdot j-1$ quantifier alternations and $j$ copies of the 
transition relation.  We consider this steep increase in formula size and 
complexity as suboptimal.  In the following, we will therefore present 
algorithms that require only one copy of the transition relation and a constant 
number of quantifier alternations in the queries to the \ac{QBF} solver.

\subsubsection{A \acs{QBF}-Based CNF Learning Algorithm}

\begin{algorithm}[tb]
\caption[\textsc{QbfWin}: Basic \acs{QBF}-based \acs{CNF} learning algorithm for
the winning region]
{\textsc{QbfWin}: Basic \acs{QBF}-based \acs{CNF} learning algorithm for
the winning region.}
\label{alg:QbfWin}
\begin{algorithmic}[1]
\ProcedureRet{QbfWin}
             {(\overline{x}, \overline{i}, \overline{c}, I, T, P)}
             {The winning region $W(\overline{x})$ in CNF or $\false$}
  \LineIf{$\propsat\bigl(I(\overline{x}) \wedge \neg P(\overline{x})\bigr)$}
    {\textbf{return} $\false$}
    \label{alg:QbfWin:init0}
  \State $F(\overline{x}) := P(\overline{x})$ \label{alg:QbfWin:f0}
  \While{$\mathsf{sat}=\true$ in 
    $(\mathsf{sat},\mathbf{x}) := \qbfsatmodel\bigl(
      \exists \overline{x},\overline{i} \scope
      \forall \overline{c} \scope
      \exists \overline{x}' \scope
      F(\overline{x}) \wedge 
      T(\overline{x}, \overline{i}, \overline{c}, \overline{x}') \wedge 
      \neg F(\overline{x}')\bigr)$}
      \label{alg:QbfWin:check}
    \State $\mathbf{x}_g := \mathbf{x}$\label{alg:QbfWin:loop0}
    \For{each literal $l$ in $\mathbf{x}$}\label{alg:QbfWin:loops}
      \State $\mathbf{x}_t := \mathbf{x}_g \setminus \{l\}$
      \If{$\neg\qbfsat\bigl(
           \exists \overline{x} \scope
           \forall \overline{i} \scope
           \exists \overline{c},\overline{x}' \scope
           \mathbf{x}_t \wedge 
           F(\overline{x}) \wedge 
           T(\overline{x}, \overline{i}, \overline{c}, \overline{x}') \wedge 
           F(\overline{x}')
           \bigr)$}\label{alg:QbfWin:gen}
        \State $\mathbf{x}_g := \mathbf{x}_t$
      \EndIf
    \EndFor\label{alg:QbfWin:loop1}
    \LineIf{$\propsat\bigl(\mathbf{x}_g \wedge I(\overline{x}) \bigr)$}
    {\textbf{return} $\false$}
    \label{alg:QbfWin:init1}
    \State $F(\overline{x}) := F(\overline{x}) \wedge \neg \mathbf{x}_g$
    \label{alg:QbfWin:ref}
  \EndWhile
  \State \textbf{return} $F(\overline{x})$
\EndProcedure  
\end{algorithmic}
\end{algorithm}

Algorithm~\ref{alg:QbfWin} shows the procedure \textsc{QbfWin}, which computes a 
\ac{CNF} representation of the winning region $W(\overline{x})$ using \acs{CNF} 
learning with a \acs{QBF} solver.  Since \textsc{QbfWin} will also be the basis 
for our algorithms that use plain \acs{SAT} solving, we discuss it here in 
detail.  Just like \textsc{SafeWin} in Algorithm~\ref{alg:SafeWin}, 
\textsc{QbfWin} takes a specification as input.  It returns either the winning 
region $W(\overline{x})$ or $\false$ in case of unrealizability.  The basic 
structure is that of the \acs{CNF} learning procedure \textsc{CnfLearn} in 
Algorithm~\ref{alg:CnfLearn}.  However, in Line~\ref{alg:QbfWin:f0}, $F$ is 
initialized to $P$ instead of $\true$ because the winning region can only be a 
subset of the safe states $P$.  Differences in counterexample computation and 
generalization are discussed in the following.

\mypara{Counterexample computation.} The equivalence query in 
Line~\ref{alg:CnfLearn:3} of the original \acs{CNF} learning procedure 
\textsc{CnfLearn} asks if the current approximation of the solution is correct. 
The corresponding line (Line~\ref{alg:QbfWin:check}) in \textsc{QbfWin} now 
checks if $F \rightarrow \FS(F)$ is valid, i.e., if another visit of $F$ can be 
enforced by the system from any state of $F$.  The \ac{QBF} query in 
Line~\ref{alg:QbfWin:check} of \textsc{QbfWin} actually asks the opposite 
question, namely if there exists a state $\mathbf{x}$ in $F$ from which the 
environment can enforce leaving $F$, i.e., if $F\wedge \FE(\neg F)$ is 
satisfiable. This is the case if there exists some state $\mathbf{x}$ in $F$ and 
some input $\mathbf{i}$ such that for all control values $\mathbf{c}$ the next 
state will be in $\neg F$.  If such a state $\mathbf{x}$ exists, 
$\qbfsatmodel$ will return it as a counterexample witnessing that $F$ is not
equal to the winning region $W$.  More specifically, this 
state $\mathbf{x}$ cannot be part of $W$, and thus needs to be 
removed from $F$.  This removal is performed in Line~\ref{alg:QbfWin:ref}.  
However, in order to reduce the number of iterations, the counterexample is 
generalized beforehand.  This is explained in the next paragraph.  
If, on the other hand, $\qbfsatmodel$ sets \textsf{sat} to $\false$ in 
Line~\ref{alg:QbfWin:check}, then this means that the implication $F 
\rightarrow \FS(F)$ holds.  In this case,  \textsc{QbfWin} terminates, 
returning $F$ as the winning region.

\mypara{Counterexample generalization.} Just like in \textsc{CnfLearn}, 
counterexample generalization is done by eliminating literals of $\mathbf{x}$ in 
the inner loop of the algorithm.  In \textsc{CnfLearn} (see 
Algorithm~\ref{alg:CnfLearn}), the final 
cube $\mathbf{x}_g \subseteq \mathbf{x}$ must not intersect with $G$ in order 
not to shrink $F$ beyond $G$.  Similarly, in \textsc{QbfWin}, 
$\mathbf{x}_g \wedge F$ must not intersect with $\FS(F)$ in order not to remove 
any states from the winning region where the system could enforce that the 
play stays in the winning region.  The reason is that the subsequent update $F 
:= F \wedge \neg 
\mathbf{x}_g$ in Line~\ref{alg:QbfWin:ref} removes exactly the states 
$\mathbf{x}_g\wedge F$.  The \ac{QBF} query in Line~\ref{alg:QbfWin:gen} is 
satisfiable if $\mathbf{x}_t\wedge F$ contains any states of $\FS(F)$, and thus 
prevents unjust state removals.  Also note that the inner loop essentially 
computes an unsatisfiable core of $\mathbf{x}$ with respect to $F \wedge 
\FS(F)$.

\mypara{Detecting unrealizability.}  Detecting unrealizability is simple.  The 
specification is unrealizable if and only if some initial state is outside of 
the winning region, i.e., if $I\not\rightarrow W$.  The reason is that no system 
implementation can prevent the environment from visiting an unsafe state from an 
initial state that is not winning.  \textsc{QbfWin} returns $\false$ as soon as 
$I\not\rightarrow F$.  Since $F=W$ eventually, this ensures that $\false$ is 
returned if $I\not\rightarrow W$.  Line~\ref{alg:QbfWin:init0} checks if 
$I\not\rightarrow F$ would hold initially.  In every iteration, 
Line~\ref{alg:QbfWin:init1} then checks if the states $\mathbf{x}_g$ that are 
going to be removed from $F$ contain an initial state.  This is potentially 
more efficient than than checking $I\not\rightarrow F$ again.

\begin{wrapfigure}[9]{r}{0.49\textwidth}
\vspace{-7mm}
\centering 
\subfloat[Counterexample computation.\label{fig:qbf_learn1}]
  {\includegraphics[width=0.16\textwidth]{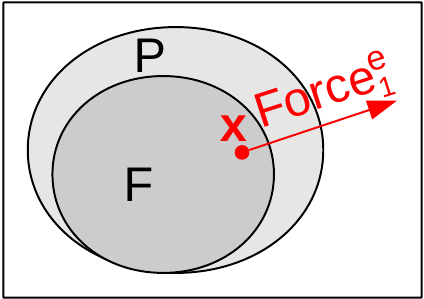}}
\hspace{0mm}
\subfloat[Generalization.\label{fig:qbf_learn2}]
  {\includegraphics[width=0.16\textwidth]{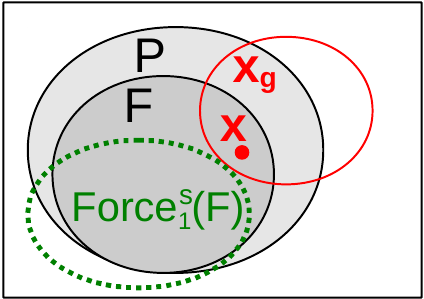}}
\hspace{0mm}
\subfloat[Update of $F$.\label{fig:qbf_learn3}]
  {\includegraphics[width=0.16\textwidth]{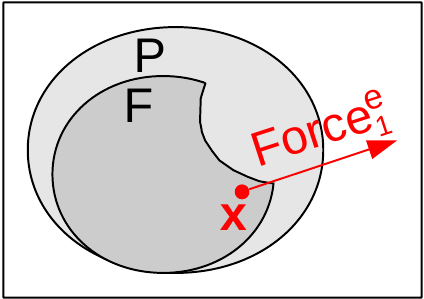}}
\caption{Working principle of \textsc{QbfWin}.}
\label{fig:qbf_learn}
\end{wrapfigure}
\mypara{Illustration.}
Figure~\ref{fig:qbf_learn} illustrates the working principle of \textsc{QbfWin} 
graphically.  A box represents the set of all states.  $F$ is always a subset of 
$P$.  In Figure~\ref{fig:qbf_learn1}, a counterexample $\mathbf{x}\models 
F\wedge \FE(\neg F)$ is computed.  It represents a state from which the 
environment can enforce that $F$ is left.  Next, the counterexample 
$\mathbf{x}$ is 
generalized into a larger region $\mathbf{x}_g$ by eliminating literals, as 
illustrated in Figure~\ref{fig:qbf_learn2}.  Every literal that can be 
eliminating from $\mathbf{x}$ doubles the size of the state region 
that is represented by $\mathbf{x}_g$.  Literals are dropped as long as 
$\mathbf{x}_g\wedge F$ does not intersect with $\FS(F)$.  Finally, as 
illustrated in Figure~\ref{fig:qbf_learn3}, the generalized counterexample 
$\mathbf{x}_g$ is removed from $F$ and the next counterexample is computed.  
This is repeated until no more counterexamples exist, or one of the initial 
states is removed.

\noindent
The following theorem summarizes these explanations into a formal correctness 
argument.

\begin{theorem}\label{thm:QbfWin_correct}
The \textsc{QbfWin} procedure in Algorithm~\ref{alg:QbfWin} returns the winning 
region $W(\overline{x})$ of a given safety specification $\mathcal{S}$, or 
$\false$ if the specification is unrealizable.
\end{theorem}
\begin{proof}
\textsc{QbfWin} enforces the invariants $F\rightarrow P$ (through 
Lines~\ref{alg:QbfWin:f0} and~\ref{alg:QbfWin:ref}) and $I\rightarrow F$ 
(through Lines~\ref{alg:QbfWin:init0} and~\ref{alg:QbfWin:init1}).  The loop 
terminates normally if $F\rightarrow \FS(F)$.  Hence, upon normal termination, 
$F$ is certainly a winning area according to Definition~\ref{def:winset}.  $F$ 
is 
also the largest possible winning area, and thereby the winning region, because 
\textsc{QbfWin} also enforces the invariant $W\rightarrow F$.  This invariant 
can be proven by induction: Initially $F=P$, so $W\rightarrow F$ holds because 
$W\rightarrow P$. Under the hypothesis that $W\rightarrow F$ holds before an 
update of $F$ in Line~\ref{alg:QbfWin:ref}, it will also hold after the update 
because Line~\ref{alg:QbfWin:ref} only removes states $\mathbf{x}_g\wedge F$ for 
which $\mathbf{x}_g\wedge F \rightarrow \FE(\neg F)$ holds.  Given that 
$W\rightarrow F$, we have that $\neg F \rightarrow \neg W$.  This means that 
$\mathbf{x}_g\wedge F \rightarrow \FE(\neg W)$, so only states that cannot be 
part of $W$ are removed.   \textsc{QbfWin} will always terminate because in 
every iteration, at least one state is removed from $F$, and when $F$ reaches 
$\false$ (or earlier) the loop necessarily terminates.  What remains to be shown 
is that \textsc{QbfWin} aborts in Line~\ref{alg:QbfWin:init0} 
or~\ref{alg:QbfWin:init1} iff $\mathcal{S}$ is unrealizable, i.e., iff 
$I\not\rightarrow W$. (Direction $\Rightarrow$:) Since $W\rightarrow F$, and 
Line~\ref{alg:QbfWin:init0} or~\ref{alg:QbfWin:init1} abort iff 
($F$ is about to be updated in such a way that) $I\not\rightarrow F$, it 
follows that \textsc{QbfWin} can only abort if $I\not\rightarrow W$.  
(Direction $\Leftarrow$:) Since $F=W$ eventually, Line~\ref{alg:QbfWin:init0} 
or~\ref{alg:QbfWin:init1} will definitely abort eventually if $I\not\rightarrow 
W$. 
\qed
\end{proof}

\mypara{Discussion.}  In contrast to the approach from 
Section~\ref{sec:hw:qbf:sf}, all \acs{QBF} queries in \textsc{QbfWin} contain 
only one copy of the transition relation and only two quantifier alternations. 
This potentially increases the scalability with respect to the size of the 
specifications.  The disadvantage is that the number of calls to the \acs{QBF} 
solver can be significantly higher.

\subsubsection{Variants and Improvements}\label{sec:qbf_var}

\noindent
In this section, we now discuss a few variants and optimizations of 
\textsc{QbfWin} as presented in Algorithm~\ref{alg:QbfWin}.

\mypara{Better generalization.} At any point in the inner loop of 
\textsc{QbfWin}, $\mathbf{x}_g$ represents states that will definitely be 
removed from $F$.  This information can be exploited already during the 
generalization loop by modifying the \acs{QBF} query in 
Line~\ref{alg:QbfWin:gen} to
$\neg\qbfsat\bigl(
           \exists \overline{x} \scope
           \forall \overline{i} \scope
           \exists \overline{c},\overline{x}' \scope
           \mathbf{x}_t \wedge 
           F(\overline{x}) \wedge \neg \mathbf{x}_g \wedge
           T(\overline{x}, \overline{i}, \overline{c}, \overline{x}') \wedge 
           F(\overline{x}') \wedge \neg \mathbf{x}_g'
           \bigr).$
This way, the generalization loop behaves as if $F$ would have been refined to 
$F(\overline{x}) \wedge \neg \mathbf{x}_g$ already (with the current version of 
$\mathbf{x}_g$).  The \acs{QBF} query becomes stricter, which can have the 
effect that more literals can be eliminated.  This can reduce the total number 
of counterexamples that have to be resolved.  In the illustration of 
Figure~\ref{fig:qbf_learn2}, this optimization shrinks $\FS(F)$ to $\FS(F\wedge 
\mathbf{x}_g)$, which allows $\mathbf{x}_g$ to grow even larger.  Since this 
optimization does not increase the number or complexity of the 
\acs{QBF} queries, we always apply it.

\mypara{Generalization until fixpoint.}  With the generalization optimization 
from the previous paragraph, the generalization check becomes non-monotonic in 
the sense that, even if a literal could not be eliminated initially, it may be 
eliminable after eliminating other literals.  Hence, it can be beneficial to 
repeat the generalization loop until a fixpoint is reached.  However, in our 
experiments, this did not result in noticeable performance improvements on 
the average over our benchmarks, so this is not done by default.

\begin{wrapfigure}[9]{r}{0.17\textwidth}
\vspace{-4mm}
\centering 
  \includegraphics[width=0.16\textwidth]{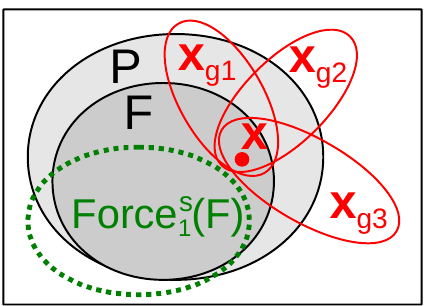}
\caption{Computing all counterexample generalizations in \textsc{QbfWin}.}
\label{fig:qbf_allgen}
\end{wrapfigure}
\mypara{Computing all counterexample generalizations.}  In our experiments we 
observed that 
counterexample computation often takes much more time than counterexample 
generalization.  Moreover, depending on the order in which the literals $l \in 
\mathbf{x}$ are processed in Line~\ref{alg:QbfWin:loops} of \textsc{QbfWin}, we 
can get different generalizations $\mathbf{x}_g$.  Motivated by these 
observations, we propose a variant that computes \emph{all} minimal 
generalizations for each counterexample.  A naive solution would just run the 
generalization loop of Line~\ref{alg:QbfWin:loops} repeatedly using all 
$|\mathbf{x}|!$ different orders of the literals in $\mathbf{x}$.  However, 
since many orderings can result in the same generalization $\mathbf{x}_g$, this 
is potentially inefficient.  Instead, we thus apply an adaption of the hitting 
set tree algorithm presented by Reiter~\cite{Reiter87}.  For the sake of 
readability, we refrain from presenting this algorithm in detail.  The 
high-level intuition is visualized in Figure~\ref{fig:qbf_allgen}.  All 
generalizations $\mathbf{x}_{g1}$, $\mathbf{x}_{g2}$ and $\mathbf{x}_{g3}$ will 
contain the original counterexample $\mathbf{x}$, and none of them may intersect 
with $\FS(F)$ inside of $F$.  Although there may be a significant overlap 
between the 
generalizations, removing all of them prunes $F$ more than removing just one of 
them.  In our experiments, we observed that the number of different 
counterexample generalizations is usually low.  Not infrequently, there is only 
exactly one minimal generalization.   Of course, computing all generalizations 
costs additional computation time.  In our experiments, it gives a solid speedup 
for some benchmarks, but slows down the computation for others.  Hence, we do 
not apply this optimization by default.
Instead of computing all generalizations, one could also compute and apply at 
most $k$ different generalizations for some value of $k$. Another option is to 
compute all generalizations but refine $F$ only with the $k$ shortest ones.  
However, in preliminary experiments, these variants did not result in 
significant performance increases either.

\subsubsection{Efficient Implementation} \label{sec:qbf_impl}

\noindent
In this section, we give a few remarks on implementing \textsc{QbfWin} 
efficiently.

\mypara{\ac{CNF} encoding.} The transition relation $T$, the characterization of 
the safe states $P$ and the formula for the initial states $I$ are transformed 
into \ac{CNF} initially. Furthermore, a \ac{CNF} representation of $\neg F$ 
needs to be computed in each iteration.  All these transformations can be done 
using the method of Plaisted and Greenbaum~\cite{PlaistedG86}.  This may 
introduce additional auxiliary variables, which are quantified existentially on 
the innermost level of the \ac{QBF} queries.  Once $T$, $P$, $I$ and $\neg F$ 
are available in \ac{CNF}, the matrices of the \ac{QBF} queries in 
Algorithm~\ref{alg:QbfWin} can be constructed by building the union of the 
respective clause sets, because the individual formula parts are all connected 
by conjunctions.

\mypara{\ac{CNF} compression.}  After some iterations, the \ac{CNF} formula $F$ 
in \textsc{QbfWin} can contain redundant clauses and literals.  First, a clause 
discovered in some later iteration can be a proper subset of some earlier 
discovered clause.  This can be checked syntactically at low costs.  Thus, 
whenever a clauses is added to $F$, we always remove all of its supersets.  
Second, a set of clauses may together imply clauses that have been added 
earlier.  The implied clauses can be eliminated without changing $F$ 
semantically.  Third, it may be possible to drop literals from clauses of $F$ in 
an equivalence-preserving manner.  The procedure \textsc{CompressCnf} in 
Algorithm~\ref{alg:CompressCnf} performs these simplifications and is explained 
in the next paragraph.  We call this procedure to simplify $F$ after every 
modification of $F$, but with literal dropping disabled (we will later use 
\textsc{CompressCnf} with literal dropping enabled in other contexts). 
\textsc{CompressCnf} is very fast compared to the \acs{QBF} solver calls in 
\textsc{QbfWin}.  Furthermore, a smaller \acs{CNF} representation of $F$ is 
particularly important for computing a compact representation of $\neg F$ using 
the method of Plaisted and Greenbaum~\cite{PlaistedG86}.  Ultimately, the more 
compact \acs{CNF} representations reduce the \acs{QBF} solving time quite 
significantly.

\begin{wrapfigure}[15]{r}{0.54\textwidth}
\vspace{-8mm}
\centering
\begin{minipage}{0.53\textwidth}
\begin{algorithm}[H]
\caption[\textsc{CompressCnf}: Removing redundant literals and clauses from a 
\acs{CNF}]
{\textsc{CompressCnf}: Removing redundant literals and clauses from a 
\acs{CNF}.}
\label{alg:CompressCnf}
\begin{algorithmic}[1]
\ProcedureRetL{CompressCnf}
             {F(\overline{x})}
             {An equivalent but potentially smaller \acs{CNF}~$G(\overline{x})$}
             {3cm}
  \If{dropping literals enabled}
    \State $G := \true$
    \For{each clause $c$ in $F$}
      \State $G := G \wedge \neg \propsatmincore(\neg c, F)$
    \EndFor
    \State $F := G$
  \EndIf
  \State $G := \true$
  \For{each clause $c$ in $F$ with increasing size}
    \If{$\propsat(G \wedge \neg c)$}
      \State $G := G \wedge c$
    \EndIf
  \EndFor
  \State \textbf{return} $G$
\EndProcedure  
\end{algorithmic}
\end{algorithm}
\end{minipage}
\end{wrapfigure}
\mypara{An algorithm for \ac{CNF} compression.}
Algorithm~\ref{alg:CompressCnf} uses a \acs{SAT} solver to remove redundant 
literals and clauses from a \ac{CNF} formula $F(\overline{x})$.  The first loop 
(if enabled) drops literals from each clause $c$ as long as the reduced clause 
$c_2\subseteq c$ is still implied by $F$.  This ensures that the reduced formula 
$G$ is implied by $F$.  Dropping literals can only make the formula stronger, 
i.e., $F$ is necessarily implied by $G$.  Hence, $G$ and $F$ are equivalent. 
Note that $F \rightarrow c_2$ iff $F  \wedge \neg c_2$ is unsatisfiable.  Hence, 
dropping the literals can be realized by computing a (minimal) unsatisfiable 
core of the cube $\neg c_2$ with respect to $F$.  Since $F$ does not change in 
this loop, all cores can be computed with incremental \acs{SAT} 
solving.

The second loop removes redundant clauses.  Non-redundant clauses are 
copied into $G$.  A clause $c$ is redundant if it is implied by $G$ already, 
i.e., if $G \wedge \neg c$ is unsatisfiable.  Clauses are processed in the order 
of increasing size because smaller clauses have a higher tendency to imply 
larger clauses than the other way around.  This second loop can also be 
accomplished with incremental solving, since clauses are only added to $G$.  
Dropping literals before eliminating clauses potentially yields better results 
than performing the operations in the reverse order. The reason is that the 
shorter clauses produced in the first loop have a higher potential for 
implying other clauses in the second loop.  Since none of the \acs{SAT} solver 
calls involves the transition relation, Algorithm~\ref{alg:CompressCnf} is 
usually very fast. It will not only be used in \textsc{QbfWin}, but also in 
other contexts.

\mypara{\ac{QBF} preprocessing.}  Using an extension~\cite{SeidlK14} of the 
popular \ac{QBF} preprocessor \bloqqer~\cite{BiereLS11} to preserve satisfying 
assignments, \ac{QBF} preprocessing can not only be applied in $\qbfsat$ but 
also in $\qbfsatmodel$ queries.  We thus perform \ac{QBF} preprocessing in every 
single \ac{QBF} query (separately).  The experimental results in 
Chapter~\ref{sec:hw:exp} will show that this is crucial for the performance.  
In 
a sense, running \textsc{CompressCnf} to simplify $F$, as explained in the 
previous paragraphs, can also be seen as \ac{QBF} preprocessing, but using 
knowledge about the structure of the final \ac{QBF}. \bloqqer~\cite{BiereLS11} 
implements way more simplification techniques, from heuristics for universal 
expansion to variable elimination, and is thus clearly not subsumed by running 
\textsc{CompressCnf}.  On the other hand, our experiments indicate that running 
\textsc{CompressCnf} in addition to \bloqqer is beneficial as well.  A possible 
reason is that we compress $F$ \emph{before} computing its negation.  This has 
advantages over applying simplifications on the final \acs{QBF}, where the 
structure is already lost.

\mypara{Incremental \ac{QBF} solving.}  We experimented with incremental 
\ac{QBF} solving using \depqbf~\cite{LonsingE14}.  We use two incremental solver 
instances, one for the queries in Line~\ref{alg:QbfWin:check} and one for 
Line~\ref{alg:QbfWin:gen} of \textsc{QbfWin}.  The queries in 
Line~\ref{alg:QbfWin:gen} are well suited for incremental solving because 
clauses are only added to $F$. The conjunction with $\mathbf{x}_t$ can be 
achieved with assumption literals, which are temporarily asserted.  In 
fact, we first let \depqbf compute an unsatisfiable core of $\mathbf{x}$ and 
minimize this core then further using a loop that attempts to eliminate more 
literals.

The check in Line~\ref{alg:QbfWin:check} of \textsc{QbfWin} is more difficult 
because it also contains the negation of the $F$, i.e., cannot be realized 
incrementally just by adding additional clauses.  We implemented three variants 
to handle $\neg F(\overline{x}')$ incrementally.  Since neither of these three 
variants performs particularly well in our experiments (see 
Chapter~\ref{sec:hw:exp}), we only sketch them briefly. The first variant uses 
the \textsf{push}/\textsf{pop} interface of \depqbf to replace the parts in the 
\ac{CNF} encoding of $\neg F(\overline{x}')$ that change from iteration to 
iteration.  The second variant updates $\neg F(\overline{x}')$ only lazily, 
namely when the check in Line~\ref{alg:QbfWin:check} becomes 
unsatisfiable.\footnote{This is similar to the procedure \textsc{SatWin1} that 
will be presented in Algorithm~\ref{alg:SatWin1} later.  We thus refer to 
Section~\ref{hw:sat1} for more details.}  In this event, a new incremental 
session of the solver is started with the latest version of $\neg 
F(\overline{x}')$.  The third variant uses a pool of variables to encode negated 
clauses in \ac{CNF}.  If a variable of this pool is not yet used, it is set to 
$\false$ using assumption literals.  Thereby, the variable essentially 
represents the negation of a tautological clause.  As clauses are added to $F$, 
the variables of the pool are equipped with constraints that make them represent 
the negation of the added clauses.  If there are no more unused variables in the 
pool, a new incremental session with a fresh pool of variables is started.  
Unfortunately, neither of these three variants performs particularly 
well in our experiments.  One reason is that incremental \ac{QBF} solving cannot 
be combined with preprocessing at the moment.  However, this may change in the 
future, which could make these approaches interesting again.

\subsection{Learning Based on \acs{SAT} Solving} \label{sec:sat_learn}

In this section, we present a learning algorithm that computes the winning
region of a safety specification $\mathcal{S} = (\overline{x}, \overline{i}, 
\overline{c}, I, T, P)$ using a plain \acs{SAT} solver.  To simplify
the presentation, this is done in two steps:  Section~\ref{hw:sat0} 
presents a basic algorithm.  Section~\ref{hw:sat1} will then discuss
a more efficient variant with better support for incremental solving.

\subsubsection{Basic Algorithm} \label{hw:sat0}

A basic solution is shown in Algorithm~\ref{alg:SatWin0}.  The working principle 
is the same as for the procedure \textsc{QbfWin} from 
Algorithm~\ref{alg:QbfWin}: starting with the initial over-approximation $F = P$ 
of the winning region $W$, counterexample-states $\mathbf{x} \models F\wedge 
\FE(\neg F)$ witnessing that $F\neq W$ are computed, generalized into a larger 
region $\mathbf{x}_g$ of states that cannot be part of the final winning region 
$W$, and finally removed from $F$.  Detecting unrealizability by checking if 
$I\not\rightarrow F$ is also done in exactly the same way as in \textsc{QbfWin}. 
Only the counterexample computation and generalization is different, and will be 
discussed in the following paragraphs.

\begin{algorithm}[tb]
\caption[\textsc{SatWin0}: Basic \acs{SAT} solver based \acs{CNF} learning 
algorithm for the winning region]
{\textsc{SatWin0}: Basic \acs{SAT} solver based \acs{CNF} learning 
algorithm for computing the winning region.}
\label{alg:SatWin0}
\begin{algorithmic}[1]
\ProcedureRet{SatWin0}
             {(\overline{x}, \overline{i}, \overline{c}, I, T, P)}
             {The winning region $W(\overline{x})$ in CNF or $\false$}
  \LineIf{$\propsat\bigl(I(\overline{x}) \wedge \neg P(\overline{x})\bigr)$}
    {\textbf{return} $\false$}
  \State $F(\overline{x}) := P(\overline{x})$,\quad
         $U(\overline{x}, \overline{i}) := \true$
  \While{$\true$}  
    \State $(\mathsf{sat},\mathbf{x},\mathbf{i}) := \propsatmodel\bigl(
            F(\overline{x}) \wedge 
            U(\overline{x}, \overline{i}) \wedge 
            T(\overline{x}, \overline{i}, \overline{c}, \overline{x}') \wedge 
            \neg F(\overline{x}')\bigr)$\label{alg:SatWin0:check}
    \If{$\neg \mathsf{sat}$}
      \State \textbf{return} $F(\overline{x})$
    \Else
      \State $(\mathsf{sat},\mathbf{c}) := \propsatmodel\bigl(
               F(\overline{x}) \wedge 
               \mathbf{x} \wedge \mathbf{i} \wedge 
               T(\overline{x},\overline{i},\overline{c},\overline{x}') \wedge 
               F(\overline{x}')\bigr)$ \label{alg:SatWin0:check1}
      \If{$\neg \mathsf{sat}$}
        \State $\mathbf{x}_g := \propsatmincore\bigl(\mathbf{x},
                F(\overline{x}) \wedge \mathbf{i} \wedge 
                T(\overline{x},\overline{i},\overline{c},\overline{x}') \wedge 
                F(\overline{x}')\bigr)$ \label{alg:SatWin0:gen}      
        \LineIf{$\propsat\bigl(\mathbf{x}_g \wedge I(\overline{x})\bigr)$}
          {\textbf{return} $\false$}
        \State $F(\overline{x}) := F(\overline{x}) \wedge \neg \mathbf{x}_g$, 
        \quad
               $U(\overline{x}, \overline{i}) := \true$\label{alg:SatWin0:ru}
      \Else
        \State $U := U
               \wedge \neg \propsatmincore\bigl(\mathbf{x} \wedge \mathbf{i},
               \mathbf{c} \wedge 
               F(\overline{x}) \wedge 
               U(\overline{x}, \overline{i}) \wedge 
               T(\overline{x},\overline{i},\overline{c},\overline{x}') \wedge
               \neg F(\overline{x}')\bigr)$\label{alg:SatWin0:ref}
      \EndIf         
    \EndIf
  \EndWhile
\EndProcedure  
\end{algorithmic}
\end{algorithm}

\mypara{Counterexample computation.}  We need to find a state $\mathbf{x}$ from 
which the environment can enforce that $F$ is left.  That is, from state 
$\mathbf{x} 
\models F$, there must exist some input $\mathbf{i}$ such that for all control 
values $\mathbf{c}$, the next state will satisfy $\neg F$.  \textsc{SatWin0} 
avoid this implicit quantifier alternation by computing such a state in several 
steps.  First, Line~\ref{alg:SatWin0:check} computes a state $\mathbf{x}$ and 
input $\mathbf{i}$ for which \emph{some} $\mathbf{c}$ would make the system 
leave $F$.  This is a necessary but not a sufficient condition for $\mathbf{x}$ 
to be a counterexample.  Hence, if the query in Line~\ref{alg:SatWin0:check} is 
unsatisfiable, no counterexample can exist, so $F$ must be the final winning 
region and the algorithm terminates.  The formula $U$ in 
Line~\ref{alg:SatWin0:check} excludes state-input combinations which cannot be 
used by the environment to enforce that $F$ is left.\footnote{Formally, $U$ 
satisfies the invariant $
\forall \overline{x}, \overline{i} \scope
\bigl(F(\overline{x}) \wedge \neg U(\overline{x},\overline{i})\bigr) 
\rightarrow
\bigl(
\exists \overline{c},\overline{x}' \scope
T(\overline{x},\overline{i},\overline{c},\overline{x}') \wedge F(\overline{x}')
\bigr)
$.
}
Initially, $U$ is $\true$, i.e., no restrictions are imposed.  The refinement of 
$U$ will be discussed further below.  For now, $U$ can be ignored. 

If the query in Line~\ref{alg:SatWin0:check} is satisfiable, the next step is to 
check if the candidate $\mathbf{x}$ is indeed a counterexample for the given 
$\mathbf{i}$.  This is investigated in Line~\ref{alg:SatWin0:check1} by 
computing some $\mathbf{c}$ for which $F$ is \emph{not} left, i.e., the next 
state is in $F$ again.  If such a $\mathbf{c}$ exists, then the environment 
cannot enforce that $F$ is left from state $\mathbf{x}$ with input $\mathbf{i}$. 
In order to prevent the same $(\mathbf{x}, \mathbf{i})$-pair from being returned 
by Line~\ref{alg:SatWin0:check} again, $U$ could be refined to $U \wedge \neg 
(\mathbf{x} \wedge \mathbf{i})$.  However, by computing the unsatisfiable core 
of $(\mathbf{x} \wedge \mathbf{i})$ in Line~\ref{alg:SatWin0:ref}, the algorithm 
may also exclude other $(\mathbf{x}, \mathbf{i})$-pairs for which $\mathbf{c}$ 
can be used by the system to prevent that $F$ is left.  Such $(\mathbf{x}, 
\mathbf{i})$-pairs are not helpful for the environment in order to enforce 
that $F$ is left.  They can thus safely be removed from $U$.  Note that the 
formula in the core computation is essentially that of 
Line~\ref{alg:SatWin0:check}. 

The remaining case is that where the formula in Line~\ref{alg:SatWin0:check1} is 
unsatisfiable.  In this case, $\mathbf{x}$ is indeed a counterexample because if 
the environment picks input $\mathbf{i}$, no system action can reach a state of 
$F$, so the next state is bound to be in $\neg F$.  As for \textsc{QbfWin}, 
$\mathbf{x}$ cannot be part of the final winning region, so it must be excluded 
from $F$.  However, before doing so, it is generalized into a larger region 
$\mathbf{x}_g$ of states that need to be excluded.  This will be explained in 
the next paragraph. As soon as $F$ changes, $U$ becomes invalid and is thus set 
to $\true$ again in Line~\ref{alg:SatWin0:ru}.  The intuitive reason is as 
follows: even if a certain state-input pair $(\mathbf{x}, \mathbf{i})$ cannot be 
used by the environment to enforce that $F$ is left, $(\mathbf{x}, \mathbf{i})$ 
may still be usable for leaving a smaller $F$ because the target region $\neg 
F(\overline{x}')$ becomes bigger.

\mypara{Counterexample generalization.}  \textsc{QbfWin} in 
Algorithm~\ref{alg:QbfWin} eliminates literals from the counterexample 
$\mathbf{x}$ as long as the reduced cube $\mathbf{x}_g\subseteq \mathbf{x}$ 
satisfies $\mathbf{x}_g\wedge F \rightarrow \FE(\neg F)$, i.e., as long as
$         \exists \overline{x} \scope
           \forall \overline{i} \scope
           \exists \overline{c},\overline{x}' \scope
           \mathbf{x}_g \wedge 
           F(\overline{x}) \wedge 
           T(\overline{x}, \overline{i}, \overline{c}, \overline{x}') \wedge 
           F(\overline{x}')
$
is unsatisfiable.  Due to the universal quantification over the inputs, a 
\acs{SAT} solver cannot be used for these checks. \textsc{SatWin0} solves this 
issue by considering only one input vector, namely the input $\mathbf{i}$ with 
which the environment can enforce that $F$ is left from $\mathbf{x}$.  For this 
input 
$\mathbf{i}$, the formula is certainly unsatisfiable for the full minterm 
$\mathbf{x}$, because this was checked in Line~\ref{alg:SatWin0:check1}. Hence, 
eliminating literals from 
$\mathbf{x}$ while 
$\mathbf{x}_g \wedge \mathbf{i} \wedge
           F(\overline{x}) \wedge 
           T(\overline{x}, \overline{i}, \overline{c}, \overline{x}') \wedge 
           F(\overline{x}')
$
is unsatisfiable is implemented in Line~\ref{alg:SatWin0:gen} by computing an 
unsatisfiable core of $\mathbf{x}$.  Considering only one input vector instead 
of all makes the formula weaker, which means that less literals may be 
eliminated.  However, the purely propositional satisfiability checks are also 
potentially faster. 

\begin{wrapfigure}[9]{r}{0.49\textwidth}
\vspace{-7mm}
\centering 
\subfloat[Counterexample candidate.\label{fig:sat_learn1}]
  {\includegraphics[width=0.16\textwidth]{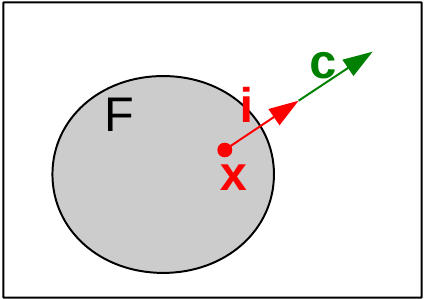}}
\hspace{0mm}
\subfloat[Check.\label{fig:sat_learn2}]
  {\includegraphics[width=0.16\textwidth]{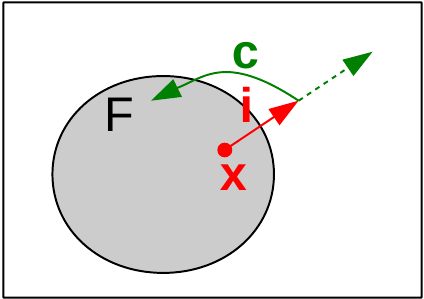}}
\hspace{0mm}
\subfloat[Generalization.\label{fig:sat_learn3}]
  {\includegraphics[width=0.16\textwidth]{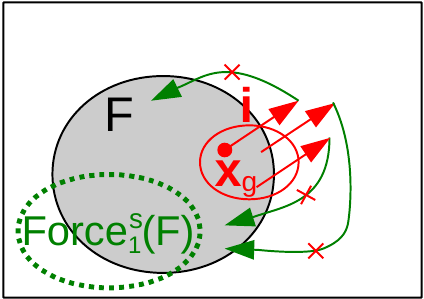}}
\caption{Working principle of \textsc{SatWin0}.}
\label{fig:sat_learn}
\end{wrapfigure}
\mypara{Illustration.}
Figure~\ref{fig:sat_learn} illustrates the working principle of \textsc{SatWin0} 
graphically.  As before, a box represents the set of all states.  In 
Figure~\ref{fig:sat_learn1}, a counterexample candidate is computed in form of 
a 
state $\mathbf{x}$ from which some input $\mathbf{i}$ and some control value 
$\mathbf{c}$ lead from $F$ to $\neg F$. This corresponds to the \acs{SAT} 
solver call in Line~\ref{alg:SatWin0:check}. In case of satisfiability, the next 
step is to check if some alternative $\mathbf{c}$ leads back to $F$ (for the 
same $\mathbf{x}$ and $\mathbf{i}$).  This is illustrated in 
Figure~\ref{fig:sat_learn2} and corresponds to the \acs{SAT} solver call in 
Line~\ref{alg:SatWin0:check1}.  In case of satisfiability, $U$ is refined in 
order not to get the same counterexample candidate again 
(Line~\ref{alg:SatWin0:ref}), and the algorithm proceeds by computing the next 
counterexample candidate as shown in Figure~\ref{fig:sat_learn1}.  In case of 
unsatisfiability, $\mathbf{x}$ is indeed a counterexample.  
Figure~\ref{fig:sat_learn3} illustrates how it is generalized into a larger 
region $\mathbf{x}_g$ for which input $\mathbf{i}$ enforces that the next state 
is in $\neg F$: it is ensured that, from any state of $\mathbf{x}_g$, with 
input 
$\mathbf{i}$, no $\mathbf{c}$ can exist such that the next state is in $F$ 
again.  This is a sufficient but not a necessary condition for $F\wedge 
\mathbf{x}_g$ not to intersect with $\FS(F)$. This generalization corresponds to 
the computation of the unsatisfiable core in Line~\ref{alg:SatWin0:gen} of 
\textsc{SatWin0}.  Finally, $\mathbf{x}_g$ is removed from $F$ and the procedure 
continues with Figure~\ref{fig:sat_learn1}.

\mypara{Discussion.}  In contrast to \textsc{QbfWin} 
(Algorithm~\ref{alg:QbfWin}), \textsc{SatWin0} potentially requires far more 
solver calls.  This has two reasons.  First, many refinements of $U$ may be 
necessary until a genuine counterexample is found.  In contrast, \textsc{QbfWin} 
computes a counterexample with one single solver call.  Second, the 
counterexample generalization in \textsc{SatWin0} is weaker and may thus drop 
fewer literals.  This can increase the number of counterexamples that needs to 
be computed.  The advantage of \textsc{SatWin0} is that all satisfiability 
checks are propositional and, thus, potentially less expensive.

The main purpose of discussing \textsc{SatWin0} from Algorithm~\ref{alg:SatWin0} 
was to prepare for a more advanced version, which 
will be presented in the next section.  Hence, we will not elaborate on 
implementation aspects or formal correctness arguments for 
Algorithm~\ref{alg:SatWin0}, but only do this for the advanced version, which is 
presented in the next section.

\subsubsection{Advanced Algorithm} \label{hw:sat1}

The basic algorithm from the previous section has two main weaknesses.  First, a 
reset of $U$ needs to be done upon \emph{every} update of $F$.  After such a 
reset, a lot of iterations may be necessary until $U$ is again restrictive 
enough for Line~\ref{alg:SatWin0:check} to produce a counterexample. Second, 
incremental solving is difficult in Line~\ref{alg:SatWin0:check} due to the 
negation of $F$: clauses are added to $F$, but this makes $\neg F$ weaker, which 
can only be expressed by (also) removing clauses from the \acs{CNF} 
representation of $\neg F$.
The procedure \textsc{SatWin1} in Algorithm~\ref{alg:SatWin1} resolves these 
weaknesses.  The differences to \textsc{SatWin0} are marked in blue.

\begin{algorithm}[tb]
\caption{\textsc{SatWin1}: Advanced \acs{SAT} solver based \acs{CNF} learning 
algorithm for computing the winning region.}
\label{alg:SatWin1}
\begin{algorithmic}[1]
\ProcedureRet{SatWin1}
             {(\overline{x}, \overline{i}, \overline{c}, I, T, P)}
             {The winning region $W(\overline{x})$ in CNF or $\false$}
  \LineIf{$\propsat\bigl(I(\overline{x}) \wedge \neg P(\overline{x})\bigr)$}
    {\textbf{return} $\false$}
    \label{alg:SatWin1:cunreal0}
    \label{alg:SatWin1:unreal0}
  \State $F(\overline{x}) := P(\overline{x})$, \label{alg:SatWin1:f0}
         $U(\overline{x}, \overline{i}) := \true$,
         \mrk{
         $G(\overline{x}) := F(\overline{x})$,
         $\textsf{precise} := \true$}
  \While{$\true$}  
    \State $(\mathsf{sat},\mathbf{x},\mathbf{i}) := \propsatmodel\bigl(
            F(\overline{x}) \wedge 
            U(\overline{x}, \overline{i}) \wedge 
            T(\overline{x}, \overline{i}, \overline{c}, \overline{x}') \wedge 
            \neg \mrk{G}(\overline{x}')\bigr)$ \label{alg:SatWin1:check}
    \If{$\neg \mathsf{sat}$}
      \LineIf{\mrk{$\textsf{precise}$}}
         {\textbf{return} $F(\overline{x})$} \label{alg:SatWin1:t0}
        \State \mrk{$U(\overline{x}, \overline{i}) := \true$,
                    $G(\overline{x}) := F(\overline{x})$,
                    $\textsf{precise} := \true$} \label{alg:SatWin1:restart}
    \Else
      \State $(\mathsf{sat},\mathbf{c}) := \propsatmodel\bigl(
               F(\overline{x}) \wedge 
               \mathbf{x} \wedge \mathbf{i} \wedge 
               T(\overline{x}, \overline{i}, \overline{c}, \overline{x}') \wedge 
               F(\overline{x}')\bigr)$\label{alg:SatWin1:check2}
      \If{$\neg \mathsf{sat}$}
        \State $\mathbf{x}_g := \propsatmincore\bigl(\mathbf{x},
                F(\overline{x}) \wedge \mathbf{i} \wedge 
                T(\overline{x},\overline{i},\overline{c},\overline{x}') \wedge 
                F(\overline{x}')\bigr)$ \label{alg:SatWin1:core}      
        \LineIf{$\propsat\bigl(\mathbf{x}_g \wedge I(\overline{x})\bigr)$} 
          {\textbf{return} $\false$}
          \label{alg:SatWin1:unreal1}
          \label{alg:SatWin1:t1}
        \State $F(\overline{x}) := F(\overline{x}) \wedge \neg \mathbf{x}_g$,
               \label{alg:SatWin1:ref}
               \mrk{$\textsf{precise}:= \false$}
      \Else
        \State $U := U
               \wedge \neg \propsatmincore\bigl(\mathbf{x} \wedge \mathbf{i},
               \mathbf{c} \wedge 
               F(\overline{x}) \wedge 
               U(\overline{x}, \overline{i}) \wedge 
               T(\overline{x},\overline{i},\overline{c},\overline{x}') \wedge
               \neg \mrk{G}(\overline{x}')\bigr)$\label{alg:SatWin1:refu}
      \EndIf         
    \EndIf
  \EndWhile
\EndProcedure  
\end{algorithmic}
\end{algorithm}

\begin{wrapfigure}[8]{r}{0.18\textwidth}
\vspace{-4mm}
\centering
  \includegraphics[width=0.17\textwidth]{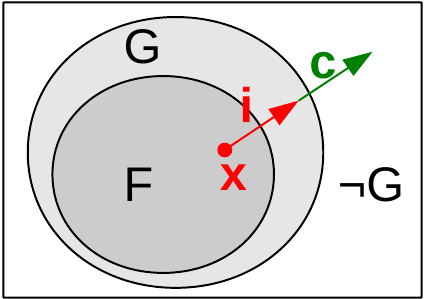}
\caption{Counterexample candidate in \textsc{SatWin1}.}
\label{fig:sat_learn_g}
\end{wrapfigure}
\mypara{Lazy updates of $F$.}
The formula $G(\overline{x})$ is a copy of $F(\overline{x})$ that is updated 
only lazily with newly discovered clauses.  Consequently, $F \rightarrow G$ 
holds at any time, i.e., $G$ always represents a superset of the states in $F$.  
The Boolean flag $\textsf{precise}$ is $\true$ whenever $G = F$.  While 
\textsc{SatWin0} computed a transition from $F$ to $\neg F$ in 
Line~\ref{alg:SatWin0:check}, \textsc{SatWin1} computes a transition from $F$ to 
$\neg G$.  This is illustrated in Figure~\ref{fig:sat_learn_g}.  A transition 
from $F$ to $\neg G$ is also a transition from $F$ to $\neg F$.  Thus, in case 
of satisfiability, nothing changes.  However, if no such transition exists, this 
does not automatically mean that no transition from $F$ to $\neg F$ exists.  
Therefore, if $G\neq F$, Line~\ref{alg:SatWin1:restart} sets $G:=F$ and the 
check 
is repeated. Only if $G=F$ (indicated by $\textsf{precise} = \true$), the 
algorithm can conclude that no more counterexample exists and returns $F$ as 
result.

\mypara{Updates of $U$.}
New clauses are only added to $F$ but not to $G$ in Line~\ref{alg:SatWin1:ref}. 
Thus, after any update of $F$, $\textsf{precise}$ must be set to $\false$.  
However, $U$ can be kept as it is.  The intuitive reason is as follows. If a 
certain $(\mathbf{x},\mathbf{i})$-pair is not helpful for the environment to 
enforce a transition from $F$ to $\neg G$, then it will definitely not be 
helpful to enforce a transition from some smaller set $F\wedge H$ of states into 
the same region $\neg G$. More formally, we have that
\[
\Bigl(
(\mathbf{x}, \mathbf{i}) \not\models 
\forall \overline{c} \scope
\exists \overline{x}'\scope
F(\overline{x}) \wedge 
T(\overline{x},\overline{i},\overline{c},\overline{x}') \wedge 
\neg G(\overline{x}')
\Bigr)
\text{ implies }
\Bigl(
(\mathbf{x}, \mathbf{i}) \not\models 
\forall \overline{c} \scope
\exists \overline{x}'\scope
F(\overline{x}) \wedge H(\overline{x}) \wedge 
T(\overline{x},\overline{i},\overline{c},\overline{x}') \wedge 
\neg G(\overline{x}')
\Bigr)
\text{ because}
\]
\[\Bigl(
\forall \overline{c} \scope
\exists \overline{x}'\scope
F(\overline{x}) \wedge H(\overline{x}) \wedge 
T(\overline{x},\overline{i},\overline{c},\overline{x}') \wedge 
\neg G(\overline{x}')
\Bigr)
\rightarrow
\Bigl(
\forall \overline{c} \scope
\exists \overline{x}'\scope
F(\overline{x}) \wedge 
T(\overline{x},\overline{i},\overline{c},\overline{x}') \wedge 
\neg G(\overline{x}')
\Bigr).
\hspace*{32mm}
\]
Only when $G$ changes in Line~\ref{alg:SatWin1:restart}, $U$ also becomes 
invalid and needs to be reset to $\true$.

\subsubsection{Correctness of the Advanced Algorithm \textsc{SatWin1}}

\noindent
We now work out a formal correctness argument for \textsc{SatWin1}, split into
several lemmas to increase readability.

\begin{lemma}\label{alg:SatWin1_term}
The \textsc{SatWin1} procedure in Algorithm~\ref{alg:SatWin1} always terminates.
\end{lemma}
\begin{proof}
Every loop iteration must end with one of the following five events: (1) the 
loop terminates in Line~\ref{alg:SatWin1:t0}, (2) $U$ is set to $\true$ in 
Line~\ref{alg:SatWin1:restart}, (3) the loop terminates in 
Line~\ref{alg:SatWin1:t1}, (4) $F$ shrinks in Line~\ref{alg:SatWin1:ref}, or (5) 
$U$ shrinks in Line~\ref{alg:SatWin1:refu}. We show that all these events lead 
to termination or eventual shrinking of $F$:  Item (2) cannot happen twice in a 
row without shrinking $F$ in between: this is prevented by having 
$\textsf{precise}=\true$.  Item (5) cannot happen infinitely often without 
shrinking $F$ in between because at some point $U$ would reach $\false$, which 
makes Line~\ref{alg:SatWin1:check} return $\textsf{sat}=\false$.  In this case, 
the algorithm either terminates in Line~\ref{alg:SatWin1:t0}, or item (2) 
occurs, and item (2) cannot occur twice without shrinking $F$ in between.  
Hence, the loop either terminates or makes some progress towards shrinking $F$.  
Before $F$ can shrink below $I$, the loop definitely terminates in 
Line~\ref{alg:SatWin1:t1}.
\qed
\end{proof}

\begin{lemma}\label{thm:SatWin1_ifp}
\textsc{SatWin1} enforces the invariant $I(\overline{x}) \rightarrow 
F(\overline{x}) \rightarrow P(\overline{x})$.
\end{lemma}
\begin{proof}
As for \textsc{QbfWin} (see Theorem~\ref{thm:QbfWin_correct}), $I 
\rightarrow F$ is enforced by Line~\ref{alg:SatWin1:cunreal0} 
and~\ref{alg:SatWin1:unreal1}; $F \rightarrow P$ is enforced by 
Line~\ref{alg:SatWin1:f0} and~\ref{alg:SatWin1:ref}.
\qed
\end{proof}

\begin{lemma}\label{thm:SatWin1_wf}
\textsc{SatWin1} enforces the invariant $W(\overline{x}) \rightarrow 
F(\overline{x})$.
\end{lemma}
\begin{proof}
Similar to Theorem~\ref{thm:QbfWin_correct}, this can be proven induction: 
Initially $F=P$, so $W\rightarrow F$ holds because $W\rightarrow P$. Given that 
$W\rightarrow F$ holds before an update of $F$ in 
Line~\ref{alg:SatWin1:ref}, it will also hold after the update because 
Line~\ref{alg:SatWin1:core} ensures that
$\mathbf{x}_g\wedge F(\overline{x})\wedge \mathbf{i} \wedge
  T(\overline{x},\overline{i},\overline{c},\overline{x}') \wedge
  F(\overline{x}')$ 
is unsatisfiable.  Consequently, we have that
$\forall \overline{x},\overline{i},\overline{c},\overline{x}'\scope
  \bigl(\mathbf{x}_g\wedge F(\overline{x})\wedge \mathbf{i} \bigr)
  \rightarrow
  \bigl(\neg T(\overline{x},\overline{i},\overline{c},\overline{x}') \vee
  \neg F(\overline{x}')\bigr).$
Because $T$ is both deterministic and complete ($\overline{x}'$ is always 
uniquely defined by $T$; see Definition~\ref{def:safety}) we can apply the 
one-point rule (\ref{eq:op1})
in order to rewrite the implication $\forall \overline{x}' \scope 
T(\overline{x},\overline{i},\overline{c},\overline{x}') \rightarrow \neg 
F(\overline{x}')$ to $\exists \overline{x}' \scope 
T(\overline{x},\overline{i},\overline{c},\overline{x}') \wedge \neg 
F(\overline{x}')$.  This gives
$\forall \overline{x},\overline{i},\overline{c} \scope
  \exists \overline{x}'\scope
  \bigl(\mathbf{x}_g\wedge F(\overline{x})\wedge \mathbf{i} \bigr)
  \rightarrow
  \bigl(T(\overline{x},\overline{i},\overline{c},\overline{x}') \wedge
  \neg F(\overline{x}')\bigr).$
Using the one point rule (\ref{eq:op0}) on $\overline{i}$, this 
formula is equivalent to
$\forall \overline{x} \scope
  \exists \overline{i} \scope
  \forall \overline{c} \scope
  \exists \overline{x}'\scope
  \mathbf{i} \wedge \bigl(
    \bigl(\mathbf{x}_g\wedge F(\overline{x})\bigr)
    \rightarrow
    \bigl(T(\overline{x},\overline{i},\overline{c},\overline{x}') \wedge
    \neg F(\overline{x}')\bigr)
  \bigr).$
This implies
$\forall \overline{x} \scope
  \exists \overline{i} \scope
  \forall \overline{c} \scope
  \exists \overline{x}'\scope
  \bigl(
    \bigl(\mathbf{x}_g\wedge F(\overline{x})\bigr)
    \rightarrow
    \bigl(T(\overline{x},\overline{i},\overline{c},\overline{x}') \wedge
    \neg F(\overline{x}')\bigr)
  \bigr)$, which can be written as
$\mathbf{x}_g\wedge F(\overline{x}) \rightarrow 
  \exists \overline{i} \scope
  \forall \overline{c} \scope
  \exists \overline{x}' \scope
  T(\overline{x},\overline{i},\overline{c},\overline{x}') \wedge
  \neg F(\overline{x}').$
By substituting the definition of $\FE$, we get $\mathbf{x}_g\wedge F 
\rightarrow \FE(\neg F)$.  Using the induction hypothesis $W\rightarrow F$, 
which can be written as $\neg F\rightarrow \neg W$, this 
means that $\mathbf{x}_g\wedge F \rightarrow \FE(\neg W)$ holds.  Thus, only 
states that cannot be part of $W$ are removed in Line~\ref{alg:SatWin1:ref}. In 
other words, $F$ cannot shrink below $W$.
\qed
\end{proof}

\noindent
The following lemma states that the formula $F(\overline{x}) \wedge \neg 
U(\overline{x}, \overline{i})$ can only represent state-input pairs for which 
the system player can reach $G$ and thus avoid ending up in $\neg G$.  In other 
words, the conjunction with $U$ in the \acs{SAT} solver call of 
Line~\ref{alg:SatWin1:check} excludes only state-input pairs for which the 
environment cannot enforce a transition from $F$ to $\neg G$.

\begin{lemma}\label{thm:SatWin1_u}
\textsc{SatWin1} enforces the invariant 
$ \forall \overline{x}, \overline{i} \scope
 \bigl(F(\overline{x}) \wedge \neg U(\overline{x}, \overline{i})\bigr)
  \rightarrow 
  \bigl(
  \exists \overline{c} , \overline{x}' \scope 
  T(\overline{x},\overline{i},\overline{c},\overline{x}') \wedge 
  G(\overline{x}')
  \bigr)
$.
\end{lemma}
\begin{proof}
$U$ is initialized to $\true$, so the invariant holds initially.  
Line~\ref{alg:SatWin1:restart} sets $U=\true$ and thus retains the invariant.  
Line~\ref{alg:SatWin1:ref} also retains the invariant because $F$ only 
gets stricter.  It remains to be shown that Line~\ref{alg:SatWin1:refu} retains 
the invariant.  Let $\mathbf{u}$ be the result of $\propsatmincore$ in 
Line~\ref{alg:SatWin1:refu}.  The update $U := U \wedge \neg \mathbf{u}$ in 
Line~\ref{alg:SatWin1:refu} changes the invariant to 
$\forall \overline{x}, \overline{i} \scope
  \bigl(F(\overline{x}) \wedge  
  (\neg U(\overline{x}, \overline{i}) \vee \mathbf{u})
  \bigr)
  \rightarrow 
  \bigl(
  \exists \overline{c} , \overline{x}' \scope 
  T(\overline{x},\overline{i},\overline{c},\overline{x}') \wedge 
  G(\overline{x}')
  \bigr),
$
which can be written as
$\forall \overline{x}, \overline{i} \scope
  \bigl(
   \bigl(F(\overline{x}) \wedge \neg U(\overline{x}, \overline{i})\bigr)
   \vee
   \bigl(F(\overline{x}) \wedge U(\overline{x}, \overline{i}) \wedge \mathbf{u} 
   \bigr)
  \bigr)
  \rightarrow 
  \bigl(
  \exists \overline{c} , \overline{x}' \scope 
  T(\overline{x},\overline{i},\overline{c},\overline{x}') \wedge 
  G(\overline{x}')
  \bigr).
$
In general, a formula $(A \vee B)\rightarrow C$ holds iff $A \rightarrow C$ and 
$B \rightarrow C$.  By induction, we know that
$ \forall \overline{x}, \overline{i} \scope
  \bigl(F(\overline{x}) \wedge \neg U(\overline{x}, \overline{i})\bigr)
  \rightarrow 
  \bigl(
  \exists \overline{c} , \overline{x}' \scope 
  T(\overline{x},\overline{i},\overline{c},\overline{x}') \wedge 
  G(\overline{x}')
  \bigr)
$ 
holds.  What remains to be shown is that 
$\forall \overline{x}, \overline{i} \scope
  \bigl(
  F(\overline{x}) \wedge U(\overline{x}, \overline{i}) \wedge \mathbf{u}  
 \bigr)
\rightarrow 
  \bigl(
  \exists \overline{c}, \overline{x}' \scope
  T(\overline{x},\overline{i},\overline{c},\overline{x}') \wedge
  G(\overline{x}')
  \bigr)
$ 
also holds.  Since
$\mathbf{u} \wedge \mathbf{c} \wedge 
 F(\overline{x}) \wedge 
 U(\overline{x}, \overline{i}) \wedge 
 T(\overline{x},\overline{i},\overline{c},\overline{x}') \wedge
 \neg G(\overline{x}')$
is unsatisfiable (enforced by Line~\ref{alg:SatWin1:refu}), we have that
$\forall \overline{x}, \overline{i} \scope
\bigl(
F(\overline{x}) \wedge U(\overline{x}, \overline{i}) \wedge \mathbf{u}  
\bigr)
\rightarrow 
\bigl(
\forall \overline{c},\overline{x}' \scope
\neg \mathbf{c} \vee 
\neg T(\overline{x},\overline{i},\overline{c},\overline{x}') \vee
 G(\overline{x}')
\bigr) 
.$
By applying the one-point rule (Eq.~(\ref{eq:op0}) for $\overline{c}$ and
Eq.~(\ref{eq:op1}) for $\overline{x}'$), this can also be 
written as
$\forall \overline{x}, \overline{i} \scope
\bigl(
F(\overline{x}) \wedge U(\overline{x}, \overline{i}) \wedge \mathbf{u} 
\bigr)
\rightarrow 
\bigl(
\exists \overline{c}, \overline{x}' \scope
\mathbf{c} \wedge 
T(\overline{x},\overline{i},\overline{c},\overline{x}') \wedge
G(\overline{x}')
\bigr).$
This formula obviously implies
$\forall \overline{x}, \overline{i} \scope
\bigl(
F(\overline{x}) \wedge U(\overline{x}, \overline{i}) \wedge \mathbf{u}  
\bigr)
\rightarrow 
\bigl(
\exists \overline{c}, \overline{x} \scope
T(\overline{x},\overline{i},\overline{c},\overline{x}') \wedge
G(\overline{x}')
\bigr),$
which was to be shown for Line~\ref{alg:SatWin1:refu} to preserve the invariant.
\qed
\end{proof}

\begin{lemma}\label{thm:SatWin1_fw}
If \textsc{SatWin1} reaches Line~\ref{alg:SatWin1:t0}, $F(\overline{x}) = 
W(\overline{x})$ holds at that point.
\end{lemma}
\begin{proof}
Line~\ref{alg:SatWin1:t0} is only reached when $G=F$ (otherwise \textsf{precise} 
is $\false$) and 
$ F(\overline{x}) \wedge 
   U(\overline{x}, \overline{i}) \wedge 
   T(\overline{x}, \overline{i}, \overline{c}, \overline{x}') \wedge 
   \neg G(\overline{x}')$
is unsatisfiable, which means that
$ \forall \overline{x}, \overline{i} \scope
   \bigl(
   F(\overline{x}) \wedge 
   U(\overline{x}, \overline{i}) 
   \bigr)
   \rightarrow
   \bigl(
   \forall \overline{c} \scope
   \forall \overline{x}' \scope
   \neg T(\overline{x}, \overline{i}, \overline{c}, \overline{x}') \vee 
   F(\overline{x}')
   \bigr)
   $
holds. By applying the one-point rule (\ref{eq:op1}), this can also be 
written 
as
$ \forall \overline{x}, \overline{i} \scope
   \bigl(
   F(\overline{x}) \wedge 
   U(\overline{x}, \overline{i}) 
   \bigr)
   \rightarrow
   \bigl(
   \forall \overline{c} \scope
   \exists \overline{x}' \scope
   T(\overline{x}, \overline{i}, \overline{c}, \overline{x}') \wedge
   F(\overline{x}') \bigr).$
In turn, this implies
$ \forall \overline{x}, \overline{i} \scope
   \bigl(
   F(\overline{x}) \wedge 
   U(\overline{x}, \overline{i}) 
   \bigr)
   \rightarrow
   \bigl(
   \exists \overline{c} \scope
   \exists \overline{x}' \scope
   T(\overline{x}, \overline{i}, \overline{c}, \overline{x}') \wedge
   F(\overline{x}')
   \bigr).$
From Lemma~\ref{thm:SatWin1_u}, we know that
$\forall \overline{x}, \overline{i} \scope
  \bigl(
  F(\overline{x}) \wedge \neg U(\overline{x}, \overline{i})
  \bigr)
  \rightarrow
  \bigl(
  \exists \overline{c} \scope \exists \overline{x}' \scope 
  T(\overline{x},\overline{i},\overline{c},\overline{x}') \wedge 
  F(\overline{x}')
  \bigr).
$
Since $A\wedge B \rightarrow C$ and $A\wedge \neg B \rightarrow C$ together 
imply $A \rightarrow C$, we can conclude that
$\forall \overline{x} \scope
  F(\overline{x})
  \rightarrow
  \forall \overline{i} \scope
  \exists \overline{c} , \overline{x}' \scope 
  T(\overline{x},\overline{i},\overline{c},\overline{x}') \wedge 
  F(\overline{x}')
$
must hold in Line~\ref{alg:SatWin1:t0}.  This means that the returned $F$ 
satisfies $F\rightarrow \FS(F)$.  From $W \rightarrow F \rightarrow P$ 
(Lemma~\ref{thm:SatWin1_wf} and~\ref{thm:SatWin1_ifp}), it follows that $F=W$.  
The reason is that $W$ is the set of \emph{all} states from which the system 
player can enforce the specification, i.e., no proper superset $H$ of $W$ can 
satisfy $H \rightarrow P$ and $H\rightarrow \FS(H)$.
\qed
\end{proof}

\begin{theorem}~\label{thm:sat_correct}
The \textsc{SatWin1} procedure in Algorithm~\ref{alg:SatWin1} returns the 
winning region $W(\overline{x})$ of a given safety specification $\mathcal{S}$, 
or $\false$ if the specification is unrealizable.
\end{theorem}
\begin{proof}
Unrealizability:
If $\mathcal{S}$ is unrealizable, $I\not\rightarrow W$.   \textsc{SatWin1} 
terminates (Lemma~\ref{alg:SatWin1_term}), but cannot terminate in 
Line~\ref{alg:SatWin1:t0} because $F=W$ (Lemma~\ref{thm:SatWin1_fw}) contradicts 
with $I\not\rightarrow W$ (unrealizability) and $I\rightarrow F$ 
(Lemma~\ref{thm:SatWin1_ifp}).  Hence, in case of unrealizability, 
\textsc{SatWin1} must terminate in Line~\ref{alg:SatWin1:unreal0} 
or~\ref{alg:SatWin1:t1} returning $\false$.

Realizability:
\textsc{SatWin1} can only return $\false$ in Line~\ref{alg:SatWin1:unreal0} 
or~\ref{alg:SatWin1:t1} if $F$ is about to be updated such that 
$I\not\rightarrow F$.  From $I\rightarrow W$ (realizability) and
$W\rightarrow F$ (Lemma~\ref{thm:SatWin1_wf}), it follows that $I\rightarrow 
F$, so this can never happen. Yet, Lemma~\ref{alg:SatWin1_term} says that 
\textsc{SatWin1} terminates, so it must reach Line~\ref{alg:SatWin1:t0} 
eventually.  By Lemma~\ref{thm:SatWin1_fw}, this will return the winning region.
\qed
\end{proof}

\subsubsection{Efficient Implementation}

\noindent
This section discusses some important aspects of implementing \textsc{SatWin1} 
efficiently.

\mypara{Incremental solving.}  We propose to use three \acs{SAT} solver 
instances incrementally.  The first one will be called \textsf{solverC} and 
stores 
 $F(\overline{x}) \wedge 
  U(\overline{x}, \overline{i}) \wedge 
  T(\overline{x}, \overline{i}, \overline{c}, \overline{x}') \wedge 
  \neg G(\overline{x}')$.
\textsf{solverC} is used in Line~\ref{alg:SatWin1:check} and 
Line~\ref{alg:SatWin1:refu}, where the conjunction with $\mathbf{c}$ is realized 
with temporarily asserted assumption literals.  Whenever 
Line~\ref{alg:SatWin1:restart} is reached, \textsf{solverC} is reset with the 
new \acs{CNF} encoding of $\neg G(\overline{x}') = \neg F(\overline{x}')$. 
Otherwise, clauses are only added to $F$ or $U$.
The second solver instance, called \textsf{solverG}, stores
 $F(\overline{x}) \wedge 
    T(\overline{x},\overline{i},\overline{c},\overline{x}') \wedge 
    F(\overline{x}')$
and is used for Line~\ref{alg:SatWin1:check2} and Line~\ref{alg:SatWin1:core}. 
Clauses are only added to $F$, so \textsf{solverG} does not have to be reset at 
all. The conjunctions with $\mathbf{i}$ and $\mathbf{x}$, which change from 
iteration to iteration, are again realized by setting assumption 
literals.  The lines~\ref{alg:SatWin1:check2} and~\ref{alg:SatWin1:core} are 
actually combined into one \acs{SAT} solver call that returns either a 
satisfying assignment $\mathbf{c}$ or an unsatisfiable core.  The third solver 
instance stores $I(\overline{x})$ and is used in 
Line~\ref{alg:SatWin1:unreal1}.\footnote{The input format in our implementation 
actually allows for only one initial state, so  
Line~\ref{alg:SatWin1:unreal1} can be realized without calling a 
\acs{SAT} solver.}  The conjunction with $\mathbf{x}_g$ is again realized with 
assumption literals.

\mypara{\ac{CNF} compression.}  Whenever \textsf{solverC} is reset with the 
current \acs{CNF} encoding of $\neg G(\overline{x}') = \neg F(\overline{x}')$ in 
Line~\ref{alg:SatWin1:restart}, we call \textsc{CompressCnf} from 
Algorithm~\ref{alg:CompressCnf} (with literal dropping disabled) in order to 
reduce the size of $F$ beforehand.  This results in a more 
compact \acs{CNF} encoding of $\neg G(\overline{x}')$ when using the method of 
Plaisted and Greenbaum~\cite{PlaistedG86}.

\mypara{Resets of \textsf{solverG}.}  By default, we only add clauses to 
\textsf{solverG}.  However, after some iterations, many of the $F$-clauses added 
to \textsf{solverG} can become redundant because they can be implied by (a 
combination of) other clauses that have been added later.  To prevent the clause 
database of \textsf{solverG} from growing unreasonably, we also reset 
\textsf{solverG} with the compressed $F$ from time to time.  As a heuristic, we 
track the number of $F$-clauses that have been added to \textsf{solverG} so far, 
and compute the difference to the number of clauses in the compressed $F$.  If 
this difference exceeds a certain limit, \textsf{solverG} is reset.  This can 
give a moderate speedup for certain \acs{SAT} solvers and benchmarks.

\subsection{Partial Quantifier Expansion} \label{sec:hw_exp}

The procedure \textsc{QbfWin} in Algorithm~\ref{alg:QbfWin} uses quantified 
formulas to compute counterexamples witnessing that $F\neq W$ and to generalize 
these counterexamples.  In contrast, the procedure \textsc{SatWin1} in 
Algorithm~\ref{alg:SatWin1} avoids the universal quantifiers.  This results in 
less expensive solver calls, but comes at the price of requiring more iterations 
of the outer loop.  In this section, we will discuss a hybrid approach which 
quantifies universally over \emph{some} (but not necessarily all) variables.  
The universal quantification is then eliminated by applying universal expansion 
so that the resulting formulas can be solved with a plain \acs{SAT} solver.  The 
hope is to find a sweet spot where the reduction in the number of iterations is 
more significant than the additional costs per solver call.

\subsubsection{Quantifier Expansion in Counterexample Computation}

The procedure \textsc{QbfWin} in Algorithm~\ref{alg:QbfWin} computes 
counterexamples to $F=W$ by solving the quantified formula
$\exists \overline{x},\overline{i} \scope
      \forall \overline{c} \scope
      \exists \overline{x}' \scope
      F(\overline{x}) \wedge 
      T(\overline{x}, \overline{i}, \overline{c}, \overline{x}') \wedge 
      \neg F(\overline{x}').$
In contrast, the \textsc{SatWin1} procedure from Algorithm~\ref{alg:SatWin1}
avoids the universal quantification of the variables $\overline{c}$ by solving 
the formula
$\exists \overline{x},\overline{i} \scope
  \exists \overline{c} \scope
  \exists \overline{x}' \scope
   F(\overline{x}) \wedge 
   U(\overline{x}, \overline{i}) \wedge 
   T(\overline{x}, \overline{i}, \overline{c}, \overline{x}') \wedge 
   \neg G(\overline{x}'),$
where $G$ is just a copy of $F$ that may not be fully up to date.  The latter 
formula does not necessarily yield a counterexample, but only a candidate.  If 
the candidate turns out to be spurious, it is excluded by refining $U$.  This 
approach can be seen as a ``lazy elimination'' of the universal quantification 
over $\overline{c}$ via $U$.  The disadvantage is that many refinements of $U$ 
may be necessary before the first genuine counterexample is found.  One 
alternative would be to eliminate $\forall \overline{c}$  in
$\exists \overline{x},\overline{i} \scope
  \forall \overline{c} \scope
  \exists \overline{x}' \scope
   F(\overline{x}) \wedge
   T(\overline{x}, \overline{i}, \overline{c}, \overline{x}') \wedge 
   \neg G(\overline{x}')$
eagerly by performing universal expansion as explained in 
Section~\ref{sec:prelim:qbf}.  Yet, this may blow up the formula size by a 
factor of $2^{|\overline{c}|}$ and may thus be infeasible.  Another alternative 
is to partition the variables of $\overline{c}$ into two subsets 
$\overline{c}_1$ and $\overline{c}_2$ and solve
\[\exists \overline{x},\overline{i} \scope
  \exists \overline{c}_1 \scope
  \forall \overline{c}_2 \scope
  \exists \overline{x}' \scope
   F(\overline{x}) \wedge 
   U(\overline{x}, \overline{i}) \wedge 
   T(\overline{x}, \overline{i}, \overline{c}, \overline{x}') \wedge 
   \neg G(\overline{x}')\]
using a \acs{SAT} solver by expanding only over the variables in 
$\overline{c}_2$.  By adjusting the relative size of $\overline{c}_2$, different
trade-offs between decreasing the number of refinements to $U$ and increasing 
the costs per solver call can be achieved.

\subsubsection{Quantifier Expansion in Counterexample Generalization}

The idea is similar to that of the previous subsection.  \textsc{QbfWin}
eliminates literals from a counterexample $\mathbf{x}$ as long as
$\exists \overline{x} \scope
  \forall \overline{i} \scope
  \exists \overline{c},\overline{x}' \scope
  \mathbf{x}_g \wedge 
  F(\overline{x}) \wedge 
  T(\overline{x}, \overline{i}, \overline{c}, \overline{x}') \wedge 
  F(\overline{x}')$
is unsatisfiable.  In contrast, \textsc{SatWin1} avoids the universal 
quantification over $\overline{i}$ by ensuring that
$\exists \overline{x} \scope
  \exists \overline{i} \scope
  \exists \overline{c},\overline{x}' \scope
  \mathbf{x} \wedge \mathbf{i} \wedge
  F(\overline{x}) \wedge
  T(\overline{x}, \overline{i}, \overline{c}, \overline{x}') \wedge 
  F(\overline{x}')$
is unsatisfiable for some concrete $\mathbf{i}$.  The latter check is 
potentially cheaper, but may result in fewer literals being eliminated from
$\mathbf{x}$.  This means that the refinement of $F$ is less substantial, so 
more iterations may be needed.  By partitioning the variables $\overline{i}$
into $\overline{i}_1$ and $\overline{i}_2$ and checking 
$\exists \overline{x} \scope
  \exists \overline{i} \scope
  \mathbf{x} \wedge \mathbf{i} \wedge
  F(\overline{x}) \wedge
  \forall \overline{i}_2 \scope
  \exists \overline{c},\overline{x}' \scope
  T(\overline{x}, \overline{i}, \overline{c}, \overline{x}') \wedge 
  F(\overline{x}')$
for unsatisfiability, different trade-offs between the generalization procedure
of \textsc{QbfWin} and that of \textsc{SatWin1} can be realized.

\subsubsection{Efficient Implementation}  

Universal expansion needs to be implemented carefully in order to avoid an 
unnecessary blow-up of the formula size, and to keep the time for the expansion 
low.  Our experience showed that even small inefficiencies can cost orders of
magnitude in both metrics.

\begin{wrapfigure}[11]{r}{0.23\textwidth}
\vspace{-4mm}
\centering
  \includegraphics[width=0.22\textwidth]{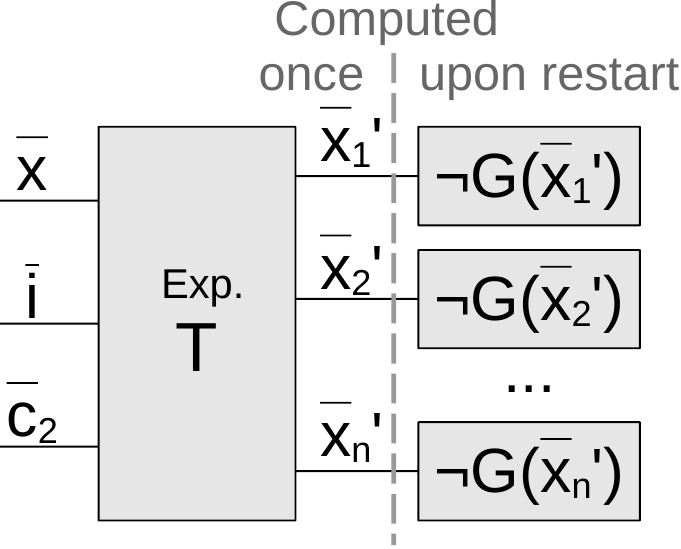}
\caption{Universal expansion for counterexample computation.}
\label{fig:sat_exp}
\end{wrapfigure}
\mypara{Expansion for counterexample computation.}  Since $F(\overline{x})$ and 
$U(\overline{x}, \overline{i})$ are independent of $\overline{c}$ and 
$\overline{x}'$, we apply universal expansion to   
$ \forall \overline{c}_2 \scope
   \exists \overline{x}' \scope
   T(\overline{x}, \overline{i}, \overline{c}, \overline{x}') \wedge 
   \neg G(\overline{x}')$
and conjoin $F(\overline{x})\wedge U(\overline{x}, \overline{i})$ afterwards. 
The transition relation $T$ is always fixed, but $\neg G(\overline{x}')$ changes 
upon every restart of \textsf{solverC} in Line~\ref{alg:SatWin1:restart} of 
\textsc{SatWin1}.  Hence, we expand $T$ only once and store the 
resulting renamings $\overline{x}'_1,\ldots \overline{x}'_n$ of 
$\overline{x}'$.  A copy of $\neg G(\overline{x}')$ is then added for each
renaming $\overline{x}'_i$ when \textsf{solverC} is initialized or restarted.  
This is illustrated in Figure~\ref{fig:sat_exp}.

\mypara{Expansion of $T$.}  In our implementation, $T$ is originally given as a 
circuit of \textsf{AND}-gates (where inputs can be negated).  We perform 
the expansion of $T$ directly on this circuit and only encode the result into 
\acs{CNF}.  This facilitates efficient constant propagation and other 
simplifications.  When expanding a certain $c\in\overline{c}_2$, we only copy 
those \textsf{AND}-gates that have $c$ in their fan-in cone.  Whenever the copy 
of some \textsf{AND}-gate has the same inputs as some existing 
\textsf{AND}-gate, the existing gate is reused.  Finally, the tool 
\abc~\cite{BraytonM10} is called to simplify the expanded circuit.  This 
involves \textsf{fraiging}, which ensures that no two nodes in the circuit can 
represent the same function over the inputs.  Hence, equivalent (copies of) 
next-state signals will be represented by the same variable, which enables a 
more substantial simplification of the $\neg G(\overline{x}')$ copies (see 
Figure~\ref{fig:sat_exp}).  Finally, duplicate renamings of the next-state 
variables are removed.  Since $T$ is expanded only once, these simplifications 
can be afforded.

\mypara{Expansion of $\neg G(\overline{x}')$.}  First, we perform an even 
more aggressive compression of $G(\overline{x}')$ than done by 
\textsc{CompressCnf} in Algorithm~\ref{alg:CompressCnf}.  \textsc{CompressCnf} 
removes a clause $c$ from a \acs{CNF} $A$ if $\bigl(A \setminus \{c\}\bigr) 
\rightarrow c$, i.e., if the clause is implied by other clauses of $A$ already. 
We now remove a clause $c$ from $G(\overline{x}')$ if 
$\bigl(F(\overline{x}) \wedge 
  T(\overline{x}, \overline{i}, \overline{c}, \overline{x}') \wedge
  \bigl(G(\overline{x}')\setminus \{c\}\bigr)\bigr)
  \rightarrow c.$
Hence, the compressed $G(\overline{x}')$ will only be equivalent to the original 
$G(\overline{x}')$ if the current state is in $F$, but this is asserted in 
all \acs{SAT} solver calls of \textsc{SatWin1} anyway.  For every renaming 
$\overline{x}'_i$ of $\overline{x}'$ that has been created during the expansion 
of $T$, we then perform the following steps:  First, $G(\overline{x}'_i)$ is 
computed by applying the renaming.  Second, $G(\overline{x}'_i)$ is simplified 
by removing tautological clauses and performing unit clause propagation.  
Finally, $G(\overline{x}'_i)$ is negated, while auxiliary variables that have 
already been introduced during the negation of other copies of 
$G(\overline{x}')$ are reused.  All these measures contribute towards reducing 
the size of the expansion of $\neg G(\overline{x}')$.

\mypara{Expansion in counterexample generalization.}
This is easier, since no negation of a \acs{CNF} is involved.  Again, $T$ is
expanded and the resulting renamings $\overline{x}'_1,\ldots \overline{x}'_n$ 
of $\overline{x}'$ are stored.  Whenever a clause $\neg \mathbf{x}_g$ is added 
to $F$, we do not only add it to \textsf{solverG} but also add all the renamed 
next-state copies of the clause to \textsf{solverG}.

\mypara{Configuration.}  In our experiments, choosing low numbers for 
$|\overline{c}_2|$ only slowed down the 
\textsc{SatWin1} procedure compared to $|\overline{c}_2|=0$.  High numbers for 
$|\overline{c}_2|$ did bring a speedup, though, with the best results achieved 
for $\overline{c}_2 = \overline{c}$.  Furthermore, we observed that an explosion 
of the formula size can be avoided in most cases by our careful implementation 
of the formula expansion.  Hence, by default, we expand over all variables in 
$\overline{c}$ and only fall back to $\overline{c}_2 = \emptyset$ if some memory 
limit is exceeded.  For counterexample generalization, a speedup could only be 
achieved with low numbers of $|\overline{i}_2|$.  Hence, by default, we only 
expand one input signal.  As a heuristic, we choose the signal that causes the least 
number of gates to be copied when expanding the transition~relation.

\mypara{Discussion.}  Our optimization of partial quantifier expansion can be 
used to realize different trade-offs between the number of \acs{SAT} solver 
calls and their costs in Algorithm~\ref{alg:SatWin1}.  We hoped to find 
a sweet spot between these two cost factors at low expansion rates, but our 
experiments suggest high rates at least for counterexample computation. While 
the basic idea of quantifier expansion is simple, such high expansion rates 
require a careful implementation, like the one discussed in this section, in 
order not to waste computational resources.

\subsection{Reachability Optimizations} \label{sec:hw_reach}

In this section, we present optimizations that exploit (un)reachability 
information when computing a winning region with query learning.  The 
optimizations can be applied to \textsc{QbfWin} (Algorithm~\ref{alg:QbfWin}) and 
to \textsc{SatWin1} (Algorithm~\ref{alg:SatWin1}), both with and without partial 
quantifier elimination.  However, to simplify the presentation, we only explain 
the optimizations for the case of \textsc{QbfWin} in detail.  The application to 
\textsc{SatWin1} works in exactly the same way.

\subsubsection{Optimization \textsf{RG}: Reachability for Counterexample 
Generalization}\label{sec:hw_rg}

Recall that a counterexample $\mathbf{x}\models F\wedge \FE(\neg F)$ in 
\textsc{QbfWin} is a state that is part of the current over-approximation $F$ of 
the winning region, but this state cannot be part of the final winning region. 
The state is represented by a minterm $\mathbf{x}$ over the state variables 
$\overline{x}$.  \textsc{QbfWin} generalizes $\mathbf{x}$ 
into a larger state region $\mathbf{x}_g$ by eliminating literals as long as 
$F\wedge \mathbf{x}_g \rightarrow \FE(\neg F)$ holds, i.e., as long as $F\wedge 
\mathbf{x}_g \wedge \FS(F)$ is unsatisfiable.  The reason is that any state 
$\mathbf{x}_a \models F\wedge \mathbf{x}_g \wedge \FS(F)$ could potentially be 
part of the winning region, and thus must not be removed from $F$.  Yet, as an 
optimization, we can still remove such a state $\mathbf{x}_a$, as long as it is 
guaranteed that $\mathbf{x}_a$ is unreachable from the initial states.  Using 
this insight, we can eliminate literals in a counterexample $\mathbf{x}$ as long 
as $R \wedge F\wedge \mathbf{x}_g \rightarrow \FE(\neg F)$, where 
$R(\overline{x})$ is an over-approximation of the reachable states in 
$\mathcal{S}$.  In \textsc{QbfWin}, this can be realized by conjoining $R$ to 
the \ac{QBF} that is checked in Line~\ref{alg:QbfWin:gen}.  This may result in 
more literals being eliminated during the generalization, which means that $F$ 
is pruned more extensively.  Ultimately, this can reduce the number of 
iterations in \textsc{QbfWin}.

\mypara{Computing reachable states.}
The states that are reachable from the initial states in a specification 
$\mathcal{S}$ can be defined inductively as follows:  All states in 
$I(\overline{x})$ are reachable.  If a state $\mathbf{x}$ is reachable, then all 
states $\mathbf{x}'\models \exists \overline{x},\overline{i},\overline{c} \scope 
\mathbf{x} \wedge T(\overline{x},\overline{i},\overline{c},\overline{x}')$ are 
also reachable.  In the synthesis setting, this definition can even be refined. 
 Any over-approximation $F$ of the winning region is itself an 
over-approximation of the reachable states, not necessarily in the 
specification $\mathcal{S}$, but definitely in the final implementation.  The 
reason is that no realization of $\mathcal{S}$ must ever leave the winning 
region $W$, and thus also not $F$. This insight can be used to compute a 
tighter set of reachable states by considering only transitions that remain in 
$F$.  In principle, the set of reachable states can easily be computed using a 
simple fixed-point algorithm. However, we consider this to be too expensive, 
and instead work with over-approximations of the reachable states.  Such 
over-approximations are also useful in formal verification, and many methods to 
compute them exist~\cite{MoonKSS99}.

\mypara{Our approach.}
We avoid computing an over-approximation of the reachable states explicitly.  
Instead, we use an idea that is inspired by the model checking algorithm 
\icthree~\cite{Bradley11}:  Let $R(\overline{x})$ be an over-approximation of 
the reachable states.  By induction, we know that a state $\mathbf{x}$ is 
definitely unreachable if $I(\overline{x}) \rightarrow \neg \mathbf{x}$ and 
$\neg \mathbf{x} \wedge R(\overline{x}) \wedge 
T(\overline{x},\overline{i},\overline{c},\overline{x}') \rightarrow \neg 
\mathbf{x}'$.  The formula says that if the current state is reachable but 
different from $\mathbf{x}$, then the next state cannot be $\mathbf{x}$ either. 
Hence, if $\mathbf{x}$ is not an initial state, then $\mathbf{x}$ can never 
be visited. The same reasoning applies if $\mathbf{x}$ is an incomplete cube (or 
any other formula) representing a set of states.  In \icthree, $\neg \mathbf{x}$ 
is said to be \emph{inductive relative to} \index{relative inductiveness} the 
current knowledge $R$ about the 
reachable states.  It can thus be used to refine $R$.  

In our synthesis setting, we take the current over-approximation $F$ as an 
over-approximation of the reachable states.  When generalizing a counterexample 
$\mathbf{x}$, literals cannot only be eliminated if $F\wedge \mathbf{x}_g 
\rightarrow \FE(\neg F)$ is preserved, but also if $\neg \mathbf{x}_g$ is 
inductive relative to $F$.  The two criteria can be combined by requiring that
\begin{equation}\label{eq:qbf_rg}
\exists \mrk{\overline{x}^*, \overline{i}^*,\overline{c}^*,} 
          \overline{x} \scope
           \forall \overline{i} \scope
           \exists \overline{c},\overline{x}' \scope
           \mrk{
           \bigl(I(\overline{x}) \vee 
            F(\overline{x}^*) \wedge
            \neg \mathbf{x}_g^* \wedge
            T(\overline{x}^*, \overline{i}^*, \overline{c}^*, \overline{x})
            \bigr)\, \wedge\,
            }
           \mathbf{x}_g \wedge 
           F(\overline{x}) \wedge 
           T(\overline{x}, \overline{i}, \overline{c}, \overline{x}') \wedge 
           F(\overline{x}') 
\end{equation}           
\begin{wrapfigure}[7]{r}{0.21\textwidth}
\vspace{-4mm}
\centering
  \includegraphics[width=0.20\textwidth]{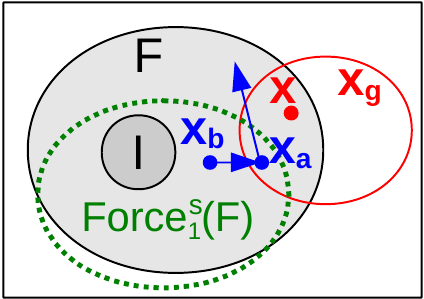}
\end{wrapfigure}
is unsatisfiable.  Only the parts of the formula that are marked in blue are 
new.  The variables $\overline{x}^*$, $\overline{i}^*$ and $\overline{c}^*$ are
previous-state copies of $\overline{x}$, $\overline{i}$ and $\overline{c}$, 
respectively.  The original version of the formula was $\true$ if some state 
$\mathbf{x}_a \models F\wedge \mathbf{x}_g \wedge \FS(F)$ exists. The improved 
formula also requires that $\mathbf{x}_a$ is either an initial state, or has a 
predecessor $\mathbf{x}_b$ in $F \wedge \neg \mathbf{x}_g$.  This is illustrated 
in the figure on the right.  If neither of these two criteria holds, then 
we 
know that $I(\overline{x}) \rightarrow \neg \mathbf{x}_a$ and $\neg \mathbf{x}_a 
\wedge F(\overline{x}) \wedge \mathbf{x}_g \wedge 
T(\overline{x},\overline{i},\overline{c},\overline{x}') \rightarrow \neg 
\mathbf{x}_a'$.  This means that $\neg \mathbf{x}_a$ is inductive relative to $F 
\wedge \neg \mathbf{x}_g$, so $\mathbf{x}_a$ is unreachable and can thus be 
removed even if it could potentially be part of the winning region.  Note that 
we do not require inductiveness relative to $F$ but rather relative to $F \wedge 
\neg \mathbf{x}_g$.  The intuitive reason is that $F$ will be updated to 
$F\wedge \neg \mathbf{x}_g$, so a predecessor $\mathbf{x}_b$ in $F\wedge 
\mathbf{x}_g$ does not count.  The following theorem states that this procedure
cannot prune $F$ too much.

\begin{theorem}\label{thm:gen_reach1}
For a realizable specification $\mathcal{S}$, if Equation~(\ref{eq:qbf_rg})
is unsatisfiable, then $F \wedge \mathbf{x}_g$ cannot contain a state 
$\mathbf{x}_a$ from which (a) the system player can enforce that $F$ is 
visited in one step, and (b) which is reachable in some implementation of 
$\mathcal{S}$.
\end{theorem}

\noindent
\textsc{Proof.}
By contradiction, assume that there exists such as state $\mathbf{x}_a$.  Any 
implementation of $\mathcal{S}$ must only visit states in $W$.  Hence, for 
$\mathbf{x}_a$ to be reachable in an implementation,  there must exist a 
play $\mathbf{x}_0,\ldots, \mathbf{x}_n, \ldots$ of the game $\mathcal{S}$ such 
that
\begin{compactitem}
\item $\mathbf{x}_0 \models I$ (the play starts in the initial states), 
\item $\mathbf{x}_n = \mathbf{x}_a$ (the play reaches $\mathbf{x}_a$ at 
some step $n$), 
\item $n\geq 1$ (that is, $\mathbf{x}_a$ cannot be initial because this would 
satisfy Equation~(\ref{eq:qbf_rg})), and 
\item $\mathbf{x}_j \models W$ for all $0\leq j\leq n$ (all states in the play 
are in the winning region and thus 
potentially~reachable).  
\end{compactitem}

\begin{wrapfigure}[7]{r}{0.21\textwidth}
\vspace{-4mm}
\centering
  \includegraphics[width=0.20\textwidth]{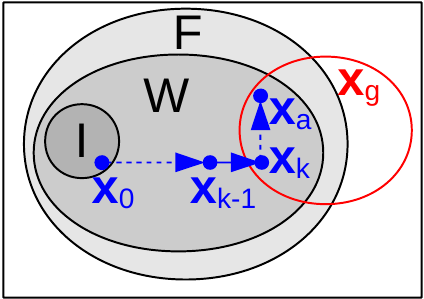}
\end{wrapfigure}
\noindent
Such a play is illustrated on the right. Since 
$\mathbf{x}_a \models F \wedge \mathbf{x}_g$, there must exist a smallest $k\leq 
n$ such that $\mathbf{x}_j \models F\wedge \mathbf{x}_g$ for all $k \leq j \leq 
n$.  Now, $\mathbf{x}_k$ is either initial or has a predecessor 
$\mathbf{x}_{k-1}$ in $F \wedge \neg \mathbf{x}_g$ (because 
$\mathbf{x}_{k-1}\models W, W\rightarrow F$, and $\mathbf{x}_{k-1}\not\models 
F\wedge \mathbf{x}_g$). Thus, $\mathbf{x}_k$ satisfies the new (blue) part of 
Equation~(\ref{eq:qbf_rg}).  Since Equation~(\ref{eq:qbf_rg}) is 
unsatisfiable, the system player cannot enforce that the play traverses from 
$\mathbf{x}_k$ to 
$F$.  Hence, $\mathbf{x}_k$ cannot be part of $W$. This contradiction means that 
such a path of reachable states ending in $\mathbf{x}_a$ cannot exist if 
Equation~(\ref{eq:qbf_rg}) is unsatisfiable.
\qed
Theorem~\ref{thm:gen_reach1} only considers the case of a realizable 
specification.  In case of unrealizability, the correctness argument is even 
simpler:  Optimization \textsf{RG} cannot make \textsc{QbfWin} identify an 
unrealizable specification as realizable because the additional conjuncts in 
Equation~(\ref{eq:qbf_rg}) can only have the effect that \emph{more} states are 
removed from $F$, thus $F$ can only shrink below $I$ faster.  Another important 
remark is that \textsc{QbfWin} no longer computes the winning region when 
optimization \textsf{RG} is enabled, but only a winning area according to 
Definition~\ref{def:winset}.  The reason is that states of $W$ may be missing in 
$F$ if they are unreachable.

\subsubsection{Optimization \textsf{RC}: Reachability for Counterexample 
Computation}\label{sec:hw_rc}

Similar to improving the generalization of counterexamples using unreachability 
information, we can also restrict their computation to potentially reachable 
states.  In addition to $\mathbf{x} \models F \wedge \FE(\neg F)$, we require 
that the counterexample $\mathbf{x}$ is either an initial state, or has a 
predecessor in $F$ that is different from~$\mathbf{x}$.  If neither of these 
two conditions is satisfied, then $\mathbf{x}$ can only be unreachable and, 
thus, does not have to be removed from $F$.

\mypara{Realization.}
In \textsc{QbfWin}, these additional constraints can be imposed by modifying 
the \acs{QBF} query in Line~\ref{alg:QbfWin:check} to
\begin{equation}\label{eq:qbf_rc}
\qbfsatmodel\bigl(
      \exists \mrk{\overline{x}^*,\overline{i}^*,\overline{c}^*,}
              \overline{x},\overline{i} \scope
      \forall \overline{c} \scope
      \exists \overline{x}' \scope \\
      \mrk{
      \bigl(
       I(\overline{x}) \vee 
       \overline{x}^*\neq \overline{x} \wedge
       F(\overline{x}^*) \wedge
       T(\overline{x}^*, \overline{i}^*, \overline{c}^*, \overline{x})       
      \bigr) \wedge
      }
      F(\overline{x}) \wedge 
      T(\overline{x}, \overline{i}, \overline{c}, \overline{x}') \wedge 
      \neg F(\overline{x}')\bigr).
\end{equation}
As before, the new parts are marked in blue, and $\overline{x}^*$, 
$\overline{i}^*$, and $\overline{c}^*$ are the previous-state copies of 
$\overline{x}$, $\overline{i}$, and $\overline{c}$, respectively.  The 
expression $\overline{x}^*\neq \overline{x}$ requires that at least one state 
variable $x\in \overline{x}$ has a different value than its 
previous-state copy.

\mypara{Consequences.}
When executing \textsc{LearnQbf} with optimization \textsf{RC} on a realizable 
specification, the returned formula $F$ may not be a winning area according to 
Definition~\ref{def:winset}: item (3) of may be violated because from some unreachable states 
of $F$, it may be that the system player cannot enforce that $F$ is reached in 
the next step.  Consequently, a system implementation can no longer be computed 
as a Skolem 
function for the variables $\overline{c}$ in the formula
$\forall \overline{x}, \overline{i} \scope
  \exists \overline{c}, \overline{x}' \scope 
  T(\overline{x},\overline{i},\overline{c},\overline{x}') \wedge
  \bigr(F(\overline{x}) \rightarrow F(\overline{x}') \bigl)
$
because this formula no longer holds true.  Still, a system implementation can 
be extracted, e.g., by computing Skolem functions for the $\overline{c}$-signals 
in the negation of Equation~(\ref{eq:qbf_rc}).  

\mypara{Configuration.}
In our experiments, we achieve a significant speedup when applying optimization 
\textsf{RG}, especially with our \acs{SAT} solver based algorithm 
\textsc{SatWin1}.  Optimization \textsf{RC} also gives some speedup for certain 
benchmarks, but does not pay off on average.  Hence, by default, we apply 
optimization \textsf{RG} but disable optimization \textsf{RC}.

\subsection{Template-Based Approach}\label{sec:hw:templ}

In the previous sections, a winning area was computed iteratively by starting 
with some initial approximation and then refining this approximation based on 
counterexamples.  This section presents a completely different approach, where 
we simply assert the constraints that constitute a winning area and compute a 
solution in one go.

\mypara{Basic idea.}
We define a generic template $H(\overline{x},\overline{k})$ \index{template} 
for the winning area 
$F(\overline{x})$ we wish to construct. $H(\overline{x},\overline{k})$ is 
a formula over the state variables $\overline{x}$ and a vector of Boolean 
variables $\overline{k}$, which act as template parameters.  Concrete values 
$\mathbf{k}$ for the parameters $\overline{k}$ instantiate a concrete formula 
$F(\overline{x}) = H(\overline{x},\mathbf{k})$ over the state variables 
$\overline{x}$.  This reduces the search for a propositional formula (the 
winning area) to a search for Boolean template parameter values.  We can now 
compute a winning area according to Definition~\ref{def:winset} with a single 
\acs{QBF} solver call $(sat, \mathbf{k}) := $
\begin{linenomath*}
\begin{align}
\qbfsatmodel\Bigl(
\exists \overline{k} \scope
\forall \overline{x}, \overline{i} \scope
\exists \overline{c}, \overline{x}' \scope
& \bigl(I(\overline{x}) \rightarrow H(\overline{x}, \overline{k})\bigr) 
  \wedge
 \bigl(H(\overline{x}, \overline{k}) \rightarrow P(\overline{x}) \bigr) 
  \wedge %
& \bigl(H(\overline{x}, \overline{k}) \rightarrow 
  \bigl(T(\overline{x}, \overline{i}, \overline{c}, \overline{x}') \wedge 
  H(\overline{x}', \overline{k})\bigr)\bigr)\Bigr)
  \label{eq:templ}
\end{align}
\end{linenomath*}
With the resulting template parameter values $\mathbf{k}$, the induced 
instantiation $F(\overline{x}) = H(\overline{x},\mathbf{k})$ of 
$H(\overline{x},\overline{k})$ is then computed.  

\mypara{Completeness of templates.}  A template $H(\overline{x},\overline{k})$ 
does not necessarily have to be complete in the sense that it can represent 
\emph{every} formula $F(\overline{x})$ over the state variables with some 
choice for the parameters $\overline{k}$.  We rather restrict the expressiveness 
of templates deliberately in order to reduce the search space for the solver.  
The underlying assumption is that many specifications have a winning area that 
can be represented as a ``simple'' formula over the state variables.  
We will use templates that are parameterized in their 
expressive power.  As a general strategy, we will start with a low value for 
some expressiveness parameter $N$, and increase $N$ as long as 
Equation~(\ref{eq:templ}) is unsatisfiable.  Detecting unrealizability is 
difficult with this approach, though.  Only if Equation~(\ref{eq:templ}) is 
unsatisfiable for a template that can represent \emph{every} function 
$F(\overline{x})$ over the state variables, we can conclude that the 
corresponding specification is unrealizable.

\mypara{Concrete realizations.}  While the basic idea of the template-based 
approach is simple, there are many ways to realize it.  One degree of freedom 
lies in the definition of the generic template $H(\overline{x},\overline{k})$ 
and its parameters.  Two concrete suggestions will be made in the following 
subsections.  Another source of freedom lies in the way to solve 
Equation~(\ref{eq:templ}).  An approach using \acs{SAT} solvers instead of a 
single call to a \acs{QBF} solver will be presented in 
Section~\ref{sec:hw:templsat}.

\subsubsection{\acs{CNF} Templates} \label{sec:cnfteml}

\begin{wrapfigure}[14]{r}{0.46\textwidth}
\vspace{-4mm}
\centering
  \includegraphics[width=0.45\textwidth]{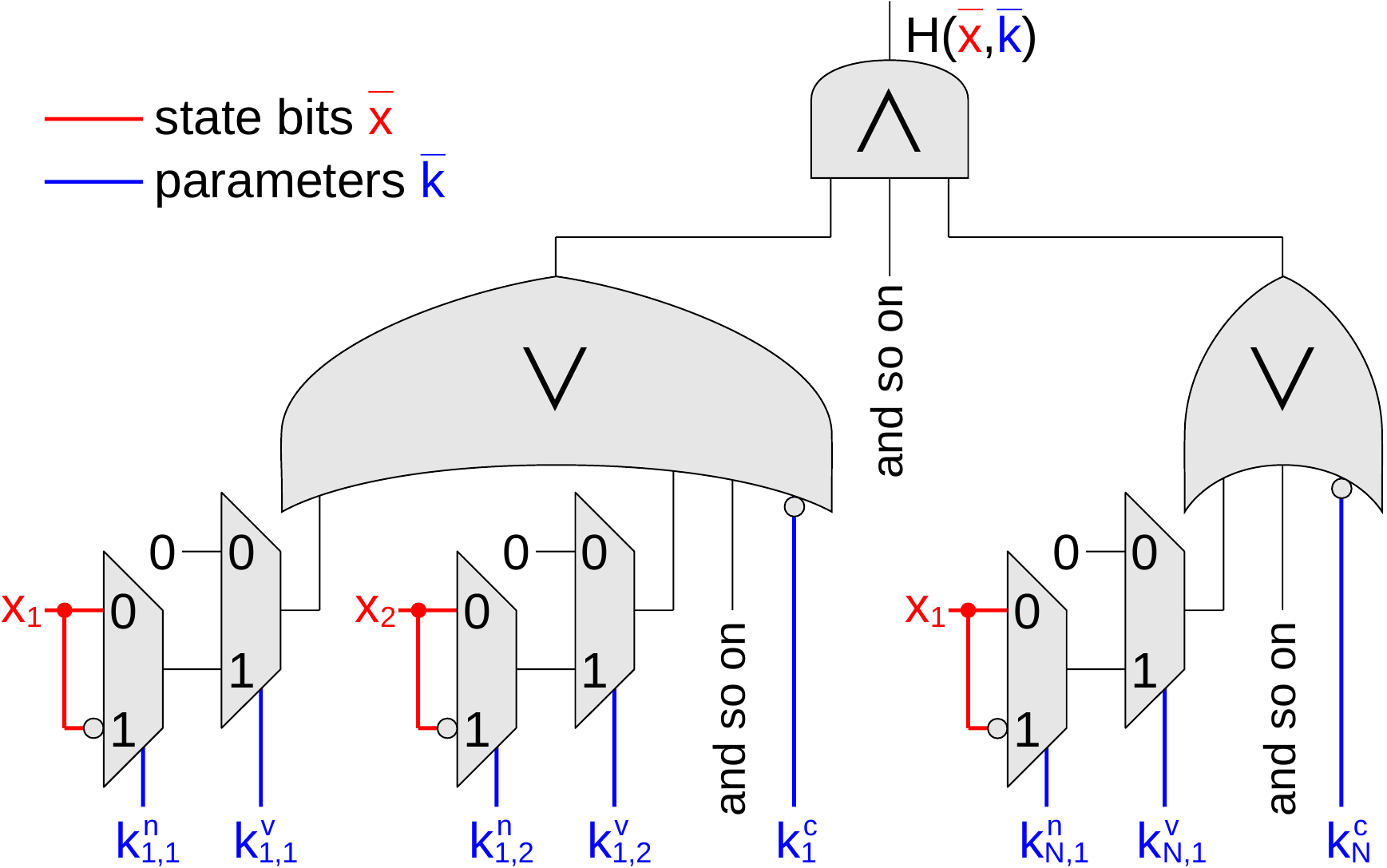}
\caption{Circuit illustration of a generic \acs{CNF} template.}
\label{fig:templ_cnf}
\end{wrapfigure}
Figure~\ref{fig:templ_cnf} shows a circuit that illustrates how the template 
$H(\overline{x}, \overline{k})$ can be defined as a parameterized \acs{CNF} 
formula over the state variables $\overline{x}$.  That is, $F(\overline{x})$ is 
represented as a conjunction of clauses over the state variables.  Template 
parameters $\overline{k}$ define the shape of the clauses.  The trapezoids in 
Figure~\ref{fig:templ_cnf} are multiplexers that select one of the inputs on the 
left depending on the signal value fed in from below.  A \acs{CNF} 
encoding of this circuit such that it can be used in Equation~(\ref{eq:templ}) is 
straightforward~\cite{Tseitin83}.

The construction in Figure~\ref{fig:templ_cnf} works as follows.  First, a 
maximum number $N$ of clauses is fixed.~This number configures the 
expressiveness of the template.  Next, three vectors $\overline{k^c}$, 
$\overline{k^v}$, $\overline{k^n}$ of template parameters are introduced.  
Together, they form $\overline{k} = \overline{k^c} \cup \overline{k^v}\cup 
\overline{k^n}$.  The meaning of the parameters is as~follows.
\begin{compactitem}
\item If parameter $k^c_i$ with $1\leq i\leq N$ is $\true$, then clause $i$ is 
used in $F(\overline{x})$, otherwise not.  This is achieved by making the clause 
$\true$ (and thus irrelevant in the conjunction of clauses) if $k^c_i$ is 
$\false$.
\item If parameter $k^v_{i,j}$ with $1\leq i\leq N$ and $1\leq j\leq 
|\overline{x}|$ is $\true$, then the state variable $x_j\in\overline{x}$ appears 
in clause $i$ of $F(\overline{x})$, otherwise not.  This is realized with a 
multiplexer that sets the corresponding literal in the clause to $\false$ 
(thus making it irrelevant in the disjunction) if $k^v_{i,j}$ is $\false$.
\item If parameter $k^n_{i,j}$ is $\true$, then the state variable 
$x_j$ can appear in clause $i$ only negated, otherwise only unnegated. This is 
realized with a multiplexer that selects between $x_j$ and $\neg x_j$.  If 
$k^v_{i,j}$ is $\false$, then $k^n_{i,j}$ is irrelevant.  
\end{compactitem}
This results in $|\overline{k}|=2\cdot N \cdot |\overline{x}| + N$ template 
parameters.
\begin{example}\label{ex:templ_cnf}
For $\overline{x} = (x_1,x_2,x_3)$ and $N=3$, the \acs{CNF} $(x_1 \vee \neg 
x_2) \wedge (\neg x_3)$ can be realized with
\begin{compactitem}
\item $k^c_1=k^c_2=\true$ and $k^c_3=\false$ (only clause $1$ and $2$ are used),
\item $k^v_{1,1}=k^v_{1,2}=\true$ and $k^v_{1,3}=\false$ (clause $1$ contains 
$x_1$ and $x_2$ but not $x_3$),
\item $k^v_{2,3}=\true$ and $k^v_{2,1}=k^v_{2,2}=\false$ (clause $2$ contains 
$x_3$ but not $x_1$ and not $x_2$),
\item $k^n_{1,1}=\false$ and $k^n_{1,2}=\true$ (clause $1$ contains 
$x_1$ unnegated and $x_2$ negated), and
\item $k^n_{2,3}=\true$ (clause $2$ contains $x_3$ negated).
\end{compactitem}
All other parameters are irrelevant.
\end{example}

Choosing $N$ is delicate.  If $N$ is too low, we will not find a solution, even 
if one exists.  If it is too high, we waste computational resources and may find 
an unnecessarily complex winning region. In our implementation, we solve this 
dilemma by starting with $N=1$ and increasing $N$ by one upon failure until we 
reach $N=4$.  From there, we double $N$ upon failure.  We stop if we get a 
negative answer for $N \ge 2^{|\overline{x}|}$ because any Boolean formula over 
$\overline{x}$ can be represented in a \acs{CNF} with less 
than $2^{|\overline{x}|}$ 
clauses.

\subsubsection{AND-Inverter Graph Templates}

Another option is to define the template $H(\overline{x}, \overline{k})$ as a 
network of AND-gates and inverters, fed by the state variables $\overline{x}$. 
The parameters $\overline{k}$ define the connections between the gates 
and the state variables, as well as the negation of signals.

\begin{wrapfigure}[12]{r}{0.61\textwidth}
\vspace{-4mm}
\centering
  \includegraphics[width=0.6\textwidth]{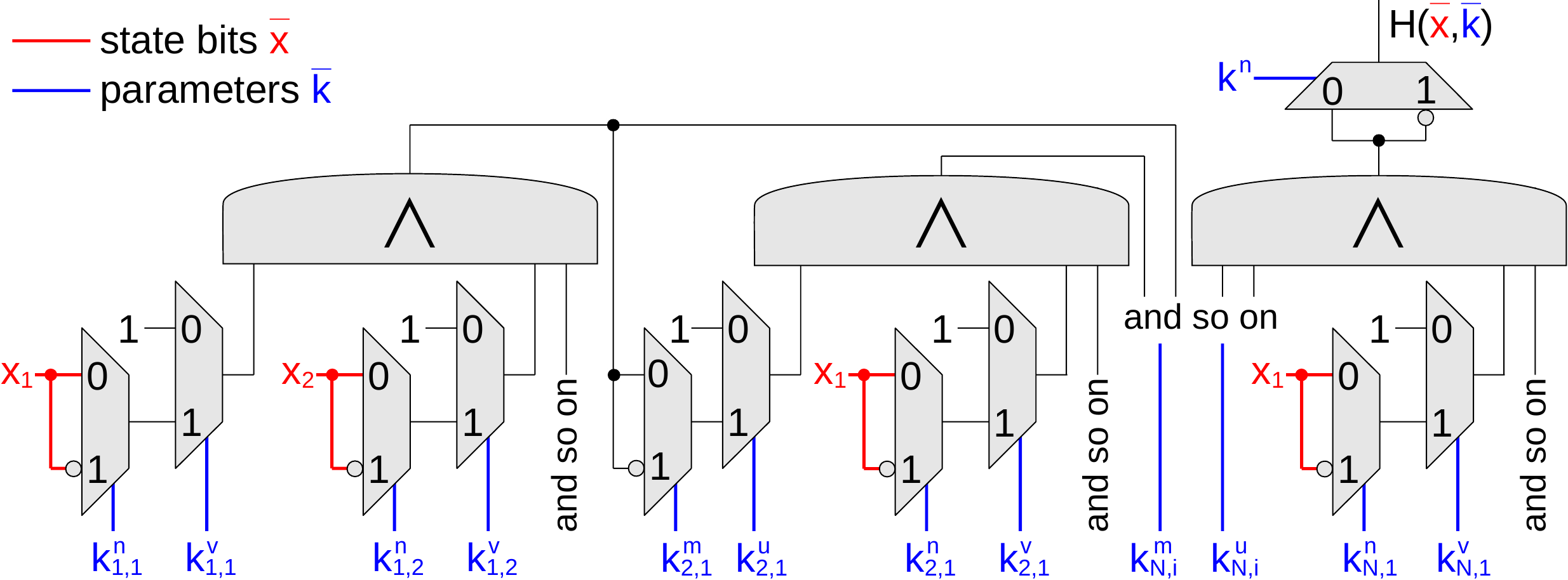}
\caption{Circuit illustration of a generic AND-inverter graph template.}
\label{fig:templ_aig}
\end{wrapfigure}
Figure~\ref{fig:templ_aig} gives a concrete proposal for defining such a 
template.  The template is again illustrated as a circuit, but can easily be 
encoded into \acs{CNF}.  A maximum number $N$ of AND-gates is chosen first.  The 
first gate can have all state variables as input, either negated or unnegated.  
The second gate can also have the output of the first gate as input. The third 
gate can have the output of the first two gates as additional inputs, and so on. 
 The output of the last AND-gate defines $H(\overline{x}, \overline{k})$, again 
with a possible negation.  The template parameters $\overline{k}$ define which 
inputs of a gates are actually used or ignored, and which inputs are used 
negated or unnegated.  We distinguish five groups of parameters.

\begin{compactitem}
\item If parameter $k^v_{i,j}$ with $1\leq i\leq N$ and $1\leq j\leq 
|\overline{x}|$ is $\true$, then $x_j\in\overline{x}$ appears 
as input of gate $i$, otherwise not. 
\item If parameter $k^n_{i,j}$ with $1\leq i\leq N$ and $1\leq j\leq 
|\overline{x}|$ is $\true$, then gate $i$ can only use the negated variable 
$x_j$ as input, otherwise only the unnegated variable.
\item If $k^u_{i,j}$ with $1\leq i\leq N$ and $1\leq j < i$ is 
$\true$, then the output of gate $j$ is an input of gate $i$, otherwise 
not. 
\item If $k^m_{i,j}$ with $1\leq i\leq N$ and $1\leq j < i$ is 
$\true$, then gate $i$ can only use the negated output of gate $j$ as input, 
otherwise only the unnegated output.
\item The single parameter $k^n$ defines if the output of the final gate defines
$H(\overline{x}, \overline{k})$ or $\neg H(\overline{x}, \overline{k})$.
\end{compactitem}
This gives $|\overline{k}| = N \cdot (2 \cdot |\overline{x}|+N-1) + 1$ 
template parameters.

\begin{example}\label{ex:templ_aig}
We continue Example~\ref{ex:templ_cnf}, where $\overline{x} = (x_1,x_2,x_3)$,
$N=3$ and $F(\overline{x}) = (x_1 \vee \neg x_2) \wedge (\neg x_3)$, which can
be rewritten to $\neg (\neg x_1 \wedge x_2) \wedge (\neg x_3)$.  This formula 
can be realized with
\begin{compactitem}
\item $k^v_{2,1}=k^v_{2,2}=\true$ and $k^v_{2,3}=\false$ (gate $2$ uses 
$x_1$ and $x_2$ as input but not $x_3$),
\item $k^n_{2,1}=\true$ and $k^n_{2,2}=\false$ (gate $2$ uses 
$x_1$ negated and $x_2$ unnegated),
\item $k^u_{2,1}=\false$ (gate $2$ ignores the output of gate $1$),
\item $k^v_{3,3}=\true$ and $k^v_{3,1}=k^v_{3,2}=\false$ (gate $3$ uses 
$x_3$ as input but not $x_1$ and not $x_2$),
\item $k^n_{3,3}=\true$ (gate $3$ uses $x_3$ negated),
\item $k^u_{3,2}=k^m_{3,2}=\true$ and $k^u_{3,1}=\false$ (gate $3$ uses the 
negated output of gate $2$ but ignores the output of gate $1$), and
\item $k^n=\false$ (the output $H(\overline{x}, \overline{k})$ is defined by 
the unnegated output of gate $3$).
\end{compactitem}
All other parameters are irrelevant.  In particular, the output of gate $1$ is 
completely ignored.
\end{example}

In our implementation, choosing $N$ works in the same way as for the \acs{CNF} 
template: starting with $N=1$, $N$ is increased by $1$ in case of 
unsatisfiability of Equation~(\ref{eq:templ}) until $N=4$ is reached.  From there, 
$N$ is doubled upon failure.  There is a straightforward way to represent a 
\acs{CNF} with $N$ clauses as a network of $N+1$ AND-gates.  Hence, the 
criterion for detecting unrealizability with \acs{CNF} templates can also be 
applied here: If Equation~(\ref{eq:templ}) is unsatisfiable for $N > 
2^{|\overline{x}|}$, the specification must be unrealizable.

\subsubsection{Implementation with \acs{SAT} Solvers}\label{sec:hw:templsat}

In this section, we present an extension of the \acf{CEGIS} approach that allows 
us to compute satisfying assignments of Equation~(\ref{eq:templ}) with \acs{SAT} 
solvers instead of a \acs{QBF} solver.

\mypara{Basic idea.}
Recall that \ac{CEGIS} (see Section~\ref{sec:prelim:cegis}) is an approach to 
compute satisfying assignments in formulas of the form
$\exists \overline{e} \scope \forall \overline{u} \scope 
  F(\overline{e}, \overline{u})$
by iterative refinements of a solution candidate.  With 
\[M(\overline{k}, \overline{x}, \overline{i}, \overline{c}, \overline{x}') =
\bigl(I(\overline{x}) \rightarrow H(\overline{x}, \overline{k})\bigr) 
  \wedge
\bigl(H(\overline{x}, \overline{k}) \rightarrow P(\overline{x}) \bigr) 
  \wedge
\Bigl(H(\overline{x}, \overline{k}) \rightarrow \notag
  \bigl(T(\overline{x}, \overline{i}, \overline{c}, \overline{x}') \wedge 
  H(\overline{x}', \overline{k})\bigr)\Bigr)
\]
being an abbreviation for the matrix of the \acs{QBF} in 
Equation~(\ref{eq:templ}), our task is now to compute a satisfying assignment for 
the parameters $\overline{k}$ in
$\exists \overline{k} \scope 
 \forall \overline{x}, \overline{i} \scope 
 \exists \overline{c}, \overline{x}' \scope 
  M(\overline{k}, \overline{x}, \overline{i}, \overline{c}, \overline{x}')$.
Hence, there is an additional existential quantifier on the innermost level.
This existential quantifier does not affect the computation of solution 
candidates significantly:  Candidates are satisfying assignments for the 
variables $\overline{k}$ in
$\bigwedge_{(\mathbf{x},\mathbf{i})\in D} 
  \exists \overline{c},\overline{x}' \scope  
  M(\overline{k}, \mathbf{x},\mathbf{i},\overline{c},\overline{x}'),$
where the existential quantification of $\overline{c}$ and $\overline{x}'$ can 
be handled by renaming these variables in every copy of $M$ and then calling a 
\acs{SAT} solver.  The computation of counterexamples, i.e., values for the 
variables $\overline{x}$ and $\overline{i}$, becomes more intricate, though.  
Instead of a satisfying assignment for $\neg F(\mathbf{e}, \overline{u})$, we 
now need to compute an assignment $\mathbf{x},\mathbf{i}$ for the variables 
$\overline{x}, \overline{i}$ in $\neg \exists \overline{c}, \overline{x}' 
\scope M(\mathbf{k}, \overline{x}, \overline{i}, \overline{c}, 
\overline{x}'),$ 
where $\mathbf{k}$ represents fixed values for the variables $\overline{k}$.
The negation turns the existential quantification into a universal one.  The 
resulting quantifier alternation prevents us from computing counterexamples with 
a single call to a \acs{SAT} solver.  A \acs{QBF} solver could be used, but the 
idea of this section is to substitute \acs{QBF} solving with plain \acs{SAT} 
solving.  Hence, we will use an iterative approach that is similar to 
\textsc{SatWin1} in Algorithm~\ref{alg:SatWin1} to compute counterexamples.

\begin{algorithm}[tb]
\caption[\textsc{TemplWinSat}: An algorithm to compute template 
instantiations using \acs{SAT} solvers]
{\textsc{TemplWinSat}: An algorithm to compute template 
instantiations using \acs{SAT} solvers.}
\label{alg:TemplWinSat}
\begin{algorithmic}[1]
\ProcedureRet{TemplWinSat}
             {H(\overline{x}, \overline{k}),
              (\overline{x}, \overline{i}, \overline{c}, I, T, P)}
             {A winning area $F(\overline{x})$ or ``\textsf{fail}''}
  \State $G(\overline{k}, \overline{t}) := \true$ \label{alg:TemplWinSat:init}
  \While{$\true$}
    \If{$\textsf{sat} = \false$ in 
        $(\textsf{sat},\mathbf{k}) := 
            \propsatmodel\bigl(G(\overline{k},\overline{t}) \bigr)$}
        \label{alg:TemplWinSat:comp}
      \State \textbf{return} ``\textsf{fail}''
    \EndIf
    \If{$\textsf{correct} = \true$ in
    $(\textsf{correct},\mathbf{x},\mathbf{i}) := 
           \textsc{Check}\bigl(
                         H(\overline{x}, \mathbf{k}),
                         (\overline{x}, \overline{i}, \overline{c}, I, T, P)
                         \bigr)$
    }
     \label{alg:TemplWinSat:check}
      \State \textbf{return} $H(\overline{x}, \mathbf{k})$
    \EndIf
    \State $\overline{t}_c := \textsf{CreateFreshCopy}(\overline{c})$, \quad
           $\overline{t}_x := \textsf{CreateFreshCopy}(\overline{x}')$
    \State $G(\overline{k}, \overline{t}) := G(\overline{k}, \overline{t}) 
\wedge
\bigl(I(\mathbf{x}) \rightarrow H(\mathbf{x},\overline{k})\bigr) 
\wedge
\bigl(H(\mathbf{x}, \overline{k}) \rightarrow P(\mathbf{x}) \bigr) 
\wedge
\bigl(H(\mathbf{x}, \overline{k}) \rightarrow
  \bigl(T(\mathbf{x}, \mathbf{i}, \overline{t}_c, \overline{t}_x) \wedge 
  H(\overline{t}_x, \overline{k})\bigr)\bigr)$
                               \label{alg:TemplWinSat:refine}
  \EndWhile
\EndProcedure  

\ProcedureRet{Check}
             {F(\overline{x}),
              (\overline{x}, \overline{i}, \overline{c}, I, T, P)}
             {$(\textsf{correct},\mathbf{x},\mathbf{i})$}
  \If{$\mathsf{sat}=\true$ in 
  $(\mathsf{sat},\mathbf{x}) := \propsatmodel\bigl( 
         (I(\overline{x}) \wedge \neg F(\overline{x})) \vee 
         (F(\overline{x}) \wedge \neg P(\overline{x})) \bigr) $
  }\label{alg:TemplWinSat:cfo}
    \State \textbf{return} $(\true, 
                             \mathbf{x}, 
                             \bigwedge_{i\in \overline{i}} \neg i)$
  \EndIf
  \State $U(\overline{x}, \overline{i}) := \true$
  \While{$\true$} 
    \If{$\mathsf{sat} = \false$ in
    $(\mathsf{sat},\mathbf{x},\mathbf{i}) := \propsatmodel\bigl(
            F(\overline{x}) \wedge 
            U(\overline{x}, \overline{i}) \wedge 
            T(\overline{x}, \overline{i}, \overline{c}, \overline{x}') \wedge 
            \neg F(\overline{x}')\bigr)$
    }\label{alg:TemplWinSat:c20}
      \State \textbf{return} $(\true, \true, \true)$
    \EndIf
    \If{$\mathsf{sat}=\false$ in
    $(\mathsf{sat},\mathbf{c}) := \propsatmodel\bigl(
               F(\overline{x}) \wedge 
               \mathbf{x} \wedge \mathbf{i} \wedge 
               T(\overline{x}, \overline{i}, \overline{c}, \overline{x}') 
               \wedge 
               F(\overline{x}')\bigr)$
    }\label{alg:TemplWinSat:c3}
      \State \textbf{return} $(\false, \mathbf{x},  \mathbf{i})$    
    \Else
      \State $U(\overline{x}, \overline{i}) := U(\overline{x}, \overline{i})
             \wedge \neg \propsatmincore\bigl(\mathbf{x} \wedge \mathbf{i},
             \mathbf{c} \wedge 
             F(\overline{x}) \wedge 
             U(\overline{x}, \overline{i}) \wedge 
             T(\overline{x},\overline{i},\overline{c},\overline{x}') \wedge
             \neg F(\overline{x}')\bigr)$ \label{alg:TemplWinSat:c21}
    \EndIf
  \EndWhile
\EndProcedure  
\end{algorithmic}
\end{algorithm}

\mypara{Algorithm.}
The procedure \textsc{TemplWinSat} in Algorithm~\ref{alg:TemplWinSat} takes as 
input a template $H(\overline{x}, \overline{k})$ for 
a winning area as well as a safety specification $\mathcal{S}$. As output, it 
returns either a concrete winning area $F(\overline{x})$ as an instantiation of 
the template $H(\overline{x}, \overline{k})$, or ``\textsf{fail}'' of no 
instantiation of $H(\overline{x}, \overline{k})$ can be a winning area.  The 
structure of the algorithm is the same as for \textsc{CegisSmt} in 
Algorithm~\ref{alg:CegisSmt}:  The formula $G(\overline{k}, \overline{t})$ 
accumulates constraints that the template parameters $\overline{k}$ have to 
satisfy, where $\overline{t}$ is a vector of auxiliary variables.  
Line~\ref{alg:TemplWinSat:comp} computes candidate template parameter values 
$\mathbf{k}$ in form of a satisfying assignment for $G$.  If the formula is 
unsatisfiable, then no template instantiation can be a winning area and the 
procedure returns ``\textsf{fail}''.  If the formula is satisfiable, a candidate 
winning area $F(\overline{x}) = H(\overline{x}, \mathbf{k})$ is computed using 
the parameter values $\mathbf{k}$.  Next, the candidate is checked in 
Line~\ref{alg:TemplWinSat:check}.  This step is different to \textsc{CegisSmt} 
in Algorithm~\ref{alg:CegisSmt} and explained in the next paragraph.  If the 
candidate is correct, it is returned.  Otherwise, the procedure \textsc{Check} 
returns a counterexample in form of a satisfying assignment 
$\mathbf{x},\mathbf{i}$ for the variables $\overline{x}, \overline{i}$.  The 
meaning of this counterexample is that
$
\exists \overline{c}, \overline{x}'\scope
\bigl(I(\mathbf{x}) \rightarrow H(\mathbf{x},\mathbf{k})\bigr) 
\wedge
\bigl(H(\mathbf{x}, \mathbf{k}) \rightarrow P(\mathbf{x}) \bigr) 
\wedge
\bigl(H(\mathbf{x}, \mathbf{k}) \rightarrow
  \bigl(T(\mathbf{x}, \mathbf{i}, \overline{c}, \overline{x}') \wedge 
  H(\mathbf{x}', \mathbf{k})\bigr)\bigr)
$
does \emph{not} hold, thus witnessing that $\mathbf{k}$ cannot be a solution to 
Equation~(\ref{eq:templ}) yet.  To make sure the candidate of the next iteration 
works also for the counterexample $\mathbf{x},\mathbf{i}$, the constraints 
on $\overline{k}$ are refined accordingly in Line~\ref{alg:TemplWinSat:refine}. 
The variables $\overline{c}$ and  $\overline{x}'$ are renamed to fresh 
auxiliary variables in order to account for their existential quantification.

\mypara{Counterexample computation.}
The procedure \textsc{Check} in Algorithm~\ref{alg:TemplWinSat} is a helper 
routine for \textsc{TemplWinSat} that checks if a given candidate 
$F(\overline{x})$ is a winning area.  It returns $\textsf{correct}=\true$ if 
this is the case.  Otherwise, it sets $\textsf{correct}=\false$ and returns a 
counterexample $\mathbf{x},\mathbf{i}$ witnessing the incorrectness. 
Line~\ref{alg:TemplWinSat:cfo} checks if the first two properties in the 
definition of a winning area $F$, namely $I \rightarrow F$ and $F \rightarrow 
P$, are satisfied (see Definition~\ref{def:winset}).  If this is not the case, 
a satisfying assignment $\mathbf{x}$ is returned as a counterexample witnessing 
this defect.  The input vector $\mathbf{i}$ returned as part of the 
counterexample is irrelevant in this case. Otherwise, \textsc{Check} turns to 
verifying the third property of a winning area, namely $F \rightarrow \FS(F)$.  
Here, we search for a counterexample $\mathbf{x},\mathbf{i}$ such that no 
value $\mathbf{c}$ can prevent the system from leaving $F(\overline{x})$ if the 
environment picks input $\mathbf{i}$ from state $\mathbf{x}\models 
F(\overline{x})$.  The same kind of counterexample computation was performed 
already by \textsc{SatWin1} in Algorithm~\ref{alg:SatWin1}, so we simply reuse 
this algorithm here.  The difference is that $F(\overline{x})$ is not refined by 
\textsc{Check}.  Thus, there is no need for lazy updates of $\neg 
F(\overline{x}')$, which renders quite some lines of Algorithm~\ref{alg:SatWin1} 
obsolete.

\mypara{An optimization.}
The check in Line~\ref{alg:TemplWinSat:cfo} of Algorithm~\ref{alg:TemplWinSat} 
can actually be omitted if we ensure that
$\forall \overline{x}, \overline{k} \scope I(\overline{x}) \rightarrow 
H(\overline{x}, \overline{k})$ 
and 
$\forall \overline{x}, \overline{k} \scope H(\overline{x}, 
\overline{k}) \rightarrow P(\overline{x})$ holds by the construction of the 
template
$H(\overline{x}, \overline{k})$.  This can easily be achieved by taking any
template $H'(\overline{x}, \overline{k})$ and defining a new template 
$H(\overline{x}, \overline{k}) = \bigl(H'(\overline{x}, \overline{k}) \wedge 
P(\overline{x})\bigr) \vee I(\overline{x})$, given that $I(\overline{x}) \wedge
\neg P(\overline{x})$ is unsatisfiable (otherwise the specification is 
trivially unrealizable).  We use this optimization in our implementation.

\mypara{Incremental solving.}
Algorithm~\ref{alg:TemplWinSat} is well suited for incremental \acs{SAT} 
solving.  We propose to use three solver instances.  The first one stores $G$ 
and is used for Line~\ref{alg:TemplWinSat:comp}. Constraints are only added to 
$G$ in Line~\ref{alg:TemplWinSat:refine}, so no re-initialization is needed.
The second solver instances stores 
$F(\overline{x}) \wedge 
   U(\overline{x}, \overline{i}) \wedge 
   T(\overline{x}, \overline{i}, \overline{c}, \overline{x}') \wedge 
   \neg F(\overline{x}')$
and is used in Line~\ref{alg:TemplWinSat:c20} and~\ref{alg:TemplWinSat:c21}.
It is (re-)initialized when \textsc{Check} is called.  After that, clauses are 
only added to $U$ in Line~\ref{alg:TemplWinSat:c21}.  Finally, the third solver 
instance stores 
$F(\overline{x}) \wedge
   T(\overline{x}, \overline{i}, \overline{c}, \overline{x}') \wedge 
   F(\overline{x}')$
and is used in Line~\ref{alg:TemplWinSat:c3}.  This instance is also 
(re-)initialized whenever \textsc{Check} is called.  This \acs{CNF} does not 
change at all during the execution of \textsc{Check}.  The conjunctions 
with $\mathbf{x}$, $\mathbf{i}$ and $\mathbf{c}$ are realized with assumption 
literals that are temporarily asserted.

\subsubsection{Discussion}

The template-based approach has a potential for finding simple winning areas 
quickly.  There may exist many winning areas that satisfy the constraints given 
by Definition~\ref{def:winset}.  The algorithms \textsc{SafeWin}, 
\textsc{QbfWin} and \textsc{SatWin1} discussed earlier will always compute the 
largest possible winning area (modulo unreachable states if used with 
optimization \textsf{RG} or \textsf{RC}).  The template-based approach is more 
flexible in this respect.  As an extreme example, suppose that there is only one 
initial state, it is safe, and the system can enforce that the play stays in 
this state. Suppose further that the winning region is complicated.  The 
template-based approach may find $F=I$ quickly, while the other approaches may 
require many iterations to compute the winning region.  

On the other hand, the template-based approach can be expected to scale poorly 
if no simple winning area exists or if the synthesis problem is 
unrealizable.  Starting with a small expressiveness parameter $N$, 
Equation~(\ref{eq:templ}) will be unsatisfiable, so $N$ is increased.  With 
increasing $N$, the search space for the solver increases, which results in 
longer execution times.  For unrealizable specifications, we can only terminate 
once $N > 2^{|\overline{x}|}$ (when using our \acs{CNF} or AND-inverter graph 
templates).  Except for specifications with a very low numbers of state 
variables, a timeout is likely to be hit before this point can be reached.
   
\subsection{Reduction to \acf{EPR}} \label{sec:red_to_epr}

The template-based approach presented in the previous section may work well if a 
simple representation of a winning area exists.  However, one drawback is the 
need to select a template, which is a delicate matter.
It would be more desirable to directly compute a winning area as a Skolem 
function of a quantified formula.  Unfortunately, the definition of a winning 
area (Definition~\ref{def:winset}) not only involves the winning area 
$F(\overline{x})$ but also its next-state copy $F(\overline{x}')$. 
Hence, we have to compute two Skolem functions, and the two functions have to be 
functionally consistent.  This problem cannot be formulated as a \acs{QBF} 
formula with a linear quantifier prefix, but requires more expressive logics.  

\subsubsection{Using Henkin Quantifiers}  
One solution is to use so-called Henkin quantifiers~\cite{Henkin1961}, 
\index{Henkin quantifier} which are 
quantifiers that are only partially ordered.  This partial order can be used to 
restrict variable dependencies.  In particular, a winning area 
$F(\overline{x})$ 
can be computed as a Skolem function for the variable $w$ in 
\[
\begin{array}{l}
\forall \overline{x} \scope
\exists w \scope
\forall \overline{i} \scope
\exists \overline{c} \scope\\
\forall \overline{x}' \scope
\exists w' \scope
\end{array}
\bigl(I(\overline{x}) \rightarrow w\bigr) \wedge
\bigl(w \rightarrow P(\overline{x})\bigr) \wedge
\bigl(w \wedge 
     T(\overline{x},\overline{i},\overline{c},\overline{x}')
     \rightarrow w'
\bigr) \wedge
\bigl((\overline{x} = \overline{x}') \rightarrow (w = w')\bigr).
\]
This formulation ensures that the Skolem function $F(\overline{x})$ for $w$ can 
only depend on $\overline{x}$, and the Skolem function $G(\overline{x}')$ for 
$w'$ can only depend on $\overline{x}'$.  The last constraint enforces 
functional consistency between $F$ and $G$, i.e., $F$ and $G$ are actually the 
same function but applied to different parameters.  The logic of applying Henkin 
quantifiers to propositional formulas is called \ac{DQBF} 
\index{DQBF@\acs{DQBF}} and was first 
described by Peterson and Reif~\cite{PetersonR79}.  Deciding 
whether a \ac{DQBF} formula is satisfiable is NEXPTIME 
complete~\cite{PetersonR79}.  In addition to this high complexity, only a few 
approaches and tools to solve \ac{DQBF} formulas have recently been 
proposed~\cite{FroehlichKB12,FrohlichKBV14}.  For this reason, we did not 
implement a \ac{DQBF}-based solution but we rather use \ac{EPR}, where 
mature solvers are available.

\subsubsection{Using \acf{EPR}} \label{sec:hw:win:eprim}

Recall from Section~\ref{sec:prelim:epr} that \ac{EPR} is the set of first-order 
logic formulas of the form $\exists \overline{x}\scope \forall \overline{y} 
\scope F(\overline{x}, \overline{y})$, where $F$ is a quantifier-free formula in 
\acs{CNF} that must not contain function symbols but can contain predicate 
symbols.  These predicate symbols are implicitly quantified existentially.  We 
seek a winning area $F(\overline{x})$ satisfying the three properties of 
Definition~\ref{def:winset}, which can be combined to
\[
\exists F \scope
\forall \overline{x}, \overline{i} \scope
\exists \overline{c}, \overline{x}' \scope
\bigl(I(\overline{x}) \rightarrow F(\overline{x})\bigr) 
  \wedge
\bigl(F(\overline{x}) \rightarrow P(\overline{x}) \bigr) 
  \wedge
\Bigl(F(\overline{x}) \rightarrow \notag
  \bigl(T(\overline{x}, \overline{i}, \overline{c}, \overline{x}') \wedge 
  F(\overline{x}')\bigr)\Bigr).
\]
In order to transform this constraint into \ac{EPR}, we need to perform several 
steps, which are similar to those by Seidl et al.~\cite{SeidlLB12} when 
transforming \acs{QBF} formulas into \acs{EPR}. 

\mypara{Step 1.}
We replace all the Boolean variables $\overline{x}, \overline{i}, \overline{c}, 
\overline{x}'$ by corresponding first-order domain variables.  Since the 
original variables can only take two different values, we introduce a unary 
predicate $V$ to represent the truth value of a domain variable. We also 
introduce two domain constants $\top$ and $\bot$ to encode $\true$ and $\false$, 
and add the axioms $V(\top)$ and $\neg V(\bot)$ to the final \acs{EPR} formula. 
 
\mypara{Step 2.}
We introduce predicate symbols $I(\overline{x})$, $P(\overline{x})$, 
$T(\overline{x}, \overline{i}, \overline{c}, \overline{x}')$ and 
$F(\overline{x})$ to represent the different parts of the formula.  The 
predicates $I$, $P$ and $T$ are equipped with additional constraints that fully 
define their truth value based on the truth values of the variables on which 
they depend.  The predicate $F$ is left unconstrained because it represents the
winning area we wish to compute.

\mypara{Step 3.}
We eliminate the existential quantification over 
$\overline{c}$ and $\overline{x}'$.   Since $T$ is 
deterministic and complete (Definition~\ref{def:safety}), the
one-point rule (\ref{eq:op1}) can be used to eliminate the 
existential quantification over $\overline{x}'$: 
\[
\exists F \scope
\forall \overline{x}, \overline{i} \scope
\exists \overline{c} \scope \forall \overline{x}' \scope
\bigl(I(\overline{x}) \rightarrow F(\overline{x})\bigr) 
  \wedge
\bigl(F(\overline{x}) \rightarrow P(\overline{x}) \bigr) 
  \wedge
\Bigl(\bigl(F(\overline{x}) \wedge
  T(\overline{x}, \overline{i}, \overline{c}, \overline{x}')\bigr)
  \rightarrow 
  F(\overline{x}')\Bigr).
\]
The existential quantification over $\overline{c}$ is eliminated by 
Skolemization: for every $c_j\in \overline{c}$, we introduce a new 
predicate $C_j(\overline{x}, \overline{i})$.  All occurrences of $V(c_j)$ in the 
definition of $T$ are then replaced by $C_j(\overline{x}, \overline{i})$.  This 
gives a formula of the form
\[
\exists F, C_1,\ldots, C_{|\overline{c}|} \scope
\forall \overline{x}, \overline{i}, \overline{x}' \scope
\bigl(I(\overline{x}) \rightarrow F(\overline{x})\bigr) 
  \wedge
\bigl(F(\overline{x}) \rightarrow P(\overline{x}) \bigr) 
  \wedge
\Bigl(\bigl(F(\overline{x}) \wedge
  T(\overline{x}, \overline{i}, \overline{x}')\bigr)
  \rightarrow 
  F(\overline{x}')\Bigr).
\]

\mypara{Step 4.} The body of the resulting formula needs to be encoded into 
\acs{CNF}.  Since we have a conjunction of implications on the top-level, this 
is mainly a matter of encoding the constraints defining $I$, $P$ and $T$ into 
\acs{CNF}.  Note that the standard Tseitin~\cite{Tseitin83} or 
Plaisted-Greenbaum~\cite{PlaistedG86} transformations introduce new auxiliary 
variables that are quantified existentially on the innermost level.  Since this 
is not allowed in \acs{EPR}, these auxiliary variables need to be eliminated 
again.  Similar to the elimination of the variables $\overline{c}$ in Step 3, we 
do this by introducing new predicates.  To increase efficiency, we do not pass 
all variables of $\overline{x}, \overline{i}, \overline{x}'$ as arguments to 
the new predicates, but rather analyze the variable dependencies structurally 
and pass only the relevant ones.

\mypara{Solving the resulting \acs{EPR} formula.}  We call \iprover on the 
resulting \acs{EPR} formula.  \iprover is an instantiation-based first-order 
theorem prover that can produce implementations for the predicates that 
occur in the formula.  This means that the solver directly returns a winning 
area 
$F(\overline{x})$.  Since we represent the truth values of the 
variables $c_j \in \overline{c}$ with predicates $C_j(\overline{x}, 
\overline{i})$, we 
can also extract an implementation from the solver 
result.\footnote{Because of the poor scalability of the \acs{EPR} approach in 
our experiments, we did not implement a parser for the predicate 
implementations returned by \iprover in our tool yet.}  That is, there is no 
need to apply the circuit construction methods that will be presented in 
Chapter~\ref{sec:hw:circ} when using the \acs{EPR} synthesis approach.

\mypara{Discussion.}  Similar to the template-based approach from 
Section~\ref{sec:hw:templ}, this approach does not compute the winning region 
but some winning area.  It can thus benefit from situations where the winning 
region is complicated but a simple winning area exists.  In contrast to the 
template-based approach, there is no need to guess a template and to 
increase the expressiveness of the template if no solution is found.  The price 
that is payed for this benefit is the higher worst-case complexity for checking 
the satisfiability of the constructed formulas because a more expressive logic 
is used.

\subsection{Parallelization} \label{sec:hw_par}

\begin{wrapfigure}[12]{r}{0.35\textwidth}
\vspace{-4mm}
\centering
  \includegraphics[width=0.34\textwidth]{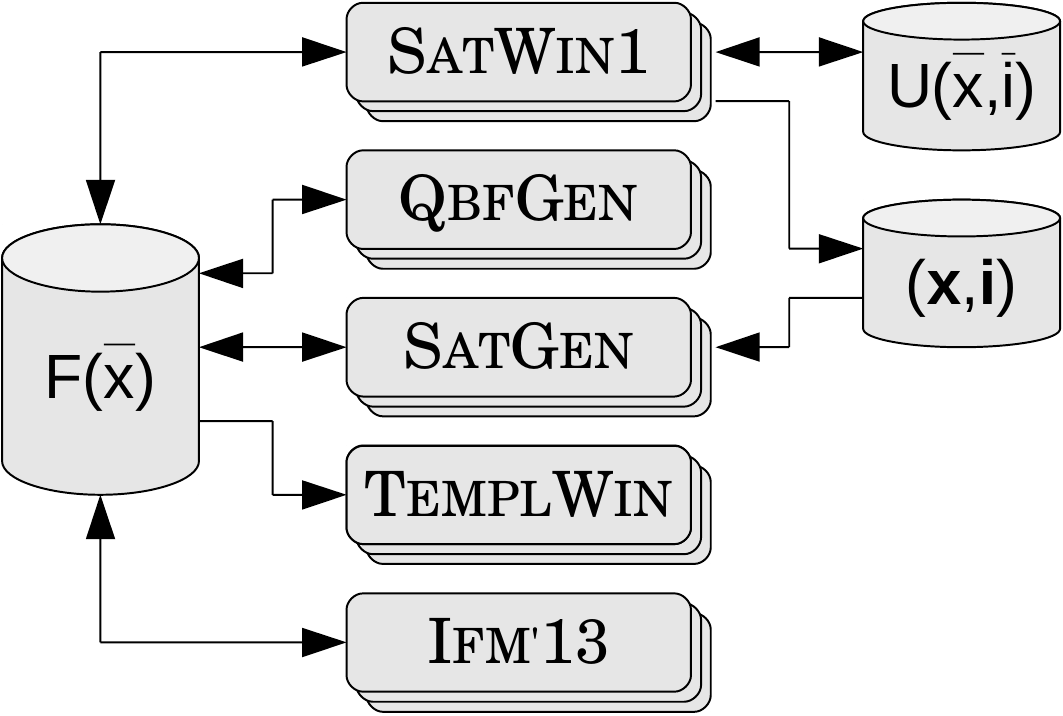}
\caption{Parallelized strategy computation.}
\label{fig:parallel}
\end{wrapfigure}

The various methods for strategy computation presented so far have different 
strengths and weaknesses and, consequently, perform well on different classes of 
benchmarks.  To a smaller extent, different characteristics can also be observed 
within one method when run with different optimizations or solvers.  
In this section, we thus combine different methods and configurations in the 
hope to inherit all their strengths while compensating their weaknesses.  We do 
this in a parallelized way, where individual methods are running in separate 
threads but share discovered information that may be helpful for others.

Figure~\ref{fig:parallel} gives a proposal for combining a promising subset of 
the methods (or fragments thereof).  Arrows denote information that is 
exchanged between threads. 

\mypara{\textsc{SatWin1} threads.}
The \textsc{SatWin1} threads execute the \textsc{SatWin1} procedure from 
Algorithm~\ref{alg:SatWin1} and can be seen as the main workhorse.  Individual 
\textsc{SatWin1} threads can be run with or without optimization \textsf{RG}, 
with or without quantifier expansion, and with different \acs{SAT} solvers.  All 
newly discovered clauses of the winning region $F(\overline{x})$ are put into a 
central database and communicated to the other threads.  Newly discovered 
$U$-clauses are also shared between \textsc{SatWin1} threads.  In order for 
this to work, the \textsc{SatWin1} 
threads need be synchronized regarding their restarts of \textsf{solverC}, i.e., 
they need to work with the same version of $\neg G(\overline{x}')$ at any time. 
If several \textsc{SatWin1} threads are running in a mode where they perform 
universal expansion, the expansion is only done by one thread (while the others 
sleep) in order not to waste resources (like stressing the memory bus 
unnecessarily).

\mypara{\textsc{QbfGen} threads.}  The \textsc{QbfGen} threads take existing 
clauses from $F$ and attempt to generalize them further by eliminating more 
literals.  This is done as in Line~\ref{alg:QbfWin:loop0} to 
Line~\ref{alg:QbfWin:loop1} of the \textsc{QbfWin} procedure in 
Algorithm~\ref{alg:QbfWin} using a \acs{QBF} solver.  If a clause could be 
shortened, the reduced clause is communicated to all other threads.  Individual 
\textsc{QbfGen} threads can be run with or without optimization \textsf{RG}, 
with or without \acs{QBF} preprocessing, and with or without incremental 
\acs{QBF} solving (the combination of incremental solving plus preprocessing is 
not available). 

\mypara{\textsc{SatGen} threads.} These threads take counterexamples 
$(\mathbf{x},\mathbf{i})$, as computed by the \textsc{SatWin1} 
threads, and compute \emph{all} generalizations using a \acs{SAT} solver (as 
illustrated in Figure~\ref{fig:qbf_allgen}).  The resulting $F$-clauses are 
shared.

\mypara{\textsc{TemplWin} threads.} These threads implement the template-based 
method from Section~\ref{sec:hw:templ}, using \acs{CNF} templates of increasing 
size.  The clauses from $F$ are considered as fixed over-approximation of the 
winning area to compute --- the threads only compute additional clauses such 
that a winning area is obtained.  A timeout of $20$ seconds makes the thread try 
again (with a potentially refined set $F$ of fixed clauses) if a solution cannot 
be found quickly.  The short timeout is justified by the observation that the 
template-based approach either finds a solution quickly or not at all.  The 
\acs{QBF}-based implementation and the \acs{SAT}-based implementation of the 
template-based approach are alternated from timeout to timeout.  The 
\textsc{TemplWin} threads are information sinks: the only information 
communicated back to other threads is a request to terminate if a solution has 
been found.

\mypara{\textsc{Ifm'13} threads.}  These threads execute a reimplementation of 
the \acs{SAT}-based synthesis method proposed by Morgenstern et 
al.~\cite{MorgensternGS13}.  This method maintains an over-approximation 
$G(\overline{x})$ of the winning region $W(\overline{x})$ as well as 
over-approximations of sets of states from which the environment can win the 
game in different numbers of steps.  We couple $G(\overline{x})$ with 
$F(\overline{x})$: If new clauses are added to $G(\overline{x})$, then they are 
also added to $F(\overline{x})$ and communicated to the other threads.  If other 
threads discover new $F$-clauses, they are also added to $G$ in the 
\textsc{Ifm'13} threads.

\mypara{Configuration.}  When only one thread is available, we make it execute 
\textsc{SatWin1} with optimization \textsf{RG}, quantifier expansion and 
\minisat as underlying \acs{SAT} solver.  If two threads are available, the 
second one executes \textsc{TemplWin} (with \depqbf, \bloqqer and \minisat).  If 
three threads are available, the third thread runs \textsc{Ifm'13} using 
\minisat.  With four threads, we also use a second instance of \textsc{SatWin1}, 
but with quantifier expansion disabled.  With five threads, we also include a 
\textsc{SatGen} thread, and with six threads we also include a \textsc{QbfGen} 
thread.  

\mypara{Variations.}  The current realization always shares \emph{all} 
discovered clauses that refine the winning region with all other threads.  
Another option is to share only \emph{small} clauses (where the number of 
literals is below some threshold) in order to reduce the communication overhead. 
In general, smaller clauses refine the winning region more substantially than 
larger ones, so this approach would focus on communicating only significant 
findings.  Another promising extension is to include also threads that run 
\acs{BDD}-based algorithms, e.g., a \acs{BDD}-based realization of 
Algorithm~\ref{alg:SafeWin}.  The \acs{BDD}-based threads can directly use 
clauses discovered by other threads to refine the \acs{BDD} that represents the 
winning region.  Communication in the other direction is possible as well: many 
\acs{BDD} libraries provide functions to convert a \acs{BDD} into \acs{CNF}.
While it may be expensive to share all clauses of such a \acs{CNF} translation,
it may still be beneficial to factor out a set of small clauses and communicate 
them.

\mypara{Discussion.}  The main purpose of our parallelization is to combine 
different methods that complement each other.  Exploiting hardware parallelism 
in only a secondary aspect because, due to the high worst-case complexities, 
even a speedup factor of, say, $10$ may have little impact on the ability of 
solving larger benchmark instances.  Furthermore, we do not claim that our 
choice of distributing workload over the threads is in any way optimal.  We 
rather selected the methods to run in individual threads quite greedily, based 
on the performance results when running methods in isolation (see 
Chapter~\ref{sec:hw:exp}) and based on experiments with subsets of the 
benchmarks.  However, there is such a plethora of possibilities for combining 
different methods, fragments thereof, optimizations, heuristics and solver 
configurations that finding particularly good configurations is quite an 
intricate task.  Hence, we rather see the main contribution of our 
parallelization in providing a ``playground'' for combining different approaches 
and configurations.  It demonstrates that a parallelized way of combining 
different \acs{SAT}-based synthesis approaches is easily possible.  This stands 
in contrast to \acs{BDD}-based synthesis algorithms, where a parallelization is 
often much more difficult to achieve.  Our parallelization goes far beyond a 
pure portfolio approach because fine-grained information about refinements of 
the winning region, discovered counterexamples and unsuccessful attempts to 
compute counterexamples is exchanged between the threads as soon as discovered.  
This information can speed up the progress in other threads and thus stimulate 
``cross-fertilization'' effects.

%% file: 04circuit.tex
\section{From Strategies to Circuits} \label{sec:hw:circ}

In Chapter~\ref{sec:hw:win}, we presented a number of \acs{SAT}-based methods 
to 
compute a strategy for defining the control signals $\overline{c}$ such that a 
given safety specification is enforced.  Recall that such a 
strategy is a formula  $S(\overline{x},\overline{i},\overline{c},\overline{x}')$ 
such that
$\forall \overline{x}, \overline{i} \scope
          \exists \overline{c}, \overline{x}' \scope
          S(\overline{x},\overline{i},\overline{c},\overline{x}').$
That is, for every state $\mathbf{x}$ and input $\mathbf{i}$, the strategy will 
contain at least one vector of control values $\mathbf{c}$ that is allowed in 
this situation.  In many situations, many control values can be allowed, though. 
The task is now to compute a system implementation in form of a function $f: 
2^{\overline{x}} \times 2^{\overline{i}} \rightarrow 2^{\overline{c}}$ to 
uniquely define the control signals $\overline{c}$ based on the current state 
variables $\overline{x}$ and the uncontrollable inputs $\overline{i}$.  The 
system implementation $f$ is supposed to implement the strategy in the sense 
that
$\forall \overline{x}, \overline{i} \scope
       \exists \overline{x}' \scope
 S\bigl(\overline{x},\overline{i},f(\overline{x},\overline{i}),\overline{x}'
 \bigr)$ 
holds.  That is, for all concrete assignments $\mathbf{x}, \mathbf{i}$, the 
control variable assignment $\mathbf{c} = f(\mathbf{x}, \mathbf{i})$ computed 
by $f$ must be allowed by the strategy $S$.  Finally, this function $f$ needs to 
be implemented as a circuit.  Obviously, we prefer fast algorithms that produce 
small circuits.  In order to achieve this, the freedom in the strategy relation 
$S$ needs to be exploited cleverly.          

A cofactor-based algorithm to solve the problem has already been presented in 
Section~\ref{sec:prelim:sasyalg}.  It can be seen as the ``standard method'' for 
computing an implementation from a strategy, and can easily be implemented using 
\acp{BDD}.  In the following subsections, we will present alternative 
approaches that use \acs{SAT}- or \acs{QBF} solvers instead.  
The presented approaches are not specific to safety specifications.  However,
in many cases, the specific structure of strategies 
$S(\overline{x},\overline{i},\overline{c},\overline{x}') =   
T(\overline{x},\overline{i},\overline{c},\overline{x}') \wedge \bigl(   
W(\overline{x}) \rightarrow W(\overline{x}')\bigr)
$
for safety specifications can be exploited.  We will thus always 
present the general approach first, and then discuss an efficient implementation 
for safety synthesis problems.  As a preprocessing step to all our methods, we 
simplify $W$ by calling \textsc{CompressCnf} (see 
Algorithm~\ref{alg:CompressCnf}) with literal dropping enabled in order to 
remove redundant literals and clauses from $W$.  As a postprocessing step to 
all our methods, we invoke the tool \abc~\cite{BraytonM10} in order to reduce 
the size of the produced circuits.  

\subsection{\acs{QBF} Certification} \label{sec:hw:circ_qbfcert}

A system implementation can be computed as Skolem function for the 
signals $\overline{c}$ in $\forall \overline{x}, 
\overline{i} \scope \exists \overline{c}, \overline{x}' \scope 
S(\overline{x},\overline{i},\overline{o}, \overline{x}')$.
The \qbfcert~\cite{NiemetzPLSB12} framework by Niemetz et al.\ computes such 
Skolem functions for satisfiable \acp{QBF} from proof traces produced by the 
\depqbf~\cite{LonsingB10} solver.  The resulting Skolem functions are produced 
as circuits in \aiger format.  Hence, in our setting, a single call to \qbfcert 
suffices to compute a system implementation in form of a circuit.

\subsubsection{Efficient Implementation for Safety Synthesis Problems} 

While the basic approach is simple, we can still apply some optimizations to 
increase the efficiency for the case of safety synthesis problems.

\noindent
\mypara{\acs{QBF} formulation.} 
Instead of computing a Skolem function for the variables 
$\overline{c}$ in the formula
\begin{equation} 
\forall \overline{x}, \overline{i} \scope 
\exists \overline{c}, \overline{x}'\scope 
T(\overline{x},\overline{i},\overline{c},\overline{x}') \wedge 
\bigl(W(\overline{x}) \rightarrow W(\overline{x}')\bigr)
\label{eq:cert00}
\end{equation} 
we rather compute a Herbrand function in its negation
$ 
\exists \overline{x}, \overline{i} \scope 
\forall \overline{c}, \overline{x}' \scope 
\neg T(\overline{x},\overline{i},\overline{c},\overline{x}') \vee \bigl(   
W(\overline{x}) \wedge \neg W(\overline{x}')\bigr).
$
Because $T$ is both deterministic and complete (Definition~\ref{def:safety}), 
the one-point rule (\ref{eq:op1}) can 
be applied to turn the universal quantification over $\overline{x}'$ into an 
existential quantification:
\begin{equation} 
\exists \overline{x}, \overline{i} \scope 
\forall \overline{c}\scope
\exists \overline{x}' \scope 
T(\overline{x},\overline{i},\overline{c},\overline{x}') \wedge    
W(\overline{x}) \wedge \neg W(\overline{x}').
\label{eq:cert01}
\end{equation} 
Just like most \ac{QBF} solvers, \qbfcert requires a \acs{PCNF} as input. Since 
most of our methods to compute a winning region (or winning area) produce $W$ in 
\acs{CNF}, we only need to transform $T$ and $\neg W(\overline{x}')$ into  
\acs{CNF}.  In contrast, using Equation~(\ref{eq:cert00}) would require an 
additional \acs{CNF} encoding of the implication $W(\overline{x}) 
\rightarrow W(\overline{x}')$.  Another advantage of using 
Equation~(\ref{eq:cert01}) lies in the size of the proofs: since the \acs{QBF} is 
now unsatisfiable, the \qbfcert framework processes a clause resolution proof 
instead of a cube resolution proof.  These clause resolution proofs are often 
smaller.

\begin{wrapfigure}[9]{r}{0.62\textwidth}
\vspace{-8mm}
\centering
\begin{minipage}{0.61\textwidth}
\begin{algorithm}[H]
\caption[\textsc{NegLearn}: Computing a \acs{CNF} representation for the 
negation of a formula $F(\overline{x})$]
{\textsc{NegLearn}: Computing a \acs{CNF} representation for the 
negation of a formula $F(\overline{x})$.}
\label{alg:NegLearn}
\begin{algorithmic}[1]
\ProcedureRet{NegLearn}
             {F(\overline{x})}
             {$\neg F(\overline{x})$ in \acs{CNF}}
  \State $N(\overline{x}) := \true$
  \While{$\mathsf{sat}=\true$ in $(\mathsf{sat},\mathbf{x}):=
  \propsatmodel\bigl(F(\overline{x}) \wedge N(\overline{x})\bigr)$}
    \State $N(\overline{x}) := N(\overline{x}) \wedge \neg 
           \propsatmincore\bigl(\mathbf{x}, \neg F(\overline{x})\bigr)$
  \EndWhile
  \State \textbf{return} $N(\overline{x})$
\EndProcedure
\end{algorithmic}
\end{algorithm}
\end{minipage}
\end{wrapfigure}
\mypara{Negation of $W(\overline{x}')$.} For complex benchmarks, the auxiliary 
files produced by \qbfcert can still grow very large (hundreds of GB). One 
reason is that a straightforward \acs{CNF} encoding of $\neg W(\overline{x}')$ 
requires many auxiliary variables and clauses.  We can reduce the size of the 
auxiliary files (by up to a factor of 30 in our experiments) by computing a 
\acs{CNF} representation of $\neg W(\overline{x}')$ without introducing 
auxiliary variables.  The procedure \textsc{NegLearn} in 
Algorithm~\ref{alg:NegLearn} computes such a negation with query learning.  It 
follows the principle of \textsc{CnfLearn}, shown in 
Algorithm~\ref{alg:CnfLearn}, and uses a \acs{SAT} solver to implement the 
queries:  As long as $N$ is not yet equivalent to $\neg F$, i.e., $F \wedge N$ 
is still satisfiable, \textsc{NegLearn} refines $N$ with a clause that excludes 
the cube $\mathbf{x}$ witnessing this insufficiency.  By taking the 
unsatisfiable core, the clause eliminates also other counterexamples.  Since 
clauses are only added to $N$, \textsc{NegLearn} is well suited for incremental 
SAT solving.

\subsubsection{Discussion} 

\mypara{Dependencies between control signals.}
In contrast to \textsc{CofSynt} from Algorithm~\ref{alg:CofSynt}, the \acs{QBF} 
certification approach computes a circuit for all control signals 
simultaneously.  This can be both an advantage and a disadvantage.  The 
advantage is that dependencies between control signals can potentially be 
handled more effectively.  \textsc{CofSynt} can only take local decisions and 
fixes an implementation for one control signal without considering the 
consequences on other control signals (as long as some solution for the other 
signals still exist).  The \acs{QBF} certification approach is free to make 
global decisions when fixing the individual circuits.  On the other hand, 
considering all control signals simultaneously instead of decomposing the 
problem into smaller subproblems can also be a scalability disadvantage.

\mypara{Dependencies on reasoning engine.}  
The performance of \qbfcert as well as the quality of the resulting circuit 
depend on the ability of \depqbf to find a compact unsatisfiability proof 
quickly.  In this sense, the technique strongly depends on the underlying 
symbolic reasoning engine.  This is similar to \textsc{CofSynt} when implemented 
using \acp{BDD}, where the ability to find a good variable ordering can 
influence the circuit size and the execution time heavily.

\subsection{\acs{QBF}-Based Query Learning} \label{sec:hw:extr:qbflearn}

In this section, we introduce an approach that is also based on \acs{QBF} 
solving, but constructs circuits for one control signal after the other.  In 
this respect, it is more similar to \textsc{CofSynt} presented in 
Algorithm~\ref{alg:CofSynt}.  However, in contrast to \textsc{CofSynt}, we 
rely on query learning to exploit implementation freedom in the strategy in 
order to obtain small circuits.

The query learning algorithms introduced in Section~\ref{sec:prelim:ql} compute 
a certain representation of a given target formula $G(\overline{x})$ precisely.  
That is, the resulting formula $F(\overline{x})$ will be equivalent to the 
target $G(\overline{x})$.  This is achieved by starting with some initial 
approximation for $F$, and refining this approximation based on counterexamples 
witnessing that $F \neq G$.  These counterexamples are also generalized to 
speed up the progress.  The same formula $G$ is used both for computing 
counterexamples and for generalizing them.  However, by using two different 
formulas $G_1$ and $G_2$ in these two phases, we can also compute a function $F$ 
such that $G_1 \rightarrow F \rightarrow G_2$.  This idea can be used to exploit 
freedom in defining $F$, where the freedom is defined by (the difference 
between) $G_1$ and $G_2$.  Note that $F$ is actually an interpolant for $G_1 
\wedge \neg G_2$ (see Section~\ref{sec:prelim:prop}).  Thus, this way of query 
learning with freedom can be seen as a special way to compute interpolants.  
However, depending on the underlying reasoning engine used in query learning, 
the formulas $G_1$ and $G_2$ do not have to be quantifier-free.  Furthermore, by 
choosing an appropriate learning algorithm, we can control the shape of $F$.  
For instance, a \acs{CNF} learning algorithm will produce $F$ in form of a 
\acs{CNF} formula.

In the following, we will present a circuit synthesis algorithm based on 
\acs{CNF} learning using a \acs{QBF} solver.  \acs{CNF} learning is particularly 
suitable in this setting because \acs{QBF} solvers require formulas in 
\acs{PCNF}, so building up the solution in \acs{CNF} reduces the overhead 
(especially in terms of formula size) imposed by \acs{CNF} transformations.  
Solutions with other learning algorithms have been proposed by Ehlers et 
al.~\cite{EhlersKH12}. After introducing the basic algorithm, 
we will also discuss an efficient realization for safety synthesis problems. 

\subsubsection{\acs{QBF}-Based \acs{CNF} Learning}

\begin{wrapfigure}[23]{r}{0.68\textwidth}
\vspace{-8mm}
\centering
\begin{minipage}{0.67\textwidth}
\begin{algorithm}[H]
\caption{\textsc{QbfSynt}: Synthesizing circuits with
\acs{QBF}-based \acs{CNF} learning.}
\label{alg:QbfSynt}
\begin{algorithmic}[1]
\Procedure{QbfSynt}
          {$S(\overline{x},\overline{i},\overline{c},\overline{x}')$}
  \For{$c_j \in \overline{c}$}
    \State $M_1(\overline{x},\overline{i}) :=
      \forall \overline{c},\overline{x}'\scope
      \neg
       S\bigl(\overline{x},\overline{i},
       (c_0,\ldots,c_{j-1},\false,c_{j+1},\ldots,c_n),
       \overline{x}'
       \bigr)$\label{alg:QbfSynt:m1}
    \State $M_0(\overline{x},\overline{i}) :=
      \forall \overline{c},\overline{x}'\scope
      \neg
       S\bigl(\overline{x},\overline{i},
       (c_0,\ldots,c_{j-1},\true,c_{j+1},\ldots,c_n),
       \overline{x}'
       \bigr)$\label{alg:QbfSynt:m0}
    \State $F_j(\overline{x},\overline{i}) := \true$ \label{alg:QbfSynt:init}
    \While{$\mathsf{sat}$ in $(\mathsf{sat},\mathbf{x}, \mathbf{i}) := 
            \qbfsatmodel\bigl(
            \exists \overline{x},\overline{i} \scope
            F_j(\overline{x},\overline{i}) \wedge
            M_0(\overline{x},\overline{i}) \bigr)$}
            \label{alg:QbfSynt:check}
      \State $\mathbf{d}_g := \textsc{QbfGeneralize}\bigl(
                              \mathbf{x} \wedge \mathbf{i},
                              M_1(\overline{x},\overline{i})
                              \bigr)$
      \State $F_j(\overline{x},\overline{i}) := 
              F_j(\overline{x},\overline{i}) \wedge
             \neg \mathbf{d}_g$
    \EndWhile \label{alg:QbfSynt:end}
  \State $\textsc{dumpCircuit}\bigl(c_j, 
          F_j(\overline{x},\overline{i})\bigr)$ 
           \label{alg:QbfSynt:dump}
  \State $S(\overline{x},\overline{i},\overline{c},\overline{x}') := 
          S(\overline{x},\overline{i},\overline{c},\overline{x}') \wedge 
          \bigl(c_j \leftrightarrow F_j(\overline{x},\overline{i})\bigr)$
          \label{alg:QbfSynt:resub}
  \EndFor
\EndProcedure
\ProcedureRetL{QbfGeneralize}
          {\mathbf{d}, M_1(\overline{x},\overline{i})}
          {$\mathbf{d}_g \subseteq \mathbf{d}$ such that $\mathbf{d}_g
            \wedge M_1$ is unsatisfiable}
          {42mm}
    \State $\mathbf{d}_g := \mathbf{d}$
    \For{each literal $l$ in $\mathbf{d}_g$}
      \State $\mathbf{d}_t := \mathbf{d}_g \setminus \{l\}$
      \If{$\neg \qbfsatmodel\bigl(
            \exists \overline{x},\overline{i} \scope
            \mathbf{d}_t \wedge
            M_1(\overline{x},\overline{i})
            \bigr)$}\label{alg:QbfSynt:gen}
        \State $\mathbf{d}_g := \mathbf{d}_t$
      \EndIf
    \EndFor
    \State \textbf{return} $\mathbf{d}_g$
\EndProcedure            
\end{algorithmic}
\end{algorithm}
\end{minipage}
\end{wrapfigure}
\textsc{QbfSynt} in Algorithm~\ref{alg:QbfSynt} presents a 
\acs{CNF} learning algorithm, implemented using a \acs{QBF} solver.  It 
synthesizes a circuit from a given strategy $S(\overline{x}, \overline{i}, 
\overline{c}, \overline{x}')$ while exploiting the freedom in $S$ in order to 
obtain small circuits.  \textsc{QbfSynt} does not return any result but directly 
dumps the produced circuits.  Individual circuits are computed for one $c_j \in 
\overline{c}$ after the other.  In this respect, \textsc{QbfSynt} is similar to 
\textsc{CofSynt} (Algorithm~\ref{alg:CofSynt}) but different from \acs{QBF} 
certification as presented in Section~\ref{sec:hw:circ_qbfcert}. 

\mypara{Definition of $M_1$ and $M_0$.}
Line~\ref{alg:QbfSynt:m1} of \textsc{QbfSynt} computes the formula 
$M_1(\overline{x},\overline{i})$, which describes all 
$(\overline{x},\overline{i})$-assignments for which the current control signal 
$c_j$ must be set to $\true$:  Recall from \textsc{CofSynt}  
(Algorithm~\ref{alg:CofSynt}) that the formula 
$C_0(\overline{x}, \overline{i}) := \exists \overline{x}', 
                    \overline{c} \scope
      S\bigl(\overline{x}, 
       \overline{i},
       (c_0,\ldots,c_{j-1},\false,c_{j+1},\ldots,c_n),
       \overline{x}'
       \bigr)$
characterizes the set of all $(\overline{x},\overline{i})$-assignments for which 
$c_j = \false$ is allowed by the strategy $S$.  Its negation 
$M_1(\overline{x},\overline{i}) = \neg C_0(\overline{x}, \overline{i})$ is thus 
the set of all situations where $c_j = \false$ is not allowed by $S$, i.e., 
where $c_j$ must be set to $\true$.  Analogously, the formula 
$M_0(\overline{x},\overline{i})$ represents the set of all 
$(\overline{x},\overline{i})$-assignments for which $c_j$ must be $\false$.

\mypara{Learning an implementation $F_j$.}  
The lines~\ref{alg:QbfSynt:init} to~\ref{alg:QbfSynt:end} compute a \acs{CNF} 
formula $F_j(\overline{x},\overline{i})$ such that 
$M_1(\overline{x},\overline{i})\rightarrow F_j(\overline{x},\overline{i}) 
\rightarrow \neg M_0(\overline{x},\overline{i})$ using a variant of 
\textsc{CnfLearn} from Algorithm~\ref{alg:CnfLearn}.  The first 
implication $M_1 \rightarrow F_j$ ensures that $F_j$ is $\true$ whenever $c_j$ must 
be $\true$.  The second implication $F_j \rightarrow \neg M_0$ ensures that 
whenever $F_j$ is $\true$, $c_j$ does not have to be $\false$.  Together, these 
two conditions fully describe a proper implementation for $c_j$.  Just like 
\textsc{CnfLearn}, we start with $F_j=\true$ (Line~\ref{alg:QbfSynt:init}).  
Next, Line~\ref{alg:QbfSynt:check} checks if $F_j$ is already correct in the 
sense that $M_1\rightarrow F_j \rightarrow \neg M_0$ holds.  The algorithm 
maintains the invariant $M_1 \rightarrow F_j$, so only $F_j \rightarrow \neg 
M_0$ needs to be checked.  This is done by calling a \acs{QBF} solver to search 
for a satisfying assignment $\mathbf{x},\mathbf{i} \models F_j \wedge M_0$ to 
the variables $\overline{x},\overline{i}$ for which $F_j$ is $\true$ but $c_j$ 
must be $\false$. Note that $M_0$ contains a universal quantification of 
$\overline{c}$ and $\overline{x}'$, so a \acs{SAT} solver cannot be used.  If no such counterexample $\mathbf{x},\mathbf{i}$ exists, the 
\textbf{while}-loop terminates. Otherwise, the counterexample cube $\mathbf{d} = 
\mathbf{x} \wedge \mathbf{i}$ is generalized into a cube $\mathbf{d}_g \subseteq 
\mathbf{d}$ by eliminating literals as long as $\mathbf{d}_g \wedge M_1$ is 
unsatisfiable.  This is done in the subroutine \textsc{QbfGeneralize} 
and ensures that $\mathbf{d}_g$ does not contain any 
$(\overline{x},\overline{i})$-assignments for which $c_j$ must be $\true$, so it 
is safe to update $F_j$ to $F_j\wedge \neg \mathbf{d}_g$ while preserving the 
invariant $M_1 \rightarrow F_j$.  This update eliminates the original 
counterexample $\mathbf{d}$ for which $F_j$ \emph{must} be $\false$.  Due to the 
generalization, other $(\overline{x},\overline{i})$-assignments for which $F_j$ 
\emph{can} be $\false$ are also mapped to $\false$.  Going with ``can be false'' 
rather than ``must be false'' in the generalization phase results in potentially 
smaller clauses being added to $F_j$.  This increases the potential for 
eliminating counterexamples before they are encountered in 
Line~\ref{alg:QbfSynt:check}.  Hence, exploiting the freedom between ``must be 
false'' and ``can be false'' --- as done by \textsc{QbfSynt} --- potentially 
not only results in a more compact \acs{CNF} representation of $F_j$ but 
also in fewer iterations.

\mypara{Circuit construction and resubstitution.}  
The remaining parts of \textsc{QbfSynt} are the same as for \textsc{CofSynt} 
(Algorithm~\ref{alg:CofSynt}): Line~\ref{alg:QbfSynt:dump} dumps the formula 
$F_j(\overline{x}, \overline{i})$ as circuit that defines $c_j$ to be 
$\true$ whenever $F_j(\overline{x}, \overline{i})$ evaluates to true.  This can 
be done by replacing every Boolean operator in $F_j$ with the 
corresponding gate.  We do not attempt to reuse existing gates while dumping 
the circuit, but leave this optimization to \abc~\cite{BraytonM10} in the 
postprocessing step.  Finally, Line~\ref{alg:QbfSynt:resub} refines the strategy 
$S$ with the solution for $c_j$ to propagate consequences of fixing $c_j$ on 
other control signals.

\mypara{Auxiliary variables.}  If the strategy formula $S$ contains auxiliary 
variables, e.g., from Tseitin-transformations~\cite{Tseitin83}, then these 
variables are all handled as if they were part of $\overline{x}$.  The 
resubstitution step in Line~\ref{alg:QbfSynt:resub} may also introduce 
additional auxiliary variables, which are also handled like $\overline{x}$.

\begin{wrapfigure}[8]{r}{0.62\textwidth}
\vspace{-7mm}
\centering 
\subfloat[First counterexample.\label{fig:qbfsynt0}]
  {\includegraphics[width=0.3\textwidth]{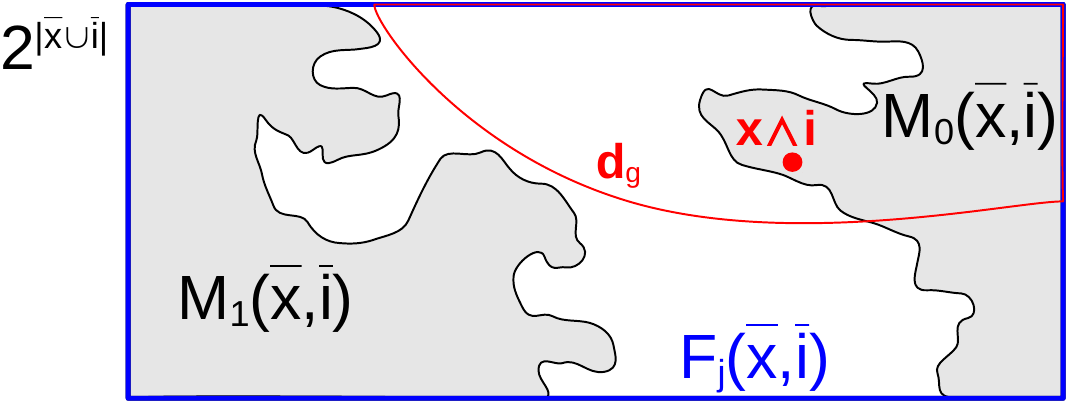}}
\hspace{1mm}
\subfloat[Second counterexample.\label{fig:qbfsynt1}]
  {\includegraphics[width=0.3\textwidth]{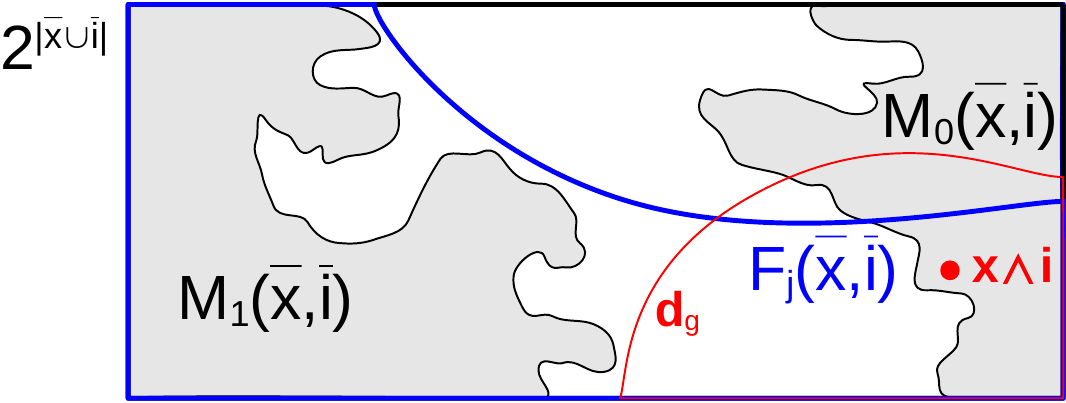}}
\caption[Working principle of \textsc{QbfSynt}]
{Working principle of \textsc{QbfWin}.}
\label{fig:qbfsynt}
\end{wrapfigure}
\mypara{Illustration.}
Figure~\ref{fig:qbfsynt} illustrates the computation of a circuit for one 
control signal $c_j$.  The boxes represent the set $2^{|\overline{x} 
\cup \overline{i}|}$ of all possible assignments to the variables $\overline{x}$ 
and $\overline{i}$.  Figure~\ref{fig:qbfsynt0} depicts the initial situation.  
The region $M_1$ represents the set of all situations where $c_j$ must be 
$\true$, and $M_0$ represents the situations where $c_j$ must be $\false$.  The 
definition of the strategy ensures that these two regions cannot overlap.  The 
current approximation $F_j$ of the solution is depicted in blue.  Initially, 
$F_j=\true$ (Line~\ref{alg:QbfSynt:init} in \textsc{QbfSynt}).  Next, a 
counterexample $\mathbf{x},\mathbf{i} \models F_j \wedge M_0$ is computed 
(Line~\ref{alg:QbfSynt:check}).  It is drawn as a red dot in 
Figure~\ref{fig:qbfsynt0}.  The counterexample cube $\mathbf{x} \wedge 
\mathbf{i}$ is then generalized into a larger region $\mathbf{d}_g$ by 
eliminating literals as long as $\mathbf{d}_g$ does not intersect with $M_1$. 
This is ensured by the check in Line~\ref{alg:QbfSynt:gen} of \textsc{QbfSynt}. 
Next, $F_j$ is refined by subtracting the resulting region $\mathbf{d}_g$. The 
refined formula $F_j$ is shown as a blue outline Figure~\ref{fig:qbfsynt1}.  
Since the first counterexample is no longer contained in $F_j \wedge M_0$, it 
cannot be encountered again.  Instead, the algorithm computes a different 
counterexample, which is generalized in the same way.  This is illustrated in 
Figure~\ref{fig:qbfsynt1}.  After subtracting the second $\mathbf{d}_g$ from 
$F_j$ (which is not shown in Figure~\ref{fig:qbfsynt}),  $F_j$ does not 
intersect with $M_0$ any more.  Hence there are no more situations where $F_j$ 
is $\true$ but must be $\false$.  Since we did not remove any situation that is 
contained in $M_1$ from $F_j$, the final solution satisfies $M_1\rightarrow F_j 
\rightarrow \neg M_0$ and the \textbf{while}-loop in \textsc{QbfSynt} 
terminates.  That is, $F_j$ exploits the freedom between $M_1$ and $M_0$.  
Compared to learning a \acs{CNF} formula for $\neg M_0$ precisely, this 
potentially reduces the number of iterations and the resulting circuit size, 
especially if $\neg M_0$ is complicated.  In Figure~\ref{fig:qbfsynt}, 
this is indicated by $M_0$ being more irregular in shape than $F_j$.

\subsubsection{Efficient Implementation for Safety Synthesis Problems}
\label{sec:hw:extr:qbf_impl}

The procedure \textsc{SafeQbfSynt} in Algorithm~\ref{alg:SafeQbfSynt} presents 
an efficient realization of \textsc{QbfSynt} for the case of safety 
specifications, where the winning strategy $S(\overline{x}, \overline{i}, 
\overline{c}, \overline{x}')$ is defined via a winning region (or a winning 
area) 
$W(\overline{x})$.  To make the \acs{QBF} queries efficient, our aim is to avoid 
disjunctions and negations of subformulas as much as possible, and to reduce 
the amount of universal quantification.

\begin{algorithm}[tb]
\caption[\textsc{SafeQbfSynt}: Synthesizes circuits from winning areas 
with \acs{QBF}-based \acs{CNF} learning]
{\textsc{SafeQbfSynt}: Synthesizes circuits from winning areas 
with \acs{QBF}-based \acs{CNF} learning.}
\label{alg:SafeQbfSynt}
\begin{algorithmic}[1]
\Procedure{SafeQbfSynt}
          {$T(\overline{x},\overline{i},\overline{c},\overline{x}')$,
           $W(\overline{x})$}
  \State $T'(\overline{x},\overline{i},\overline{c},\overline{x}') :=
          T(\overline{x},\overline{i},\overline{c},\overline{x}')$, \quad
         $\overline{c}_b := \overline{c}$, \quad
         $\overline{c}_a := \emptyset$
         \label{alg:SafeQbfSynt:in}
  \For{all $j$ from $1$ to $|\overline{c}|$}
    \State $\overline{c}_b := \overline{c}_b \setminus \{c_j\}$
           \label{alg:SafeQbfSynt:ca}   
    \State $M_1(\overline{x}, \overline{i}) :=
       \forall \overline{c}_b \scope
       \exists \overline{c}_a, \overline{x}' \scope
       T'\bigl(\overline{x},\overline{i},
       \overline{c}_a,\false,\overline{c}_b,
       \overline{x}'
       \bigr)\wedge 
       W(\overline{x}) \wedge
       \neg W(\overline{x}')$
       \label{alg:SafeQbfSynt:m1}
    \State $M_0(\overline{x},
                 \overline{i}) :=
       \forall \overline{c}_b \scope
       \exists \overline{c}_a, \overline{x}' \scope
       T'\bigl(\overline{x},\overline{i},
       \overline{c}_a,\true,\overline{c}_b,
       \overline{x}'
       \bigr) \wedge 
       W(\overline{x}) \wedge
       \neg W(\overline{x}')$
       \label{alg:SafeQbfSynt:m0}        
    \State $F_j(\overline{x},\overline{i}) := \true$
    \While{$\mathsf{sat}$ in $(\mathsf{sat},\mathbf{x}, \mathbf{i}) := 
            \qbfsatmodel\bigl(
            \exists \overline{x},\overline{i} \scope
            F_j(\overline{x},\overline{i}) \wedge
            M_0(\overline{x},
                 \overline{i})\bigr)$}
            \label{alg:SafeQbfSynt:check}
      \State $\mathbf{d}_g := \textsc{QbfGeneralize}\bigl(
                  \mathbf{x} \wedge \mathbf{i},
                  M_1(\overline{x},\overline{i})
                  \bigr)$
      \State $F_j(\overline{x},\overline{i}) := 
              F_j(\overline{x},\overline{i}) \wedge
             \neg \mathbf{d}_g$\label{alg:SafeQbfSynt:ref}
    \EndWhile
  \State $\textsc{dumpCircuit}\bigl(c_j, 
          F_j(\overline{x},\overline{i})\bigr)$ 
           \label{alg:SafeQbfSynt:dump}
  \State $T'(\overline{x},\overline{i},\overline{c},\overline{x}') := 
          T'(\overline{x},\overline{i},\overline{c},\overline{x}') \wedge 
          \bigl(c_j \leftrightarrow F_j(\overline{x},\overline{i})\bigr)$
          \label{alg:SafeQbfSynt:resub}
  \State $\overline{c}_a := \overline{c}_a \cup \{c_j\}$
         \label{alg:SafeQbfSynt:cb}
  \EndFor
\EndProcedure
\ProcedureRet{QbfGeneralize}
          {\mathbf{d}, M_1(\overline{x},\overline{i})}
          {$\mathbf{d}_g \subseteq \mathbf{d}$ such that $\mathbf{d}_g
            \wedge M_1$ is unsatisfiable}
      \State $\mathbf{d}_g := \mathbf{x} \wedge \mathbf{i}$
      \For{each literal $l$ in $\mathbf{d}$}\label{alg:SafeQbfSynt:loop}
        \State $\mathbf{d}_t := \mathbf{d}_g \setminus \{l\}$
        \If{$\neg \qbfsatmodel\bigl(
            \exists \overline{x},\overline{i} \scope
              \mathbf{d}_t \wedge
              M_1(\overline{x},\overline{i})
              \bigr)$}\label{alg:SafeQbfSynt:gen}
          \State $\mathbf{d}_g := \mathbf{d}_t$
        \EndIf
      \EndFor
    \State \textbf{return} $\mathbf{d}_g$
\EndProcedure
\end{algorithmic}
\end{algorithm}

\mypara{Grouping of control variables.}
In every iteration, \textsc{SafeQbfSynt} splits the control variables 
$\overline{c}$ into three groups $\overline{c}_a$, $c_j$, $\overline{c}_b$:  The 
single variable $c_j$ is the one for which a circuit is constructed in the 
current iteration, $\overline{c}_a$ contains all variables for which a circuit 
has already been computed, and $\overline{c}_b$ contains all control variables 
for which a circuit will be computed in some future iteration.  This split is 
performed in the Lines \ref{alg:SafeQbfSynt:in}, \ref{alg:SafeQbfSynt:ca} and 
\ref{alg:SafeQbfSynt:cb}, and will allow us to reduce the amount of universal 
quantification.

\mypara{Definition of $M_1$ and $M_0$.}  With 
$S(\overline{x},\overline{i},\overline{c},\overline{x}') = 
T(\overline{x},\overline{i},\overline{c},\overline{x}') \wedge \bigl(\neg 
W(\overline{x}) \vee W(\overline{x}')\bigr)$, we can apply the following 
transformations to compute a \acs{CNF} for $M_1$ more efficiently.
\begin{linenomath*}
\begin{align*}
M_1(\overline{x},\overline{i}) = &
      \forall \overline{c},\overline{x}'\scope
      \neg
       S\bigl(\overline{x},\overline{i},
       (c_0,\ldots,c_{j-1},\false,c_{j+1},\ldots,c_n),
       \overline{x}'
       \bigr)
\\=&
      \forall \overline{c}_b,\overline{c}_a,\overline{x}'\scope
      \neg \Bigl(
       T(\overline{x},\overline{i},
       \overline{c}_a,\false,\overline{c}_b,
       \overline{x}') \wedge 
       \bigl(\neg W(\overline{x}) \vee W(\overline{x}')\bigr)
       \Bigr)
\\=&
      \forall \overline{c}_b,\overline{c}_a,\overline{x}'\scope
      \Bigl(
       T(\overline{x},\overline{i},
       \overline{c}_a,\false,\overline{c}_b,
       \overline{x}') \rightarrow 
       \bigl(W(\overline{x}) \wedge \neg W(\overline{x}')\bigr)
       \Bigr)
\end{align*}
\end{linenomath*}
\textsc{SafeQbfSynt} keeps a copy $T'$ of the transition relation $T$.  It is 
updated in such a way that the variables $\overline{c}_a, \overline{x}'$ 
are~defined uniquely by $T'$.  For $\overline{x}'$, this holds 
initially. For $\overline{c}_a$, this is ensured by 
Line~\ref{alg:SafeQbfSynt:resub}. Thus, by using $T'$ instead of $T$ and 
applying the one-point rule (\ref{eq:op1}), the universal 
quantification of $\overline{c}_a, \overline{x}'$ can be turned into an 
existential one:
\[
M_1(\overline{x},\overline{i}) = 
      \forall \overline{c}_b\scope
      \exists \overline{c}_a,\overline{x}'\scope
      \Bigl(
       T'(\overline{x},\overline{i},
       \overline{c}_a,\false,\overline{c}_b,
       \overline{x}') \wedge 
       W(\overline{x}) \wedge \neg W(\overline{x}')
       \Bigr)
\]
The computation of $M_0(\overline{x},\overline{i})$ works analogously. As a 
result, only the control signals $\overline{c}_b$, for which no solution has 
been computed yet, are quantified universally in the \acs{QBF} queries of 
Line~\ref{alg:SafeQbfSynt:check} and~\ref{alg:SafeQbfSynt:gen}.  The variable 
vector $\overline{c}_b$ becomes shorter from iteration to iteration, which means 
that the formula gets ``more propositional''.  In the last iteration, a 
\acs{SAT} solver can actually be used instead of a \acs{QBF} solver.

\mypara{\acs{CNF} conversion.}
The \acs{QBF} queries in Line~\ref{alg:SafeQbfSynt:check} 
and~\ref{alg:SafeQbfSynt:gen} contain only conjunctions.  The formula $F_j$ is 
always in \acs{CNF}.  Furthermore, most of our methods to compute a winning 
region or a winning area produce $W(\overline{x})$ in \acs{CNF}.  Hence, we only 
need to compute a \acs{CNF} representation of $T'$ and $\neg W(\overline{x}')$. 
\textsc{NegLearn} (Algorithm~\ref{alg:NegLearn}), which negates a 
formula without introducing auxiliary variables, was beneficial in the \acs{QBF} 
certification approach but does not pay off in the learning-based approach.  
Hence, we apply the method of Plaisted and Greenbaum~\cite{PlaistedG86} to 
compute a \acs{CNF} for $\neg W(\overline{x}')$. 

\mypara{\acs{QBF} preprocessing.}  With our extension of 
\bloqqer~\cite{SeidlK14} to preserve satisfying assignments, \acs{QBF} 
preprocessing can be applied both for counterexample computation and 
generalization.  However, while preprocessing was vital in our methods for 
computing a winning region, it does not give a significant speedup for 
\textsc{SafeQbfSynt} (see Chapter~\ref{sec:hw:exp}).

\mypara{Incremental \acs{QBF} solving.} 
\textsc{SafeQbfSynt} is very well suited for incremental \acs{QBF} solving, 
especially with a solver interface such as the one provided by 
\depqbf~\cite{LonsingE14,LonsingE14b}.  We propose to use two solver instances 
incrementally.  The first instance stores $F_j \wedge M_0$ and is used for 
Line~\ref{alg:SafeQbfSynt:check}.  Since Line~\ref{alg:SafeQbfSynt:ref} only 
adds clauses to $F_j$, this solver instance is only re-initialized when a mayor 
iteration (synthesizing the next $c_j$) is started.  The second solver instance 
stores $M_1$, is used for Line~\ref{alg:SafeQbfSynt:gen}, and is also 
re-initialized when a mayor iteration starts.  Before executing the loop in 
Line~\ref{alg:SafeQbfSynt:loop}, we let the second solver instance compute an 
unsatisfiable core $\mathbf{d}_g$ of $\mathbf{d} = \mathbf{x} \wedge \mathbf{i}$ 
and only reduce this core further in the loop.  The conjunction with 
$\mathbf{d}_t$ is realized with assumption literals.

\subsubsection{Discussion}

\mypara{Greediness.} 
\textsc{QbfSynt} is greedy in exploiting implementation freedom.  When 
synthesizing a circuit for one control signal $c_j$, \textsc{QbfSynt} ensures 
that \emph{some} solution for the remaining control signals still exists.  
However, the algorithm does not specifically attempt to retain implementation 
freedom for the remaining control signals.  This can have the effect that the 
signals synthesized early have a small implementation, which is found after only 
a few refinements.  Yet, for the signals synthesized later, the implementation 
freedom may already be ``exhausted'' and large implementations may be produced 
after many refinements.  Consequently, the performance may also strongly depend 
on the order in which control signals are processed.  This is similar to the 
standard \textsc{CofSynt} procedure, but different from the \acs{QBF} 
certification approach from Section~\ref{sec:hw:circ_qbfcert}, which computes 
circuits for all control signals simultaneously.

\mypara{Independence of symbolic representation.}  In contrast to 
\textsc{CofSynt} and \acs{QBF} certification, the \acs{QBF}-based 
learning approach is rather independent of the symbolic strategy representation 
and the reasoning engine.  Only the concrete counterexamples computed by 
Line~\ref{alg:QbfSynt:check} may differ, and our experience in trying to develop 
heuristics for computing good counterexamples indicates that one counterexample 
is usually just as good as any other.  Consequently, the number of iterations 
and the resulting circuit will be similar, independent of whether the 
strategy formula is encoded efficiently or not.  When implemented using 
\acp{BDD}, the variable ordering has little impact on these metrics too.

\mypara{Circuit depth.}
Another advantage of the \acs{QBF}-based \acs{CNF} learning algorithm presented 
in this section is that the produced circuits have a low depth.  This can be an 
important property because the circuit depth determines the maximum clock 
frequency with which the circuit can be operated.  The formulas $F_j$ defining 
the control signals $c_j$ are computed in \acs{CNF}.  When these formulas are 
transformed into circuits in the straightforward way, this yields circuits with 
a depth of at most $3$: every signal $\overline{x},\overline{i}$ needs to pass 
at most one inverter, one OR-gate and one AND-gate.  Depending on the gates 
available in the standard cell library, it may not be feasible to realize the 
circuit in this straightforward way.  However, experiments~\cite{EhlersKH12} 
with a simplistic standard cell library suggest that the circuit depth is 
usually much lower than when using the standard \textsc{CofSynt} 
procedure with \acp{BDD}.

\subsection{Interpolation} \label{sec:hw:int}

Jiang et al.~\cite{JiangLH09} present an interpolation-based approach to 
synthesize circuits from strategies.  Similar to the cofactor-based approach 
presented in Algorithm~\ref{alg:CofSynt} and the \acs{QBF}-based learning 
approach from Algorithm~\ref{alg:QbfSynt}, it computes circuits for one control 
signal $c_j \in \overline{c}$ after the other. However, in contrast to these 
previous algorithms, the interpolation-based approach avoids quantifier 
alternations by temporarily considering other control signals for which no 
circuits have been computed yet as if they were inputs.  
We will define the approach by Jiang et al.~\cite{JiangLH09} as an algorithm for 
our setting in Section~\ref{sec:hw:int_basic}.  After that, we will present 
optimizations and an efficient realization for safety specifications.  In 
Section~\ref{sec:hw:extr:satlearn}, we will furthermore combine the approach 
with query learning.

\subsubsection{Basic Algorithm}\label{sec:hw:int_basic}

\begin{algorithm}[tb]
\caption[\textsc{InterpolSynt}~\cite{JiangLH09}: Synthesizing circuits from 
strategies using interpolation]
{\textsc{InterpolSynt}~\cite{JiangLH09}: Synthesizing circuits from 
strategies using interpolation.}
\label{alg:InterpolSynt}
\begin{algorithmic}[1]
\Procedure{InterpolSynt}
          {$S(\overline{x},\overline{i},\overline{c},\overline{x}')$}
  \State $\overline{c}_a := \overline{c}$, \quad  
         $\overline{c}_b := \emptyset$
  \For{all $j$ from $|\overline{c}|$ to $1$}
    \State $\overline{c}_a := \overline{c}_a \setminus \{c_j\}$
    \State $M_1(\overline{x},\overline{i},\overline{c}_a) :=
           \bigl(
           \exists \overline{c}_b, \overline{x}' \scope 
           S(\overline{x},
             \overline{i},
             \overline{c}_a,
             \true,
             \overline{c}_b,
             \overline{x}')
           \bigr)
           \wedge
           \bigl(
           \neg \exists \overline{c}_b, \overline{x}' \scope 
           S(\overline{x},
             \overline{i},
             \overline{c}_a,
             \false,
             \overline{c}_b,
             \overline{x}')
           \bigr)$
           \label{alg:InterpolSynt:m1}
    \State $M_0(\overline{x},\overline{i},\overline{c}_a) :=
           \bigl(
           \exists \overline{c}_b, \overline{x}' \scope 
           S(\overline{x},
             \overline{i},
             \overline{c}_a,
             \false,
             \overline{c}_b,
             \overline{x}')
           \bigr)
           \wedge
           \bigl(
           \neg \exists \overline{c}_b, \overline{x}' \scope 
           S(\overline{x},
             \overline{i},
             \overline{c}_a,
             \true,
             \overline{c}_b,
             \overline{x}')
           \bigr)$
           \label{alg:InterpolSynt:m0}
    \State $F_j(\overline{x},\overline{i},\overline{c}_a) := \interpol\bigl(
              M_1(\overline{x},\overline{i},\overline{c}_a),
              M_0(\overline{x},\overline{i},\overline{c}_a)
           \bigr)$\label{alg:InterpolSynt:int}
    \State $\textsc{dumpCircuit}\bigl(c_j, 
            F_j(\overline{x},\overline{i},\overline{c}_a)\bigr)$ 
           \label{alg:InterpolSynt:dump}
    \State $S(\overline{x},\overline{i},\overline{c},\overline{x}') := 
           S(\overline{x},\overline{i},\overline{c},\overline{x}')
           \wedge \bigl(c_j \leftrightarrow 
           F_j(\overline{x},\overline{i},\overline{c}_a)\bigr)$ 
    \label{alg:InterpolSynt:resub}
    \State $\overline{c}_b := \overline{c}_b \cup \{c_j\}$
    \label{alg:InterpolSynt:cb}
  \EndFor
\EndProcedure
\end{algorithmic}
\end{algorithm}

Algorithm~\ref{alg:InterpolSynt} 
illustrates the approach by Jiang et al.~\cite{JiangLH09} in our setting.  
As before, the input is a strategy formula $S(\overline{x}, \overline{i}, 
\overline{c}, \overline{x}')$.  The procedure does not return any result but 
directly dumps the produced circuits defining $\overline{c}$.  

\mypara{Variable dependencies.}
Similar to \textsc{QbfSynt} in Algorithm~\ref{alg:QbfSynt}, the variables 
$\overline{c}=(c_1,\ldots,c_n)$ are split into three groups $\overline{c}_a, 
c_j, \overline{c}_b$.  Here, $c_j$ is the variable for which a circuit is 
computed in the current iteration.  The algorithm starts with the last control 
signal $c_n$ and proceeds with decreasing indices.\footnote{The order is 
actually irrelevant, but fixing some order simplifies the discussion.}
Line~\ref{alg:InterpolSynt:cb} makes sure that the variable vector 
$\overline{c}_b$ contains all control variables for which a circuit has 
been computed in some previous iteration.  Finally, $\overline{c}_a$ contains 
all control variables for which a circuit needs to be computed in one of 
the following iterations.  The variables in $\overline{c}_a$ 
\begin{wrapfigure}[10]{r}{0.51\textwidth}
\vspace{-3mm}
\centering
  \includegraphics[width=0.5\textwidth]{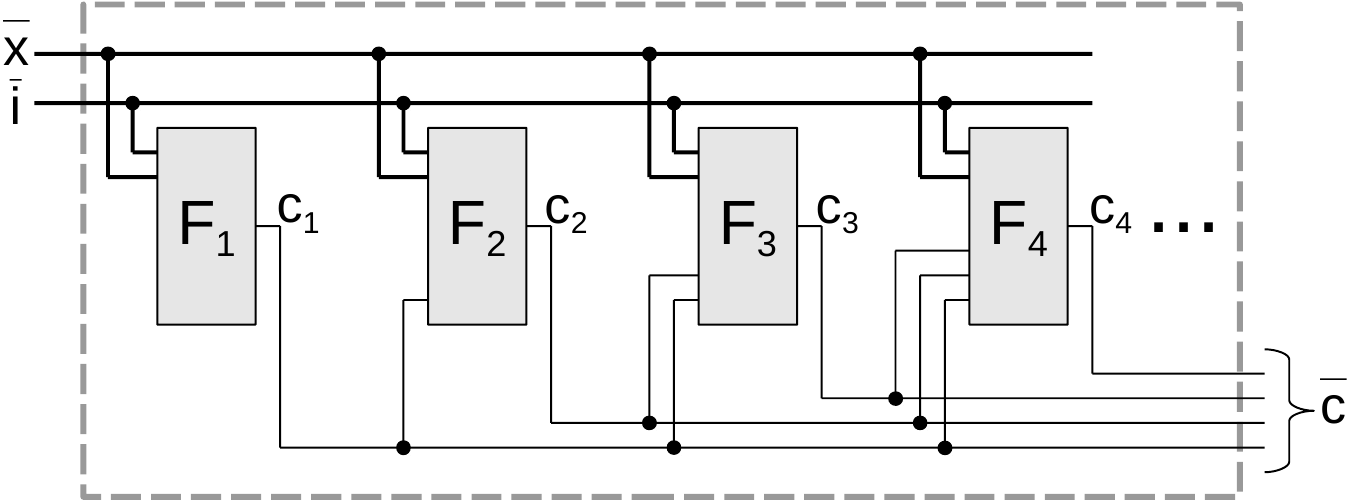}
\caption{Variable dependencies in interpolation-based circuit synthesis.}
\label{fig:interpol}
\end{wrapfigure}
are treated as if 
they were inputs.  That is, the circuit defining $c_j$ may not only reference 
variables from $\overline{x}$ and $\overline{i}$, but also all $c_k$ with $k < 
j$ for which no circuit has been computed yet.  This is illustrated in 
Figure~\ref{fig:interpol}: $c_n$ can also take all signals $c_1,\ldots,c_{n-1}$ 
as input, the circuit for $c_{n-1}$ can also take 
$c_1,\ldots,c_{n-2}$ as input, etc.  Finally, $c_1$ cannot depend on any other 
variables of $\overline{c}$.  This ensures that there are no circular 
dependencies.  Furthermore, when the circuits for all $c_j\in \overline{c}$ are 
built together, the signals $\overline{c}$ effectively depend on $\overline{x}$ 
and $\overline{i}$ only.

\mypara{Definition of $M_1$ and $M_0$.}
Let $\overline{d} = \overline{x} \cup \overline{i} \cup \overline{c}_a$ be the 
vector of all variables on which the current control signal $c_j$ may depend. 
Line~\ref{alg:InterpolSynt:m1} of \textsc{InterpolSynt} computes 
$M_1(\overline{d})$, which characterizes the set of all 
$\overline{d}$-assignments for which $c_j$ must be $\true$.  This is done as 
follows.  The subformula 
$C_1(\overline{d}) = \exists \overline{c}_b, \overline{x}' \scope 
           S(\overline{x},
             \overline{i},
             \overline{c}_a,
             \true,
             \overline{c}_b,
             \overline{x}')$ 
characterizes the set of all $\overline{d}$-assignments for which $c_j=\true$ 
is 
allowed by $S$.  This is essentially the positive cofactor of $S$ regarding 
$c_j$, but the variables $\overline{c}_b, \overline{x}'$ are also quantified 
existentially, which means that their concrete value is irrelevant as long as 
some value exists.  Similarly, the subformula 
$C_0(\overline{d}) = \exists \overline{c}_b, \overline{x}' \scope 
           S(\overline{x},
             \overline{i},
             \overline{c}_a,
             \false,
             \overline{c}_b,
             \overline{x}')$ 
characterizes the set of all $\overline{d}$-assignments for which $c_j=\false$ 
is allowed by $S$.  Hence, $M_1$ represents the set of all 
$\overline{d}$-assignments for which $\true$ is allowed, but $\false$ is not 
allowed.  Analogously, Line~\ref{alg:InterpolSynt:m0} computes the formula 
$M_0$, which characterizes the $\overline{d}$-assignments for which $c_j$ must 
be $\false$.  In principle, $M_1$ and $M_0$ can easily be transformed into a 
propositional \acs{CNF} formula by renaming or expanding the existentially 
quantified variables.  An efficient solution to do so will be presented in 
Section~\ref{sec:hw_interpol_impl}, but for now we focus on understandability 
rather than efficiency.  

\mypara{Differences to \textsc{QbfSynt}.}  
Note that the procedure \textsc{QbfSynt} from Algorithm~\ref{alg:QbfSynt} 
computes $M_1$ and $M_0$ differently in two respects.  First, 
$M_1(\overline{x},\overline{i})$ and $M_0(\overline{x},\overline{i})$ do not 
contain $\overline{c}_a$ as free variables in \textsc{QbfSynt}.  Second, 
$M_1(\overline{x},\overline{i})$ is computed as $\neg 
C_0(\overline{x},\overline{i})$ in \textsc{QbfSynt} instead of 
$C_1(\overline{d}) \wedge \neg C_0(\overline{d})$ (and similar for 
$M_0(\overline{x},\overline{i})$).  The additional conjunction with 
$C_1(\overline{d})$ in \textsc{InterpolSynt} is necessary for the following 
reason.  We have that $\neg C_0(\overline{x},\overline{i}) \rightarrow 
C_1(\overline{x},\overline{i})$ and $\neg C_1(\overline{x},\overline{i}) 
\rightarrow C_0(\overline{x},\overline{i})$ in \textsc{QbfSynt} because 
$\forall \overline{x},\overline{i} \scope 
 \exists \overline{c},\overline{x}'\scope 
S(\overline{x},\overline{i},\overline{c},\overline{x}')$ 
is guaranteed by the strategy.  In other words, for every 
$(\overline{x},\overline{i})$-assignment, any control signal $c_j$ can either be 
$\true$ or $\false$ (or both).  Hence, the additional conjunct 
$C_1(\overline{x},\overline{i})$ would be of no use in 
$M_1(\overline{x},\overline{i}) = \neg C_0(\overline{x},\overline{i})$ as 
defined by \textsc{QbfSynt}, because it is implied anyway.  Yet, $\neg 
C_0(\overline{d}) \rightarrow C_1(\overline{d})$ and $\neg C_1(\overline{d}) 
\rightarrow C_0(\overline{d})$ do \emph{not} hold in \textsc{InterpolSynt}: 
there may be $\overline{d}$-assignment for which neither $c_j=\true$ nor 
$c_j=\false$ is allowed by the strategy.  The reason is that we also consider 
the signals $\overline{c}_a$ as if they were inputs, but
$\forall \overline{x},\overline{i},\overline{c}_a\scope 
 \exists c_j, \overline{c}_b,\overline{x}'\scope 
S(\overline{x},\overline{i},\overline{c}_a,c_j,\overline{c}_b,\overline{x}')$
does not hold in general.  For $\overline{d}$-assignments for which neither 
$c_j=\true$ nor $c_j=\false$ is allowed, the definition of $M_1 
= C_1(\overline{d}) \wedge \neg C_0(\overline{d})$ and $M_0 = C_0(\overline{d}) 
\wedge \neg C_1(\overline{d})$ allows both values for $c_j$. This is justified 
by the fact that the circuits synthesized for $\overline{c}_a$ in subsequent 
iterations will make sure that such $\overline{d}$-assignments will never occur 
as input of the circuit defining $c_j$.  We refer to Jiang et 
al.~\cite{JiangLH09} for details on this technical subtlety.

\mypara{Interpolation.}
The conjunction 
$M_1(\overline{d})\wedge M_0(\overline{d}) = 
C_1(\overline{d}) \wedge \neg C_0(\overline{d}) \wedge 
C_0(\overline{d}) \wedge \neg C_1(\overline{d})
$ 
is trivially unsatisfiable, so an interpolant $F_j(\overline{d})$ can be 
computed in Line~\ref{alg:InterpolSynt:int}.  The properties of an interpolant 
(see Section~\ref{sec:prelim:prop}) ensure that $M_1 \rightarrow F_j 
\rightarrow 
\neg M_0$.  The first implication means that $F_j$ is $\true$ whenever $c_j$ 
must be $\true$.  The second implication means that if $F_j$ is $\true$, then 
$c_j$ does not have to $\false$.  This means that $F_j(\overline{d})$ is a 
proper implementation for $c_j$.  

\mypara{Circuit construction and resubstitution.} 
The remaining steps of the \textsc{InterpolSynt} procedure are the same as for 
\textsf{CofSynt} (Algorithm~\ref{alg:CofSynt}) and \textsf{QbfSynt} 
(Algorithm~\ref{alg:QbfSynt}).  Line~\ref{alg:InterpolSynt:dump} constructs a 
circuit which sets $c_j=\true$ if an only if $F_j(\overline{x}, \overline{i}, 
\overline{c}_a)$ evaluates to $\true$. Finally, 
Line~\ref{alg:InterpolSynt:resub} refines the strategy formula $S$ with the 
concrete implementation for $c_j$.

\mypara{Auxiliary variables.}  If the strategy formula $S$ is defined using 
auxiliary variables, these can all be put into $\overline{c}_b$.  This also 
applies to auxiliary variables that may be introduced in the resubstitution
in Line~\ref{alg:InterpolSynt:resub}.

\mypara{Variations.}  Jiang et al.~\cite{JiangLH09} propose to perform a second 
pass over all control signals, where the circuits for all $c_j$ are recomputed 
using interpolation, while fixing the implementation for the other control 
signals.  This has the potential for producing smaller circuits because the 
recomputed interpolants can now rely on some concrete realization for the other 
control signals.  However, in preliminary experiments for our setting, this 
second pass did not result in considerable circuit size improvements (but rather 
increased the circuit size for many cases).  Since such a second pass also 
increases the computation time, we do not perform it.  Jiang et 
al.~\cite{JiangLH09} also propose a second interpolation-based approach which 
does not treat other control signals as if they were inputs but rather 
quantifies them universally and applies universal expansion to eliminate the 
quantifiers.  However, this can blow up the formula size significantly.  
Preliminary experiments with this second approach were not promising in our 
setting either.

\subsubsection{Dependency Optimization}\label{sec:hw_dep}

For some specifications, the performance of \textsc{InterpolSynt} strongly 
depends on the order in which the control signals 
$\overline{c}=(c_1,\ldots,c_n)$ are processed.  One reason is that this order 
defines which signal $c_j$ may depend on which other signals $c_k$.  The aim of 
the optimization presented in this section is to increase the set of variables 
on which a certain signal $c_j$ can depend.  This increases the freedom for the 
interpolation procedure (the interpolant $F_j$ may still choose the ignore the 
additional signals) and can lead to smaller interpolants and shorter execution 
times.

\mypara{Basic idea.}
The basic idea is as follows.  As illustrated in Figure~\ref{fig:interpol}, the 
interpolant $F_n$ computed first can reference all other control signals 
$c_1,\ldots,c_{n-1}$.  The interpolant $F_{n-1}$ computed in the second 
iteration cannot depend on $c_{n}$, though.  The reason is that $F_n$, which 
defines $c_{n}$, could in turn reference $c_{n-1}$, which would result in a 
circular dependency.  Yet, the concrete interpolant $F_n$ may choose to ignore 
$c_{n-1}$ completely.  In this case, $F_{n-1}$ \emph{can} in fact be allowed to 
reference $c_{n}$.  The reason is that there is no danger to introduce a 
circular dependency --- the result would be the same as if $c_{n}$ and $c_{n-1}$ 
would have been processed by \textsc{InterpolSynt} in reverse order.

\mypara{Realization.}
In the iteration synthesizing a solution for $c_j$, we analyze which other 
signals $c_k$ with $k > j$ do not transitively depend on $c_j$.  This is done on 
a syntactic level by checking if $c_j$ occurs in the fan-in cone of $c_k$ when 
the circuits for $F_{j+1},\ldots, F_{n}$ are combined.  If $c_j$ does not appear 
in the fan-in cone of $c_k$, then $c_k$ is moved (temporarily) from 
$\overline{c}_b$ to $\overline{c}_a$.  Thus, $F_j(\overline{x}, \overline{i}, 
\overline{c}_a)$ can reference $c_k$.

\mypara{Dependencies on auxiliary variables.}
Depending on the realization of $\interpol$, the computed interpolants 
$F_{j}(\overline{x}, \overline{i}, \overline{c}_a)$ may be represented using 
auxiliary variables (e.g., introduced by a 
Tseitin-transformation~\cite{Tseitin83}) that act as abbreviation for some 
subformulas over $\overline{x}$, $\overline{i}$ and $\overline{c}_a$.  As 
mentioned in the previous subsection, all auxiliary variables are put into 
$\overline{c}_b$, so they cannot be referenced by the computed interpolants by 
default.  However, the dependency analysis cannot only be performed for the 
final output of each $F_k$  with $k > j$, but also on their auxiliary variables: 
if the current $c_j$ does not appear in the fan-in cone of some auxiliary 
variable $t$, then $t$ can be moved from $\overline{c}_b$ to $\overline{c}_a$.

\subsubsection{Efficient Implementation for Safety Synthesis Problems} 
\label{sec:hw_interpol_impl}

The procedure \textsc{SafeInterpolSynt} in Algorithm~\ref{alg:SafeInterpolSynt} 
shows an efficient implementation of \textsc{InterpolSynt} if the winning 
strategy is defined via a winning region (or winning area) $W(\overline{x})$ of 
a safety specification.  The dependency optimization is not included for the 
sake of readability.

\begin{algorithm}[tb]
\caption[\textsc{SafeInterpolSynt}: Synthesizing circuits from winning areas 
using interpolation]
{\textsc{SafeInterpolSynt}: Synthesizing circuits from winning areas 
using interpolation.}
\label{alg:SafeInterpolSynt}
\begin{algorithmic}[1]
\Procedure{SafeInterpolSynt}
          {$T(\overline{x},\overline{i},\overline{c},\overline{x}')$,
           $W(\overline{x})$}
  \State $T'(\overline{x},\overline{i},\overline{c},\overline{x}') :=
          T(\overline{x},\overline{i},\overline{c},\overline{x}')$, \quad
         $\overline{c}_a := \overline{c}$, \quad  
         $\overline{c}_b := \emptyset$
  \For{all $j$ from $|\overline{c}|$ to $1$}
    \State $\overline{c}_a := \overline{c}_a \setminus \{c_j\}$
    \State $\overline{c}_{b1}, \overline{c}_{b2}, 
            \overline{c}_{b3}, \overline{c}_{b4} :=
           \textsf{create4FreshCopies}(\overline{c}_{b})$
    \State $\overline{x}_1',\, \overline{x}_2',\, 
            \overline{x}_3',\, \overline{x}_4'\,\, :=
           \textsf{create4FreshCopies}(\overline{x}')$
    \State $M_1'(\overline{x},\overline{i},\overline{c}_a,
                \overline{c}_{b1},\overline{c}_{b2},
                \overline{x}_1',\overline{x}_2') :=
           T'(\overline{x},
             \overline{i},
             \overline{c}_a,
             \true,
             \overline{c}_{b1},
             \overline{x}_1')
           \wedge
           W(\overline{x}_1')
           \wedge
           T'(\overline{x},
             \overline{i},
             \overline{c}_a,
             \false,
             \overline{c}_{b2},
             \overline{x}_2')
           \wedge
           W(\overline{x})
           \wedge
           \neg W(\overline{x}_2')$
           \label{alg:SafeInterpolSynt:m1}
    \State $M_0'(\overline{x},\overline{i},\overline{c}_a,
                \overline{c}_{b3},\overline{c}_{b4},
                \overline{x}_3',\overline{x}_4') :=
           T'(\overline{x},
             \overline{i},
             \overline{c}_a,
             \false,
             \overline{c}_{b3},
             \overline{x}_3')
           \wedge
           W(\overline{x}_3')
           \wedge
           T'(\overline{x},
             \overline{i},
             \overline{c}_a,
             \true,
             \overline{c}_{b4},
             \overline{x}_4')
           \wedge
           W(\overline{x})
           \wedge
           \neg W(\overline{x}_4')$
           \label{alg:SafeInterpolSynt:m0}
    \State $F_j(\overline{x},\overline{i},\overline{c}_a) := \interpol\bigl(
              M_1'(\overline{x},\overline{i},\overline{c}_a,
                  \overline{c}_{b1},\overline{c}_{b2},
                  \overline{x}_1',\overline{x}_2'),
              M_0'(\overline{x},\overline{i},\overline{c}_a,
                  \overline{c}_{b3},\overline{c}_{b4},
                  \overline{x}_3',\overline{x}_4')
           \bigr)$\label{alg:SafeInterpolSynt:int}
    \State $\textsc{dumpCircuit}\bigl(c_j, 
            F_j(\overline{x},\overline{i},\overline{c}_a)\bigr)$ 
           \label{alg:SafeInterpolSynt:dump}
    \State $T'(\overline{x},\overline{i},\overline{c},\overline{x}') := 
            T'(\overline{x},\overline{i},\overline{c},\overline{x}')
            \wedge \bigl(c_j \leftrightarrow 
            F_j(\overline{x},\overline{i},\overline{c}_a)\bigr)$ 
    \label{alg:SafeInterpolSynt:resub}
    \State $\overline{c}_b := \overline{c}_b \cup \{c_j\}$
    \label{alg:SafeInterpolSynt:cb}
  \EndFor
\EndProcedure
\end{algorithmic}
\end{algorithm}

\mypara{Computation of $M_1$ and $M_0$.}  With
$S(\overline{x},\overline{i},\overline{c},\overline{x}') = 
T(\overline{x},\overline{i},\overline{c},\overline{x}') \wedge \bigl(\neg 
W(\overline{x}) \vee W(\overline{x}')\bigr)$, we can apply the following 
transformations to compute a more compact \acs{CNF} for $M_1$
as $M_1(\overline{x},\overline{i},\overline{c}_a) =$
\begin{linenomath*}
\begin{align*}
& \;
           \bigl(
           \exists \overline{c}_b, \overline{x}' \scope 
           S(\overline{x},
             \overline{i},
             \overline{c}_a,
             \true,
             \overline{c}_b,
             \overline{x}')
           \bigr)
           \wedge
           \bigl(
           \neg \exists \overline{c}_b, \overline{x}' \scope 
           S(\overline{x},
             \overline{i},
             \overline{c}_a,
             \false,
             \overline{c}_b,
             \overline{x}')
           \bigr)
\\=&\;
    \Bigl(
    \exists \overline{c}_b, \overline{x}' \scope 
    T(\overline{x},
      \overline{i},
      \overline{c}_a,
      \true,
      \overline{c}_b,
      \overline{x}') 
    \wedge 
    \bigl(\neg W(\overline{x}) \vee W(\overline{x}')\bigr)
    \Bigr)
    \wedge
    \Bigl(
    \neg \exists \overline{c}_b, \overline{x}' \scope 
    T(\overline{x},
      \overline{i},
      \overline{c}_a,
      \false,
      \overline{c}_b,
      \overline{x}') 
    \wedge 
    \bigl(\neg W(\overline{x}) \vee W(\overline{x}')\bigr)
    \Bigr)
\\=&\;
    \Bigl(
    \exists \overline{c}_b, \overline{x}' \scope 
    T(\overline{x},
      \overline{i},
      \overline{c}_a,
      \true,
      \overline{c}_b,
      \overline{x}') 
    \wedge 
    \bigl(\neg W(\overline{x}) \vee W(\overline{x}')\bigr)
    \Bigr)
    \wedge
\Bigl(
    \forall \overline{c}_b, \overline{x}' \scope 
    T(\overline{x},
      \overline{i},
      \overline{c}_a,
      \false,
      \overline{c}_b,
      \overline{x}') 
    \rightarrow 
    \bigl(W(\overline{x}) \wedge \neg W(\overline{x}')\bigr)
    \Bigr)
\end{align*}
\end{linenomath*}
That is, the negation turns the existential quantification over 
$\overline{c}_b, \overline{x}'$ into a universal one.  Yet, just like 
\textsc{SafeQbfSynt} (Algorithm~\ref{alg:SafeQbfSynt}), 
\textsc{SafeInterpolSynt} also keeps a copy $T'$ of the transition relation $T$
that defines all variables in $\overline{c}_b$ and $\overline{x}'$ uniquely 
based on
the other variables.  For the variables $\overline{x}'$, this holds initially. 
For $\overline{c}_b$, this is ensured by Line~\ref{alg:SafeInterpolSynt:resub}.
Thus, by using $T'$ instead of $T$ and by applying the one-point rule 
(\ref{eq:op1}), the universal quantification can be turned into an 
existential one:
\begin{linenomath*}
\begin{align*}
    \Bigl(
    \exists \overline{c}_b, \overline{x}' \scope 
    T'(\overline{x},
      \overline{i},
      \overline{c}_a,
      \true,
      \overline{c}_b,
      \overline{x}') 
    \wedge 
    \bigl(\neg W(\overline{x}) \vee W(\overline{x}')\bigr)
    \Bigr)
    \wedge
\Bigl(
    \exists \overline{c}_b, \overline{x}' \scope 
    T'(\overline{x},
      \overline{i},
      \overline{c}_a,
      \false,
      \overline{c}_b,
      \overline{x}') 
    \wedge 
    W(\overline{x}) \wedge \neg W(\overline{x}')
    \Bigr)
\end{align*}
\end{linenomath*}
By renaming the variables $\overline{c}_b$ and $\overline{x}'$, the two 
subformulas can be merged into one block of quantifiers:
\begin{linenomath*}
\begin{align*}
    \exists \overline{c}_{b1},\overline{c}_{b2},
            \overline{x}_1',\overline{x}_2' \scope %
    T'(\overline{x},
      \overline{i},
      \overline{c}_a,
      \true,
      \overline{c}_{b1},
      \overline{x}_1') 
    \wedge 
    \bigl(\neg W(\overline{x}) \vee W(\overline{x}_1')\bigr)
    \wedge
    T'(\overline{x},
      \overline{i},
      \overline{c}_a,
      \false,
      \overline{c}_{b2},
      \overline{x}_2') 
    \wedge 
    W(\overline{x}) \wedge \neg W(\overline{x}_2')
\end{align*}
\end{linenomath*}
Finally, $\bigl(\neg W(\overline{x}) \vee W(\overline{x}_1')\bigr) \wedge
W(\overline{x})$ can be simplified to $W(\overline{x}) \wedge 
W(\overline{x}_1')$, which is fortunate because negations and disjunctions
are expensive to perform in \acs{CNF}. This gives 
\[M_1(\overline{x},\overline{i},\overline{c}_a) =
    \exists \overline{c}_{b1},\overline{c}_{b2},
            \overline{x}_1',\overline{x}_2' \scope 
    T'(\overline{x},
      \overline{i},
      \overline{c}_a,
      \true,
      \overline{c}_{b1},
      \overline{x}_1') 
    \wedge 
    W(\overline{x}_1')
    \wedge
    T'(\overline{x},
      \overline{i},
      \overline{c}_a,
      \false,
      \overline{c}_{b2},
      \overline{x}_2') 
    \wedge 
    W(\overline{x}) \wedge \neg W(\overline{x}_2'). 
\]
In \textsc{SafeInterpolSynt}, the existential quantification is not applied. 
Instead, the variables $\overline{c}_{b1}$, $\overline{c}_{b2}$, 
$\overline{x}_1'$, $\overline{x}_2'$ occur freely in $M_1'$.  Similarly, other 
fresh copies $\overline{c}_{b3}, \overline{c}_{b4}, \overline{x}_3', 
\overline{x}_4'$ of the same variables occur freely in $M_0'$.  The properties 
of an interpolant (see Section~\ref{sec:prelim:prop}) ensure that $F_j$, 
computed in Line~\ref{alg:SafeInterpolSynt:int}, can only reference the 
variables $\overline{x},\overline{i},\overline{c}_a$ occurring both in $M_1'$ 
and in $M_0'$.  Hence, these free variables cannot be referenced in the 
resulting circuit.

\mypara{\acs{CNF} conversion.}
The formulas in Line~\ref{alg:SafeInterpolSynt:m1} 
and~\ref{alg:SafeInterpolSynt:m0} contain only conjunctions.  Most of our 
methods to compute a winning region or a winning area produce $W(\overline{x})$ 
in \acs{CNF}.  Hence, just as for \acs{QBF} certification and \acs{QBF}-based 
\acs{CNF} learning, we only need to compute a \acs{CNF} 
representation of $T'$ and $\neg W(\overline{x}')$. 

\mypara{Simplification of interpolants.}
The computed interpolants $F_j$ refine $T'$ in 
Line~\ref{alg:SafeInterpolSynt:resub}.  Hence, complicated representations of 
$F_j$ result in more complicated formulas for $T'$, which can increase the time 
for interpolation (and may result in even more 
complicated formulas for the subsequent interpolants).  Besides optimizing the 
final circuit regarding size, we therefore also optimize every single 
interpolant using the tool \abc~\cite{BraytonM10} after it has been computed.

\subsubsection{Discussion} \label{sec:intdisc}

\mypara{Exploiting implementation freedom.}
\textsc{InterpolSynt} is rather conservative in exploiting implementation 
freedom when computing a circuit for some control signal $c_j$:  to the extend 
where this is feasible, the circuit $F_j$ defining $c_j$ will work for 
\emph{any} realization of the control signals $\overline{c}_a$ that have not 
been synthesized yet.  The reason is that the variables of $\overline{c}_a$ are 
handled as if they were inputs.  This stands in contrast to \textsc{QbfSynt}, 
which is more greedy by exploiting implementation freedom as long as \emph{some} 
solution for the other signals still exists.  Both strategies have their 
advantages.  Preserving implementation freedom can result in smaller circuits 
for control signals that are synthesized later.  The greedy strategy can be 
better in preventing that implementation freedom is left unexploited.

\mypara{Dependencies between control signals.}  In contrast to \textsc{CofSynt} 
and \textsc{QbfSynt}, \textsc{InterpolSynt} constructs a circuit in such a way 
that the implementation for one control signal can be reused in the 
definition of others (see Figure~\ref{fig:interpol}).  This can result in a 
smaller total circuit size.  As an extreme example, one control signal $c_j$ 
could be required to be an exact copy of some other control signal $c_k$.  While 
\textsc{InterpolSynt} may find the implementation $c_j=c_k$ quickly, both 
\textsc{CofSynt} and \textsc{QbfSynt} would have to construct the same 
(potentially complicated) circuit based on the variables $\overline{x}$ and 
$\overline{i}$ twice.  The circuit optimization techniques we apply as a 
postprocessing step may optimize one copy away, so the final circuit may 
actually be the same.  Nevertheless, computing the same circuit twice is at 
least a waste of resources.

\mypara{Dependence on the interpolation procedure.}
With \textsc{InterpolSynt}, the size of the resulting circuits strongly depends 
on the ability of the interpolation procedure $\interpol$ to exploit the 
freedom between $M_1$ and $\neg M_0$.  When the interpolant is computed from an 
unsatisfiability proof returned by a \acs{SAT} solver, we must rely on the 
heuristics in the solver to yield a compact proof that can be used to derive a 
simple interpolant, which can then be implemented in a small circuit.  In 
contrast, \textsc{QbfSynt} is more independent of the underlying reasoning 
engine.  The next section will present an approach to reduce this dependency 
of \textsc{InterpolSynt} on the underlying solver.

\subsection{Query Learning Based on \acs{SAT} Solving} 
\label{sec:hw:extr:satlearn}

In this section, we combine query learning with the idea by Jiang et 
al.~\cite{JiangLH09} to temporarily treat control signals as if they were 
inputs.  This eliminates the need for universal quantification and allows us to 
implement the query learning approach from Section~\ref{sec:hw:extr:qbflearn} 
with a \acs{SAT} solver instead of a \acs{QBF} solver.
In the following subsection, we will present a solution based on \acs{CNF} 
learning.  Applying other learning algorithms from~\cite{EhlersKH12} publication
is possible, but imposes more overhead 
for encoding formula parts into \acs{CNF}. After introducing the basic 
algorithm, we will again present an efficient realization for safety synthesis 
problems and discuss the differences to the other algorithms.

\subsubsection{\acs{CNF} Learning Based on \acs{SAT} Solving}

In Section~\ref{sec:hw:extr:qbflearn}, we have discussed that query 
learning can be used as a special interpolation procedure if different formulas 
are used for counterexample computation and generalization.  While 
Section~\ref{sec:hw:extr:qbflearn} used this idea to compute interpolants 
between quantified formulas using a \acs{QBF} solver, we use it here to compute 
interpolants for propositional formulas using \acs{CNF} learning.  

\begin{wrapfigure}[11]{r}{0.60\textwidth}
\vspace{-8mm}
\centering
\begin{minipage}{0.59\textwidth}
\begin{algorithm}[H]
\caption{\textsc{CnfInterpol}: Computing an interpolant using \acs{CNF} 
learning with a \acs{SAT} solver.}
\label{alg:CnfInterpol}
\begin{algorithmic}[1]
\ProcedureRetL{CnfInterpol}{M_1(\overline{d},\overline{t}_1),
                           M_0(\overline{d},\overline{t}_0)}
                       {A \acs{CNF} $F(\overline{d})$ with 
                        $M_1 \rightarrow F \rightarrow \neg M_0$}
                       {2cm}
  \State $F(\overline{d}) := \true$
    \While{$\mathsf{sat}$ in $(\mathsf{sat}, \mathbf{d}) := 
    \propsatmodel\bigl(M_0(\overline{d},\overline{t}_0) \wedge
                       F(\overline{d})\bigr)$}
          \label{alg:CnfInterpol:check}
      \State $F(\overline{d}) := 
              F(\overline{d}) \wedge \neg 
               \propsatmincore\bigl(\mathbf{d},
                               M_1(\overline{d},\overline{t}_1)
               \bigr)$
      \label{alg:CnfInterpol:core}
    \EndWhile
    \State \textbf{return} $F(\overline{d})$
\EndProcedure
\end{algorithmic}
\end{algorithm}
\end{minipage}
\end{wrapfigure}

\mypara{Algorithm.}
We keep the basic structure of the \textsc{InterpolSynt} procedure from 
Algorithm~\ref{alg:InterpolSynt}, but replace the call to $\interpol$ in 
Line~\ref{alg:InterpolSynt:int} by a call to \textsc{CnfInterpol}, which is 
defined in Algorithm~\ref{alg:CnfInterpol}.  The interface of 
\textsc{CnfInterpol} is the same as that of any interpolation procedure: given 
two formulas $M_1(\overline{d},\overline{t}_1)$ and 
$M_0(\overline{d},\overline{t}_0)$ such that $M_1\wedge M_0$ is unsatisfiable, 
it returns a formula $F(\overline{d})$ over 
the shared variables $\overline{d}$ such that $M_1 \rightarrow F \rightarrow 
\neg M_0$.  The implementation of \textsc{CnfInterpol} is simple. It starts with 
the initial approximation $F = \true$ and enforces the invariant $M_1 
\rightarrow F$.  Line~\ref{alg:CnfInterpol:check} checks if $F\rightarrow \neg 
M_0$, which is the case if and only if $F \wedge M_0$ is unsatisfiable.  If so, 
then $M_1 \rightarrow F \rightarrow \neg M_0$ holds, so the loop terminates and 
$F$ is returned as result.  Otherwise a counterexample $\mathbf{d} \models F 
\wedge M_0$ is extracted for which $F$ is $\true$ but must be $\false$.  The 
computation of the unsatisfiable core in Line~\ref{alg:CnfInterpol:core} 
generalizes the cube $\mathbf{d}$ by dropping literals as long as $\mathbf{d}$ 
does not intersect with $M_1$.  Consequently, the update of $F$ in 
Line~\ref{alg:CnfInterpol:core} preserves the invariant $M_1 \rightarrow F$ and 
resolves the counterexample.

\mypara{Exploiting freedom.}
As in \textsc{QbfSynt} (Algorithm~\ref{alg:QbfSynt}) using $M_0$ in 
counterexample computation makes sure that refinements of $F$ are only triggered 
if some $\overline{d}$-assignment $\mathbf{d}$ \emph{must} be mapped to 
$\false$.  Using $M_1$ in counterexample generalization entails that other 
$\overline{d}$-assignments are also mapped to $\false$ as long as they 
\emph{can} be mapped to $\false$.  Using ``can'' instead of ``must'' during 
generalization potentially eliminates more counterexamples before they are 
actually encountered by Line~\ref{alg:CnfInterpol:check}.

\subsubsection{Efficient Implementation for Safety Synthesis Problems}
\label{sec:hw:extrsat_impl}

\noindent
Following the transformations presented for \textsc{SafeInterpolSynt} in 
Section~\ref{sec:hw_interpol_impl}, \textsc{Cnf\-Interpol} is 
called
with
\begin{linenomath*}
\begin{align*}
&M_1(\overline{x},\overline{i},\overline{c}_a,
      \overline{c}_{b1},\overline{c}_{b2},\overline{x}_1',\overline{x}_2') =
    T'(\overline{x},
      \overline{i},
      \overline{c}_a,
      \true,
      \overline{c}_{b1},
      \overline{x}_1') 
    \wedge 
    W(\overline{x}_1')
    \wedge
    T'(\overline{x},
      \overline{i},
      \overline{c}_a,
      \false,
      \overline{c}_{b2},
      \overline{x}_2') 
    \wedge 
    W(\overline{x}) \wedge \neg W(\overline{x}_2') \text{ and}
\\&
M_0(\overline{x},\overline{i},\overline{c}_a,
      \overline{c}_{b3},\overline{c}_{b4},\overline{x}_3',\overline{x}_4') =
    T'(\overline{x},
      \overline{i},
      \overline{c}_a,
      \false,
      \overline{c}_{b3},
      \overline{x}_3') 
    \wedge 
    W(\overline{x}_3')
    \wedge
    T'(\overline{x},
      \overline{i},
      \overline{c}_a,
      \true,
      \overline{c}_{b4},
      \overline{x}_4') 
    \wedge 
    W(\overline{x}) \wedge \neg W(\overline{x}_4').
\end{align*}
\end{linenomath*}
Since \textsc{CnfInterpol} does not perform any negations nor 
disjunctions, only $\neg W(\overline{x}')$ needs to be transformed into 
\acs{CNF}.

\mypara{Dependency optimization.}
We can apply the dependency optimization presented in Section~\ref{sec:hw_dep}. 
However, on top of allowing dependencies on other control signals, we also allow 
dependencies on auxiliary variables that are used for defining the transition 
relation $T'$ as long as this does not result in circular dependencies.

\mypara{Incremental solving.} 
\textsc{CnfInterpol} is well suited for incremental \acs{SAT} solving.  A simple 
solution uses two solver instances, which are initialized whenever 
\textsc{CnfInterpol} is called.  The first solver instance stores $M_0 \wedge F$ 
and is used for Line~\ref{alg:CnfInterpol:check}. The second one stores $M_1$ 
and is used for Line~\ref{alg:CnfInterpol:core}.  A more radical solution uses 
only one solver instance throughout \emph{all} calls to \textsc{CnfInterpol}.  
Note that $M_1$ differs from $M_0$ only by having $c_j$ (in two copies) set to 
different truth constants.  Hence, switching between $M_1$ and $M_0$ can be 
achieved by setting (the two copies of) $c_j$ differently with assumption 
literals.  Furthermore, the clauses of some $F_j$ are all disjoined with some 
fresh activation variable $a_j$ before they are asserted in the solver.  This 
way, $F_j$ can be enabled or disabled by setting the assumption literal $\neg 
a_j$ or $a_j$, respectively.  Finally, $T'$ changes between major iterations of 
\textsc{SafeInterpolSynt} (see Line~\ref{alg:SafeInterpolSynt:resub}).  However, 
additional constraints are only added in this update, so this does not pose any 
challenge for incremental solving.

\mypara{Minimizing the final solution.}  Recall from the interpolation-based 
method from Section~\ref{sec:hw:int} that a second pass over all control signals 
can be performed, in which the circuits for all $c_j$ are recomputed while the 
implementation for the other control signals is fixed.  In principle, this has 
the potential for reducing the circuit size because the recomputed circuits can 
now rely on some concrete realization for the other control signals.  The same 
idea can also be applied in our \acs{SAT} solver based \acs{CNF} learning 
approach.  
However, similar to interpolation, recomputing individual circuits by learning 
them from scratch did not result in circuit size reductions, but more often in 
circuit size increases in our experiments.  Yet, instead of recomputing a 
circuit from scratch, we can also start with the existing solution $F_j$, which 
is given as a \acs{CNF} formula, and simplify it by dropping literals and 
clauses as long as correctness is still preserved.  The idea is similar to 
\textsc{CompressCnf} (Algorithm~\ref{alg:CompressCnf}), but the simplification 
is not equivalence preserving but only correctness preserving.  We propose to 
postprocess all $F_j$ in the order of decreasing $j$.  The reason is that $F_n$ 
was computed first, without any knowledge about the implementation of the other 
$F_j$.  Hence, intuitively, $F_n$ has the greatest potential for simplifications 
relying on the concrete realization of all other $F_j$.  Each $F_j$ satisfies 
$M_1 \rightarrow F_j \rightarrow \neg M_0$ initially, where $M_1$ and $M_0$ are 
now defined using the concrete implementation for the other $F_k$.  We propose 
to simplify each $F_j$ in two phases.  The first phase drops literals from 
clauses of $F_j$ as long as $M_1 \rightarrow F_j$ is preserved (because 
dropping literals can make $F_j$ only stronger).  Similar to 
\textsc{CompressCnf}, this can be realized by computing unsatisfiable cores, 
utilizing incremental \acs{SAT} solving.  The second phase drops clauses from 
$F_j$, starting with the longest ones, as long as $F_j \rightarrow \neg M_0$ is 
preserved (because dropping clauses can make $F_j$ only weaker).  Since we 
only drop literals and clauses from the existing implementations, this 
postprocessing can only make the resulting circuits smaller but never larger.

\subsubsection{Discussion}

The \acs{SAT} solver based \acs{CNF} learning approach is very similar to the 
interpolation-based method from the previous section, and thus inherits most of 
its strength and weaknesses.  However, using the learning algorithm instead of 
interpolation makes the approach less dependent on the underlying solver.  This 
is similar to \textsc{QbfSynt}.  Also similar to \textsc{QbfSynt} is the fact 
that individual circuits are computed as formulas in \acs{CNF}.  However, 
because the individual circuits are cascaded as illustrated in 
Figure~\ref{fig:interpol}, the final circuit depth will in general be higher 
than that of circuits produced by \textsc{QbfSynt}.  Still, the circuit depths 
can be expected to be lower compared to \textsc{InterpolSynt} in most cases.  
The reason is that interpolants derived from an unsatisfiability proof can have 
a depth that is much higher than $3$, and the procedure for building the 
individual circuits together is the same.

\subsection{Parallelization} \label{sec:hw:par_extr}

We have already discussed that different methods for circuit synthesis have 
different characteristics.  The experimental results in 
Chapter~\ref{sec:hw:exp} 
will indicate that this results in different methods and optimizations 
performing well on different classes of benchmarks.  Similar to strategy 
computation (see Section~\ref{sec:hw_par}) we thus propose a parallelization 
that executes different methods and optimizations in different threads.  The aim 
is to combine the strengths and compensate the weaknesses of the individual 
methods.

\mypara{Realization.} In contrast to Section~\ref{sec:hw_par}, our 
parallelization for synthesizing circuits from strategies follows a rather 
simple portfolio approach, where each thread solves the circuit synthesis 
problem without any information from other threads.  The first thread implements 
the \acs{SAT} solver based learning algorithm from 
Section~\ref{sec:hw:extr:satlearn} 
with the dependency optimization from Section~\ref{sec:hw_dep}.  If our 
parallelization is executed with two threads, the second thread performs 
\acs{QBF}-based \acs{CNF} learning (Section~\ref{sec:hw:extr:qbflearn}) with 
incremental \acs{QBF} solving.  If executed with three threads, the third thread 
again performs learning using a \acs{SAT} solver, but without the dependency 
optimization.  

\mypara{Heuristics.}
In order to achieve a good balance between low execution time and small 
circuits, the user can inform our parallelization about a timeout.  A heuristic 
then uses this information to decide whether to perform a minimization of the 
final solution, as explained in Section~\ref{sec:hw:extrsat_impl}, or 
not.\footnote{For the experiments, we used a very conservative heuristic: if the 
remaining time available is more than $10$ times the time used so far for 
computing a circuit from the strategy, then the minimization of the final 
solution will be performed.}  Furthermore, if one thread finishes, it does not 
stop the other threads immediately but only if the user-defined timeout is 
approaching or the ratio between waiting time and working time exceeds a certain 
threshold ($0.25$ in our experiments).  The reason is that, from all threads 
that terminated, we finally select the circuit with the lowest number of gates. 
Hence, even if one thread has already found a solution, waiting for other 
threads to finish their computation can be beneficial for the final circuit 
size.

\mypara{Alternatives.}  As for strategy computation, there is a plethora of 
possibilities to combine different methods while sharing information in a more 
fine-grained way.  Since most of the methods compute circuits for one control 
signal after the other, the final solutions for each control signal can be 
exchanged.  Each thread can then continue with the smallest solution that has 
been found for the respective signal.  Since several methods are based on 
counterexample-guided refinements of solution candidates, the respective threads 
can also exchange counterexamples and the corresponding blocking 
clauses.  Furthermore, it can be beneficial to have different threads 
synthesizing circuits for control signals in different order.  We leave an 
exploration of such fine-grained parallelization approaches for future work.

%% file: 05experiments.tex
\section{Experimental Results}\label{sec:hw:exp}

In this section, we will first sketch our implementation of the \acs{SAT}-based 
synthesis algorithms introduced so far.  After that, we will describe benchmarks 
that will be used in our experimental evaluation (Section~\ref{sec:hw:bench}).  
The core of this section is formed by our performance evaluation for computing 
strategies (Section~\ref{sec:hw:ex:strati}) and for constructing circuits from 
strategies (Section~\ref{hw:res:extr}).  The section concludes with a discussion 
of the central results (Section~\ref{sec:hw:exp:concl}).

\subsection{Implementation} \label{sec:dem_impl}

We have implemented the synthesis methods presented in Chapter~\ref{sec:hw:win} 
and Chapter~\ref{sec:hw:circ} in a synthesis tool called \demiurge.  It is 
written in \textsf{C++} and compatible with the rules for the 
\syntcomp~\cite{sttt_syntcomp} synthesis competition.  \demiurge has won 
two gold medals in this synthesis competition: one in 2014 and one in 2015, both 
in the parallel synthesis track.  The input of \demiurge is a safety 
specification in \aiger format.  The 
synthesis result is a circuit in \aiger format as well.  Since the synthesis 
process does not involve any interaction with the user except for setting 
parameters, \demiurge does not come with a \acs{GUI}, but is started from the 
command-line.  So far, our synthesis tool has only been tested on \textsf{Linux} 
operating systems.  \demiurge is freely available under the \textsf{GNU} Lesser 
General Public License version 3, and can be downloaded from
\begin{center}
\url{https://www.iaik.tugraz.at/content/research/opensource/demiurge/}.
\end{center}
All experiments presented in this article have been performed using version 
\texttt{1.2.0}.  The downloadable archive contains all scripts to reproduce the 
experiments, as well as spreadsheets with more detailed data (such as execution 
times for individual steps of the algorithms, numbers of iterations, etc.).

\begin{wrapfigure}[13]{r}{0.51\textwidth}
\centering
  \includegraphics[width=0.5\textwidth]{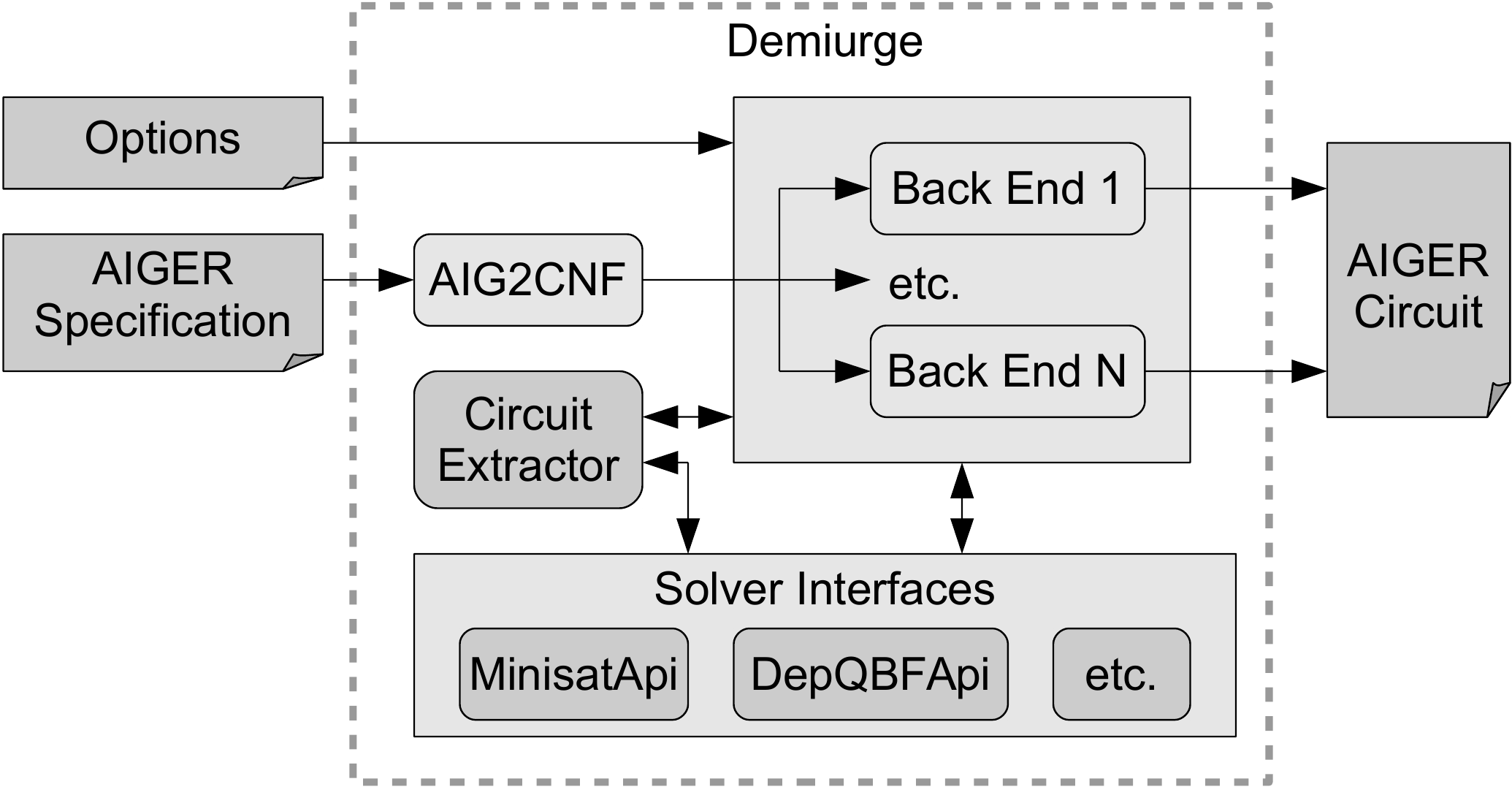}
\caption{Architecture of the \acs{SAT}-based synthesis tool \demiurge.}
\label{fig:impl}
\end{wrapfigure}
\mypara{Architecture.}
The architecture of \demiurge is outlined in Figure~\ref{fig:impl}.  The 
\textsf{AIG2CNF} module parses the specification into \acs{CNF} formulas 
representing the transition relation $T$ and the set of safe states $P$.  Only 
one initial state is allowed in the input format, so the initial states $I$ in 
our definition of a safety specification are represented as a minterm.  Next, 
the back end selected by the user via command-line options is executed.  The 
back ends mostly differ in their method for computing the winning region (or a 
winning area), and can be parameterized with a method for computing the circuit 
from the induced winning strategy.  Furthermore, the back ends can be configured 
with options to enable or disable optimizations or optional steps. The back ends 
can access a number of different solvers via uniform interfaces.  That is, 
multiple \acs{SAT} solvers can be accessed via the same abstract interface, 
which hides the concrete solver from the application.  Various \acs{QBF} solvers 
are accessible via a second interface (which is similar to the interface for 
\acs{SAT} solvers).  The concrete solvers that shall be used are again 
configured via command-line options.  Due to this extensible architecture, 
\demiurge can also be seen as a framework for implementing new synthesis 
algorithms or optimizations with low effort:  A lot of infrastructure such as 
the parser, interfaces to solvers and entire synthesis steps (like computing a 
circuit from a strategy) can be reused.

\noindent
\mypara{External tools.}
In version \texttt{1.2.0}, \demiurge has interfaces to 
\begin{compactitem}
\item the \acs{SAT} solver \minisat~\cite{EenS03} in version \texttt{2.2.0} via 
its \acs{API},
\item the \acs{SAT} solver \picosat~\cite{Biere08} in version \texttt{960} via 
its \acs{API},
\item the \acs{SAT} solver \lingeling~\cite{lingeling} in version \texttt{ayv} 
via its \acs{API},
\item the \acs{QBF} solver \depqbf~\cite{LonsingB10,LonsingE14b} in version 
\texttt{3.04} via its \acs{API}, both with and without preprocessing by 
\bloqqer~\cite{BiereLS11,SeidlK14} 
version \texttt{34},
\item the \acs{QBF} solver \rareqs~\cite{JanotaKMC12} in version \texttt{1.1} 
via a self-made \acs{API},
\item the \acs{QBF} solver \qube~\cite{GiunchigliaNT01} in version \texttt{7.2} 
with communication via files,
\item the tool \abc~\cite{BraytonM10} (commit \texttt{d3db71b}) for optimizing 
\aiger circuits with communication via files, and
\item the first-order theorem prover \iprover~\cite{Korovin08} in version 
\texttt{1.0} with communication via files.
\end{compactitem}

\subsection{Benchmarks} \label{sec:hw:bench}

\begin{table}
\setlength{\tabcolsep}{5.5mm}
\centering
\caption[Summary of benchmark sizes]
{Summary of benchmark sizes. The suffix k multiplies by 1000. The suffix M 
multiplies by one million. 
}
\label{tab:benchsize}
\begin{tabular}{lccccc}
\toprule
Name 
 & parameter range
 & $|\overline{x}|$ 
 & $|\overline{i}|$ 
 & $|\overline{c}|$ 
 & Gates defining $T$\\
\midrule
\add$ko$  
  & $k=2$ to $20$  
  & $2$ 
  & $2\cdot k$ 
  & $k$ 
  & $17$ to $365$\\
\mult$k$  
  & $k=2$ to $16$  
  & $0$ 
  & $2\cdot k$ 
  & $2\cdot k$ 
  & $24$ to $2450$\\
\cnt$ko$  
  & $k=2$ to $30$  
  & $k+1$ 
  & $1$ 
  & $1$ 
  & $11$ to $450$\\
\mv$ko$  
  & $k=2$ to $28$  
  & $k+1$ 
  & $k-1$ 
  & $k-1$ 
  & $10$ to $469$\\
\bs$ko$  
  & $k=8$ to $128$  
  & $k+1$ 
  & $\ld(k)$ 
  & $1$ 
  & $80$ to $3202$\\
\stay$ko$  
  & $k=2$ to $24$  
  & $k+2$ 
  & $k$ 
  & $k+1$ 
  & $17$ to $4104$\\
\amba$kl$  
  & $k=2$ to $10$  
  & $28$ to $76$ 
  & $2\cdot k+3$
  & $8$ to $19$
  & $177$ to $630$\\
\genbuf$kl$  
  & $k=1$ to $16$  
  & $21$ to $73$ 
  & $k+4$
  & $6$ to $24$
  & $134$ to $733$\\
\fa$mnkc$ 
  & special selection 
  & $20$ to $54$ 
  & $10$ to $40$
  & $8$ to $12$
  & $122$ to $594$\\
\mo$klm$ 
  & $k=l=8$ to $128$ 
  & $19$ to $41$ 
  & $12$ to $34$
  & $5$
  & $306$ to $830$\\
\driver$kl$ 
  & $l=5$ to $8$ 
  & $55$ to $326$ 
  & $16$ to $98$
  & $24$ to $82$
  & $435$ to $1942$\\
\demo$kl$ 
  & $k=1$ to $25$ 
  & $12$ to $280$ 
  & $1$ to $4$
  & $1$ to $4$
  & $43$ to $2055$\\
\gb$k$
  & $k=1$ to $4$ 
  & $11$k to $23$k 
  & $4$
  & $4$
  & $867$k to $1.7$M\\
\load$kl$
  & $k=2$ to $3$ 
  & $96$ to $296$ 
  & $3$ to $4$
  & $2$ to $3$
  & $1092$ to $3156$\\
\ltltodba$kl$
  & $k=1$ to $20$ 
  & $44$ to $484$ 
  & $2$ to $7$
  & $1$
  & $194$ to $5482$\\
\ltltodpa$k$
  & $k=1$ to $18$ 
  & $44$ to $340$ 
  & $1$ to $3$
  & $2$ to $4$
  & $191$ to $3866$\\
\bottomrule
\end{tabular}
\end{table}

We used benchmarks from the \syntcompy~\cite{sttt_syntcomp} benchmark set\footnote{We used \emph{all} 
benchmark instances from this set with two exceptions: From the \amba\ and 
\genbuf\ benchmarks, we did not select the unoptimized and the unrealizable 
instances to keep the number of instances manageable and balanced.  Second, we 
also included a \driver\ benchmark that is not contained in the \syntcompy 
benchmark set.} to evaluate the performance of our different methods to compute 
strategies as well as circuits implementing these strategies.  Most of the 
benchmarks are parameterized.  In the following, we briefly summarize their 
main characteristics as far as this is helpful for interpreting the 
performance results.  The size of the benchmarks is summarized in 
Table~\ref{tab:benchsize}. All in all, we included $350$ benchmark instances, of 
which $40$ instances are unrealizable.

The \add$ko$ benchmark specifies a combinational adder for two $k$-bit 
numbers. The parameter $o\in\{\texttt{y},\texttt{n}\}$ indicates if the 
benchmark file has been optimized with \abc~\cite{BraytonM10} for circuit size 
(value $\texttt{y}$) or not (value $\texttt{n}$).  This benchmarks is 
realizable.  
Since it is mostly combinational, it challenges circuit synthesis more than 
strategy computation.  

The \mult$k$ benchmark specifies a combinational multiplier for two $k$-bit 
numbers and is thus similar to \add.

The \cnt$ko$ benchmark specifies a $k$-bit counter that must not reach its 
maximum value.  At value $2^{k-1}-1$, the counter can be reset if the only 
control signal is set to $\true$.  The parameter $o\in\{\texttt{y},\texttt{n}\}$ 
again indicates if the benchmark was optimized. This benchmark is realizable and 
can be challenging for strategy computation because it may require many 
iterations to find the winning region.  It is trivial for circuit synthesis, 
though, because hardwiring the only control signal to $\true$ suffices.

The \mv$ko$ benchmark also contains a $k$-bit counter that must not reach its 
maximum value. However, when the most significant counter bit is set, the 
counter can be reset if the XOR sum of all control signals is $\true$.  Hence, 
there exists an implementation that hardwires all control signals to constant 
values. Again, $o\in\{\texttt{y},\texttt{n}\}$ indicates if the benchmark was 
optimized. The benchmark is realizable and can be challenging for circuit 
synthesis because it contains many interdependent control signals.

The realizable benchmark \bs$ko$ applies a barrel shifter to a $k$-bit register, 
which is initialized to some constant value and must never reach specific 
values. The amount of shifting is defined by uncontrollable inputs, but the 
shifting can be disabled with a control signal.  Barrel shifters can be 
particularly challenging for \acp{BDD}.

The benchmark \stay$ko$ again contains a $k$-bit counter that must not reach its 
maximum value.  Whether the counter is incremented or not depends on complicated 
logic, involving an arithmetic multiplication of the control signals with the 
uncontrollable inputs.  Yet, when setting one specific control signal always to 
$\false$, the specification is always satisfied.  Hence, the crux with this 
benchmark is whether the algorithms can find and exploit this ``backdoor''.

The benchmark \amba$kl$ specifies an arbiter for ARM's AMBA AHB 
bus~\cite{BloemGJPPW07} with $k$ bus masters.  The parameter 
$l\in\{\texttt{b},\texttt{c},\texttt{f}\}$ describes the method that has been 
used for transforming liveness properties in the original formulation of the 
benchmark~\cite{BloemGJPPW07} into safety properties.  We refer to Jacobs et 
al.~\cite{sttt_syntcomp} for a description of these three transformations. All 
benchmark instances are available in an optimized and in an unoptimized form.  
Additionally, all benchmark instances are available in an unrealizable variant.  
However, since the performance difference between all these variants are rather 
small, we only ran our experiments with the realizable and optimized versions.  

The benchmark \genbuf$kl$ specifies a generalized buffer~\cite{BloemGJPPW07} 
connecting $k$ senders to two receivers.  The parameter 
$l\in\{\texttt{b},\texttt{c},\texttt{f}\}$ is the same as for the 
\amba\ benchmarks.  Similar to \amba, we only ran our experiments with the 
realizable and optimized versions in order to reduce the number of instances.

The \fa$mnkc$ benchmark specifies a factory line with $m$ tasks that need to 
be performed by two manipulation arms on a continuous stream of objects. The 
factory belt has $n$ places and rotates every $k$ cycles by one place, thereby 
delivering an object.  The parameter $c$ is a maximum number of errors in the 
setup of the processed objects that needs to be tolerated by the factory line.
Some of the included benchmark instances are unrealizable.

The \mo$klm$ benchmark specifies a robot that has to move in a 
two-dimensional grid of $k \times l$ cells while avoiding collisions with a 
moving obstacle.  By default, the obstacle can only move in every second step. 
However, at most $m$ times, the obstacle can also move in consecutive time 
steps.  For every grid size, our benchmark set contains an unrealizable and a 
realizable instance.

The benchmark \driver$kl$ specifies an IDE hard drive controller based on an 
operating system interface specification~\cite{RyzhykWKLRSV14}.  The parameter 
$k\in\{\texttt{a}, \texttt{b}, \texttt{c}, \texttt{d}\}$ encodes the level of 
manual abstraction that has been applied when translating the benchmark into 
a safety specification.  The value $\texttt{a}$ means that no abstraction has 
been applied, and the value $\texttt{d}$ means that many details have been 
simplified. The parameter $l\in\{5,6,7,8\}$ is a bound on the reaction time. 
The benchmark is only realizable for $l=8$.

The remaining benchmarks are \acs{LTL} formulas that are contained as examples 
in the distribution of the synthesis tool \acaciap~\cite{BohyBFJR12}. They have 
been translated into safety specifications using the approach by Filiot et 
al.~\cite{FiliotJR11}.
The \demo$kl$ benchmarks represent \acs{LTL} formulas that have originally 
been used as benchmarks for the synthesis tool \lily~\cite{JobstmannB06}.  
Here, $k$ is just a running number without any special meaning and $l$ is a 
bound for the liveness-to-safety transformation.  Some of these benchmarks are 
unrealizable.
The benchmark \gb$k$ represents a different formulation of the generalized 
buffer benchmark \genbuf\ for two senders and two receivers.  The parameter $k$ 
is here a bound for the liveness-to-safety transformation.  One of these 
instances is unrealizable.
The benchmark \load$kl$ contains a specification of a load balancing 
system~\cite{Ehlers12} for $k$ clients that has been used as a case study for 
the \unbeast synthesis tool~\cite{Ehlers12}.  The parameter $l$ is again a 
bound for the liveness-to-safety transformation. One of these instances is 
unrealizable. 
Finally, the benchmarks \ltltodba$kl$ and \ltltodpa$k$ from the 
\acaciap~\cite{BohyBFJR12} examples have been translated. Here, $k$ is 
just a running index without any special meaning, and $l$ is again a parameter 
of the translation.  From these benchmarks, some instances are also 
unrealizable.

\subsection{Strategy Computation Results} \label{sec:hw:ex:strati}

In this section, we compare different methods for strategy computation.  Methods 
for computing circuits that implement a given strategy will be evaluated in 
Section~\ref{hw:res:extr}.  First, we will describe the compared methods and 
their configuration.  Section~\ref{sec:hw:exp:big} will then present performance 
results on the average over all our benchmarks.  A more detailed investigation 
for the individual benchmark classes is then performed in 
Section~\ref{sec:hw:exp:class}.  Section~\ref{sec:hw:exp:fu} will finally 
highlight other interesting observations.   All experiments reported in this 
section were performed on an Intel Xeon E5430 CPU with 4 cores running at 
$2.66$\,GHz, and a 64 bit Linux.

\subsubsection{Evaluated Configurations} \label{hw:res:win:conf}

\noindent
Table~\ref{tab:winconfig} summarizes the methods and their configurations we 
compare in this thesis.

\mypara{Baseline.}
\wbdd denotes a \acs{BDD}-based implementation of the standard \textsc{SafeWin} 
procedure presented in Algorithm~\ref{alg:SafeWin}.  It has been implemented by 
students and won a synthesis competition that has been carried out in a lecture. 
It is fairly optimized: it uses dynamic variable reordering, forced reorderings 
at certain points, combined \acs{BDD} operations, and a cache to speed up the 
construction of the transition relation.  See Section~\ref{sec:prelim:bdds} for 
more background.  
\wifm denotes a reimplementation of the approach by Morgenstern et 
al.~\cite{MorgensternGS13}.  It is inspired by the model checking algorithm 
\icthree~\cite{Bradley11} and based on \acs{SAT} solving.  
\abssynthe is a \acs{BDD}-based synthesis tool that uses abstraction and 
refinement\footnote{Abstraction and refinement are applied (roughly) in the 
following way. Only a subset of the state variables are considered. Based on 
this subset, an under-approximation and an over-approximation of the mixed 
preimage operator $\FS$ are defined.  These are used to compute an 
over-approximation $W\!\!\uparrow$ and an under-approximation 
$W\!\!\downarrow$ of the winning region. If the initial state is in 
$W\!\!\downarrow$, the specification is realizable.  If it is not contained in 
$W\!\!\uparrow$, the specification is unrealizable.  Otherwise, the 
abstraction is refined by considering additional state variables.} as well as 
other advanced optimizations~\cite{BrenguierPRS14}. It won the sequential 
synthesis track in the \syntcompy~\cite{sttt_syntcomp} competition.  In version 
2.0 (the version we compare to), \abssynthe has also been extended with an 
approach for compositional synthesis.  \abssynthe can therefore be considered as 
one of the leading state-of-the-art synthesis tools for safety specifications.
Together with \wifm and \wbdd, it serves as a 
baseline for our comparison.  Since this section only evaluates the strategy 
computation, the circuit extraction is disabled in all baseline tools for now.

\mypara{\acs{QBF}-based learning.}
The configurations starting with a \textsf{Q} represent different realizations 
of the \textsc{QbfWin} procedure shown in Algorithm~\ref{alg:QbfWin}.  This 
includes the basic algorithm with \acs{QBF} preprocessing (\wqb) and without 
preprocessing (\wq), a version (\wqgb) using optimization \textsf{RG} (see 
Section~\ref{sec:hw_rg}), and a version (\wqgcb) that also uses optimization 
\textsf{RC} (see Section~\ref{sec:hw_rc}).  Furthermore, we present results for 
an implementation (\wqgab) that computes all counterexample generalizations 
instead of just one (see Section~\ref{sec:qbf_var}), and for one of our three 
approaches (named \wqi) for incremental \acs{QBF} solving (see 
Section~\ref{sec:qbf_impl}).  The results for the other two methods using 
incremental \acs{QBF} solving are similar and can be found in the downloadable 
archive.  The downloadable archive also contains other combinations of the different options
and optimizations (20 in total).

\begin{table}
\centering
\caption[Configurations for computing a winning strategy]
{Configurations for computing a winning strategy.}
\label{tab:winconfig}
\begin{tabular}{lll}
\toprule
Name & Algorithm and Optimizations & Solver \\
\midrule
\wbdd  
& \textsc{SafeWin} (Alg.~\ref{alg:SafeWin})
& \textsf{CuDD} \\
\wifm
& Re-implementation of~\cite{MorgensternGS13}
& \minisat\\
\wabs
& \abssynthe 2.0~\cite{BrenguierPRS14}
& \textsf{CuDD}\\
\wq
& \textsc{QbfWin} (Alg.~\ref{alg:QbfWin})
& \depqbf\\
\wqb  
& \textsc{QbfWin} (Alg.~\ref{alg:QbfWin})
& \depqbf + \bloqqer\\
\wqgb
& \textsc{QbfWin} (Alg.~\ref{alg:QbfWin}) + 
  Opt.~\textsf{RG} (Sect.~\ref{sec:hw_rg})
& \depqbf + \bloqqer\\
\wqgab
& \wqgb + computing all generalizations (Sect.~\ref{sec:qbf_var})
& \depqbf + \bloqqer\\ 
\wqgcb  
& \textsc{QbfWin} (Alg.~\ref{alg:QbfWin}) + 
  Opt.~\textsf{RG} and \textsf{RC} (Sect.~\ref{sec:hw_reach})
& \depqbf + \bloqqer\\
\wqi
& Incremental \textsc{QbfWin} with variable pool (Sect.~\ref{sec:qbf_impl}) 
& Incremental \depqbf\\
\ws
& \textsc{SatWin1} (Alg.~\ref{alg:SatWin1})
& \minisat\\
\wsg 
& \textsc{SatWin1} (Alg.~\ref{alg:SatWin1})
  + Opt.~\textsf{RG} (Sect.~\ref{sec:hw_rg})
& \minisat\\  
\wsgc  
& \textsc{SatWin1} (Alg.~\ref{alg:SatWin1})
  + Opt.~\textsf{RG} and \textsf{RC} (Sect.~\ref{sec:hw_reach})
& \minisat\\  
\wse  
& \textsc{SatWin1} (Alg.~\ref{alg:SatWin1})
  + Expansion (Sect.~\ref{sec:hw_exp})
& \minisat\\  
\wsge 
& \textsc{SatWin1} (Alg.~\ref{alg:SatWin1})
  + Opt.~\textsf{RG}
  + Expansion
& \minisat\\
\wtqc   
& Eq.~(\ref{eq:templ}) + \acs{CNF} Templates (Sect.~\ref{sec:cnfteml})
& \depqbf\\
\wtbc  
& Eq.~(\ref{eq:templ}) + \acs{CNF} Templates (Sect.~\ref{sec:cnfteml})
& \depqbf + \bloqqer\\
\wtsc 
& Eq.~(\ref{eq:templ}) + \acs{CEGIS} (Alg.~\ref{alg:TemplWinSat}) + \acs{CNF} 
Templates
& \minisat\\
\wepr
& Reduction to \acs{EPR} (Sect.~\ref{sec:hw:win:eprim}) & 
\iprover \\
\wptwo
& Parallel (Sect.~\ref{sec:hw_par}) with 2 threads & 
\minisat + \depqbf + \bloqqer \\
\wpthree  
& Parallel (Sect.~\ref{sec:hw_par}) with 3 threads & 
\minisat + \depqbf + \bloqqer \\
\bottomrule
\end{tabular}
\end{table}

\mypara{Learning based on \acs{SAT} solvers.}  All configurations of 
the \acs{SAT} solver based learning procedure \textsc{SatWin1}, presented in 
Algorithm~\ref{alg:SatWin1}, are all named with an \textsf{S} as first letter. 
Our comparison contains a plain implementation (\ws), a variant (\wsg) with 
optimization  \textsf{RG} (see Section~\ref{sec:hw_rg}), and a version (\wsgc) 
that also performs optimization \textsf{RC} (see Section~\ref{sec:hw_rc}).  The 
former two are also combined with our heuristic for performing universal 
expansion (see Section~\ref{sec:hw_exp}), named \wse and \wsge respectively. To 
simply the matters, we only present result using the \acs{SAT} solver \minisat. 
Results using \lingeling and \picosat can be found in the downloadable archive.
Both \lingeling and \picosat can be faster than \minisat for individual 
benchmark instances, but \minisat yields better results on average.

\mypara{Template-based approach.}  All configurations of our template-based 
approach (see Section~\ref{sec:hw:templ}) start with a \textsf{T}.  A 
\acs{QBF}-based implementation with and without \acs{QBF} preprocessing is 
realized in \wtbc and \wtqc, respectively.  \wtsc denotes a \acs{SAT} 
solver based realization using our variant of the \acs{CEGIS} algorithm (see 
Algorithm~\ref{alg:TemplWinSat}).  We only present results using \acs{CNF} 
templates.  Results using AND-inverter graph templates are similar and 
can be found in the archive.

\mypara{Reduction to \acs{EPR}.}  The configuration realizing the approach of 
Section~\ref{sec:hw:win:eprim} is named \wepr.

\mypara{Parallelization.}  The results produced by our parallelization with one 
thread are essentially the same as for \wsge because our parallelization 
executes \wsge when used with one thread.  The additional communication overhead 
is negligible.  The configurations with two and three threads are named \wptwo 
and \wpthree, respectively.  The additional speedup we achieve with more than 
three threads is rather insignificant.  Thus, we do not present any results 
with more threads.

\subsubsection{The Big Picture} \label{sec:hw:exp:big}

We executed the configurations listed in Table~\ref{tab:winconfig} with a 
timeout of $10\,000$ seconds per benchmark instance and a memory limit of $8$ 
GB.  Figure~\ref{fig:win_cactus} gives an overview of the resulting execution 
times in form of a \index{cactus plot} cactus plot.  The horizontal axis 
contains the benchmarks, sorted in the order of increasing execution times 
(individually for each configuration).  The vertical axis shows the 
corresponding execution time on a logarithmic scale. Hence, the lines for the 
individual configurations can only rise, and the steeper a line rises, the worse 
is its scalability.  Another way to read cactus plots is as follows:  For a 
given time limit on the vertical axis, the horizontal axis contains the number 
of benchmarks that can be solved within this time limit.  We omitted some of the 
exotic configurations from Table~\ref{tab:winconfig} (namely \wq, \wqgab, 
\wqgcb, \wsgc, and \wse) to keep the plot readable.  In the following 
paragraphs, we will focus on the most important observations based on 
Figure~\ref{fig:win_cactus}.  A more detailed comparison will be given in the 
next subsections.  

\begin{figure}
\centering
  \includegraphics[width=0.8\textwidth]{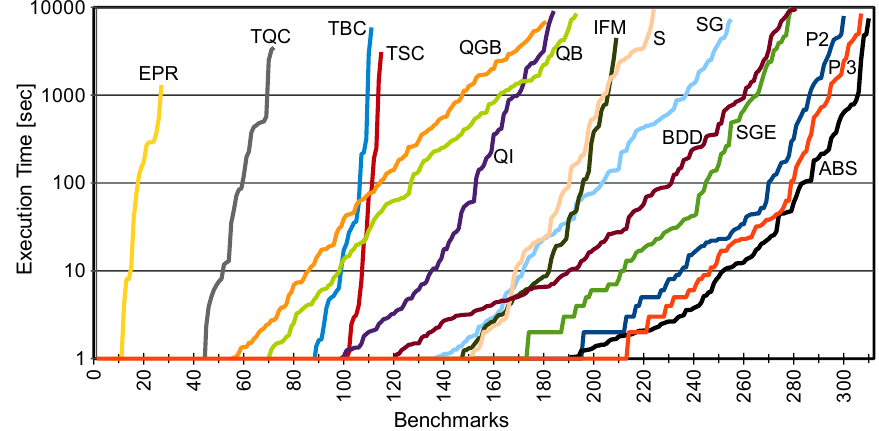}
\caption{A cactus plot summarizing the execution times for computing a winning strategy 
with different methods and configurations.}
\label{fig:win_cactus}
\end{figure}

\mypara{Our reduction to \acs{EPR} does not scale well.} \wepr could only solve 
$27$ instances.  In none of the cases, a timeout was hit.  For all instances 
that could not be solved, \iprover ran out of memory.

\mypara{Our template-based configurations solve only few instances.}  By 
comparing the lines for \wtqc and \wtbc, we can see that \acs{QBF} preprocessing 
improves the scalability of our \acs{QBF}-based realization of the 
template-based approach quite significantly.  Our implementation \wtsc using 
\acs{CEGIS} and \acs{SAT} solving can even solve a few more instances. 
In all three cases, the lines rise very steeply.  Slightly oversimplified, this 
means that the template-based methods either find a solution quickly or not at 
all.  Unfortunately, the latter case happens more often.  In total, \wtsc solves 
only 115 instances, which is low compared to the other methods.  However, the 
solved instances include some that cannot be solved by any other method, so the 
template-based approach can complement other techniques.  We will elaborate on 
this aspect in the next section.  Except for the (very large) \gb~benchmarks, 
the memory limit was never exceeded.

\mypara{Incremental \acs{QBF} solving gives a solid speedup for simple 
benchmark instances.} Compared to \wqb and \wqgb, the realization \wqi using 
incremental \acs{QBF} solving is faster on average by more than one order of 
magnitude for simple benchmark instances.  For example, the $130$ 
simplest  instances for \wqb can all be solved by \wqb in less than $137$ 
seconds each, while the $130$ simplest instances for \wqi can be solved by \wqi 
in less than $6.2$ seconds each.  Yet, for more complex instances, \wqi 
falls behind \wqb. One possible reason is the lack of \acs{QBF} preprocessing in 
\wqi, which appears to be a promising future research direction.

\mypara{\acs{SAT} solvers can outperform \acs{QBF} solvers when learning a 
winning region.}   All our \acs{QBF}-based methods are outperformed 
significantly even by the plain \acs{SAT} solver based implementation \ws.  The 
observation that it can be beneficial to solve \acs{QBF} problems with plain 
\acs{SAT} solving is not new~\cite{MorgensternGS13,JanotaKMC12}, and hence 
not completely surprising.  The plain implementation \ws can already solve more 
instances than our reimplementation \wifm  of the approach by Morgenstern 
et~al.~\cite{MorgensternGS13}.

\mypara{Optimization \textsf{RG} yields a speedup of roughly one order of 
magnitude for method \ws.}  This can be observed, at least for larger benchmark 
instances, when comparing the lines for \wsg and \ws in 
Figure~\ref{fig:win_cactus}.  For example, the $224$ simplest instances for \ws 
can each be solved by \ws in at most $9450$ seconds. On the other hand, \wsg can 
solve its $224$ simplest instances in at most $470$ seconds, which is $20$ times 
shorter.  Interestingly, optimization \textsf{RG} is \emph{not} beneficial when 
applied to our \acs{QBF}-based implementation (compare \wqgb versus \wqb) on the 
average over all our benchmarks, though.  But even in the \acs{QBF} case, it 
still yields a significant speedup for certain benchmark instances.  While 
optimization \textsf{RG} turned out to be very effective,  optimization 
\textsf{RC} does not have a positive effect on average in our experiments: the 
number of solved instances decreases from $255$ to $236$ when switching from 
\wsg to \wsgc (this is not shown in Figure~\ref{fig:win_cactus} but can be seen 
in Table~\ref{tab:winsolved}).  But optimization \textsf{RC} is also beneficial 
for individual benchmark instances and, thus, not useless either.

\mypara{Our heuristic for quantifier expansion gives a speedup of roughly one 
more order of magnitude.}  This can be seen by comparing the line for \wsge with 
that for \wsg.  Nailed down by numbers, \wsg solves its $254$ simplest 
benchmarks in at most $6800$ seconds each, while \wsge solves its $254$ simplest 
benchmarks in at most $268$ seconds, which is $25$ times shorter.  The 
configuration \wsge is already on a par with \wbdd.

\mypara{Our parallelization achieves a speedup of more than one additional order 
of magnitude.} \wsge solves its $279$ simplest benchmarks in at most $9920$ 
seconds each.  \wptwo solves its $279$ simplest benchmarks in 
at most $295$ seconds, which is $33$ times shorter.  \wpthree never requires 
more than $105$ seconds on its $279$ simplest benchmarks, which can even be seen 
as a speedup by a factor of $95$ over \wsge.  Obviously, these speedups 
are not primarily the result of exploiting hardware parallelism.  They rather 
stem from combining different approaches that complement each other.  Although 
\abssynthe uses advanced techniques such as abstraction/refinement, \wpthree 
is not far~behind (\wpthree solves $4$ instances less).

\subsubsection{Performance per Benchmark Class} \label{sec:hw:exp:class}

The previous section discussed the performance of the individual methods and 
configurations on average over all our benchmarks.  In this section, we will 
perform a more fine-grained analysis for the different classes of benchmarks.

Table~\ref{tab:winsolved} lists the number of solved benchmark instances per 
benchmark class.  The 
first line gives the total number of benchmarks in the class.  The 
last column gives the total number of benchmarks for which a winning strategy 
could be computed by the respective method within $10\,000$ seconds and a memory 
limit of $8$ GB.
Recall that a description of the compared methods and their 
configurations can be found in Table~\ref{tab:winconfig}.  Some statistics on 
the benchmarks can be found in Table~\ref{tab:benchsize}.  Table~\ref{tab:winsolved} marks the ``best'' configuration for a 
certain benchmark class in blue:  If several methods solve the same amount of 
instances, we marked the one with the lowest total execution time.  For cases 
where the difference in the total execution time is insignificant, we marked 
several configurations.  If most of the configurations solve all instances of a 
certain benchmark class in an insignificant amount of time, we refrain from 
marking them.  Moreover, we do not include \wabs and the parallelizations in 
this ranking because they combine several techniques.

\begin{table}
\setlength{\tabcolsep}{1.13mm}
\centering
\caption{Computing a winning strategy: solved instances per benchmark class.}
\label{tab:winsolved}
\begin{tabular}{lccccccccccccccccc}
\toprule
         &{\add}
         &{\mult}   
         &{\cnt}    
         &{\mv}     
         &{\bs}     
         &{\stay}   
         &{\amba}   
         &{\genbuf} 
         &{\fa}     
         &{\mo}     
         &{\driver} 
         &{\demo}   
         &{\gb}     
         &{\load}   
         &{\ltltodba}     
         &{\ltltodpa} 
         &{Total} \\
\midrule
Total    &20      &14      &28      &32      &10      &24      &27      &48      &15      &16      &16      &50      &4       &5       &23            &18        &350   \\
\midrule
\wbdd    &\tm{20} &8       &26      &32      &4       &16      &21      &\tm{48} &\tm{12} &\tm{11} &6       &46      &0       &3       &18            &10        &281   \\
\wifm    &8       &7       &16      &32      &10      &10      &3       &7       &2       &2       &\tm{16} &\tm{50} &0       &\tm{5}  &23            &18        &209   \\
\wq      &8       &4       &24      &30      &10      &10      &2       &7       &1       &0       &2       &44      &0       &2       &20            &15        &179   \\
\wqb     &10      &8       &22      &32      &10      &12      &3       &7       &3       &0       &8       &43      &0       &2       &19            &14        &193   \\
\wqgb    &10      &8       &22      &32      &10      &12      &3       &11      &3       &0       &5       &42      &0       &2       &14            &7         &181   \\
\wqgab   &10      &8       &22      &32      &10      &12      &3       &11      &3       &0       &4       &42      &0       &1       &6             &9         &173   \\
\wqgcb   &10      &8       &22      &32      &10      &12      &3       &12      &3       &0       &7       &44      &0       &4       &18            &12        &197   \\
\wqi     &8       &4       &22      &30      &10      &10      &3       &7       &2       &0       &2       &45      &0       &3       &22            &16        &184   \\
\ws      &8       &7       &24      &32      &10      &18      &15      &13      &2       &2       &3       &46      &0       &4       &23            &17        &224   \\
\wsg     &8       &7       &24      &32      &10      &17      &18      &29      &2       &2       &10      &50      &0       &5       &23            &18        &255   \\
\wsgc    &10      &7       &22      &32      &10      &12      &15      &25      &2       &1       &6       &49      &0       &\tm{5}  &23            &17        &236   \\
\wse     &20      &9       &24      &32      &10      &12      &19      &20      &3       &2       &3       &48      &0       &4       &23            &18        &247   \\
\wsge    &20      &9       &24      &32      &10      &12      &\tm{21} &37      &5       &3       &10      &\tm{50} &0       &\tm{5}  &\tm{23}       &\tm{18}   &279   \\
\wtqc    &10      &9       &24      &12      &4       &10      &0       &0       &0       &0       &0       &3       &0       &0       &0             &0         &72    \\
\wtbc    &\tm{20} &\tm{14} &24      &25      &6       &18      &0       &0       &0       &0       &0       &4       &0       &0       &0             &0         &111   \\
\wtsc    &8       &7       &\tm{28} &32      &10      &\tm{24} &0       &0       &0       &0       &0       &6       &0       &0       &0             &0         &115   \\
\wepr    &4       &2       &11      &8       &0       &2       &0       &0       &0       &0       &0       &0       &0       &0       &0             &0         &27    \\
\midrule
\wptwo   &20      &14      &28      &32      &10      &24      &22      &37      &5       &2       &10      &50      &0       &5       &23            &18        &300   \\
\wpthree &20      &14      &28      &32      &10      &24      &23      &37      &5       &2       &16      &50      &0       &5       &23            &18        &307   \\
\wabs    &20      &9       &24      &32      &10      &16      &27      &48      &11      &11      &7       &50      &0       &5       &23            &18        &311   \\
\bottomrule
\end{tabular}
\end{table}

\mypara{\add.}  Neither \wbdd nor the template-based method \wtbc require more 
than $0.2$ seconds for any instance of the \add\ benchmark.  The \acs{SAT} 
solver based learning approaches using universal expansion (\wse, \wsge) solve 
all instances as well, but require up to $42$ seconds.  Without expansion (\ws, 
\wsg, \wsgc), \textsc{SatWin1} requires many iterations to refine $U$ before a 
counterexample is found or to conclude that no counterexample exists (see 
Algorithm~\ref{alg:SatWin1}).  For instance, for \add6\texttt{y}, roughly $4000$ 
counterexample candidates are computed.  This takes only one second.  For 
\add8\texttt{y}, \textsc{SatWin1} already computes $65\,000$ counterexample 
candidates, which takes $90$ seconds.  For \add10\texttt{y} we hit the 
timeout.  In contrast, the \acs{QBF}-based learning methods (with names starting 
with \textsf{Q}) require only two iterations, but cannot solve significantly 
more instances either.  This illustrates that the number of iterations alone is 
often not a good measure for estimating the performance of different algorithms 
relative to each other.

\mypara{\mult.}  The results for this benchmark are similar to \add.  The main 
difference is that the \acs{BDD}-based implementation does not perform well, but
this is not surprising since multipliers are known to be challenging for 
\acsp{BDD} (see Section~\ref{sec:prelim:bdds}).  Even \wabs, which is highly
optimized but also \acs{BDD}-based, cannot solve all \mult~instances.  The 
template-based configuration \wtbc performs best.

\mypara{\cnt.}  When the winning region is computed iteratively for this 
benchmark, this requires many iterations.  More specifically, around $2^{k-1}$ 
refinements of the winning region are required for \cnt$ko$.  For $k=30$, this 
already gives around half a billion iterations.  Even though the time per 
iteration is very low for all configurations, this still results in timeouts for 
large values of $k$.    In contrast, the template-based realizations require 
only one iteration.  In particular, the configuration \wtsc solves all \cnt\ 
instances in less than $8$ seconds.

\mypara{\mv.}  Even though this benchmark has a relatively high number of inputs 
and control signals, most methods can solve all its instances within a fraction 
of a second.  This benchmark will only be challenging for some of our circuit 
computation methods in Section~\ref{hw:res:extr}.

\mypara{\bs.}  This benchmark contains a barrel shifter and is thus challenging
for \wbdd.  Most of the other methods solve all instances within a fraction of
a second.

\mypara{\stay.}  This benchmark contains a counter and a multiplier, and thus
combines the characteristics of \mult\ and \cnt.  Hence, it is not 
surprising that one of the template-based configurations performs best.

\mypara{\amba\ and \genbuf.}  While the previous benchmarks are basically toy 
examples designed to challenge the synthesis methods in different ways, the 
\amba\ and \genbuf\ benchmarks specify realistic 
hardware designs.  \wbdd performs very well on both these benchmarks.  One 
circumstance contributing to this success may be that these benchmarks have been 
translated from input files for the \acs{BDD}-based synthesis tool 
\ratsy~\cite{BloemCGHKRSS10}, where they have been tweaked for efficient 
synthesizability.  Yet, the \acs{SAT}-based learning method \wsge solves the 
same amount of \amba\ instances as \wbdd, an is even slightly faster on the 
solved instances.  For \genbuf, \wbdd is unrivaled in our experiments.

\mypara{\fa\ and \mo.}  None of our \acs{SAT}-based methods can compete with 
\acsp{BDD} on these benchmarks.

\mypara{\driver.}  The \wifm method by Morgenstern et al.~\cite{MorgensternGS13} 
solves all instances of the \driver\ benchmark in a fraction of a second.  This 
is remarkable because with up to $326$ state variables, these benchmarks are 
quite large. The \acs{SAT} solver based learning methods \wsg and \wsge are 
ranked second when run with optimization \textsc{RG}.  Without optimization 
\textsc{RG}, only few instances can be solved.

\mypara{\demo.}  Both \wifm and \wsge can solve all instances in at most $40$ 
seconds.  With up to $280$ state variables, the \demo\ benchmarks contain quite
large instances as well.  The number of inputs is always relatively low, though.

\mypara{\gb.}  These benchmarks are far beyond reach with any of our methods.
Even \wabs fails.

\mypara{\load, \ltltodba\ and \ltltodpa.}  \wsge performs best, solving most of 
these instances in less than a second.  With $138$ seconds, the longest 
execution time with \wsge is also quite low.

\mypara{Conclusions.}  Our \acs{QBF}-based learning algorithms are dominated by 
our \acs{SAT} solver based realizations across all benchmarks classes.  \wepr 
is even dominated by all other configurations.  On the other hand, no single 
methods dominates all the other methods on all benchmark classes.  We thus 
conclude that it is important to have different synthesis approaches available.  
Our experiments suggest that our novel \acs{SAT}-based synthesis methods form an 
important contribution to the portfolio of available methods, complementing 
existing \acs{BDD}-based methods (like \wbdd and \wabs) but also existing 
\acs{SAT}-based methods (like \wifm).

\subsubsection{Further Observations} \label{sec:hw:exp:fu}

\noindent
This section highlights interesting observations that are more specific to 
certain methods.

\begin{figure}
\centering
\begin{minipage}{.49\textwidth}
  \includegraphics[width=\textwidth]{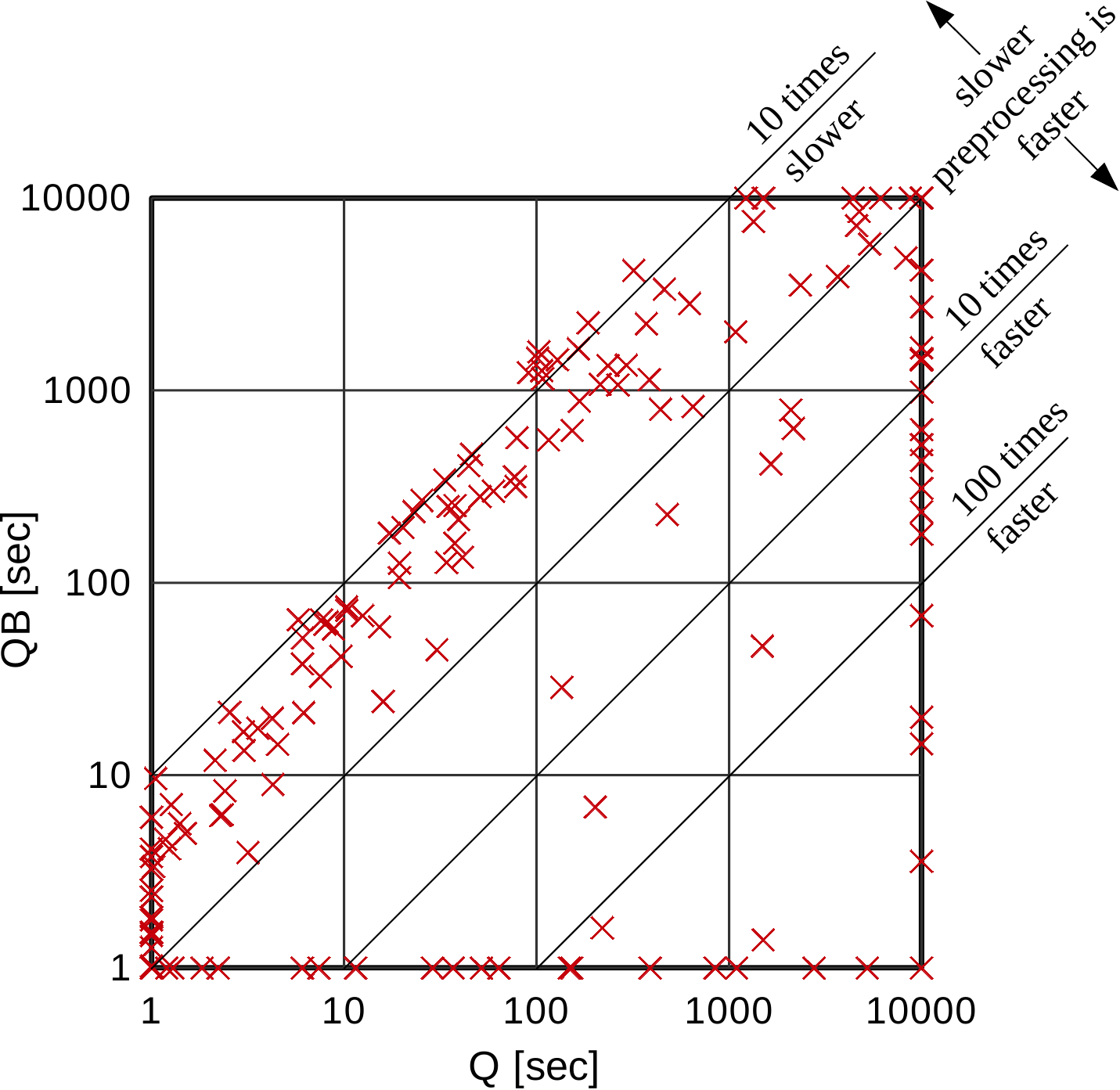}
\caption[The speedup due to \acs{QBF} preprocessing in \acs{QBF}-based learning]
{A scatter plot illustrating the speedup due to \acs{QBF} preprocessing in 
\acs{QBF}-based learning (\wq versus \wqb).}
\label{fig:win_pre_scatter1}
\end{minipage}
\hspace*{1mm}
\begin{minipage}{.49\textwidth}
  \includegraphics[width=\textwidth]{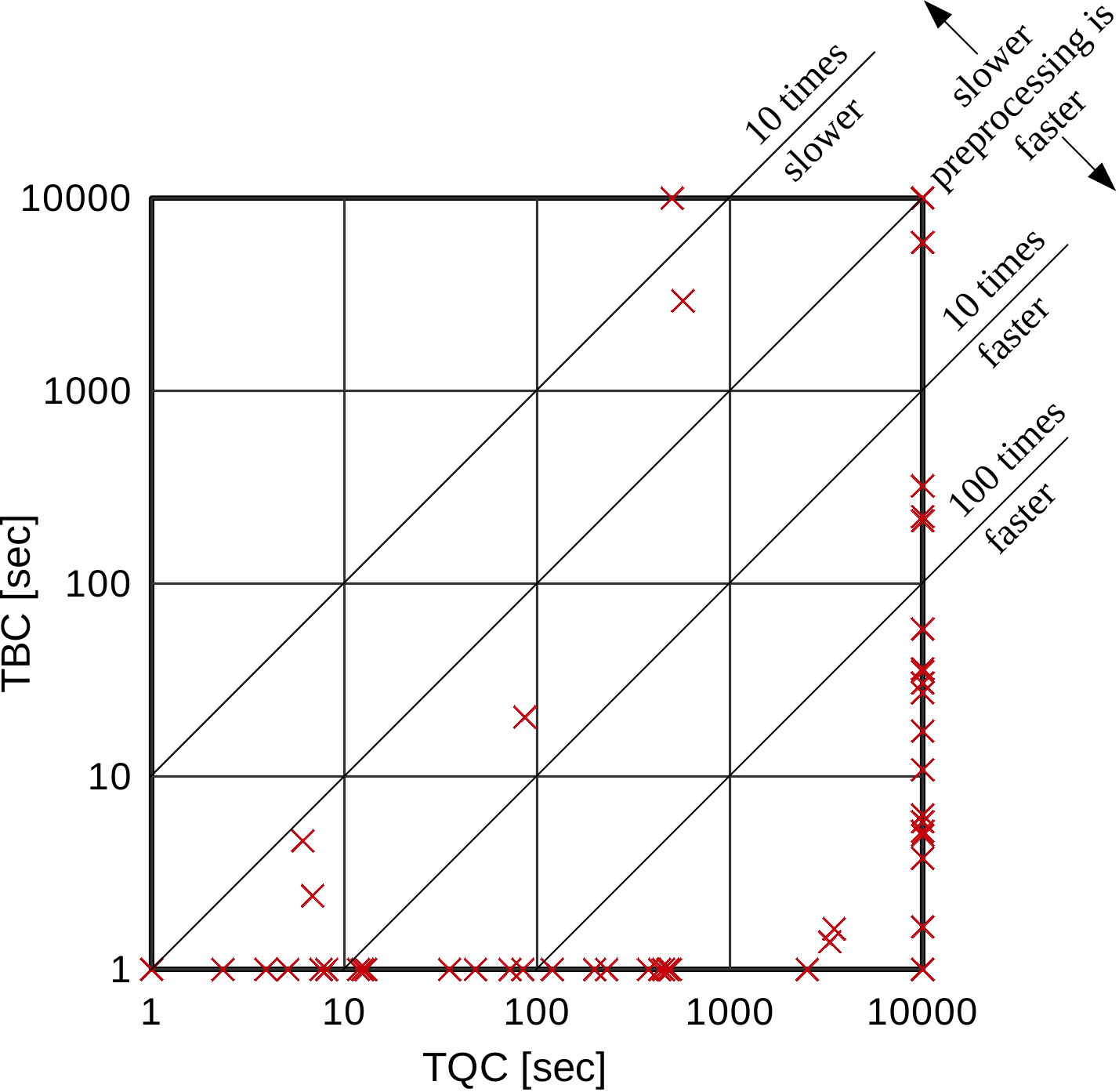}
\caption[The speedup due to \acs{QBF} preprocessing in our template-based 
approach]
{A scatter plot illustrating the speedup due to \acs{QBF} preprocessing in our 
template-based approach (\wtqc versus \wtbc).}
\label{fig:win_pre_scatter2}
\end{minipage}
\end{figure}

\mypara{\acs{QBF} preprocessing is important.} 
Figure~\ref{fig:win_pre_scatter1} compares the execution times with and without 
\acs{QBF} preprocessing for the \acs{QBF}-based learning approach in a 
\index{scatter plot} scatter plot.  Each point in the diagram corresponds to one 
benchmark instance.  The horizontal axis gives the execution time for the 
benchmark without preprocessing, and the vertical axis the corresponding 
execution time with preprocessing.  Hence, all points below the diagonal 
represent a speedup due to preprocessing, and all points above are instances 
with a slowdown.  Note that both axes are scaled logarithmically.  We can see a 
slowdown by up to around one order of magnitude for many instances.  However, 
there are also $20$ points on the x-axis, indicating instances that can be 
solved in less than one second due to preprocessing.  Furthermore, there are 
$19$ points on the right border of the diagram, indicating cases where we had a 
timeout without preprocessing but get a solution when preprocessing is enabled.  
Two instances are even located in the lower right corner, representing an 
improvement from a timeout to less than one second.  The number of solved 
instances increases from $179$ to $193$ due to preprocessing in the 
\acs{QBF}-based learning method (see Table~\ref{tab:winsolved}).  The results 
for the template-based method (\wtqc versus \wtbc) are illustrated in 
Figure~\ref{fig:win_pre_scatter2} and are even more impressive. A noticeable 
slowdown can only be observed for two cases.  There are $44$ points on the 
x-axis, for which preprocessing reduced the execution time to less than one 
second.  For $40$ cases, a timeout is avoided due to preprocessing.  Finally, 
there are $23$ points in the bottom right corner of 
Figure~\ref{fig:win_pre_scatter2}, for which a timeout is turned into a 
successful execution that takes less than one second.

\mypara{Our optimizations for quantifier expansion can avoid a formula size 
explosion in many cases.}  For most of our \acs{SAT}-based methods, the memory 
consumption is rather insignificant.  As an exception, \wsge can consume quite 
some memory due the expansion of universal quantifiers (see 
Section~\ref{sec:hw_exp}).  However, our implementation can also fall back to 
\wsg if some memory limit is exceeded.  In our experiments, this happened only 
for large instances of \mult, \stay, \gb, and \driver.  One reason is our 
careful implementation of the expansion, which aggressively applies 
simplifications to reduce the formula blow-up. As an example, for 
\genbuf15\texttt{b}, a straightforward implementation would produce $652 \cdot 
2^{23} \approx 5 \cdot 10^9$ AND gates to define the expanded transition 
relation.  With $3\cdot 4=12$ byte per AND gate, this gives $56$ GB.  
With our simplification techniques, the expanded transition relation has 
``only'' $1.5$ million AND gates.  In addition, the final 
over-approximation $F$ of the winning region $W$ for \genbuf15\texttt{b} 
contains $707$ clauses and $8517$ literals (after simplification).  After 
negation, this makes $8517$ clauses with $17739$ literals.  With $4$ byte per 
literal, combining $2^{23}$ copies of this \acs{CNF} in a straightforward way 
would require more than $500$ GB of memory.  Due to our simplifications, the 
expanded \acs{CNF} for $\neg F$ has only $8.5$ million literals and \wsge solves 
this benchmark instance without falling back to \wsg.  The maximum memory 
consumption is only $680$ MB.

\mypara{Our parallelization is more than a portfolio approach.}  When executed 
with two threads, our parallelization combines the template-based approach (in a 
mix of \wtbc and \wtsc) with the learning-based approach \wsge.  The two 
approaches do not only run in isolation, but share information: clauses 
discovered by the \wsge-thread are communicated to the template-based thread 
and 
are considered as fixed part of the winning region there (see 
Section~\ref{sec:hw_par}).  This exchange of information can have a positive 
effect. For example, for the \genbuf\ benchmark, the template-based approach 
fails to solve even the simplest instances when applied in isolation.  However, 
in our parallelization, the final winning region for certain instances is 
actually found by the template-based thread.  This includes even very large 
instances such as \genbuf15\texttt{b}.  The speedup of our parallelization 
\wptwo in comparison to \wsge for such instances is rather moderate (e.g. 
$\approx 10\,\%$ for \genbuf15\texttt{b}), but this still illustrates that 
complementary methods can benefit from each other in our parallelization.

\subsection{Circuit Synthesis Results} \label{hw:res:extr}

We now compare our different methods (from Chapter~\ref{sec:hw:circ}) to 
construct a circuit from a given strategy.  Again, we first describe the 
evaluated configurations and the experimental setup.  
Section~\ref{sec:hw:ex:big2} then discusses the results on the average over all 
our benchmarks.  Section~\ref{sec:hw:exextr_class} dives into more details by 
investigating the performance for different benchmark classes.  Other
interesting observations will be highlighted in Section~\ref{sec:hw:exextr_obs}.

\subsubsection{Evaluated Configurations}

Table~\ref{tab:extrconfig} lists the different methods and their configurations
compared in this section. All our \acs{SAT}-based methods use 
the tool \abc~\cite{BraytonM10} 
in a postprocessing step to further reduce the circuit size.\footnote{If the
\aiger circuit has less than $2\cdot 10^5$ AND gates before optimization, then 
we execute the command sequence \texttt{strash; refactor -zl; rewrite -zl;} 
three times, followed by \texttt{dfraig; rewrite -zl; dfraig;}.  Between $2\cdot 
10^5$ and $10^6$ AND gates, we only execute the sequence \texttt{strash; 
refactor -zl; rewrite -zl;} twice. For more than $10^6$ AND gates, we perform 
it only once.}

\mypara{Baseline.} \ebdd denotes a \acs{BDD}-based implementation of the 
standard cofactor-based approach presented in Algorithm~\ref{alg:CofSynt}.  It 
is implemented in the tool that has already been discussed in 
Section~\ref{hw:res:win:conf}, which won a synthesis competition that has been 
carried out in the course of a lecture.  Besides dynamic variable reordering 
(with method SIFT~\cite{Rudell93,cudd}), it also performs a forced reordering 
with a more expensive heuristic (SIFT\_CONV~\cite{Rudell93,cudd}) before 
circuits are extracted from the strategy.  Furthermore, it uses a cache that 
maps \acs{BDD} nodes to corresponding signals in the circuit constructed so far. 
Whenever new circuitry is added, the cache is consulted to reuse existing 
signals.  Consequently, no two signals in the constructed circuit will be 
equivalent.  The configuration \eabs denotes the circuit synthesis step as 
implemented in \abssynthe version 2.0~\cite{BrenguierPRS14}.  The basic 
algorithm is the same as that of \ebdd, but additional optimizations are 
applied.  The \wifm method by Morgenstern et al.~\cite{MorgensternGS13}, which 
has been used as a baseline in Section~\ref{sec:hw:ex:strati}, is not included 
here because it can only compute a winning strategy but not a circuit 
implementing it.

\mypara{\acs{QBF}-based methods.} Our approach using \acs{QBF} certification 
(see Section~\ref{sec:hw:circ_qbfcert}) is named \ecert.  A variant where we 
compute the negation of the winning region using the procedure 
\textsc{NegLearn} (Algorithm~\ref{alg:NegLearn}) is denoted by \ecertn.  Our 
\acs{QBF}-based learning approach from Algorithm~\ref{alg:SafeQbfSynt} is used 
in three configurations: \eq denotes a plain implementation using \depqbf, \eqb 
also uses \acs{QBF} preprocessing by \bloqqer, and \eqi uses the \depqbf solver 
in an incremental fashion (see Section~\ref{sec:hw:extr:qbf_impl}).

\mypara{Interpolation-based method.} An implementation of the 
interpolation-based method from Algorithm~\ref{alg:SafeInterpolSynt} is denoted 
by \ei.  It applies the dependency optimization presented in 
Section~\ref{sec:hw_dep} and uses \mathsat version \texttt{5.2.12} as 
interpolation engine.  \mathsat supports several interpolation methods.  In our 
experiments we use McMillan's system~\cite{McMillan03}.  Results with other 
interpolation methods are rather similar, though.  We also implemented our own 
interpolation engine by processing \picosat proofs in the 
\href{http://fmv.jku.at/tracecheck/}{\textsf{TraceCheck}} 
format.  However, for larger 
benchmark instances, the proof files grew prohibitively large with this 
approach.  Our realization using \mathsat does not have this problem.

\mypara{Learning based on \acs{SAT} solvers.} Configuration \es implements 
the \acs{SAT} solver based learning approach from 
Section~\ref{sec:hw:extr:satlearn} without the dependency optimization 
(Section~\ref{sec:hw_dep}) and without minimizing the final solution 
(Section~\ref{sec:hw:extrsat_impl}).  \esd denotes a similar configuration, but 
with the dependency optimization enabled.  Finally, the \esdm configuration also 
applies a minimization of the final solution by attempting to eliminate literals 
and clauses from the computed solutions (Section~\ref{sec:hw:extrsat_impl}). 
 All these configurations use activation variables to perform incremental 
solving across all calls to \textsc{CnfInterpol} (see 
Section~\ref{sec:hw:extrsat_impl}).  \lingeling is slightly faster on average 
than \minisat and \picosat in all these configuration. Results for other 
configurations ($28$ in total) can be found in the downloadable archive.

\begin{table}
\centering
\caption[Configurations for computing a circuit that implements a given 
strategy]
{Configurations for computing a circuit that implements a given 
strategy.}
\label{tab:extrconfig}
\begin{tabular}{lll}
\toprule
Name & Algorithm and Optimizations & Solver \\
\midrule
\ebdd  
& \textsc{CofSynt} (Alg.~\ref{alg:CofSynt})
& \textsf{CuDD} \\
\wabs
& \abssynthe 2.0~\cite{BrenguierPRS14}
& \textsf{CuDD}\\
\ecert
& QBF Certification (Sect.~\ref{sec:hw:circ_qbfcert})
& \qbfcert\\
\ecertn
& QBF Certification (Sect.~\ref{sec:hw:circ_qbfcert})
  + \textsc{NegLearn} (Alg.~\ref{alg:NegLearn})
& \qbfcert\\
\eq
& \textsc{SafeQbfSynt} (Alg.~\ref{alg:SafeQbfSynt})
& \depqbf\\
\eqb
& \textsc{SafeQbfSynt} (Alg.~\ref{alg:SafeQbfSynt})
& \depqbf + \bloqqer\\
\eqi
& \textsc{SafeQbfSynt} (Alg.~\ref{alg:SafeQbfSynt})
& Incremental \depqbf\\
\ei
& \textsc{SafeInterpolSynt} (Alg.~\ref{alg:SafeInterpolSynt}) + 
  Dep.~Opt.~(Sect.~\ref{sec:hw_dep})
& \mathsat\\
\es
& \textsc{SafeInterpolSynt} + 
  \textsc{CnfInterpol} (Alg.~\ref{alg:CnfInterpol})
& \lingeling\\
\esd
& \textsc{SafeInterpolSynt} + 
  \textsc{CnfInterpol} (Alg.~\ref{alg:CnfInterpol})
  + Dep.~Opt.~(Sect.~\ref{sec:hw_dep})
& \lingeling\\
\esdm
& \textsc{SafeInterpolSynt} + 
  \textsc{CnfInterpol} (Alg.~\ref{alg:CnfInterpol})
& \lingeling\\
& + Dep.~Opt. +
  Minimizing the final solution (Sect.~\ref{sec:hw:extrsat_impl}) & \\
\eptwo
& Parallel (Sect.~\ref{sec:hw:par_extr}) with 2 threads & 
\lingeling + \depqbf \\  
\epthree
& Parallel (Sect.~\ref{sec:hw:par_extr}) with 3 threads & 
\lingeling + \depqbf \\  
\bottomrule
\end{tabular}
\end{table}

\subsubsection{Experimental Setup}\label{sec:hw:extrsetup}

Again, all experiments were performed on an Intel Xeon E5430 CPU with 4 cores 
running at $2.66$\,GHz, using a 64 bit Linux as operating system. A timeout was 
set to $10\,000$ seconds for all circuit synthesis runs.  The available main 
memory was limited to $8$ GB.  The maximum size for auxiliary files to be 
written to the hard disk was set to $20$ GB.  

\mypara{Sanity checks.}
All synthesized circuits were model checked using \icthree~\cite{Bradley11}.  
\icthree never found a counterexample but in some cases hit a timeout.  We thus 
also ran a bounded model checker (\textsf{BLIMC}, which is distributed with 
\lingeling~\cite{lingeling}) to get bounded correctness guarantees for such 
cases.

\mypara{Winning strategies.}
For all our \acs{SAT}-based circuit computation methods, we used the winning 
strategies as computed by configuration \wpthree (see 
Table~\ref{tab:winconfig}).  Preliminary experiments with other strategy 
computation methods suggest that the impact on the performance in circuit 
synthesis is rather small.  One reason is that we simplify the computed winning 
region (or winning area) by calling \textsc{CompressCnf} 
(Algorithm~\ref{alg:CompressCnf}) as a preprocessing step to circuit extraction 
(see Chapter~\ref{sec:hw:circ}).  We thus refrain from running experiments with 
all combinations of our strategy- and circuit computation methods.  
Furthermore, we stored the winning strategies computed by \wpthree into files 
and loaded them for our circuit computation experiments in order to ensure that 
all our methods operate on exactly the same strategy (and to save computational 
resources for recomputing the strategy each time).  For \ebdd and \eabs, we used 
the winning regions as computed by these tools as a starting point for circuit 
synthesis.

\mypara{Benchmarks.}
From the $350$ benchmark instances used to evaluate our strategy computation 
methods (see Section~\ref{sec:hw:bench}), only $267$ instances have been used to 
compare our circuit computation methods. This has two reasons.  First, $40$ 
instances are unrealizable, so they have no winning strategy.  Second, for some 
instances, no winning strategy could be computed (even with a timeout of $10^5$ 
seconds) for at least one of the compared methods.  In order to have a fair 
comparison, we thus used only those benchmarks for which \wpthree, \ebdd and 
\eabs succeeded in computing a winning strategy.  In detail, \wpthree failed to 
compute a winning strategy for $25$ instances.  For $19$ additional instances, 
\ebdd could not find a winning strategy within $10^5$ seconds.  This includes 
\cnt$30$\texttt{n} and \cnt$30$\texttt{y}, for which we estimated \ebdd's 
circuit synthesis time (to be $0.1$ seconds) and the circuit size (to be $32$ 
gates) based on observations from smaller instances.  This leaves $17$ excluded 
instances.  \eabs could not compute a winning strategy for $11$ more benchmarks. 
However, for two \cnt~instances, we could estimate the time and size to $0.1$ 
seconds and $1$ gate based on results for smaller instances.  For eight 
instances of the \stay$ko$ benchmark, we also estimated $0.1$ seconds and $k+1$ 
gates.  What remains is one \driver~instance to exclude.  This results 
in $350-40-25-17-1=267$ instances used for the comparison.  The number of used 
instances per benchmark class can be seen from Table~\ref{tab:extrsolved} (see 
the line labeled ``Total'').  

\mypara{More detailed comparisons.}
The downloadable archive also contains more detailed pairwise comparisons on 
larger subsets of the benchmark instances.  This includes charts to compare our 
\acs{SAT}-based methods with \eabs on all $281$ benchmarks for which both \eabs 
and \wpthree were able to compute a winning strategy.  Charts comparing our 
\acs{SAT}-based methods with \ebdd on all $268$ instance on which both \ebdd and 
\wpthree were able could find a winning region are included as well.  However, 
since the results are almost identical to our three-way comparison, we refrain 
from presenting them in this article.

\subsubsection{The Big Picture} \label{sec:hw:ex:big2}

The Figures~\ref{fig:extr_cactus_time} and \ref{fig:extr_cactus_size} contain 
cactus plots illustrating the execution time and the resulting circuit size for 
the method configurations from Table~\ref{tab:extrconfig}.  Configuration 
\ecert is omitted because both the execution time and the 
circuit size is similar to \ecertn (the difference is mostly in the memory 
consumption).  Configuration \eq performs slightly worse than \eqb and \eqi 
and is also omitted to make the plots more legible.  The following 
paragraphs discuss the most important observations based on these two figures.  
A detailed analysis is done in Section~\ref{sec:hw:exextr_class} 
and~\ref{sec:hw:exextr_obs}.

\begin{figure}
\centering
  \hspace*{3mm}%
  \includegraphics[scale=1.45]{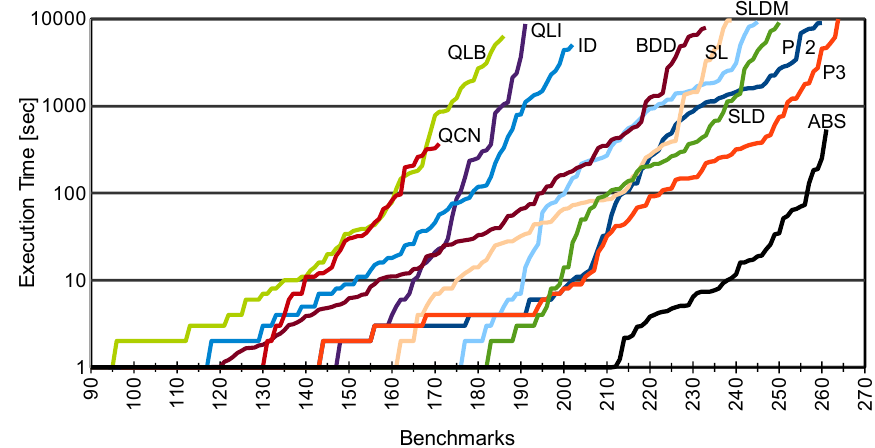}
\caption{A cactus plot summarizing execution times for computing a 
circuit from a strategy.}
\label{fig:extr_cactus_time}
\vspace{5mm}
  \includegraphics[scale=1.45]{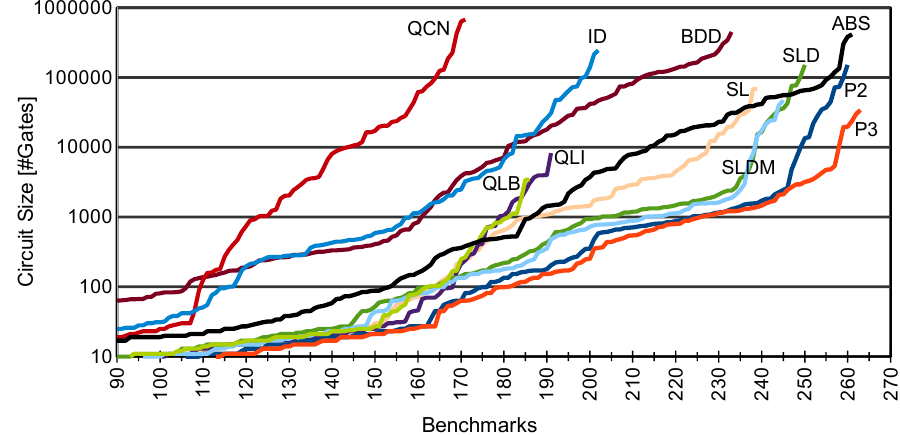}
\caption{A cactus plot summarizing the resulting circuit sizes.}
\label{fig:extr_cactus_size}
\end{figure}

\mypara{\acs{QBF} certification does not perform well.}  From 
Figure~\ref{fig:extr_cactus_size}, we can see that the \acs{QBF} certification 
method \ecertn produces the largest circuits.  Figure~\ref{fig:extr_cactus_time} 
illustrates that \ecertn is on average slightly faster than \eqb, but still 
solves less instances within the given resource limits.  The reason is that 
\ecertn often exceeds the $20$ GB limit for auxiliary files because the proof 
traces produced by \depqbf in the \qbfcert framework can grow very large 
(several hundreds of GB when run without limits).

\mypara{\acs{QBF}-based learning is slow but produces small circuits.}  
Especially when used with incremental \acs{QBF} solving, \acs{QBF}-based 
learning can outperform \acs{QBF} certification both regarding execution time 
and circuit size (compare \eqi versus \ecertn).  Still, in comparison with the 
other methods, \eqi is on average way slower.  Regarding circuit size, \eqi and 
\eqb are on a par and outperform the interpolation-based method \ei as well as 
the \acs{BDD}-based implementation \ebdd by almost one order of magnitude 
on~average.

\mypara{The interpolation-based approach does not outperform \acsp{BDD} in our 
experiments.}  Regarding circuit size, the interpolation-based configuration 
\ei yields similar results as \ebdd on average.  However, \ei is noticeably 
slower than \ebdd, especially for more complex benchmark instances.

\mypara{Our \acs{SAT} solver based learning approach outperforms all our other 
non-parallel methods.}  This holds true for both the execution time and the 
circuit size.  Configuration \esd turns out to be our best non-parallel 
option on the average over all our benchmarks.  It is already faster than \ebdd 
on average by more than one order of magnitude.  For example, \ebdd can solve 
its $233$ simplest benchmark instances in at most $7931$ seconds each.  \esd 
never needs more than $461$ seconds for its $233$ simplest instances.  This is a 
factor of $17$.  With respect to circuit size, the situation is even more 
extreme.  The $233$ smallest circuits produced by \ebdd have at most half a 
million AND gates each.  The $233$ smallest circuits produced by \esd have at 
most $2383$ gates each.  This is smaller by a factor of $188$.  \esdm produces 
even slightly smaller circuits, with an improvement factor of $240$ compared to
\ebdd.  This is not only because \ebdd 
produces large circuits for benchmarks that cannot be solved by \esd.  The most 
extreme instance is \driver\texttt{d8}\footnote{This instance is not included in 
Figure~\ref{fig:extr_cactus_size} because \eabs could not compute a winning 
strategy for this instance.}, for which \ebdd produces a circuit with $3.7$ 
million AND gates, while \esd produces a circuit with only $66$ gates.

\mypara{Our parallelization is competitive with the state of the art.} Our 
parallelization \epthree increases the number of solved instances compared to 
\esd from $250$ to $263$ by combining different methods and optimizations. The 
circuit sizes also decrease slightly, which is partly due to our heuristics 
performing additional circuit minimizations if there is sufficient time left 
(see Section~\ref{sec:hw:par_extr}), and due to selecting the smallest solution 
from all threads.  Our parallelization already solves $2$ instances more than 
the state-of-the-art tool \abssynthe.  When comparing 
\epthree against \abssynthe in Figure~\ref{fig:extr_cactus_time}, we can see 
that \abssynthe solves many instances in less than one second but the execution 
times grow steeper for more difficult instances.  Thus, \abssynthe can be 
superior if the timeout is short, while \epthree shows a more steady pace.  
Regarding circuit size, our parallelization outperforms \abssynthe by more than 
one order of magnitude on average (compare \epthree versus \eabs in 
Figure~\ref{fig:extr_cactus_size}).

\mypara{Execution time and circuit size often correlate.}  With the exception 
of 
the \acs{QBF}-based learning methods, we can observe a correlation between 
execution time and circuit size in our experiments.  Methods that are fast have 
a tendency to also producing small circuits and vice versa.  At the first 
glance, this may be surprising because, intuitively, one could expect that we 
have to find a good trade-off between these performance metrics.  One reason 
for 
the correlation is that most of our methods (all except \ecert and \ecertn) 
compute circuits iteratively for one control signal after the other.  After 
every iteration, the strategy formula is refined with the solutions for the 
control signals that have been synthesized so far.  If these solutions are 
complicated, then this results in more complicated strategy formulas for the 
next iterations, which can increase the computation times.  For the 
learning-based methods, the size of the \acs{CNF} formulas defining the control 
signals directly corresponds to the number of iterations that were needed to 
compute them: Every clause results from a mayor iteration involving a 
counterexample computation.  Every literal in a clause witnesses a failed 
attempt to eliminate this literal with a \acs{SAT}- or \acs{QBF} solver call. A 
correlation between the circuit size and the execution time is thus natural.

\mypara{Computing circuits from strategies is by no means a negligible step in 
the synthesis process.}  Let us compare the total strategy computation time 
against the total circuit computation time for all instances where both steps 
terminate within the timeout of $10\,000$ seconds.  For \wpthree, this 
comparison reveals that $52$ percent of the total synthesis time is spent on 
strategy computation and $48$ percent is consumed by circuit computation.  For 
\wbdd, the distribution is $60\,\%$ to $40\,\%$.  Only for \wabs, the 
distribution is $90\,\%$ to $10\,\%$, which may be due to the 
abstraction/refinement techniques implemented in \wabs.

\subsubsection{Performance per Benchmark Class} \label{sec:hw:exextr_class}

This section analyzes the performance of our methods for circuit synthesis for
the different benchmark classes.  We will see that the configuration \esd is
not always superior.

Table~\ref{tab:extrsolved} lists the number of benchmark instances that could be 
solved per benchmark class by the different configurations.  The first line gives 
the total number of benchmark instances in the respective class.  The last 
column gives the total number of instances for which a circuit could be computed 
by the respective method within the given resource limits ($10\,000$ seconds, 
$8$ GB of main memory, $20$ GB for auxiliary files).
The fastest 
configuration is marked in blue.  If the same amount of instances are solved by 
several configurations, we marked the one with the lowest total execution time. 
In case the total execution time is very similar, we sometimes marked several 
configurations.  For benchmark classes where most of the configurations solve 
all instances, we did not mark any configuration.  Again, we do not include 
\wabs and the parallelizations in this ranking because they combine several 
techniques.

\begin{table}
\setlength{\tabcolsep}{1.32mm}
\centering
\caption{Computing a circuit from a strategy: solved instances per benchmark 
class.}
\label{tab:extrsolved}
\begin{tabular}{lccccccccccccccccc}
\toprule
        &\add    &\mult   &\cnt    &\mv     &\bs     &\stay   &\amba   &\genbuf &\fa     &\mo     &\driver &\demo   &\load   &\ltltodba &\ltltodpa &Total \\
\midrule
Total   &20      &14      &28      &32      &10      &24      &23      &42      &3       &2       &0       &37      &3       &19        &10        &267   \\
\midrule
\ebdd   &\tm{20} &7       &28      &32      &4       &16      &13      &41      &\tm{3}  &\tm{2}  &-       &36      &3       &18        &10        &233   \\
\ecert  &6       &4       &28      &32      &10      &24      &0       &6       &2       &0       &-       &30      &2       &7         &10        &161   \\
\ecertn &6       &4       &28      &32      &10      &24      &3       &10      &3       &1       &-       &31      &2       &9         &8         &171   \\
\eq     &8       &4       &28      &32      &10      &24      &3       &10      &2       &0       &-       &35      &3       &19        &10        &188   \\
\eqb    &8       &3       &28      &32      &10      &24      &3       &9       &2       &0       &-       &35      &3       &19        &10        &186   \\
\eqi    &8       &4       &28      &32      &10      &24      &3       &11      &3       &1       &-       &35      &3       &19        &10        &191   \\
\ei     &20      &5       &28      &28      &10      &24      &4       &12      &2       &1       &-       &36      &3       &\tm{19}   &\tm{10}   &202   \\
\es     &12      &5       &28      &25      &10      &24      &\tm{20} &\tm{41} &\tm{3}  &2       &-       &\tm{37} &\tm{3}  &\tm{19}   &\tm{10}   &239   \\
\esd    &\tm{20} &\tm{14} &28      &22      &10      &24      &17      &41      &3       &2       &-       &37      &\tm{3}  &\tm{19}   &\tm{10}   &250   \\
\esdm   &\tm{20} &\tm{14} &28      &22      &10      &24      &14      &40      &3       &1       &-       &37      &\tm{3}  &\tm{19}   &\tm{10}   &245   \\
\midrule
\eptwo  &20      &14      &28      &32      &10      &24      &17      &41      &3       &2       &-       &37      &3       &19        &10        &260   \\
\epthree&20      &14      &28      &32      &10      &24      &20      &41      &3       &2       &-       &37      &3       &19        &10        &263   \\
\eabs   &20      &8       &28      &32      &10      &24      &23      &42      &3       &2       &-       &37      &3       &19        &10        &261   \\
\bottomrule
\end{tabular}
\end{table}

\mypara{\add.} The configurations \ebdd, \esd and \esdm solve all instances of 
the \add\ benchmark within one second.  The difference in circuit size is 
moderate (at most $171$ gates with \esd and \esdm; at most $416$ gates with 
\ebdd). The interpolation-based method \ei solves all instances as well but 
requires at most $40$ seconds.  The good results of \esd, \esdm and \ei are 
mostly due to our dependency optimization (see Section~\ref{sec:hw_dep} and 
Section~\ref{sec:hw:extrsat_impl}):  Without the dependency optimization, the 
\acs{SAT} solver based learning method solves only $12$ instances (\esd versus 
\es).  

\mypara{\mult.} This benchmark is similar to \add\ in spirit, but the circuit to 
be synthesized is more complex.  \esd and \esdm still perform well, again due to 
the dependency optimization.  However, \ebdd and \ei fall back noticeably.  The 
difference in circuit size also grows more significant: For example, \ebdd 
implements \mult9 with more than $10^5$ gates, while \esd and \esdm require only 
$633$ gates.  One reason is that multipliers cannot be represented by small 
(monolithic) \acsp{BDD} with any variable ordering (see 
Section~\ref{sec:prelim:bdds}).  Since the \ebdd method dumps \acsp{BDD} as a 
network of multiplexers to obtain the resulting circuit, the \acs{BDD} size does 
not only affect the computation time but also the resulting circuit size.  Our 
\acs{SAT} solving based methods \esd and \esdm do not suffer from this issue.  
They even outperform \abssynthe significantly on this benchmark. 

\mypara{\cnt, \bs~ and \stay.}  These benchmarks can be solved by all our 
methods in a few seconds.  Only \ei requires up to $46$ seconds on larger 
instances of \stay. \ebdd performs well on \cnt, but cannot solve all instances 
of \bs\ and \stay.  The latter two benchmarks contain barrel shifters and 
multipliers, which are known to be challenging for \acsp{BDD}.

\mypara{\mv.} The \mv\ benchmark is an interesting case.  Most of our methods 
can solve all instances of this benchmark in less than one second.  However, for 
the interpolation-based method \ei as well as the \acs{SAT} solver based 
learning methods \es, \esd and \esdm, this benchmark is challenging.  All these 
methods 
are based on \textsc{InterpolSynt} (Algorithm~\ref{alg:InterpolSynt}).  The crux 
with the \mv\ benchmark is that the XOR sum of all control signals must be 
$\true$. \textsc{InterpolSynt} starts by building a circuit to fix the value of 
the last control signal based on all other control signals such that this is 
ensured.  Since this circuit needs to react properly to all possible values of 
all other control signals, it can be very large.  In particular, the 
\acs{SAT} solver based learning methods build this circuit in a \acs{CNF} 
representation without introducing auxiliary variables.  A \acs{CNF} 
formula that computes the XOR sum of $n$ variables without introducing new 
auxiliary variables requires $2^{n-1}$ clauses.  For \mv28\texttt{y}, this gives 
$2^{27} \approx 134\cdot 10^6$ clauses.\footnote{For the resubstitution step in 
Line~\ref{alg:InterpolSynt:resub} of Algorithm~\ref{alg:InterpolSynt}, this 
\acs{CNF} also needs to be negated, which can even result in running out of 
memory.}  Only in the last iteration, when the algorithm processes the first 
control signal, \textsc{InterpolSynt} discovers that this signal can actually be 
set to a constant value.  This has the effect that all the computed circuits for 
the other control signals also collapse to constant values.  The root cause 
for this behavior is that \textsc{InterpolSynt} is very conservative with 
exploiting implementation freedom (see Section~\ref{sec:intdisc} for a 
discussion).  In contrast, the \acs{QBF}-based learning algorithm 
\textsc{QbfSynt} (Algorithm~\ref{alg:QbfSynt}) exploits the available freedom 
greedily.  It sets each control signal to a constant value right away, because 
this is sufficient to ensure that a solution for the remaining control signals 
still exists.

\mypara{\amba\ and \genbuf.}  For these benchmarks, the \acs{SAT} solver based 
learning configuration \es performs best.  That is, the dependency optimization 
implemented in \esd and \esdm does not pay off.  \esd, \es and \ebdd can solve 
the same amount of \genbuf\ instances, but \esd is slower than \es by a factor 
of $2$ in total, and \ebdd is even slightly slower than \esd in total.  The sum 
of the circuit sizes for all \genbuf\ instances is $44$ times smaller when using 
\es instead of \ebdd.  For \amba, the factor is $21$ when counting only the 
instances that can be solved by both \es and \ebdd.

\mypara{\fa\ and \mo.}  Both \ebdd and \es can solve all \fa\ instances in less 
than $10$ seconds per instance.  The \mo\ instances are solved by \ebdd in at 
most $160$ seconds per instance.  The second fastest configuration for \mo\ is 
\es, but it requires already $4500$ seconds.

\mypara{\driver.} \wabs cannot compute a winning strategy for any of the 
\driver\ instances, so this benchmark is not included in the comparison of 
Table~\ref{tab:extrsolved}.  Our \acs{SAT} solver based learning methods \es, 
\esd and \esdm can solve all \driver\ instances in less than $10$ seconds.  The 
circuit size with these methods is at most $600$ gates.  \ebdd can only handle 
the smallest \driver\ instance, but takes already half an hour to produce a 
circuit with $3.7$ million gates.  With up to $326$ state variables and $98$ 
inputs, the \driver\ benchmark certainly offers plenty of possibilities for 
building complicated circuitry.  Yet, in contrast to \ebdd, our learning-based 
methods appear to perform well in exploiting the implementation freedom to avoid 
overly complicated solutions.

\mypara{\demo.}  Only \eabs and our \acs{SAT} solver based learning methods 
\es, \esd and \esdm can solve all instances.  The dependency optimization is 
not beneficial: \esd is slower than \es by a factor of $3.2$.

\mypara{\load.}  Again, the \acs{SAT} solver based learning methods \es, \esd 
and \esdm perform best: they solve all instances in at most $2$ seconds.  \ei 
requires up to $18$ seconds.  The fastest \acs{QBF}-based learning method is 
\eqi, requiring up to $87$ seconds for the \load\ instances. \ebdd 
requires up to $10$ minutes.

\mypara{\ltltodba\ and \ltltodpa.}  The configurations \ei, \es, \esd, and 
\esdm require at most $4$ seconds on these benchmarks.  Other configurations 
that can also solve all these benchmarks are slightly slower.

\subsubsection{Further Observations} \label{sec:hw:exextr_obs}

\mypara{The effect of our postprocessing with \abc~\cite{BraytonM10} is rather 
insignificant.}  For \esd, \abc manages to reduce the average circuit size from 
$9500$ gates to around $2700$ gates in our experiments.  However, this average 
is strongly influenced by the \mv\ benchmark, where circuits with up to half a 
million gates are reduced to circuits were all control signals are driven by 
constants. See Section~\ref{sec:hw:exextr_class} for an explanation why this 
happens.  This reduction for the \mv\ benchmark could also be achieved with a 
simple constant propagation.  When omitting the \mv\ benchmark, the average 
circuit size is reduced from $4100$ gates to $2900$ gates, which is a reduction 
by around $30$ percent.  In relation to the circuit size differences between our 
methods, which can be in the range of several orders of magnitude (see 
Figure~\ref{fig:extr_cactus_size}), this is rather insignificant.  On the other 
hand, in the case of \esd, only $0.6$ percent of the total execution time for 
all benchmarks is spent by \abc. By modifying the sequence of minimization 
commands executed by \abc, other trade-offs between the execution time and the 
resulting circuit size improvements are possible.  Yet, our experiments suggest 
that postprocessing cannot easily compensate the large circuit size differences 
between the methods. In other words, exploiting the implementation freedom 
cleverly while computing the circuits appears to be much more effective than 
investing more effort into postprocessing.

\begin{figure}
\centering
\begin{minipage}{.49\textwidth}
  \includegraphics[width=\textwidth]{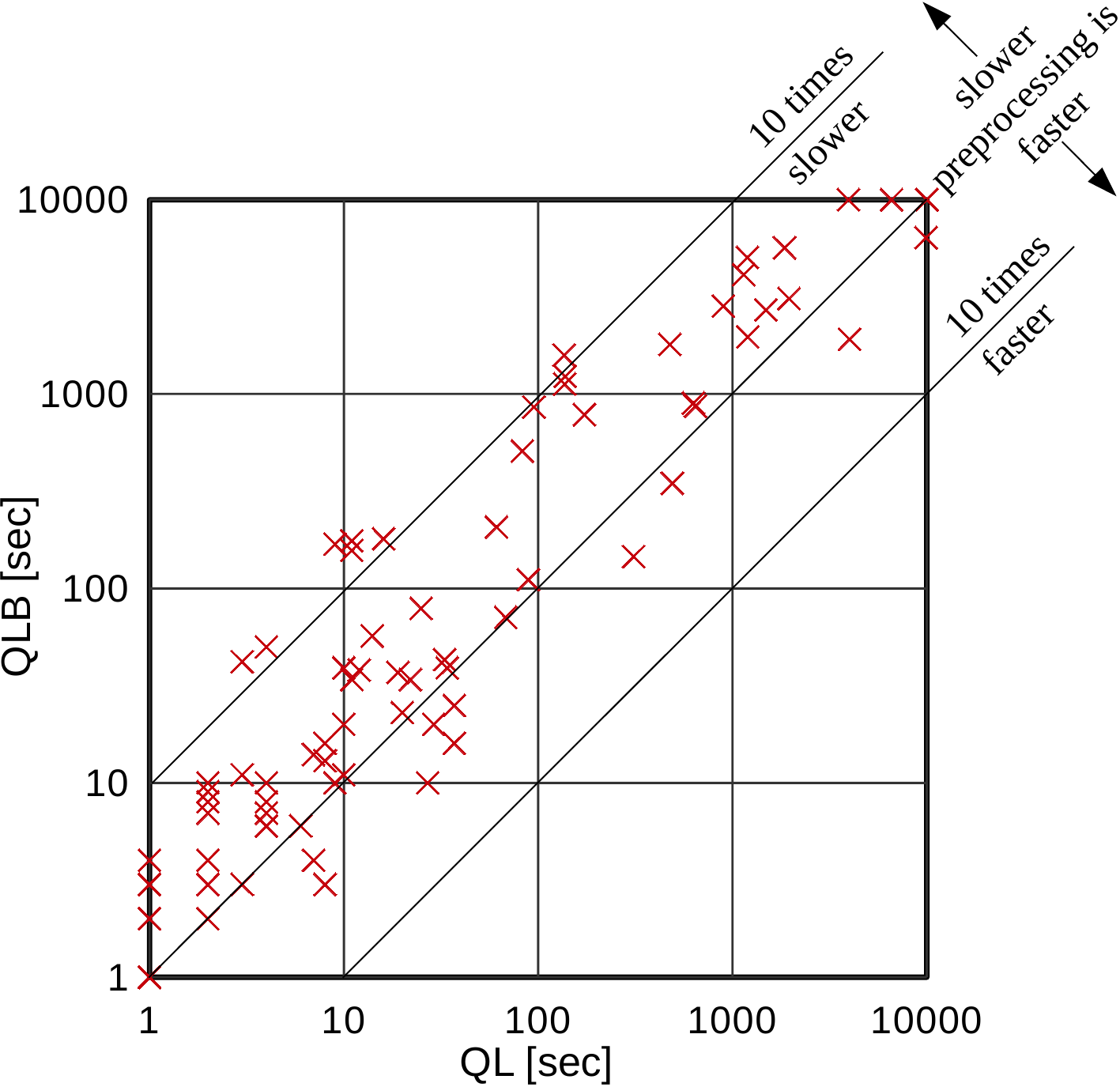}
\caption{The effect of \acs{QBF} preprocessing in circuit synthesis.}
\label{fig:extr_pre_scatter1}
\end{minipage}
\hspace*{1mm}
\begin{minipage}{.49\textwidth}
  \includegraphics[width=\textwidth]{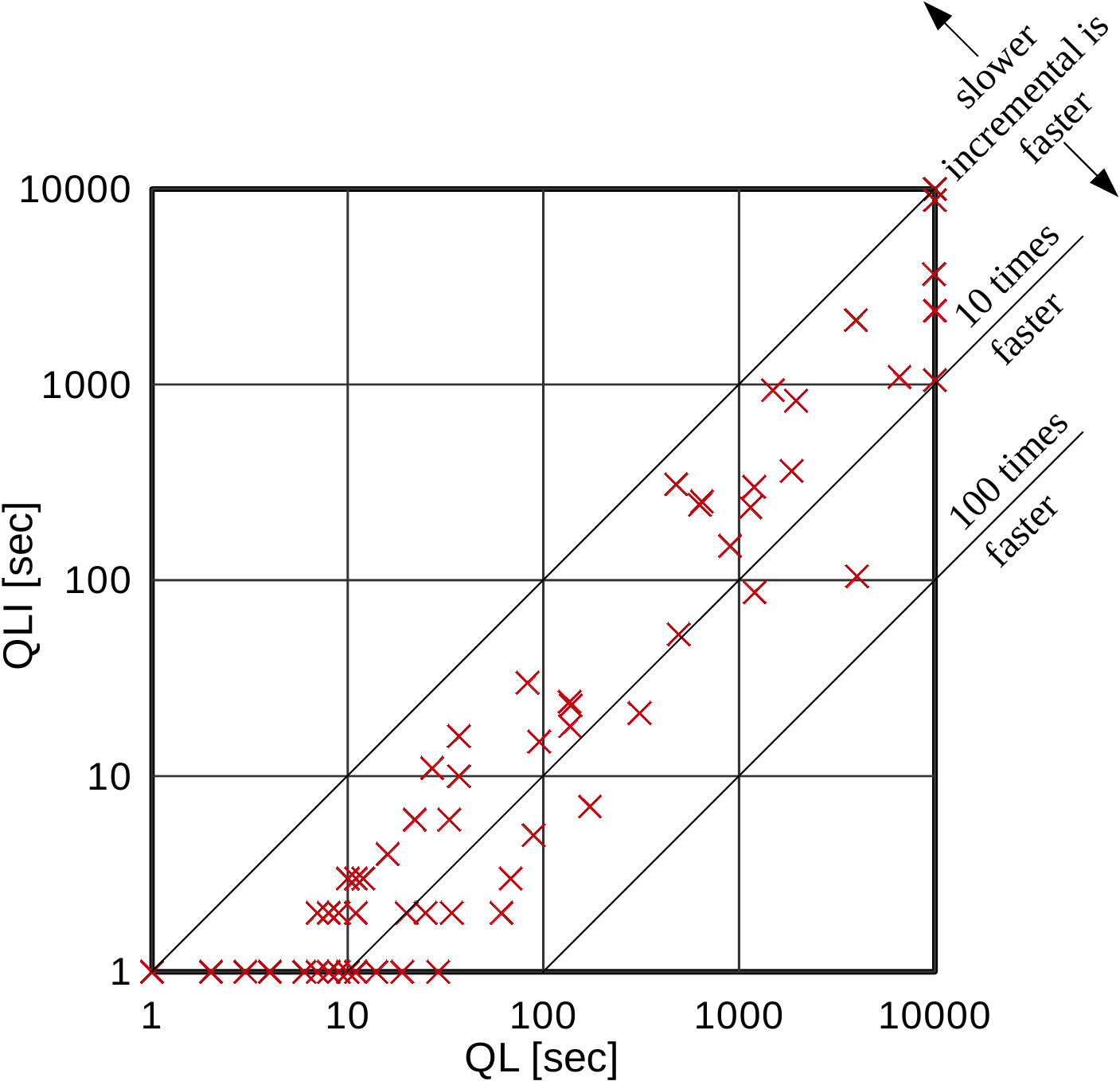}
\caption{The effect of incremental \acs{QBF} solving in circuit synthesis.}
\label{fig:extr_pre_scatter2}
\end{minipage}
\end{figure}

\mypara{Incremental \acs{QBF} solving outperforms \acs{QBF} preprocessing in our 
circuit synthesis experiments.}  Figure~\ref{fig:extr_pre_scatter1} illustrates 
the effect of \acs{QBF} preprocessing in our \acs{QBF}-based learning method for 
circuit synthesis by comparing \eqb against \eq in a scatter plot.  We see a 
negative effect for most instances.  The number of solved instances even 
decreases from $188$ to $186$ (see Table~\ref{tab:extrsolved}).  By trend, 
preprocessing is more beneficial for more complex instances.  Some of the more 
complex instances have been left out in the comparison because either \wbdd or 
\wabs failed to compute a winning region.  If we consider all $285$ instances on 
which \wpthree managed to compute a winning strategy, the number of solved 
instances actually increases from $190$ to $193$ due to \acs{QBF} preprocessing. 
 Hence, preprocessing also has its merits. On the other hand, incremental 
\acs{QBF} solving has an exclusively positive impact in our experiments.  It is 
visualized in Figure~\ref{fig:extr_pre_scatter2}, comparing \eqi against \eq.  
There is not a single instance where incremental \acs{QBF} solving increased the 
computation time.  In $3$ cases, a timeout is avoided.  When counting only the 
instances where \eq terminates successfully, the average execution time is 
reduced from $204$ to $59$ seconds, which is a speedup of factor~$3.5$.

\mypara{Using \textsc{NegLearn} reduces the memory required by \qbfcert.}  As 
already mentioned, \qbfcert can consume quite some memory.  This applies to both 
main memory as well as disk space for auxiliary files.  As a consequence, \ecert 
encounters a timeout for only two benchmark instances.  For the other instances 
that cannot be solved, the reason is in exceeding the memory limit.  When using 
\textsc{NegLearn} (Algorithm~\ref{alg:NegLearn}) in order to compute the 
negation of the winning region without introducing auxiliary variables, the size 
of the auxiliary files is reduced by up to a factor of $30$.  On the other hand, 
for more than $15$ instances, running \textsc{NegLearn} only trades a memory 
issue for a timeout.  Still, the total number of solved instances increases from 
$161$ to $171$ in our \acs{QBF} certification approach (compare \ecert with 
\ecertn in Table~\ref{tab:extrsolved}).  For our other methods, running 
\textsc{NegLearn} does not pay off, though.

\mypara{Minimizing the final solution in \acs{SAT} solver based learning yields 
moderate circuit reductions.}  When counting only the benchmark instances where 
both \esd and \esdm terminate, the average circuit size is reduced by $33\%$ 
(from $1500$ to $1000$ gates) due to the final minimization step discussed in 
Section~\ref{sec:hw:extrsat_impl}.  On the other hand, the average circuit 
computation time increases from $106$ seconds to $311$ seconds, which is almost 
a factor of $3$.  For individual benchmark instances, the cost/benefit ratio 
can be lower, though.  For example, in the case of \driver\texttt{b}8, the 
circuit size is reduced from $594$ gates to $152$ gates in only a few extra 
seconds.
 
\subsection{Discussion} \label{sec:hw:exp:concl}

\aclp{BDD} are increasingly displaced by \acs{SAT}-based methods in the formal 
verification of hardware circuits.  In synthesis, however, outperforming 
\acsp{BDD} is challenging.  

\mypara{Outperforming \acsp{BDD} in strategy computation.} For the strategy 
computation step, our \acs{SAT} solver based learning approach is competitive 
with the \acs{BDD}-based reference implementation in our experiments, but only 
when making clever use of incremental solving, unsatisfiable cores, our 
optimization for exploiting unreachable states, and our heuristic for expanding 
quantifiers.  With our parallelization, we can even solve significantly more 
benchmark instances than the \acs{BDD}-based implementation.  This is achieved 
by complementing the \acs{SAT} solver based learning approach with our 
template-based approach.  
An advantage of \acsp{BDD} is that they can handle both universal and 
existential quantification.  This also holds true for \acs{QBF} solvers.  
However, a \acs{QBF} solver only computes one satisfying assignment, while 
\acsp{BDD} eliminate the quantifiers to represent all satisfying assignments 
simultaneously.  Our \acs{QBF}-based algorithms have to compensate for that with 
more iterations.  The performance of our \acs{QBF}-based algorithms is rather 
limited compared to our \acs{SAT} solver based realizations.  This is somewhat 
surprising because the lack of universal quantification induces even more 
iterations.  However, considering that \acs{QBF} is still a rather young 
research discipline, this situation may change in the future.  A combination of 
incremental \acs{QBF} solving with preprocessing appears to be a particularly 
promising direction.
Our parallelization is on a par with the state-of-the-art synthesis tool 
\abssynthe, which is also \acs{BDD}-based but uses abstraction/refinement and 
other advanced optimizations.  Adopting optimizations like 
abstraction/refinement from \abssynthe appears to be a promising direction for 
future work.

\mypara{Outperforming \acsp{BDD} in circuit computation.} For the second 
synthesis step, where circuits implementing the computed strategies are 
constructed, our satisfiability-based methods are even more beneficial on
average over all our benchmarks.  In particular, our combination of 
interpolation with \acs{SAT} solver based learning outperforms the 
\acs{BDD}-based reference implementation by roughly one order of magnitude on 
average.  Moreover, it produces circuits that are smaller by around two orders 
of magnitude.  One reason is that the learning techniques we apply 
seem to be good at exploiting implementation freedom.  Ehlers et 
al.~\cite{EhlersKH12} showed that learning techniques can also improve the 
circuit sizes when used with \acsp{BDD}, but only at the cost of additional 
computation time.  The experiments in this article suggest that the 
combination of learning algorithms with decision procedures for satisfiability 
is more promising.  Our plain \acs{SAT} solver based learning approach is still 
significantly slower than \abssynthe.  However, our parallelization can already 
solve more instances by combining different \acs{SAT}-based methods. Moreover, 
it produces circuits that are smaller than those from \abssynthe by more than 
one order of magnitude on average in our experiments.

\mypara{Conclusion.} \acsp{BDD} are far from obsolete in the synthesis of 
reactive systems.  Yet, \acs{SAT}-based methods can be 
competitive, and even outperform \acs{BDD}-based implementations on average 
when designed carefully.  We also observed that \acs{SAT}-based methods can 
solve classes of benchmarks that are hard to deal with for \acp{BDD}.  We 
therefore believe that our novel synthesis methods form an 
important contribution to the portfolio of available approaches. 

%% file: 06relwork.tex
\section{Related Work}\label{sec:rel}

Related work on which this thesis builds has already been discussed throughout 
the document, and especially in Chapter~\ref{sec:prelim}.  This chapter 
discusses alternative approaches and points out similarities and differences.

\subsection{\acs{SAT}-Based Reactive Synthesis Approaches} \label{sec:rel:sat}

\noindent
Reactive synthesis is a broad research area, but approaches based on decision 
procedures for satisfiability are rare.

\mypara{Incremental induction.} Morgenstern et al.~\cite{MorgensternGS13} 
present a \acs{SAT} solver based synthesis algorithm for safety specifications 
that is inspired by the model checking algorithm \icthree~\cite{Bradley11} and 
its principle of incremental induction.  The basic idea is to lazily compute the 
\emph{rank} of the initial state of the specification, which is the maximum 
number of steps in which the environment can enforce to visit an unsafe state.  
If this rank is found to be finite, the specification is unrealizable.  If it is 
found to be infinite, the specification is realizable.  Hence, strictly 
speaking, the paper only presents a decision procedure for realizability.  
However, computing a winning strategy and a circuit implementing this strategy 
is also possible.  We used a reimplementation of this algorithm as a baseline 
in our experimental evaluation.  It was very fast on certain benchmark 
instances, but outperformed significantly by our new algorithms on average.

\mypara{Property-directed synthesis.} Chiang and Jiang~\cite{ChiangJ15} present 
a similar approach, which is also inspired by \icthree~\cite{Bradley11}.  While 
Morgenstern et al.~\cite{MorgensternGS13} solve the game from the perspective of 
the environment taking the unsafe states as anchor, the approach by Chiang and 
Jiang~\cite{ChiangJ15} takes the perspective of the controller to be 
synthesized.  It uses the initial states as anchor and tries to avoid ending 
up in an unsafe state.  This yields promising results for a (rather small) 
subset of the \syntcompy benchmarks.  An integration into our parallelization 
would be interesting.

\mypara{Strategy computation without preimages.}
Narodytska et al.~\cite{NarodytskaLBRW14} propose an algorithm to compute 
strategies for reachability specifications, where a set of target states needs 
to be visited at least once.  The general idea is to apply a 
counterexample-guided backtracking search in order to find a set of executions 
that is sufficient to reach the target states within some number $n$ of steps.  
This set of executions is then generalized into a winning strategy in the form 
of a tree that defines control actions based on previous inputs.  If no strategy 
is found for a particular bound $n$, then $n$ is increased.  A \acs{SAT} solver 
is used both to compute and to generalize executions.  Hence, in comparison to 
our work, this approach operates on a different specification class 
(reachability rather than safety), and it computes a winning strategy directly 
rather than deriving it from a winning region.

\mypara{Implementing strategy trees.}
E{\'{e}}n et al.~\cite{EenLNR15} complete the work discussed in the previous 
paragraph by proposing a method to compute circuits implementing the obtained 
winning strategies.  Just like one of our methods, it uses interpolation.  
However, since the strategies are represented as trees rather than relations, 
the use of interpolation is quite different compared to our work.

\mypara{\acs{QBF}-based approaches.}
Staber and Bloem~\cite{StaberB07} present a \acs{QBF}-based synthesis method for 
safety specifications.  The general principle of unrolling the transition 
relation has already been discussed along with its drawbacks in 
Section~\ref{sec:hw:qbf:sf} as a motivation for our learning-based algorithms.  
A solution for B\"uchi objectives (where some set of states needs to be visited 
infinitely often) is presented by Staber 
and Bloem~\cite{StaberB07} as well.  Alur et al.~\cite{AlurMN05} propose a 
similar solution for bounded reachability specifications (where a set of target 
states needs to be reached within at most $n$ steps).  That 
paper~\cite{AlurMN05} also proposes an optimization that uses only one copy of 
the transition relation. However, all variables are still copied for all time 
steps and the high number of quantifier alternations (linear in $n$) remains.  
In contrast, our learning-based methods use only one copy of the transition 
relation and two quantifier alternations in all \acs{QBF} solver calls (at the 
cost of a potentially higher number of solver calls).

\mypara{ALLQBF \index{ALLQBF} solving.}
Becker et al.~\cite{BeckerELM12} explain how \acs{QBF} solvers can be used to 
compute not only one but \emph{all} satisfying assignments of a \acs{QBF} in the 
form of a compact (quantifier-free) formula.  Similar to some of our 
satisfiability-based synthesis methods, query learning is used to solve this 
problem.  The paper also points out that such an ALLQBF engine can be used as a 
direct replacement of \acsp{BDD} to compute the winning region of various 
specification classes using fixpoint algorithms.  For instance, 
Algorithm~\ref{alg:SafeWin} can be realized with an ALLQBF engine in order to 
compute the winning region of a safety specification.  While our \acs{QBF}-based 
algorithm \textsc{QbfWin} (Algorithm~\ref{alg:QbfWin}) is similar in spirit, 
there are also some important differences.  We apply query learning directly to 
the specification rather than the preimage computations, which allows for better 
generalizations.  Furthermore, we extend the basic algorithm with additional 
optimizations such as our reachability optimization from 
Section~\ref{sec:hw_reach}.

\mypara{\acs{QBF} as a game.}
Synthesis can be seen as a game between two players: the system controlling the 
outputs and trying to satisfy the specification, and the environment controlling 
the inputs and trying to violate the specification.  Similarly, \acs{QBF} 
solving can also be seen as a game between two players: one player controls the 
existentially quantified variables and tries to satisfy the formula, the other 
player controls the universal variables and tries to falsify the 
formula.  This idea is followed by Janota et al.~\cite{JanotaKMC12} in the 
\acs{QBF} solver \rareqs.  Following the principle of counterexample-guided 
refinement of solution candidates, it uses two competing \acs{SAT} solvers to 
build a \acs{QBF} solver: one \acs{SAT} solver computes candidates in the form 
of assignments to existential variables, the other one refutes them 
with assignments for the universal variables.  We followed the same principle 
when traversing from our \acs{QBF}-based synthesis algorithm to \acs{SAT} solver 
based algorithms (cp.\ Algorithm~\ref{alg:QbfWin} with 
Algorithm~\ref{alg:SatWin1}).  However, we apply the idea on the level of the 
synthesis algorithm rather than for realizing individual \acs{QBF} solver calls. 
This allows for additional optimizations.  Another connection to this work is 
in coming to the same conclusion, namely that solving quantified problems with 
\acs{SAT} solvers instead of \acs{QBF} solvers can be beneficial.

\mypara{\acs{SMT}-based \index{bounded synthesis} bounded synthesis.}
Bounded synthesis~\cite{FinkbeinerS13} by Finkbeiner and Schewe has the 
objective of synthesizing a reactive system from a given 
\acf{LTL}~\cite{Pnueli77} specification.  First, the \acs{LTL} specification 
$\varphi$ is transformed into a (universal co-B\"uchi tree) automaton.  A given 
system implementation satisfies $\varphi$ if there exists a special annotation 
that maps each (automaton state, system state)-pair to a natural number.  The 
idea is now to search for such an annotation and a system implementation 
simultaneously using an \acs{SMT} solver: An upper bound on the system size is 
fixed but the system behavior is left open by using uninterpreted 
functions for the transition relation and the definition of the system outputs.  
Along with the annotations, the \acs{SMT} solver then searches for concrete 
realizations of these uninterpreted functions.  In case of unsatisfiability, 
the bound on the system size is increased until a solution is found.  Although 
this 
synthesis approach is also \acs{SAT}-based, it is quite different from the 
algorithms presented in this thesis.  The basic philosophy of enumerating 
constraints that have to be satisfied by the final solution is similar to our 
template-based approach and our reduction to \acs{EPR}, though.

\mypara{Parameterized synthesis.}  The tool \textsf{PARTY}~\cite{KhalimovJB13} 
uses \acs{SMT}-based bounded synthesis to solve the parameterized synthesis 
\index{parameterized synthesis} problem~\cite{JacobsB14}, which asks to 
synthesize systems with a parametric number of isomorphic components.  The 
approach is based on so-called \index{cutoff} \emph{cutoffs}~\cite{EmersonN03}, 
saying that the verification of parametric systems with an arbitrary number of 
isomorphic components can be reduced to the verification of systems with a fixed 
size (the cutoff size) if the specification has a certain structure.

\mypara{Controller synthesis using uninterpreted functions.}  Hofferek et 
al.~\cite{HofferekB11, HofferekGKJB13, Hofferek14} present an approach to 
synthesize controllers for aspects that are hard to engineer in concurrent 
systems.  A sequential reference implementation acts as a specification.  
Uninterpreted functions are used to abstract complex datapath elements.  
Interpolation over \acs{SMT} formulas is used as the core technology for 
computing a controller implementation.  This includes a method to compute 
multiple interpolants from a single unsatisfiability 
proof~\cite{HofferekGKJB13}.  The approach is implemented in the tool 
\textsf{Suraq}~\cite{HofferekG14}.  While there are similarities with our 
interpolation-based algorithms, we apply interpolation on the propositional 
level, we do not use abstraction using uninterpreted functions, and we compute 
one interpolant after the other.  These differences appear to be interesting 
directions for future work, though.

\subsection{Other Reactive Synthesis Approaches and Tools} \label{sec:rel:ohw}

\acsp{BDD} can be considered as the dominant data structure for symbolic 
synthesis algorithms.  However, there are also other alternatives.

\mypara{Antichains.}  Given a set of partially ordered elements, an 
\emph{antichain} \index{antichain} is a subset of elements that are all pairwise 
incomparable.  Just like \acsp{BDD}, antichains can be used as compact 
representations of large state sets: for a given partial order among states, 
an antichain represents the set of all states that are less than or equal to one 
antichain element with respect to the partial order. Besides decision procedures 
for satisfiability, antichains provide another successful alternative to 
\acsp{BDD} in synthesis~\cite{FiliotJR11,RaskinCDH07,BerwangerCWDH10}.  The 
following paragraphs describe such approaches in more detail.

\mypara{Antichains for \acs{LTL} synthesis.}  Filiot et al.~\cite{FiliotJR11} 
present a synthesis approach for \acs{LTL} specifications that uses antichains 
as data structure.  It translates the specification into a (universal co-\buchi 
word) automaton and enforces that the rejecting states of the automaton are 
visited at most $n$ times.  This effectively gives a safety game and is thus 
similar to bounded synthesis~\cite{FinkbeinerS13} as discussed earlier.  The 
approach has been implemented in the tool \textsf{Acacia+}~\cite{BohyBFJR12}.  
While the similarities to our work are small, the procedure of reducing 
\acs{LTL} specifications to safety games can be used to apply our 
\acs{SAT}-based synthesis methods also to \acs{LTL} specifications.  In fact, 
this approach was followed in the \syntcomp competition to translate \acs{LTL} 
benchmarks into safety specifications automatically~\cite{sttt_syntcomp}.

\mypara{Antichains for synthesis with imperfect information.}  In certain 
settings, the system to be synthesized may not be able to observe all internals 
of other components.  Synthesis algorithms for imperfect information address 
this issue.  Raskin et al.~\cite{RaskinCDH07} present algorithms to determine 
the realizability of such synthesis problems using antichains. Berwanger et 
al.~\cite{BerwangerCWDH10} extend this work by proposing a method to also 
extract winning strategies for parity games with imperfect information. This 
approach has been implemented in the tool 
\textsf{Alpaga}~\cite{BerwangerCWDH09}.  As an optimization, this tool uses 
\acsp{BDD} to represent antichains in such a way that efficient quantification 
is possible.

\mypara{Explicit representations.}  The tool \lily~\cite{JobstmannB06} 
synthesizes reactive systems from \acs{LTL} specifications by a serious of 
automata transformations that are based on work by Kupferman and 
Vardi~\cite{KupfermanV05}.\footnote{Similar to the antichain-based approach by 
Filiot et al.~\cite{FiliotJR11} and the bounded synthesis approach by Finkbeiner 
and Schewe~\cite{FinkbeinerS13}, the \acs{LTL} specification is translated into 
a universal universal co-\buchi tree automaton first.  Following an approach by 
Kupferman and Vardi~\cite{KupfermanV05}, this automaton is then translated into 
an alternating weak tree automaton and further on to a nondeterministic \buchi 
tree automaton.}  A witness to the non-emptiness of the final (nondeterministic 
\buchi tree) automaton constitutes an implementation of the original 
specification.   Jobstmann and Bloem~\cite{JobstmannB06} present a  multitude of 
optimizations to improve the performance of this approach.  \lily implements 
them on top of \textsf{Wring}~\cite{SomenziB00}.  \lily does not represent 
automata symbolically but operates on explicit representations.  The 
similarities to our \acs{SAT}-based synthesis algorithms are thus rather small.

\mypara{\acs{BDD}-based tools}.  
We only give a brief and incomplete overview of \acs{BDD}-based synthesis tools 
and approaches. \textsf{Anzu}~\cite{JobstmannGWB07} is a \acs{BDD}-based 
synthesis tool for \acs{GR(1)} 
specifications~\cite{BloemJPPS12}.  It has later been reimplemented in 
\ratsy~\cite{BloemCGHKRSS10}.  The same synthesis algorithm is also implemented 
in the \acs{BDD}-based tools
\href{https://github.com/LTLMoP/slugs}{\textsf{slugs}}, 
\href{http://slivingston.github.io/gr1c/}{\textsf{gr1c}}, and 
\href{http://es.fbk.eu/technologies/nugat-game-solver/}{\textsf{NuGAT}}, which is a game solver built on top of the model checker 
\textsf{NuSMV}~\cite{CimattiCGGPRST02}.  \textsf{Unbeast}~\cite{Ehlers12} is a 
tool for synthesis from \acs{LTL} specifications that also builds on the 
principle of bounded synthesis~\cite{FinkbeinerS13}.  The reduction from 
\acs{LTL} to safety games is similar to that by Filiot et al.~\cite{FiliotJR11} 
but the resulting safety game is solved using \acsp{BDD} instead of antichains.
Except for our own submission \demiurge, all tools that competed in the 
\syntcompy competition~\cite{sttt_syntcomp} are 
\acs{BDD}-based.  This includes \abssynthe~\cite{BrenguierPRS14}, which has been 
used as a baseline for comparison in our experimental results, \textsf{Basil} by 
R\"udiger Ehlers, \textsf{realizer} by Leander Tentrup, and the 
\href{https://github.com/adamwalker/syntcomp/}{\textsf{Simple \acs{BDD} Solver}}
by Leonid Ryzhyk and Adam Walker.

%% file: 07conclusion.tex
\section{Conclusions}
\label{sec:conc}

Chapter~\ref{sec:hw:win} and~\ref{sec:hw:circ} already discussed the strengths 
and weaknesses of the different algorithms and optimizations while they were 
presented.  Moreover, Section~\ref{sec:hw:exp:concl} summarized the most 
important conclusions that can be drawn from our experiments.  In this section, 
we will not repeat this discussion but rather focus on the most important 
conclusions from a high-level point of view.  This will also form the basis to 
our suggestions for future work. 

\mypara{Exploiting solver features.}  In contrast to verification, decision 
procedures that can only give a yes/no answer are of no use in synthesis.  
Fortunately, many decision procedures for satisfiability are based on the search 
for satisfying structures.  These artifacts can in turn be used to build an 
implementation for a given specification in synthesis.  Modern \acs{SAT}-, 
\acs{QBF}- and \acs{SMT} solvers offer additional features that can be exploited 
in synthesis as well. This includes the computation of unsatisfiable cores, 
which can be used to generalize discovered facts.  Another example is 
incremental solving, which can be used to answer sequences of similar queries 
much more efficiently. Our synthesis algorithms utilize such solver features by 
design, which turned out to be crucial for being competitive with \acsp{BDD}.

\mypara{Counterexample-guided refinement.}   The algorithmic principle of 
refining solution candidates iteratively based on counterexamples turned out to 
be a good match with decision procedures for satisfiability. We used this 
concept in two flavors: query learning and \acf{CEGIS}.  Our extension of 
\acs{CEGIS} outperformed \acs{QBF} solving in our template-based approach. 
Overall, query learning combined with \acs{SAT} solving proved to be our best 
approach in our synthesis experiments.  This applies both to the first step of 
computing a winning strategy as well as to the second step of constructing a 
circuit.  In the second step, query learning also produced circuits that were 
smaller by more than one order of magnitude on average compared to other 
techniques such as interpolation, \acs{QBF} certification, or the 
\acs{BDD}-based cofactor approach.  This suggests that query learning performs 
well at exploiting available implementation freedom.  

\mypara{Handling quantifiers.}  The game-based approach to synthesis 
inherently involves dealing with both universal and existential quantifiers.  
The support for both quantifiers is also among the reasons for the sustained 
success of \acsp{BDD} in reactive synthesis. When switching from \acsp{BDD} to 
decision procedures for satisfiability, one could thus expect that \acs{QBF} 
solvers are the most suitable choice.  Yet, in our experiments, our algorithms 
using plain \acs{SAT} solving outperformed the \acs{QBF}-based algorithms 
significantly, even though (often far) more solver calls are necessary to 
compensate for the lack of universal quantifiers.  Our heuristic for quantifier 
expansion reduces this amount of iterations at the cost of larger formulas for 
the \acs{SAT} solver, which gives a speedup of one more order of magnitude. 
 This suggests that the current state in \acs{QBF} solving is still lacking 
behind its potential, at least for the specific kinds of \acs{QBF} problems we 
encounter in our synthesis algorithms.  However, considering that \acs{QBF} is 
still a rather young research discipline compared to \acs{SAT}, this situation 
may change in the future.  

\mypara{More expressive logics.}  The scalability of our approach based on 
reduction to \acs{EPR}, which is a more expressive logic, is even worse than 
when using \acs{QBF} in our experiments.  Together with the statement from the 
previous paragraph, this suggests that breaking the synthesis problem into 
simple solver queries in a lean logic is a better strategy than delegating 
bigger chunks of the problem to the underlying solver.

\mypara{Parallelizability.}  Since our satisfiability-based methods for reactive 
synthesis mostly break the synthesis problem down to many small solver queries 
that do not crucially depend on each other, they are also well suited for 
fine-grained application-level parallelization.  This stands in contrast to 
symbolic algorithms realized with \acsp{BDD}, which are often intrinsically hard 
to parallelize~\cite{EzekielLS11}.  We presented parallelizations that do not 
only exploit hardware parallelism but also combine different (variants of) 
algorithms in different threads.  This way, we achieved average speedups of 
around one order of magnitude with only three threads.

\mypara{Outperforming BDDs.}  Due to our heuristics and optimizations, careful 
utilization of solver features, and our parallelization, our 
satisfiability-based methods managed to outperform a \acs{BDD}-based synthesis 
tool by more than one order of magnitude regarding execution time, and even two 
orders of magnitude regarding circuit size on average in our experiments.  Our 
parallelization is even competitive with \abssynthe, a highly optimized 
state-of-the-art tool implementing advanced optimizations such as 
abstraction/refinement.  
These results confirm that decision procedures for the satisfiability of 
formulas can indeed be used to build scalable synthesis algorithms.

\mypara{There is no silver bullet.}  Despite the excellent performance results 
we achieved on average in our experiments, we observed that different techniques 
perform well on different classes of benchmarks.  We thus see our main 
contribution in extending the portfolio of available synthesis approaches with 
new algorithms that complement existing techniques. 

\mypara{Safety specifications.}  Our reactive synthesis algorithms 
operate on safety specifications.  Many of the benchmarks used in our 
experimental evaluation originally contained liveness properties that have been 
translated to safety specifications by imposing fixed bounds on the reaction 
time.  While choosing low bounds for the reaction time (such that the 
specification is still realizable) can have the advantage of producing systems 
that react faster, the translation may have a negative performance impact 
compared to handling liveness properties directly in the synthesis algorithm.

\section{Future Work} \label{sec:future}

Our suggestions for future work in satisfiability-based reactive synthesis 
range from improvements in the underlying reasoning engines up to extensions for 
different classes of specifications.

\mypara{\acs{QBF} preprocessing.}  
While our \acs{QBF}-based synthesis algorithms were not among the best solutions 
in our experiments, we still observed that using incremental \acs{QBF} solving 
and \acs{QBF} preprocessing both can have a very positive performance impact.  
Researching ways to combine these techniques therefore seems to be a 
particularly promising direction to support the success of \acs{QBF} in 
synthesis.  
Furthermore, in our circuit computation method based on \acs{QBF} certification, 
preprocessing could not be applied because existing tools only preserve 
satisfying assignments for existentially quantified variables~\cite{SeidlK14}, 
but are in general not certificate preserving.  Research on such 
certificate-preserving preprocessing solutions could thus boost the performance 
of \acs{QBF} certification (not only) in synthesis.

\mypara{Solver parameters.}
We used all solvers with default parameters in our experiments.  It is 
not unlikely that a solid speedup can be achieved by tuning solver parameters to 
the specific kinds of decision problems encountered in our algorithms.  For 
instance, our algorithms based on \acs{SAT} solving usually 
make huge amounts of rather simple queries.  Yet, the default parameters of the 
\acs{SAT} solvers may be tuned to more complex instances from \acs{SAT} 
competitions.

\mypara{Other logics.}
Our approach based on reduction to \acs{EPR} did not perform well in our 
experiments.  For this reason, we did not explore the alternative of using 
\acs{DQBF} instead.  Yet, recent progress~\cite{FroehlichKB12, FrohlichKBV14, 
BalabanovCJ14, GitinaWRSSB15} in theory and tools for \acs{DQBF} makes this 
approach interesting as well.

\mypara{Computing multiple interpolants.}  Some of our methods to compute 
circuits from given strategies are based on interpolation.  As mentioned in 
Section~\ref{sec:rel:sat}, it would be interesting to also implement the 
approach by Hofferek et al.~\cite{HofferekGKJB13} for computing multiple 
interpolants from a single proof.

\mypara{Reachability optimization.}  Our reachability optimization is rather 
simplistic and still has a very positive performance impact.  Other variants may 
thus yield even bigger speedups.  In particular, our reachability optimization 
avoids the explicit computation of an over-approximation of the reachable 
states.  Exploring this option based on existing work in 
verification~\cite{MoonKSS99} can be worthwhile.

\mypara{Parallelization.}  
Our parallelized synthesis method demonstrates that parallelization is easily 
possible and beneficial for our \acs{SAT}-based synthesis algorithms.  
However, it is in no way optimal.  First, there is a plethora of possibilities 
to combine different algorithms, optimizations and solver configurations in 
different threads.  Second, there are numerous ways for exchanging information 
between threads.  A thorough exploration of possibilities is still to be done.

\mypara{\aiger as symbolic data structure.} Another alternative to \acsp{BDD} is 
to use \aiger circuits as a data structure for formulas.  Boolean connectives 
($\wedge, \vee, \rightarrow, \ldots$) are easy to realize by 
adding gates accordingly.  Universal and existential quantification can be 
realized by expansion.  Circuit simplification techniques as implemented in 
\abc~\cite{BraytonM10} can be applied to reduce the size of the symbolic 
representation after applying operations (similar to variable reordering in 
\acsp{BDD}). A \acs{SAT} solver can be used for equivalence or inclusion checks. 
In contrast to \acsp{BDD}, such a symbolic representation is not canonical.  It 
may thus be more compact in cases where \acsp{BDD} explode in size (see 
Section~\ref{sec:prelim:bdds}).

\mypara{Specification preprocessing.}  Inspired by the formidable performance 
impact of preprocessing in \acs{QBF} solving, research on preprocessing 
techniques for specifications in synthesis can be another angle from which the 
scalability issue can be tackled.  For specifications defined as \aiger 
circuits, one first idea would be to develop heuristics for identifying 
auxiliary variables (outputs of AND-inverter gates defining the transition 
relation) that can be controlled fully and independently by either the system or 
the environment.   As a simple example, some auxiliary variable $t$ may be 
defined as a function over some vector $\overline{i}_t \subseteq \overline{i}$ 
of uncontrollable inputs, and the inputs $\overline{i}_t$ are used nowhere else. 
Such auxiliary variables can be replaced by new controllable or uncontrollable 
inputs, and their respective cone of influence can be removed.  Another idea is 
to detect monotonic dependencies of the error output on inputs or latches and to 
replace them with constants. Existing techniques for circuit simplification can 
also be applied, of course.

\mypara{Other specifications.}  Our satisfiability-based synthesis algorithms 
operate on safety specifications.  A natural point for 
future work is thus to extend them to other types of specifications.  
Interesting cases would include reachability specifications (some states must be 
visited at least once), B\"uchi specifications (some states must be visited 
infinitely often), or even \acs{GR(1)}~\cite{BloemJPPS12}.  Our methods to 
compute a circuit from a given strategy are rather agnostic against the 
specification from which the strategy has been constructed.  Here, future work 
would mostly be in working out an efficient implementation.  For the computation 
of strategies, the situation is different though.  Learning-based algorithms are 
not difficult to define for other specification formats in principle.  If and 
how they can be applied \emph{efficiently} remains to be explored, though.